\documentclass[12pt,reqno,a4paper]{amsart}
\usepackage[dvips]{color}
\usepackage{amsmath}
\usepackage{amsxtra}
\usepackage{amscd}
\usepackage{amsthm}
\usepackage{amsfonts}
\usepackage{amssymb}
\usepackage{nccmath}
\usepackage{eucal}
\usepackage{epsfig}
\usepackage{graphics}
\usepackage{dirtytalk}
\usepackage{ulem}
\usepackage[breakable]{tcolorbox}
\usepackage{ifthen}
\usepackage{varioref}
\usepackage{mathrsfs}

\usepackage{cancel}

\usepackage{xcolor}
\definecolor{superlightred}{HTML}{F5F5F5}

\allowdisplaybreaks %This is important for splitting an aligned equation in two pages. 

\usepackage{tikz}
\usetikzlibrary{decorations.pathreplacing,
	calligraphy,% must be loaded after decorations.pathreplacing
	matrix}

\usepackage{bbold}
\usepackage{mlmodern}

\usepackage{enumerate}
\usepackage{hhline}
\usepackage{multirow}

\usepackage[nosort]{cite}
\usepackage{xspace}
\usepackage{setspace}

\usepackage{titlesec}
\titleformat{\section}[block]{\large\scshape\centering}{§\,\thesection}{1em}{}[\HRule{3pt}] 
\titleformat{\subsection}[block]{\normalsize\bfseries}{\thesubsection}{1em}{}
\titleformat{\subsubsection}[block]{\normalsize\bfseries}{\thesubsubsection}{1em}{}

\usepackage[framemethod=TikZ]{mdframed}

\newcommand{\beq}{\begin{equation}}
	\newcommand{\eeq}{\end{equation}}
\newcommand{\beqa}{\begin{eqnarray}}
	\newcommand{\eeqa}{\end{eqnarray}}

\usepackage[T1]{fontenc}

\newcommand*\circled[1]{\tikz[baseline=(char.base)]{
		\node[shape=circle,draw,inner sep=2pt] (char) {#1};}}

\newcommand{\HRule}[1]{\rule{\linewidth}{#1}} 	% Horizontal rule
\newcommand{\comp}{\circ}

\newcommand{\tens}{\otimes}

\newcommand{\ddsum}{\bigoplus}

\newcommand{\cals}[1]{\mathcal{#1}}
\newcommand{\fraks}[1]{\mathfrak{#1}}

\newcommand{\normord}[1]{:\mathrel{#1}:}
\newcommand{\bb}[1]{\mathbb{#1}}

\newcommand{\yMTo}[1]{\mathbin{\,\tikz[baseline] \draw[-stealth,line width=.4pt] (#1,0.4ex) -- (-0.8ex,0.4ex);}}
\newcommand{\Plim}{%
	\mathchoice
	{\lim_{\yMTo{2.7ex}}}% \displaystyle
	{\lim_{\yMTo{2.5ex}}}% \textstyle
	{\lim_{\yMTo{2.0ex}}}% \scriptstyle
	{\lim_{\yMTo{2.0ex}}}% \scriptscriptstyle
}
\newcommand{\inverselimm}[1]{\mathop{\Plim}_{#1}}

\numberwithin{equation}{section}
\newtheorem{thm}{Theorem}[section]
\newtheorem{prop}[thm]{Proposition}
\newtheorem{lem}[thm]{Lemma}

\newtheorem{cor}[thm]{Corollary}

\newtheorem{dfn}[thm]{Definition}
\newtheorem{dfn-prp}[thm]{Definition-Proposition}

\newtheorem{nota}[thm]{Notation}
\newtheorem{exa}[thm]{Example}
\newmdtheoremenv[backgroundcolor=superlightred]{assump}[thm]{Assumption}

\newcommand{\supermac}{\overline{SP}}

\newcommand{\dualmap}{
\widetilde{\Psi}^{(q,\xi)}_{\lambda}
}

\theoremstyle{definition}
\newtheorem{rem}[thm]{Remark}

\usepackage{xpatch}
\xpatchcmd{\proof}{\itshape}{\normalfont\proofnamefont}{}{}
\newcommand{\proofnamefont}{}
\renewcommand{\proofnamefont}{\bfseries}

\usepackage{ytableau}
\ytableausetup{mathmode, boxframe=normal, boxsize=2.6em,centertableaux,notabloids}

\begin{document}
	
	\baselineskip = 18pt 
	
	\begin{titlepage}
		
		\bigskip
		\hfill\vbox{\baselineskip12pt
			\hbox{}
		}
		\bigskip
		\begin{center}
			\Large{ \scshape
				\HRule{3pt}
				Quantum Corner VOA and the Super Macdonald Polynomials 
				%; \\ Intertwiner, $R$ matrix and q-KZ equation 
				\HRule{3pt}
			}
		\end{center}
		\bigskip
		\bigskip

		\begin{center}
			\large 
			Panupong Cheewaphutthisakun$^{a}$\footnote{panupong.cheewaphutthisakun@gmail.com},
			Jun'ichi Shiraishi$^{b}$\footnote{shiraish@ms.u-tokyo.ac.jp},
%			and
			Keng Wiboonton$^{a}$\footnote{keng.w@chula.ac.th}
			\\
			\bigskip
			\bigskip
			$^a${\small {\it Department of Mathematics and Computer Science, Faculty of Science, Chulalongkorn University, 
					\\
					Bangkok, 10330, Thailand}}\\
			$^b${\small {\it Graduate School of Mathematical Sciences, University of Tokyo, Komaba, Tokyo 153-8914, Japan}} \\
		\end{center}
		\bigskip
		\bigskip
		
		\begin{flushright}
		\textit{Dedicated to the memory of Prof. Masatoshi Noumi }
		\end{flushright}
		{\vskip 6cm}
		{\small
			\begin{quote}
				\noindent {\textbf{\textit{Abstract}}.}
%				It is well-known that there is an intimate relation between conformal algebras and symmetric polynomials, and this relation plays an important role for the construction of quantum $W_N$ algebra. 
				In this paper, we establish a relation between the quantum corner VOA  
				$q\widetilde{Y}_{L,0,N}[\Psi]$, which can be regarded as a generalization of quantum $W_N$ algebra, 
				and Sergeev-Veselov super Macdonald polynomials. We demonstrate precisely that, under a specific map, 
				the correlation functions of the currents of $q\widetilde{Y}_{L,0,N}[\Psi]$, coincide with
				the Sergeev-Veselov super Macdonald polynomials. 
		\end{quote}}

	\end{titlepage}
	
%	 \tableofcontents
	% \newpage
	
\section{Introduction}
\label{intro}

Representation theory of infinite-dimensional algebras is fundamental in contemporary physics research. A significant example is conformal field theory, where the representation theory of its symmetry algebra is essential for understanding the theory. The symmetry algebra of conformal field theory, containing fields of higher spin $2, \dots, N$, is known as the $W_N$ algebra \cite{BPZ}  \cite{BS1993} \cite{Zamo}. Consequently, the $W_N$ algebra and its representations are of significant importance in contemporary physics research.

Motivated by the quantum deformation of the universal enveloping algebra $U_q(\fraks{g})$, independently introduced by Drinfeld \cite{DR} and Jimbo \cite{Jim-original}, it is natural to ask whether the $W_N$ algebra admits a quantum deformation. The quantum deformation of the $W_N$ algebra (referred to as the 
quantum $W_N$ algebra) was constructed in \cite{AKOS-950} \cite{SKAO95}, by requiring that its singular vectors
be described by Macdonald polynomials. This requirement was motivated by the earlier result that the singular vectors of the $W_N$ algebra are described by Jack polynomials \cite{awata-excited} \cite{mimachi-yamada-1995-singular}, and the Jack polynomials possess a well-known quantum deformation, the Macdonald polynomials \cite{macbook-1998} (see also \cite{noumi-mac}). Therefore, it was reasonable to expect 
that the singular vectors of the quantum $W_N$ algebra are described by the Macdonald polynomials.

In terms of generators and relations, the quantum $W_N$ algebra is generated by the currents $\{\widetilde{T}_i(z)\,|\, i =1,\dots,N\}$ subject to the relations, called quadratic relations, which take the following form schematically:
\begin{align}
	&f_{r,m}\left(
	\frac{w}{z}
	\right)
	\widetilde{T}_r(z)\widetilde{T}_m(w) 
	- f_{m,r}\left(\frac{z}{w}\right)
	\widetilde{T}_m(w)\widetilde{T}_r(z)
	\notag \\
	&= 
	\sum_{k = 1}^{r}
	\biggl\{
	(\cdots)
	\delta\left(q_3^k\frac{w}{z}\right)
	\widetilde{T}_{r-k}(q_3^{-k}z)
	\widetilde{T}_{m+k}(q_3^kw)
	- 
	(\cdots)
	\delta\left(q_3^{r-m-k}\frac{w}{z}\right)
	\widetilde{T}_{r-k}(z)
	\widetilde{T}_{m+k}(w)
	\biggr\}.
	\label{eqn11-1039-1apr}
\end{align}
Here we adopt the convention that $\widetilde{T}_i(z) = 0$ for $i > N$. 
The explicit expression of $f_{r,m}(z)$ and $\widetilde{T}_r(z)$ can be obtained by substituting $\vec{c} = (3^N)$ into equations \eqref{eqn337-2055-30mar} and \eqref{eqn314-1625} in this paper. 
In fact, the equation \eqref{eqn314-1625} is a direct consequence of Miura transformation (\textbf{Definition \ref{def38-2133-12jan}}), which describes how to express the currents $\widetilde{T}_r(z)$ in terms of the vertex operators $\widetilde{\Lambda}^{(3^N)}_i(z) \, (i = 1,\dots,N)$. Note that the 
vertex operators $\widetilde{\Lambda}^{(3^N)}_i(z) \, (i = 1,\dots,N)$ are expressed in terms of free bosons, and this expression can be deduced from the $N$-fold tensor product of horizontal Fock representations of the quantum toroidal algebra of type $\fraks{gl}_1$ (a.k.a. Ding-Iohara-Miki algebra) \cite{Awata-note}  \cite{BS}  \cite{DI}   \cite{FHH} \cite{Sch} (see also \cite{	AFS, AKMM17 , AKMM17-2 , AKMM18 , BFM17, BFM17-2 ,BJ, CK,CK2,Ma,Zen}). This corresponds to the fact that 
quantum toroidal algebra of type $\fraks{gl}_1$ can be regarded as the quantum $W_{1 + \infty}$ algebra \cite{Miki}. Recall that the horizontal Fock representation of the quantum toroidal $\fraks{gl}_1$ algebra depends on the choice of the \say{color} parameters $q_1,q_2,q_3$. By using this fact, Bershtein, Feigin, and Merzon studied the tensor product representations of the horizontal Fock modules with mixed choice of the color parameters \cite{Misha}, thereby introducing extensions of the quantum $W_N$ algebra. This technique is deeply investigated in the context of (quantum) corner VOA via the Miura transformation \cite{HMNW} \cite{PR18}. As for the details, see explanations in the forthcoming paragraphs. 

%It is important to note that the quadratic relations can be deduced from the Miura transformation of the quantum $W_N$ algebra, which tells us how to express the currents $\{\widetilde{T}_i\}$ in terms of the vertex operators $\widetilde{\Lambda}_i(z)$. 

Similar to the $W_N$ algebra, the quantum $W_N$ algebra is crucial in physics, particularly in the context of the 
five-dimensional version \cite{AY} \cite{AY2} \cite{Taki} of 
AGT correspondence \cite{AGT} \cite{Wyl}. This correspondence establishes a relation between the chiral conformal blocks of quantum $W_N$ algebra and the $K$-theoretic Nekrasov partition function \cite{nekrasov} of the $SU(N)$ gauge theory on $\bb{R}^4 \times S^1$. 

The relation between the quantum $W_N$ algebra and Macdonald polynomials is further illustrated by the computation of the correlation function of the currents $\widetilde{T}_{1}(z)$ of the quantum $W_N$ algebra \cite{FHSSY}. Specifically, it has been shown that the correlation function of these currents, under a certain map, equals $\psi_{T}(q,t)$, a coefficient appearing in the combinatorial expansion formula of the Macdonald polynomials. Despite minor differences, the result of \cite{FHSSY} is essentially equivalent to the case $(N,0)$ of \textbf{Lemma \ref{lemm44-1106-29jan}} in this paper.

Recently, a four-parameter family of vertex operator algebras, denoted $\widetilde{Y}_{L,M,N}[\Psi] \,\, (L,M,N \in \bb{Z}^{\geq 0}, \Psi \in \bb{C})$ and called corner vertex operator algebras (corner VOAs), was constructed by Gaiotto and Rap\v{c}\'{a}k in \cite{GR17} by considering the algebra arising from the brane configuration of orthogonal D3 branes. Following this work, Procházka and Rapčák discovered various equivalent constructions of corner VOAs \cite{PR17} \cite{PR18}. However, the construction most relevant to our work is the Miura transformation \cite{PR18}. 

The corner VOA $\widetilde{Y}_{L,M,N}[\Psi]$ generalizes the $W_N$ algebra, as $\widetilde{Y}_{0,0,N} = W_N$. 
Consequently, it is natural to apply the quantum deformation program to the corner VOA $\widetilde{Y}_{L,M,N}[\Psi]$. This was achieved by Harada, Matsuo, Noshita, and Watanabe,
in \cite{HMNW}, where the Miura transformation of the quantum corner VOA $q\widetilde{Y}_{L,M,N}[\Psi]$ was presented. Furthermore, the authors of \cite{HMNW} explained how to obtain the free boson representation of the vertex operators $\widetilde{\Lambda}_i^{(3^N1^L2^M)}(z)$ of $q\widetilde{Y}_{L,M,N}[\Psi]$
from the tensor product of horizontal Fock representations of the quantum toroidal $\fraks{gl}_1$ algebra.
Similar to the quantum $W_N$ algebra, the quantum corner VOA $q\widetilde{Y}_{L,M,N}[\Psi]$ is determined by quadratic relations (see \textbf{Proposition \ref{prop39-1328-8ap}}). The explicit form of these relations was initially conjectured in \cite{HMNW} and subsequently rigorously proved in \cite{pc2406}.

In this paper, motivated by the relation between the quantum $W_N$ algebra and the Macdonald polynomials, we propose a relation between quantum corner VOA $q\widetilde{Y}_{L,0,N}[\Psi]$ and super Macdonald polynomials introduced by Sergeev and Veselov in \cite{sv2007-supermac}. The precise formulation of this result is given in \textbf{Theorem \ref{thm42-main-1022}}, which constitutes the main theorem of this paper.

\subsubsection*{Organization of material}

This paper is structured as follows: Section \ref{sec2} provides a comprehensive review of fundamental concepts in the theory of symmetric polynomials that are essential for our subsequent analysis. We begin in subsection \ref{subsec21-1940-4a} by revisiting the definitions of partitions, reverse semi-standard Young tableaus, and reverse semi-standard Young bitableaus. Following this, subsection \ref{subsec22-1941-4apr} reviews the definition and construction of rings of symmetric polynomials and symmetric functions. Finally, subsections \ref{subsec23-1942-4apr} and \ref{subsec24-1943-4apr} are dedicated to reviewing the definitions and key properties of Macdonald polynomials and super Macdonald polynomials, which are central to this paper. 

In section \ref{sec3-1945-4apr}, we delve into the algebraic structures underlying our work and define the quantum corner VOA. Subsection \ref{subsec31-1953-4apr} reviews the definition and fundamental properties of the quantum toroidal algebra of type $\fraks{gl}_1$ and its horizontal Fock representation, which are necessary for constructing the vertex operators used in the Miura transformation. Subsequently, subsection \ref{subsec32-1955-4apr} presents the definition of the quantum corner VOA via the Miura transformation.

The main result of this paper is presented as \textbf{Theorem \ref{thm42-main-1022}} in section \ref{sec4-1950-4apr}. We further demonstrate in \textbf{Lemma \ref{lemm43-1156-22jan}} that the correlation function of currents can be expressed as a summation over the set of reverse semi-standard Young bitableaus (reverse SSYBTs). The detailed proof of this lemma is provided in appendix \ref{appA-1155-22jan}. By utilizing \textbf{Lemma \ref{lemm43-1156-22jan}} in conjunction with the combinatorial formula for super Macdonald polynomials explained in section \ref{sec2}, we establish that proving \textbf{Theorem \ref{thm42-main-1022}} is reduced to proving \textbf{Lemma \ref{lemm44-1106-29jan}}. A comprehensive proof of \textbf{Lemma \ref{lemm44-1106-29jan}} necessitates considering three distinct cases: (1) $(N,0)$, (2) $(N,M)$, (3) $(0,M)$, where $N, M \in \bb{Z}^{>0}$. These cases are treated individually in sections \ref{sec5-1041-7apr}, \ref{sec6-2226-13mar}, and \ref{sec7-1042-7apr}, respectively.

Section \ref{sec5-1041-7apr} is dedicated to the proof of \textbf{Lemma \ref{lemm44-1106-29jan}} for the $(N,0)$ case, which proceeds by induction on $N$. The inductive step relies on proving \textbf{Lemma \ref{lemm54-1412-2feb}}, whose detailed proof, while straightforward, is lengthy and thus relegated to appendix \ref{secappB-1410-2feb}. We note that a statement essentially equivalent to the case $(N,0)$ of \textbf{Lemma \ref{lemm44-1106-29jan}} has been previously proved in  \cite{FHSSY} using a different approach. However, given the inductive nature of our complete proof for the $(N,M)$ case, presenting our alternative proof for the simpler $(N,0)$ case in section \ref{sec5-1041-7apr} helps to illustrate the key ideas of the inductive approach.

In section \ref{sec6-2226-13mar}, we provide the proof of \textbf{Lemma \ref{lemm44-1106-29jan}} for the $(N,M)$ case, employing induction on $M \in \bb{Z}^{> 0}$. It is important to note that the basis step, $M = 1$, relies on the result established for the $(N,0)$ case. This is the reason why we need to prove the case $(N,0)$ before proving the case $(N,M)$. The inductive step in this section requires proving \textbf{Lemma \ref{lem68-1212-13mar}}, the detailed proof of which is provided in appendix \ref{appC-1214-13mar}.

Finally, section \ref{sec7-1042-7apr} presents the proof of \textbf{Lemma \ref{lemm44-1106-29jan}} for the $(0,M)$ case, where $M \in \bb{Z}^{> 0}$. For this case, induction is not required. By directly applying the result obtained for the $(N,M)$ case, the statement is readily proved.

\subsubsection*{Conventions}
Throughout this paper, we adopt the following conventions regarding summations and products:

\begin{enumerate}[(1)]
\item For $m < n$, the sum from $n$ to $m$ is defined as zero:
\begin{align*}
\displaystyle \sum_{i = n}^{m}(\text{any expression}) = 0. 
\end{align*}
\item For $m < n$, the product from $n$ to $m$ is defined as one:
\begin{align*}
\displaystyle \prod_{i = n}^{m}(\text{any expression}) = 1. 
\end{align*}
\end{enumerate}

Throughout this paper, we will encounter summations and products of the form 
\begin{align*}
\underbrace{	\sum_{a \in A}			}_{
\substack{
(1) \,\, \dots\dots
\\
\vdots
\\
(n) \,\, \dots\dots
}
}
\hspace{1cm}
\text{
and 
}
\hspace{1cm}
\underbrace{	\prod_{a \in A}			}_{
	\substack{
	(1) \,\, \dots\dots
	\\
	\vdots
	\\
	(n) \,\, \dots\dots
	}
}
\end{align*}
This notation means that the summation or product is taken over the elements $a$ of a set $A$ that satisfy all of the conditions written in $(1),(2),\dots,(n)$. 

Also, in this paper, the evaluation operation, denoted by $\, \bigg| \,$, is consistently placed to the left of the expression it acts upon. For example, for $a_1,a_2,a_3 \in \bb{C}$, we have:
\begin{align*}
\bigg|_{
	\substack{
		x_1 = a_1, \\
		x_2 = a_2,\\
		x_3 = a_3\\
	}
}
\left(
x_1^2 + x_2^3 + x_2x_3
\right)
= 
a_1^2 + a_2^3 + a_2a_3 \in \bb{C}. 
\end{align*}

\subsubsection*{Acknowledgement}
This research project is supported by the Second Century Fund (C2F), Chulalongkorn University. Our work is supported in part by Grants-in-Aid for Scientific Research (Kakenhi); 24K06753 (J.S.), 21K03180 (J.S.).

\section{Super Macdonald Polynomials}
\label{sec2}

In this section, we will review the definitions of Macdonald polynomials and super Macdonald polynomials, as well as the propositions that will be used in later sections. The main reference for Macdonald polynomials is the book \cite{macbook-1998}, while the main reference for the super Macdonald polynomials is the paper \cite{sv2007-supermac}. We will begin by reviewing the definitions of partitions, (reverse) semi-standard Young tableaus, and (reverse) semi-standard Young bitableaus

\subsection{Partitions, (reverse) semi-standard Young tableaus, (reverse) semi-standard Young bitableaus}
\label{subsec21-1940-4a}

\begin{dfn}
A sequence $\lambda = (\lambda_1,\lambda_2,\dots)$ of nonnegative integers is called a partition if the following conditions are satisfied:
\begin{enumerate}[(1)]
\item For all (except finitely many) positive integers $k$, we have $\lambda_k = 0$. 
\item $\lambda_1 \geq \lambda_2 \geq \cdots$.  
\end{enumerate}
The set of all partitions is denoted by $\operatorname{Par}$.

If $\lambda = (\lambda_1,\lambda_2,\dots)$ is a partition such that $|\lambda| := \sum_{k = 1}^{\infty}\lambda_k = n$, then $\lambda$ is said to be a partition of $n$. The set of all partitions of $n$ is denoted by $\operatorname{Par}(n)$. 
\end{dfn}

It is clear that $\operatorname{Par} = \bigsqcup_{n = 0}^{\infty}\operatorname{Par}(n)$. 

\begin{dfn}
Let $\lambda, \mu \in \operatorname{Par}$. We say that $\mu \subseteq \lambda$ if for any $i \in \bb{Z}^{\geq 1}$, $\mu_i \leq \lambda_i$. 
\end{dfn}

\begin{dfn}
Let $\lambda, \mu \in \operatorname{Par}$. We say that $\mu \leq_{d.o.} \lambda$ if 
\begin{enumerate}[(1)]
\item $|\lambda| = |\mu|$.
\item For any $i \in \bb{Z}^{\geq 1}$, $\mu_1 + \cdots + \mu_i \leq \lambda_1 + \cdots + \lambda_i$. 
\end{enumerate}
The relation $\leq_{d.o.}$ is a partial order relation on the set $\operatorname{Par}$, and it is called dominance ordering. 
\end{dfn}

\begin{dfn}
Let $\lambda  \in \operatorname{Par}$. The number of nonzero elements in the sequence $\lambda$ is called the length of $\lambda$, denoted by $\ell(\lambda)$. 
\end{dfn}

It is clear that for any $\lambda \in \operatorname{Par}$, we can always write $\lambda$ uniquely as 
\begin{align}
\lambda 
&= (\underbrace{			\lambda_1,\dots,\lambda_1		}_{m_1}, 
\underbrace{		\lambda_{m_1 + 1},\dots, \lambda_{m_1 + 1}			}_{m_2},\dots, 
\underbrace{	\lambda_{\sum_{i = 1}^{n-1}m_i+ 1}, \dots, \lambda_{\sum_{i = 1}^{n-1}m_i+ 1}		}_{m_n},
0,0,\dots 
)
\label{eqn21-1749-1jan}
\\
&:= (\lambda_1^{m_1}\lambda_{m_1 + 1}^{m_2}\cdots
\lambda_{\sum_{i = 1}^{n-1}m_i+ 1}^{m_n}
)
\notag 
\end{align}
where $\lambda_1, \lambda_{m_1 + 1}, \dots, \lambda_{\sum_{i = 1}^{n-1}m_i+ 1} \in \bb{Z}^{\geq 1}$ and $m_1,m_2,\dots, m_n \in \bb{Z}^{\geq 1}$ satisfying $m_1 + \cdots + m_n = \ell(\lambda)$. In the situation as in the equation \eqref{eqn21-1749-1jan}, we say that the partition $\lambda$ consists of $n$ \textbf{\textit{row intervals}}. Also, for each $r \in \{1,\dots,n\}$, we will call the row $(m_1 + \cdots + m_r)^{\text{th}}$
\textbf{\textit{the last row of the row interval $r^{\text{th}}$}}.

%Also, for $\lambda$ as in the equation \eqref{eqn21-1749-1jan}, we define 
%\begin{align}
%z_\lambda := 
%\prod_{j = 1}^{n}
%\bigg(
%\lambda_{\sum_{i = 1}^{j-1}m_i+ 1}^{m_j}
%\times
%m_j!
%\bigg)
%\label{eqn22-1146}
%\end{align}

One can present a partition by using the Young diagram. For $\lambda \in \operatorname{Par}$, the Young diagram with shape $\lambda$ is a collection of boxes, provided that 
the number of boxes in the row $i$ of the Young diagram is equal to $\lambda_i$. For example, 
the Young diagram with shape $\lambda = (4,4,2,1)$ is 
\begin{align}
\ydiagram{4,4,2,1}
\label{eqn22-2056-1jan}
\end{align}
It is clear from this Young diagram that $\lambda = (4,4,2,1)$ has $3$ row intervals. 

\begin{dfn}
Let $\lambda, \mu \in \operatorname{Par}$ and $\mu \subseteq \lambda$. The Young diagram with skew shape 
$\lambda \backslash \mu$ is an object obtained by removing the boxes appearing in the Young diagram with shape $\mu$
from the Young diagram with shape $\lambda$. 
\end{dfn}

\begin{rem}
In symmetric polynomials literatures, the skew shape of a Young diagram is usually denoted by $\lambda / \mu$. 
However, in this paper, we will denote it by $\lambda \backslash \mu$ to avoid confusion with the $\lambda/\mu$ that appears in the combinatorial formula of the Macdonald polynomials 
(see equation \eqref{eqn27-2052} and \textbf{Proposition \ref{thm225-2054}} below). 
\end{rem}

\begin{exa}
Let $\lambda = (4,4,2,1)$ and $\mu = (3,1)$. Then, 
\begin{align}
\ydiagram{3+1,1+3,2,1}
\end{align}
is a Young diagram with skew shape $(4,4,2,1) \backslash (3,1)$. 
\end{exa}

\begin{dfn-prp}[Conjugation of $\lambda$]
Let $\lambda = (\lambda_1,\lambda_2,\dots) \in \operatorname{Par}$. Define $\lambda^\prime = (\lambda^\prime_1,\lambda^\prime_2,\dots)$ where for each $k \in \bb{Z}^{\geq 1}$, $\lambda^\prime_k$ is the number of boxes in the $k^{\text{th}}$ column of the Young diagram with shape $\lambda$. 
Then, $\lambda^\prime \in \operatorname{Par}$, for any $\lambda \in \operatorname{Par}$.
The partition $\lambda^\prime$ is called the conjugation of the partition $\lambda$. 
\end{dfn-prp}

For example, for $\lambda = (4,4,2,1)$ whose diagram is as shown in equation \eqref{eqn22-2056-1jan}, we have $\lambda^\prime = (4,3,2,2)$.

\begin{dfn}[Arm length and leg length]
Let $\lambda = (\lambda_1,\lambda_2,\dots) \in \operatorname{Par}$, and let $s$ be the box that appears in the $i^{\text{th}}$ row and $j^{\text{th}}$ column in the Young diagram of $\lambda$. Then, the arm length and leg length of the box $s$, denoted by $a_\lambda(s)$ and $\ell_\lambda(s)$ respectively, are defined by 
\begin{align*}
a_\lambda(s) = \lambda_i - j, 
\hspace{1.6cm}
\ell_\lambda(s) = \lambda^\prime_j - i. 
\end{align*}
\end{dfn}

\begin{nota}
To indicate that box $s$ is located in $i^{\text{th}}$ row  and $j^{\text{th}}$ column, we can use the shorthand notation $\operatorname{row}(s) = i$ and $\operatorname{col}(s) = j$. 
\end{nota}

\begin{dfn}
Let $\lambda \in \operatorname{Par}$ and $N \in \bb{Z}^{\geq 1}$. 
A semi-standard Young tableau (SSYT) of type $N$ with shape $\lambda$ is a Young diagram of shape $\lambda$ where each box is filled with a positive integer from the set $\{1,\dots,N\}$, satisfying the following conditions:
\begin{enumerate}[(1)]
\item 
The numbers in each row are weakly increasing from left to right.
\item 
The numbers in each column are strictly increasing from top to bottom. 
\end{enumerate}
The set of all SSYTs of type $N$ with shape $\lambda$ is denoted by $\operatorname{SSYT}(N;\lambda)$. 
\end{dfn}

\begin{dfn}
Let $\lambda, \mu \in \operatorname{Par}$, $\mu \subseteq \lambda$, and $N \in \bb{Z}^{\geq 1}$. 
A semi-standard Young tableau (SSYT) of type $N$ with skew shape $\lambda \backslash \mu$ is a Young diagram of skew shape $\lambda \backslash \mu$ where each box is filled with a positive integer from the set $\{1,\dots,N\}$, satisfying the following conditions:
\begin{enumerate}[(1)]
	\item 
	The numbers in each row are weakly increasing from left to right.
	\item 
The numbers in each column are strictly increasing from top to bottom. 
\end{enumerate}
The set of all SSYTs of type $N$ with skew shape $\lambda \backslash \mu$ is denoted by $\operatorname{SSYT}(N;\lambda \backslash \mu)$. 
\end{dfn}

\begin{dfn}
\label{dfn210-1443}
Let $\lambda, \mu \in \operatorname{Par}$, $\mu \subseteq \lambda$, and $N \in \bb{Z}^{\geq 1}$. 
A reverse semi-standard Young tableau (reverse SSYT) of type $N$ with skew shape $\lambda\backslash\mu$ is a Young diagram with skew shape $\lambda\backslash\mu$ where each box is filled with a positive integer from the set $\{1,\dots,N\}$, satisfying the following conditions:
\begin{enumerate}[(1)] 
	\item 
	The numbers in each row are weakly decreasing from left to right.
	\item 
	The numbers in each column are strictly decreasing from top to bottom. 
\end{enumerate}
The set of all reverse SSYTs of type $N$ with skew shape $\lambda \backslash \mu$ is denoted by $\operatorname{RSSYT}(N;\lambda\backslash \mu)$. 
\end{dfn}

\begin{exa}
	\begin{align}
		\begin{ytableau}
			4 & 3 & 2 & 2 \\
			3 & 2
		\end{ytableau}
	\end{align}
	is an element of $\operatorname{RSSYT}(4; (4,2) )$. 
\end{exa}

\begin{exa}
	\label{example212-2115}
	\begin{align}
		\begin{ytableau}
			\none  & \none & \none & 2 \\
			\none & 3 & 3 & 1\\ 
			2 & 1 \\
			1
		\end{ytableau}
	\end{align}
	is an element of $\operatorname{RSSYT}(3; 	(4,4,2,1) \backslash (3,1)		 )$. 
\end{exa}

\begin{dfn}
\label{dfn211-1443}
Let $\lambda, \mu \in \operatorname{Par}$ such that $\mu \subseteq \lambda$, and let $N, M \in \bb{Z}^{\geq 1}$. 
A reverse semi-standard Young bitableau (reverse SSYBT) of type $(N,M)$ with  skew shape $\lambda\backslash\mu$
is a filling of the Young diagram of skew shape $\lambda\backslash\mu$ with positive integers satisfying the following conditions:
\begin{enumerate}[(1)]
\item Each box in the Young diagram is assigned a positive integer from the set $\{1,\dots,N,N+1,\dots,N+M		\}$. The elements of the subset $\{1,\dots,N	\}$ are referred to as ordinary numbers, while the elements of the subset $\{N+1,\dots,N+M	\}$ are referred to as super numbers. 
\item 
The numbers in each row are weakly decreasing from left to right.
\item 
The numbers in each column are weakly decreasing from top to bottom. 
\item 
The super numbers within each row are strictly decreasing from left to right
\item 
The ordinary numbers within each column are strictly decreasing from top to bottom. 
\end{enumerate}
The set of all reverse SSYBT of type $(N,M)$ with skew shape $\lambda\backslash\mu$ is denoted by $\operatorname{RSSYBT}(N,M;\lambda\backslash\mu)$. 
\end{dfn}

For clarity and ease of calculation, super numbers in the reverse SSYBT will be represented in blue throughout this paper. It is important to note that by substituting $\mu = (0,0,\dots)$ (the empty partition) into \textbf{Definitions \ref{dfn210-1443}} and \textbf{\ref{dfn211-1443}}, we directly obtain the definitions of reverse SSYT and reverse SSYBT with shape $\lambda$.

\begin{exa}
\begin{align}
	\begin{ytableau}
		\textcolor{blue}{4}  & \textcolor{blue}{3} & 2 & 2 \\
		\textcolor{blue}{4} & 2
	\end{ytableau}
\end{align}
is an element of $\operatorname{RSSYBT}(2,2; (4,2) )$. 
\end{exa}

\begin{exa}
\begin{align}
	\begin{ytableau}
		\textcolor{blue}{4}  & \textcolor{blue}{4} & 2 & 2 \\
		\textcolor{blue}{3} & 2
	\end{ytableau}
\end{align}
is \textbf{not} an element of $\operatorname{RSSYBT}(2,2; (4,2) )$. This is because the super numbers in the first row is not strictly decreasing. 
\end{exa}

\begin{exa}
\label{example217-2115}
\begin{align}
	\begin{ytableau}
		\none  & \textcolor{blue}{4} & 2 & 2 \\
		\textcolor{blue}{3} & 2
	\end{ytableau}
\end{align}
is an element of $\operatorname{RSSYBT}(2,2; (4,2) \backslash (1))$. 
\end{exa}

\begin{prop}
Let $T \in \operatorname{RSSYBT}(N,M;\lambda)$, and let $T_1$ be the subdiagram of $T$ consisting of the boxes with super numbers. Then, the shape of $T_1$ is a Young diagram of a partition. 
\end{prop}

\begin{dfn}\mbox{}
\label{def220-1540-19mar}
\begin{enumerate}[(1)]
\item Let $T \in \operatorname{RSSYBT}(N,M;\lambda)$, and let $\beta \in \{1,\dots,N+M\}$. Define 
\begin{align}
\operatorname{Row}(T|\beta) := 
\left\{
\alpha \in \{1,\dots,\ell(\lambda)\}
\;\middle\vert\;
\begin{array}{@{}l@{}}
\alpha \text{ is a row of $T$ which contains a box with number $\beta$}
\end{array}
\right\}
\end{align}
\item Let $\lambda, \mu \in \operatorname{Par}$ where $\mu \subseteq \lambda$. Define 
\begin{align}
\operatorname{Row}( \mu \rightarrow \lambda ) := 
\left\{
\alpha \in \{1,\dots,\ell(\lambda)\}
\;\middle\vert\;
\begin{array}{@{}l@{}}
\lambda_\alpha \neq \mu_\alpha
\end{array}
\right\}
\end{align}
\end{enumerate}
\end{dfn}

\subsection{The algebra of symmetric polynomials}
\label{subsec22-1941-4apr}

\begin{dfn}
Let $S_N$ be the symmetric group\footnote{
Recall that the elements of $S_N$ is the bijections $\sigma : \{1,\dots,N\} \rightarrow \{1,\dots,N\}$. 
}. Define the algebra $\Lambda_N$ over the field $\bb{C}$ to be 
\begin{align}
\Lambda_N := 
\left\{
f(x_1,\dots,x_N) \in \bb{C}[x_1,\dots,x_N]
\;\middle\vert\;
\begin{array}{@{}l@{}}
^\forall \sigma \in S_N \,\, f(x_{\sigma(1)},\dots,x_{\sigma(N)}) = f(x_1,\dots,x_N)
\end{array}
\right\}.
\end{align}
The elements of $\Lambda_N$ are called symmetric polynomials of $N$ variables over $\bb{C}$. 
\end{dfn}

The algebra $\Lambda_N$ admits a grading by degree of the polynomials. Specifically, we can write 
$\Lambda_N = \ddsum_{r = 0}^{\infty} \Lambda_N^r$ where 
\begin{align}
\Lambda_N^r := 
\left\{
f(x_1,\dots,x_N) \in \Lambda_N
\;\middle\vert\;
\begin{array}{@{}l@{}}
\deg f(x_1,\dots,x_N) = r 
\end{array}
\right\}. 
\end{align}
For each $N \in \bb{Z}^{\geq 1}$, we define a homomorphism $\pi^{(r)}_N : \Lambda_N^r \rightarrow \Lambda_{N-1}^r$ by the evaluation $f(x_1,\dots,x_N)  \mapsto f(x_1,\dots,x_{N-1},0)$. 
We define $\Lambda^r$ to be the inverse limit of the following inverse system:
\begin{align}
\cdots \xrightarrow{\pi^{(r)}_{N+1}} \Lambda_N^r \xrightarrow{\pi^{(r)}_N} \Lambda_{N-1}^r \xrightarrow{\pi^{(r)}_{N-1}} \cdots.
\end{align}
That is, $\displaystyle \Lambda^r := \inverselimm{N}\Lambda_N^r$. Then, the ring of symmetric functions, denoted by $\Lambda$, is defined as 
\begin{align}
\Lambda = \ddsum_{r = 0}^{\infty}\Lambda^r. 
\end{align}

%For each fixed degree $r \in \bb{Z}^{\geq 0}$, we have the following inverse system
%\begin{align}
%\cdots \rightarrow \Lambda_{n+1}^r \xrightarrow{\pi_{n+1}} \Lambda_{n}^r \rightarrow  \cdots \rightarrow 
%\end{align}
%where $\pi_{n+1}(f(x_1,\dots,x_{n},x_{n+1})) = f(x_1,\dots,x_n,0)$.
%The inverse limit of this inverse system will be denoted by $\Lambda^r$, i.e. 
%\begin{align}
%\Lambda^r := 
%\Prlimm{n}
%\Lambda_n^r. 
%\end{align}
%
%\begin{dfn}
%Define $\Lambda := \ddsum_{r = 0}^{\infty} \Lambda^r$. The elements of the algebra $\Lambda$ are called symmetric functions over $\bb{Q}$.
%\end{dfn}
%
%\begin{rem}
%It is important to note that elements of $\Lambda$ are not strictly functions. However, they are commonly referred to as symmetric functions based on historical convention. 
%\end{rem}

\begin{exa}[Monomial symmetric polynomials]
For each partition $\lambda \in \operatorname{Par}$ with length $\ell(\lambda) \leq N$, define the monomial symmetric polynomials of $N$ variables as 
\begin{align}
m_\lambda(x_1,\dots,x_N) := \sum_{\alpha \in \operatorname{Per}_N(\lambda)}x^\alpha, 
\end{align}
where
\begin{align}
\operatorname{Per}_N(\lambda)
:=
\left\{
(\alpha_1,\dots,\alpha_N)
\;\middle\vert\;
\begin{array}{@{}l@{}}
^\exists \sigma \in S_N \text{ such that }
(\alpha_1,\dots,\alpha_N)
= 
(\lambda_{\sigma(1)},\dots,\lambda_{\sigma(N)})
\end{array}
\right\},
\end{align}
and $x^\alpha := x_1^{\alpha_1}\cdots x_N^{\alpha_N}$. It is evident that $m_\lambda(x_1,\dots,x_N) \in \Lambda_N$. The corresponding symmetric functions, known as monomial symmetric functions, are defined as follows: for each partition $\lambda \in \operatorname{Par}$, 
\begin{align}
	m_\lambda(x) := \sum_{\alpha \in \operatorname{Per}(\lambda)}x^\alpha, 
\end{align}
where 
\begin{align}
	\operatorname{Per}(\lambda)
	:=
	\left\{
	(\alpha_i)_{i = 1}^{\infty}
	\;\middle\vert\;
	\begin{array}{@{}l@{}}
		^\exists \text{bijection }
		\sigma : \bb{Z}^{\geq 1} \rightarrow \bb{Z}^{\geq 1} \text{ such that }
		^\forall i \in \bb{Z}^{\geq 1} \, \alpha_i = \lambda_{\sigma(i)}
	\end{array}
	\right\},
\end{align}
and $x^\alpha := \prod_{i = 1}^{\infty}x_i^{\alpha_i}$. Clearly, $m_\lambda(x) \in \Lambda$. 
\end{exa}

\begin{exa}[Power sum symmetric polynomials]
\label{exa223-1225-26mar}
For each integer $r \in \bb{Z}^{\geq 1}$, define 
\begin{align}
p_r(x_1,\dots,x_N) := \sum_{i = 1}^{N}x_i^r. 
\label{eqn215-1543-24mar}
\end{align}
By convention, we define $p_0(x_1,\dots,x_N) := 1$. For each partition $\lambda = (\lambda_1,\lambda_2,\dots)$, define 
\begin{align}
p_\lambda(x_1,\dots,x_N) := \prod_{i = 1}^{\ell(\lambda)}p_{\lambda_i}(x_1,\dots,x_N). 
\label{eqn216-1543-24mar}
\end{align}
It is clear that $p_\lambda(x_1,\dots,x_N) \in \Lambda_N$, and these are called the power sum symmetric polynomials of $N$ variables. 
The symmetric functions analogues of equations \eqref{eqn215-1543-24mar} and \eqref{eqn216-1543-24mar} are defined as 
\begin{align}
p_r(x) := \sum_{i = 1}^{\infty}x_i^r, 
\end{align}
and 
\begin{align}
p_\lambda(x) := \prod_{i = 1}^{\ell(\lambda)}p_{\lambda_i}(x),
\end{align}
respectively. Clearly, $p_\lambda(x) \in \Lambda$, and these are called the power sum symmetric functions. 
\end{exa}

\begin{prop}
\label{prp224-1352-25mar}
The set of power sum symmetric functions, 
$\{	p_\lambda(x) ~|~ \lambda \in \operatorname{Par}		\}$, forms a $\bb{C}$-basis of the algebra of symmetric functions 
$\Lambda$. 
\end{prop}
\begin{proof}
See \cite{macbook-1998}. 
\end{proof}

%\begin{exa}[Power sum symmetric polynomials]
%\label{exam216}
%\end{exa}

\subsection{Macdonald  polynomials}
\label{subsec23-1942-4apr}
In this subsection, we will review the fundamental definitions and propositions concerning Macdonald  polynomials. It is important to note that, throughout this paper, Macdonald  polynomials will always be considered over the field of complex numbers $\bb{C}$. 

\begin{dfn}[Macdonald-Ruijsenaars operator]
Let $q,t \in \bb{C}\backslash\{0\}$, and let $\Lambda_{N}$
be the algebra of symmetric polynomials of $N$ variables over the field $\bb{C}$. 
We define the Macdonald-Ruijsenaars operator $D^{(N)} : \Lambda_{N} \rightarrow \Lambda_{N}$ as the linear 
operator given by 
\begin{align}
D^{(N)}
:= \sum_{i = 1}^{N}
\bigg[
\bigg(
\prod_{
\substack{
1 \leq j \leq N
\\
j \neq i 
}
}
\frac{tx_i - x_j}{x_i - x_j}
\bigg)
T_{q,x_i}
\bigg], 
\end{align}
where $(T_{q,x_i}f)(x_1,\dots,x_N) = f(x_1,\dots,qx_i,\dots,x_N)$. 
\end{dfn}

\begin{dfn-prp}[Macdonald polynomials]
Let $q,t \in \bb{C}\backslash\{0\}$ satisfy the condition that
\begin{align}
\text{
if $a,b \in \bb{Z}$ and $q^at^b = 1$, then $a = b = 0$.}
\label{eqn214-1341-13jan}
\end{align}
Then, for each partition $\lambda \in \operatorname{Par}$ with length $\ell(\lambda) \leq N$,
there exists a unique 
$P_\lambda(x_1,\dots,x_N;q,t) \in \Lambda_{N}$ satisfying the following conditions:
\begin{enumerate}[(1)]
\item There exists coefficients 
$
\left\{
u_{\lambda\mu} \in \bb{C}
\;\middle\vert\;
\begin{array}{@{}l@{}}
	\mu <_{d.o.} \lambda
\end{array}
\right\}
$
such that 
\begin{align}
	P_\lambda(x_1,\dots,x_N;q,t) = m_\lambda(x_1,\dots,x_N) + \sum_{\mu <_{d.o.} \lambda}u_{\lambda\mu}m_{\mu}(x_1,\dots,x_N). 
\end{align}
\item $D^{(N)}P_\lambda(x_1,\dots,x_N;q,t) = (	\sum_{i = 1}^{N}t^{N-i}q^{\lambda_i}		)P_\lambda(x_1,\dots,x_N;q,t)$
\end{enumerate}
The symmetric polynomials $P_\lambda(x_1,\dots,x_N;q,t) \,\, (\lambda \in \operatorname{Par})$ are called the Macdonald polynomials. 
\end{dfn-prp}

\begin{dfn-prp}[Macdonald functions]
Let $q,t \in \bb{C}\backslash\{0\}$ satisfy the condition stated in equation \eqref{eqn214-1341-13jan}.
Then, the Macdonald polynomials satisfy the following relation:
\begin{align}
	\pi_N\left(	P_\lambda(x_1,\dots,x_N;q,t)				\right) = P_\lambda(x_1,\dots,x_{N-1};q,t), 
\end{align}
where $\pi_N := \ddsum_{r = 0}^{\infty}\pi_N^{(r)} : \ddsum_{r = 0}^{\infty} \Lambda^r_{N} \rightarrow \ddsum_{r = 0}^{\infty} \Lambda^r_{N-1}$. As a consequence, there uniquely exists a symmetric function $P_\lambda(x;q,t) \in \Lambda$ such that 
\begin{align*}
	\bigg|_{
		\substack{
			x_{N+1} = x_{N+2} = \cdots = 0  
		}
	}
	P_\lambda(x;q,t) = P_\lambda(x_1,\dots,x_N;q,t). 
\end{align*}
The family of symmetric functions $\left(		P_\lambda(x;q,t)	\right)_{\lambda \in \operatorname{Par}}$ are called the Macdonald functions. 
\end{dfn-prp}

\begin{rem}
From now on, whenever we refer to the Macdonald polynomials $P_\lambda(x_1,\dots,x_N;q,t)$ and Macdonald functions $P_\lambda(x;q,t)$, we will always assume that $q,t$ satisfy the condition stated in equation \eqref{eqn214-1341-13jan}.
\end{rem}

In the following, we will discuss the combinatorial formula for Macdonald polynomials, which provides a method for computing Macdonald polynomials. First, we introduce the definition of $\psi_{T}(q,t)$.

\begin{dfn}
\label{dfn223-2118}
\begin{enumerate}[(1)]
\item 
Let $\lambda, \mu \in \operatorname{Par}$ and $\mu \subseteq \lambda$. Define 
\begin{align}
	\psi_{\lambda/\mu}(q,t) := \prod_{1 \leq i \leq j \leq \ell(\mu)}
	\frac{
		f_{q,t}(q^{\mu_i - \mu_j}t^{j - i})
		f_{q,t}(q^{\lambda_i - \lambda_{j+1}}t^{j - i})
	}{
		f_{q,t}(q^{\lambda_i - \mu_j}t^{j - i})
		f_{q,t}(q^{\mu_i - \lambda_{j+1}}t^{j - i})
	}, 
\label{eqn27-2052}
\end{align}
where $\displaystyle f_{q,t}(u) := \frac{(tu;q)_\infty}{(qu;q)_\infty}$. Here $(x;q)_\infty := \prod_{i = 0}^{\infty}(1 - xq^i)$. 
\vspace{0.2cm}
\item 
Let $\lambda, \mu \in \operatorname{Par}$ and $\mu \subseteq \lambda$.
For each $T \in \operatorname{RSSYT}(N;\lambda\backslash\mu)$, we define
\begin{align}
	\psi_{T}(q,t) := \prod_{k = 1}^{N}\psi_{T^{(k)}/T^{(k+1)}}(q,t), 
\label{eqn28-1539-3jan}
\end{align}
where for each $i \in \{1,\dots,N\}$, $T^{(i)}$ is the Young diagram formed by the union of the Young diagram $\mu$ and the boxes containing number greater than or equal to $i$ in T. By convention, for any $i \geq N+1$, we have 
$T^{(i)} := \mu$. 
\end{enumerate}
\end{dfn}

The next example illustrates the definition of $T^{(i)}$.

\begin{exa}
Let $T$ be the Young diagram in the example \ref{example212-2115}. According to the definition \ref{dfn223-2118}, we get that 

\small 
\begin{align}
T^{(4)} &= 
\ydiagram{3,1} \,\,, 
\hspace{3cm}
T^{(3)} = 
\ydiagram{3,3} \,\,,
\\
T^{(2)} &= 
\ydiagram{4,3,1} \,\,, 
\hspace{2.2cm}
T^{(1)} = 
\ydiagram{4,4,2,1}. 
\end{align}
\normalsize
\end{exa}

\begin{prop}
\label{prop226-1127-19mar}
Let $\lambda, \mu \in \operatorname{Par}$ and $\mu \subseteq \lambda$. Define
\begin{align}
\psi^\prime_{\lambda/\mu}(q,t) := 
\underbrace{		\prod_{	1 \leq  i < j  \leq \ell(\lambda)			}					}_{
\substack{
(1) \,\, \lambda_i = \mu_i 
\\
(2) \,\, \lambda_j = \mu_j + 1
}
}
\frac{
(1 - q^{\mu_i - \mu_j}t^{j - i - 1})(1 - q^{\lambda_i - \lambda_j}t^{j - i +1})
}{
(1 - q^{\mu_i - \mu_j}t^{j - i })(1 - q^{\lambda_i - \lambda_j}t^{j - i})
}
\end{align}
Then, we have $\psi^\prime_{\lambda/\mu}(q,t) = \psi_{\lambda^\prime/\mu^\prime}(t,q)$. 
\end{prop}
\begin{proof}
See \cite{macbook-1998}. 
\end{proof}

\begin{lem}
\label{lem225-1201-11mar}
Let $\lambda, \mu \in \operatorname{Par}$ and $\mu \subseteq \lambda$. Then, for each $1 \leq i \leq j \leq \ell(\mu)$, we get that 
\begin{align}
\frac{
	f_{q,t}(q^{\mu_i - \mu_j}t^{j - i})
	f_{q,t}(q^{\lambda_i - \lambda_{j+1}}t^{j - i})
}{
	f_{q,t}(q^{\lambda_i - \mu_j}t^{j - i})
	f_{q,t}(q^{\mu_i - \lambda_{j+1}}t^{j - i})
}
= 
\frac{
	f_{q^{-1},t^{-1}}((q^{-1})^{\mu_i - \mu_j}(t^{-1})^{j - i})
	f_{q^{-1},t^{-1}}((q^{-1})^{\lambda_i - \lambda_{j+1}}(t^{-1})^{j - i})
}{
	f_{q^{-1},t^{-1}}((q^{-1})^{\lambda_i - \mu_j}(t^{-1})^{j - i})
	f_{q^{-1},t^{-1}}((q^{-1})^{\mu_i - \lambda_{j+1}}(t^{-1})^{j - i})
}
\label{eqn220-1153-11mar}
\end{align}
\end{lem}
\begin{proof}
By using \textbf{Definition \ref{dfn223-2118}}, we can show that 
\begin{align}
\frac{
f_{q^{-1},t^{-1}}((q^{-1})^{\mu_i - \mu_j}(t^{-1})^{j - i})
}{
f_{q^{-1},t^{-1}}((q^{-1})^{\lambda_i - \mu_j}(t^{-1})^{j - i})
}
&= 
\left(\frac{q}{t}\right)^{\lambda_i - \mu_i}
\frac{
f_{q,t}(q^{\mu_i - \mu_j}t^{j - i})
}{
f_{q,t}(q^{\lambda_i - \mu_j}t^{j - i})
}, 
\\
\frac{
	f_{q^{-1},t^{-1}}((q^{-1})^{\lambda_i - \lambda_{j+1}}(t^{-1})^{j - i})
}{
	f_{q^{-1},t^{-1}}((q^{-1})^{\mu_i - \lambda_{j+1}}(t^{-1})^{j - i})
}
&=
\left(\frac{t}{q}\right)^{\lambda_i - \mu_i}
\frac{
	f_{q,t}(q^{\lambda_i - \lambda_{j+1}}t^{j - i})
}{
	f_{q,t}(q^{\mu_i - \lambda_{j+1}}t^{j - i})
}. 
\end{align}
From this, we immediately obtain equation \eqref{eqn220-1153-11mar}. 
\end{proof}

\begin{cor}
\label{cor226-1214-11mar}
For each $T \in \operatorname{RSSYT}(N;\lambda)$, we have $\psi_{T}(q,t) = \psi_{T}(q^{-1},t^{-1})$. 
\end{cor}
\begin{proof}
From \textbf{Lemma \ref{lem225-1201-11mar}}, we can see that for each $j \in \{1,\dots,N - 1\}$, 
\begin{align}
\psi_{T^{(j)}/T^{(j+1)}}(q,t)
= 
\psi_{T^{(j)}/T^{(j+1)}}(q^{-1},t^{-1}). 
\end{align}
Therefore, $\psi_{T}(q,t) = \prod_{j = 1}^{N - 1}\psi_{T^{(j)}/T^{(j+1)}}(q,t) = \prod_{j = 1}^{N - 1}\psi_{T^{(j)}/T^{(j+1)}}(q^{-1},t^{-1}) = \psi_{T}(q^{-1},t^{-1})$. 
\end{proof}

\begin{dfn}
\label{dfn227-2239-13mar}
Let $N, \widetilde{N} \in \bb{Z}^{\geq 1}$, $T \in \operatorname{RSSYT}(N;\lambda\backslash\mu)$, and $\widetilde{T} \in \operatorname{RSSYT}(\widetilde{N};\lambda\backslash\mu)$. We say that $T$ and $\widetilde{	T		}$ has the same structure if 
\begin{align}
	\left\{
	T^{(i)}	
	\;\middle\vert\;
	\begin{array}{@{}l@{}}
		i \in \bb{Z}^{\geq 1}
	\end{array}
	\right\}
	= 
	\left\{
	\widetilde{T}^{(i)}
	\;\middle\vert\;
	\begin{array}{@{}l@{}}
		i \in \bb{Z}^{\geq 1}
	\end{array}
	\right\}. 
\end{align}
\end{dfn}

\begin{exa}
Let $\lambda = (5,3,1)$ and $\mu = (1)$. One can show that 
\begin{align*}
T = 
\begin{ytableau}
	\none  & 4 & 4 & 4 & 2 \\
	2 & 2 & 1 \\
	1
\end{ytableau}
\in 
\operatorname{RSSYT}(4,\lambda\backslash\mu)
\end{align*}
and 
\begin{align*}
\widetilde{T} = 
\begin{ytableau}
	\none  & 7 & 7 & 7 & 4 \\
	4 & 4 & 1 \\
	1
\end{ytableau}
\in 
\operatorname{RSSYT}(8,\lambda\backslash\mu)
\end{align*}
has the same structure. 
\end{exa}

\begin{prop}
\label{prp226-1326-15feb}
Let $N, \widetilde{N} \in \bb{Z}^{\geq 1}$, $T \in \operatorname{RSSYT}(N;\lambda\backslash\mu)$, and $\widetilde{T} \in \operatorname{RSSYT}(\widetilde{N};\lambda\backslash\mu)$. If $T$ and $\widetilde{	T		}$ has the same structure, then $\psi_{T}(q,t) = \psi_{\widetilde{	T		}}(q,t)$. 
\end{prop}
\begin{proof}
This is clear from equation \eqref{eqn28-1539-3jan}. 
\end{proof}

\begin{prop}[Combinatorial formula for Macdonald polynomials]
\mbox{}
\label{thm225-2054}
\begin{align}
P_\lambda(x_1,\dots,x_N;q,t)
= 
\sum_{T \in \operatorname{RSSYT}(N;\lambda)}
\left(
\psi_T(q,t) \times \prod_{s \in \lambda} x_{T(s)}
\right).
\label{eqn232-2040-25mar}
\end{align}
Here $T(s)$ denotes the number on the box $s$. 
\end{prop}
\begin{proof}
See equation $(7.13^\prime)$ in chapter VI of \cite{macbook-1998}
\end{proof}

Motivated from \textbf{Proposition \ref{thm225-2054}}, the skew Macdonald polynomials $P_{\lambda \backslash \mu}(x_1,\dots,x_N;q,t)$ is defined by the formula similar to those in equation \eqref{eqn232-2040-25mar}. 

\begin{dfn}[Skew Macdonald polynomials]\mbox{}
Let $\lambda, \mu \in \operatorname{Par}$ and $\mu \subseteq \lambda$. The skew Macdonald polynomials $P_{\lambda \backslash \mu}(x_1,\dots,x_N;q,t)$ are defined by 
\begin{align}
	P_{\lambda \backslash \mu}(x_1,\dots,x_N;q,t)
	:= 
	\sum_{T \in \operatorname{RSSYT}(N;\lambda \backslash \mu)}
	\left(
	\psi_T(q,t) \times \prod_{s \in \lambda \backslash \mu} x_{T(s)}
	\right),
\label{eqn233-2040-26mar}	
\end{align}
where $T(s)$ denotes the number on the box $s$. 
\end{dfn}

\begin{rem}
The skew Macdonald polynomials $P_{\lambda \backslash \mu}(x_1,\dots,x_N;q,t)$ can also be defined by a different, but equivalent, method. We refer the reader to the book \cite{macbook-1998} for a thorough description of this alternative method. 
\end{rem}

\subsection{Super Macdonald polynomials}
\label{subsec24-1943-4apr}

In this subsection, we recall the construction given in \cite{sv2007-supermac}. 
For convenience, the notation $P_{N,M} := \bb{C}[x_1,\dots,x_N,y_1,\dots,y_M]$ is introduced. 

\begin{dfn}
Define the algebra $\Lambda_{N,M,q,t}$ over the field $\bb{C}$ to be 
\small 
\begin{align}
&\Lambda_{N,M,q,t}
\\
&:= 
\left\{
f(x_1,\dots,x_N,y_1,\dots,y_M) \in P_{N,M}
\;\middle|\;
\begin{aligned}
&\circled{1} \,\,
^\forall \sigma \in S_N \,\, ^\forall \tau \in S_M,
\\
&\,\, f(x_{\sigma(1)},\dots,x_{\sigma(N)},y_{\tau(1)},\dots,y_{\tau(M)}) = f(x_1,\dots,x_N,y_1,\dots,y_M). 
\\
&\circled{2} \,\,
^\forall i \in \{1,\dots,N\} \,\, ^\forall j \in \{1,\dots,M\}, 
\\
& \,\, 
\bigg|_{
	\substack{
	x_i = y_j 
	}
}
T_{q,x_i}(f) 
= 
\bigg|_{
	\substack{
		x_i = y_j 
	}
}
T_{t,y_j}(f).
\end{aligned}
\right\}
\notag 
\end{align}
\end{dfn}

\begin{exa}
For each $r \in \bb{Z}^{\geq 1}$, define 
\begin{align}
p^{(N,M)}_r(x,y,q,t) := \sum_{i = 1}^{N}x^r_i + \frac{1 - q^r}{1 - t^r}\sum_{j = 1}^{M}y^r_j. 
\end{align}
It is easy to show that $p^{(N,M)}_r(x,y,q,t) \in \Lambda_{N,M,q,t}$. 
\end{exa}

\begin{thm}
\label{thm237-1353-25mar}
The $\bb{C}$-algebra $\Lambda_{N,M,q,t}$ is generated by the subset $\{		p^{(N,M)}_r(x,y,q,t) ~|~ r \in \bb{Z}^{\geq 1}	\}$. 
\end{thm}
\begin{proof}
See theorem 5.8 in \cite{sv2007-supermac}. 
\end{proof}

\begin{dfn}
\label{dfn238-1401-25mar}
Define $\varphi_{q,t} : \Lambda \rightarrow \Lambda_{N,M,q,t}$ to be a surjective algebra homomorphism determined by 
\begin{align}
\varphi_{q,t}\left(	p_r(z)	\right) = p^{(N,M)}_r(x,y,q,t). 
\end{align}
\end{dfn}

Note that the surjectivity of the algebra homomorphism $\varphi_{q,t} : \Lambda \rightarrow \Lambda_{N,M,q,t}$ can be directly seen from \textbf{Theorem \ref{thm237-1353-25mar}}.

\begin{thm}
\label{thm239-1508-25mar}
Define $H_{N,M} := \{\lambda \in \operatorname{Par}~|~ \lambda_{N+1} \leq M\}$. 
%Then, $\operatorname{ker}(\varphi_{q,t})$ is spanned by subset $\{	P_\lambda(z;q,t^{-1})~|~ \lambda \in \operatorname{Par}\backslash H_{N,M}		\}$. 
If $\lambda \in \operatorname{Par}\backslash H_{N,M}$, then
\begin{align*}
\varphi_{q,t}(P_\lambda(z;q,t^{-1})) = 0. 
\end{align*}
\end{thm}

The proof of \textbf{Theorem \ref{thm239-1508-25mar}} was given in theorem 5.6 in \cite{sv2007-supermac}. 
We will not rewrite it here. However, we would like to record two lemmas that are crucial to the proofs of \textbf{Theorem \ref{thm239-1508-25mar}} and \textbf{Proposition \ref{prp228-1020-29jan}}.

\begin{lem}
\label{lem245-1514-21apr}
Define an algebra homomorphism $\sigma : \Lambda \rightarrow \Lambda$ by 
\begin{align}
\sigma_{q,t}(p_r) = \frac{1-q^r}{1-t^r}p_r. 
\label{eqn236-1618-25mar}
\end{align}
Then, we obtain that 
\begin{align}
\sigma_{q,t}\left(	P_\lambda(z;q,t^{-1})				\right) = \frac{H(\lambda,q,t)}{H(\lambda^\prime,t,q)}
P_{\lambda^\prime}(z;t,q^{-1}).
\label{eqn237-1527-25mar}
\end{align}
Here $H(\lambda,q,t) := q^{n(\lambda)}t^{n(\lambda^\prime)}
\prod_{s \in \lambda}\left(
q^{a(s) + 1} - t^{\ell(s)}
\right)$, where $n(\lambda) := \sum_{i = 1}^{\infty}(i-1)\lambda_i$. 
\end{lem}
\begin{proof}
First, we collect some facts from \cite{macbook-1998} which will be used to prove this lemma. Define algebra homomorphism $\omega_{q,t} : \Lambda \rightarrow \Lambda$ by 
\begin{align}
	\omega_{q,t}(p_r) = (-1)^{r-1}\frac{1 - q^r}{1 - t^r}p_r. 
\label{eqn238-1618-25mar}
\end{align}
Define 
\begin{align}
Q_{\lambda}(z;q,t) := b_\lambda(q,t)P_{\lambda}(z;q,t)
\label{eqn241-2036-30mar}
\end{align}
where 
\begin{align}
	b_\lambda(q,t) = \prod_{s \in \lambda}\frac{1 - q^{a_\lambda(s)}t^{\ell_\lambda(s) + 1}}{1 - q^{a_\lambda(s) + 1}t^{\ell_\lambda(s)}}. 
\end{align}
From \cite{macbook-1998} (page 327, equation (5.1)), it is known that $\omega_{q,t}P_{\lambda}(z;q,t) = Q_{\lambda^\prime}(z;t,q)$. From the definition of $H(\lambda,q,t)$ and $b_\lambda(q,t)$, it can be shown that
\begin{align}
	\frac{H(\lambda,q,t)}{H(\lambda^\prime,t,q)} = (-q)^{|\lambda|}b_{\lambda^\prime}(t,q^{-1}) = (-t)^{-|\lambda|}b_{\lambda^\prime}(t^{-1},q).
\end{align}
As a consequence, we obtain that 
\begin{align}
	\frac{H(\lambda,q,t)}{H(\lambda^\prime,t,q)}
	P_{\lambda^\prime}(z;t,q^{-1}) 
	=
	(-q)^{|\lambda|}Q_{\lambda^\prime}(z;t,q^{-1})
	= 
	(-t)^{-|\lambda|}Q_{\lambda^\prime}(z;t^{-1},q). 
\end{align}

Now, we are ready to derive the equation \eqref{eqn237-1527-25mar}. From equations \eqref{eqn236-1618-25mar} and \eqref{eqn238-1618-25mar}, it is clear that 
\begin{align}
	\omega_{q,t^{-1}}(p_\lambda) 
	&= (-1)^{|\lambda|-\ell(\lambda)}
	p_\lambda
	\prod_{i = 1}^{\ell(\lambda)}\frac{1 - q^{\lambda_i}}{1 - t^{-\lambda_i}}
	= (-t)^{|\lambda|}
	p_\lambda
	\prod_{i = 1}^{\ell(\lambda)}\frac{1 - q^{\lambda_i}}{1 - t^{\lambda_i}}
	=
	(-t)^{|\lambda|}
	\sigma_{q,t}(p_\lambda).
\end{align}
According to the theory of Macdonald polynomials, we know that 
\begin{align}
	P_{\lambda}(z;q,t) =
	\sum_{
		\substack{
			\mu \in \operatorname{Par}
			\\
			|\mu|  = |\lambda|
		}
	}
	u_{\mu}(q,t)p_\mu(z),
\end{align}
where $u_\mu(q,t) \in \bb{C}$. Therefore, 
\begin{align}
	\omega_{q,t^{-1}}
	\left(
	P_{\lambda}(z;q,t^{-1})
	\right)
	&= 
	\sum_{
		\substack{
			\mu
			\\
			|\mu|  = |\lambda|
		}
	}
	u_{\mu}(q,t^{-1})\omega_{q,t^{-1}}\left(		p_\mu(z)	\right)
	=
	(-t)^{|\lambda|}
	\sum_{
		\substack{
			\mu
			\\
			|\mu|  = |\lambda|
		}
	}
	u_{\mu}(q,t^{-1})\sigma_{q,t}\left(		p_\mu(z)	\right)
	\\
	&=
	(-t)^{|\lambda|}
	\sigma_{q,t}
	\left(
	P_{\lambda}(z;q,t^{-1})
	\right).
\end{align}
In other word, 
\begin{align}
	\sigma_{q,t}
	\left(
	P_{\lambda}(z;q,t^{-1})
	\right)
	&= 
	(-t)^{-|\lambda|}
	\omega_{q,t^{-1}}
	\left(
	P_{\lambda}(z;q,t^{-1})
	\right)
	= (-t)^{-|\lambda|}Q_{\lambda^\prime}(z;t^{-1},q)
	= 
	\frac{H(\lambda,q,t)}{H(\lambda^\prime,t,q)}
	P_{\lambda^\prime}(z;t,q^{-1}). 
\end{align}
\end{proof}

\begin{lem}\mbox{}
\label{lem246-1514-21apr}
\begin{align}
	\varphi_{q,t}\left(	P_\lambda(z;q,t^{-1})		\right)
	&= 
	\sum_{
		\substack{
			\mu \in \operatorname{Par}
			\\
			\mu \subseteq \lambda
		}
	}
	P_{\lambda \backslash \mu}(x_1,\dots,x_N,q,t^{-1})
	\frac{H(\mu,q,t)}{H(\mu^\prime,t,q)}
	P_{\mu^\prime}(y_1,\dots,y_M;t,q^{-1}).
\label{eqn238-1527-25mar}
\end{align}
\end{lem}
\begin{proof}
First, note that the power sum symmetric functions, described in \textbf{Example \ref{exa223-1225-26mar}}, was written by using only one class of variables. However, there is nothing prevents us to express the power sum symmetric functions by using several classes of variables. For example,
\begin{align}
p_r(x,y) := \sum_{i = 1}^{\infty}x_i^r + \sum_{j = 1}^{\infty}y_j^r. 
\end{align}
is the power sum symmetric function expressed in two classes of variables $\{x_i\}_{i \in \bb{Z}^{\geq 1}}$ and $\{y_j\}_{j \in \bb{Z}^{\geq 1}}$.

Since $P_{\lambda}(z;q,t) =
\sum_{
	\substack{
		\mu \in \operatorname{Par}
		\\
		|\mu|  = |\lambda|
	}
}
u_{\mu}(q,t)p_\mu(z)$, the Macdonald functions with two classes of variables can be defined as 
\begin{align}
	P_{\lambda}(x,y;q,t) :=
	\sum_{
		\substack{
			\mu \in \operatorname{Par}
			\\
			|\mu|  = |\lambda|
		}
	}
	u_{\mu}(q,t)p_\mu(x,y). 
\end{align}
Also, from \cite{macbook-1998} (page 345, equation ($7.9^\prime$)), it is known that 
\begin{align}
P_\lambda(x,y;q,t) := \sum_{
\substack{
\mu \in \operatorname{Par}
\\
\mu \subseteq \lambda
}
}
P_{\lambda \backslash \mu}(x;q,t)P_{\mu}(y;q,t). 
\end{align}

It is clear that 
\begin{align}
	&
	\varphi_{q,t}(p_\mu(z))
	=
	\bigg[ \,\,
	\bigg|_{
		\substack{
			x_{N+1} = x_{N+2} = \cdots = 0 
			\\
			y_{M+1} = y_{M+2} = \cdots = 0 
		}
	}
	\comp 
	\left(
	\operatorname{id}_x \tens \sigma_{q,t}
	\right)
	\bigg]
	p_{\mu}(x,y). 
\end{align}
Here $\operatorname{id}_x \tens \sigma_{q,t}$ is the map which applies $\sigma_{q,t}$ to the part of variable $y$, while apply identity map $\operatorname{id}$ to the part of variable $x$. 
Thus, we get that 
\begin{align}
\varphi_{q,t}\left(
P_\lambda(z;q,t^{-1})
\right)
&=
\bigg[ \,\,
\bigg|_{
	\substack{
		x_{N+1} = x_{N+2} = \cdots = 0 
		\\
		y_{M+1} = y_{M+2} = \cdots = 0 
	}
}
\comp 
\left(
\operatorname{id}_x \tens \sigma_{q,t}
\right)
\bigg]
P_{\lambda}(x,y;q,t^{-1})
\\
&= 
\bigg|_{
	\substack{
		x_{N+1} = x_{N+2} = \cdots = 0 
		\\
		y_{M+1} = y_{M+2} = \cdots = 0 
	}
}
\left(
\sum_{
	\substack{
		\mu \in \operatorname{Par}
		\\
		\mu \subseteq \lambda
	}
}
P_{\lambda \backslash \mu}(x;q,t^{-1})
\sigma_{q,t}(P_{\mu}(y;q,t^{-1}))
\right)
\notag 
\\
&= 
\sum_{
	\substack{
		\mu \in \operatorname{Par}
		\\
		\mu \subseteq \lambda
	}
}
P_{\lambda \backslash \mu}(x_1,\dots,x_N,q,t^{-1})
\frac{H(\mu,q,t)}{H(\mu^\prime,t,q)}
P_{\mu^\prime}(y_1,\dots,y_M;t,q^{-1}). 
\notag 
\end{align}

\end{proof}

\begin{rem}
There are a few minor errata in the formulas corresponding to equations \eqref{eqn237-1527-25mar} and \eqref{eqn238-1527-25mar} in \cite{sv2007-supermac}. For this reason, we include detailed proof for \textbf{Lemmas \ref{lem245-1514-21apr}}
and \textbf{\ref{lem246-1514-21apr}}
above. Nevertheless, these minor errata do not affect the validity of \textbf{Theorem \ref{thm239-1508-25mar}}. Thus, \textbf{Theorem \ref{thm239-1508-25mar}} remains valid. 
\end{rem}

%\begin{rem}
%The formulas in equations \eqref{eqn237-1527-25mar} and \eqref{eqn238-1527-25mar} appearing in \cite{sv2007-supermac} contain a few minor mistakes. Because of this, we provide detailed proofs 
%to confirm that the formulas in our paper is correct. 
%\end{rem}

%\begin{rem}
%The formulas for equations \eqref{eqn237-1527-25mar} and \eqref{eqn238-1527-25mar} in \cite{sv2007-supermac} contain a few minor mistakes. That is, in \cite{sv2007-supermac}, the formulas were written as A = 2 and B = 3. However, these small typos do not change the result of theorem 2.39. 
%\end{rem}

According to \textbf{Theorem \ref{thm239-1508-25mar}}, we know that if $\lambda \in \operatorname{Par}\backslash H_{N,M}$, then $\varphi_{q,t}(P_\lambda(z;q,t^{-1})) = 0$. Thus, in the following definition, we consider only $\lambda \in H_{N,M}$.

\begin{dfn}[Super Macdonald polynomials]
For each $\lambda \in H_{N,M}$, the super Macdonald polynomials $\supermac_\lambda(x_1,\dots,x_N;y_1,\dots,y_M;q,t)$ are defined by 
\begin{align}
\supermac_\lambda(x_1,\dots,x_N;y_1,\dots,y_M;q,t) 
:=  
\varphi_{q,t^{-1}}\left(	P_\lambda(z;q,t)		\right). 
\label{eqn254-1458-26mar}
\end{align}
\end{dfn}

\begin{prop}[Combinatorial formula for super Macdonald polynomials]\mbox{}
\label{prp228-1020-29jan}
For each $\lambda \in H_{N,M}$, we have 
\begin{align}
	\supermac_\lambda(x_1,\dots,x_N;y_1,\dots,y_M;q,t) =  
	\sum_{T \in \operatorname{RSSYBT}(N,M;\lambda)}
	\left(
	\widetilde{\psi}_T(q,t^{-1}) \times \prod_{s \in \lambda} x_{T(s)}
	\right).
	\label{eqn261-1741-17aug}
\end{align}
Here $\displaystyle \widetilde{\psi}_{T}(q,t) := \psi_{T^\prime_1}(t,q^{-1})\psi_{T_0}(q,t^{-1})\frac{H(\mu,q,t)}{H(\mu^\prime,t,q)}$, 
where $T_0$ and $T_1$ are the sub-diagrams of $T$ consisting of boxes assigned with ordinary numbers, and super numbers, respectively. Also, in \eqref{eqn261-1741-17aug}, for each $j \in \{1,\dots,M\}$, we write $x_{N+j} = y_j$. 
\end{prop}
\begin{proof}
From equations 
\eqref{eqn232-2040-25mar}
\eqref{eqn233-2040-26mar}
\eqref{eqn238-1527-25mar} \eqref{eqn254-1458-26mar}, the statement of this proposition follows directly. 
\end{proof}

Before ending this section, we note that if $\lambda \in H_{N,M}$, then $\operatorname{RSSYBT}(N,M;\lambda) \neq \emptyset$. This fact will be used in \textbf{Lemma \ref{lemm44-1106-29jan}} when referring to an element in 
$\operatorname{RSSYBT}(N,M;\lambda)$ when $\lambda \in H_{N,M}$.

\section{Quantum Corner VOA}
\label{sec3-1945-4apr}

In this section, we review the definition of the quantum corner VOA via the Miura transformation \cite{HMNW} \cite{PR18}. As the Miura transformation involves vertex operators $\Lambda^{\vec{c},\vec{u}}_i(z)$ which can be defined through the tensor product of horizontal Fock representations of the quantum toroidal $\fraks{gl}_1$ algebra, we begin by recalling the definition of the quantum toroidal $\fraks{gl}_1$ algebra and its representation. In writing this section, we follow the convention of \cite{pc2406}.

\subsection{Quantum toroidal algebra of type $\fraks{gl}_1$}
\label{subsec31-1953-4apr}

\begin{dfn}
\label{def31-1621-12jan}
Let $q_1,q_2,q_3$ be complex numbers satisfying the following conditions:
\begin{enumerate}[(1)]
\item $q_1q_2q_3 = 1$
\item If $a,b,c \in \bb{Z}$ which makes $q_1^aq_2^bq_3^c = 1$, then $a = b =c$. 
\end{enumerate}
The quantum toroidal algebra of type $\fraks{gl}_1$, denoted by $U_{q_1,q_2,q_3}(\widehat{\widehat{\fraks{gl}}}_1)$, is a unital associative algebra generated by the generators 
\begin{align}
E_k, F_k, K^{\pm}_0, H_{\pm r}, C  \hspace{0.3cm} (k \in \bb{Z}, r \in \bb{Z}^{\geq 1}), 
\end{align}
subject to the following relations:
\begin{gather}
	C \text{ is a central element}, 
	\\
	K^{\pm}(z)K^{\pm}(w) = 		K^{\pm}(w)K^{\pm}(z),
	\label{2.3main}
	\\
	K^{+}(z)K^-(w) = \frac{\cals{G}(w/C z)}{\cals{G}(C w/z)} K^-(w)K^{+}(z),
	\\
	K^+(z)E(w) = \cals{G}(w/z)E(w)K^+(z),
	\\
	K^-(C z)E(w) = \cals{G}(w/z)E(w)K^-(C z),
	\\
	K^+(C z)F(w) = \cals{G}(w/z)^{-1}F(w)K^+(C z),
	\\
	K^-(z)F(w) = \cals{G}(w/z)^{-1}F(w)K^-(z),
	\\
	\label{EEexchage}
	E(z)E(w) = \cals{G}(w/z)E(w)E(z),
\end{gather}
\begin{gather}
	F(z)F(w) = \cals{G}(w/z)^{-1}F(w)F(z),
	\\
	[E(z),F(w)] = \frac{1}{(q_1 - 1)(q_2 - 1)(q_3 - 1)}
	\bigg(
	\delta\Big(\frac{Cw}{z}\Big)K^+(z) - \delta\Big(\frac{Cz}{w}\Big)K^-(w) 
	\bigg),
	\label{2.11main}
\end{gather}
where 
$\delta(z) := \sum_{k \in \bb{Z}}
z^k$, 
\begin{gather}
	E(z) := \sum_{k\in\mathbb{Z}} E_k z^{-k}, \quad
	F(z) := \sum_{k\in\mathbb{Z}} F_k z^{-k}, \quad
	K^{\pm}(z) := K_0^{\pm}\exp\left(
	\pm \sum_{r=1}^{\infty} H_{\pm r} z^{\mp r}
	\right), 
	\label{2.12}
	\\
	\cals{G}(z) := 
	\frac{(1 - q_1^{-1}z)(1 - q_2^{-1}z)(1 - q_3^{-1}z)}{
	(1 - q_1z)(1 - q_2z)(1 - q_3z)
	}. 
\end{gather}
\end{dfn}
\vspace{0.2cm}

It is well-known that $U_{q_1,q_2,q_3}(\widehat{\widehat{\fraks{gl}}}_1)$ has the structure of a Hopf algebra. This means it has a coproduct, a counit, and an antipode. However, for the purpose of this paper, only the coproduct will play a significant role. Therefore, we will focus only on the coproduct. The formula for the coproduct of $U_{q_1,q_2,q_3}(\widehat{\widehat{\fraks{gl}}}_1)$ is explicitly given in \textbf{Proposition \ref{prp32-1748-9jan}} below. 

\begin{prop}
\label{prp32-1748-9jan}
The map $\Delta : U_{q_1,q_2,q_3}(\widehat{\widehat{\fraks{gl}}}_1)
\rightarrow 
U_{q_1,q_2,q_3}(\widehat{\widehat{\fraks{gl}}}_1) \tens U_{q_1,q_2,q_3}(\widehat{\widehat{\fraks{gl}}}_1)$
determined by 
\begin{align}
	\Delta\big(E(z)\big) &= E(z) \tens 1 + K^-(C_1z) \tens E(C_1z),
	\label{copro1}
	\\
	\Delta\big(F(z)\big) &= F(C_2z) \tens K^+(C_2z) + 1 \tens F(z),
	\\
	\Delta\big(K^+(z)\big) &= K^+(z) \tens K^+(C_1^{-1}z),
	\\
	\Delta\big(K^-(z)\big) &= K^-(C_2^{-1}z) \tens K^-(z),
	\label{copro4}
	\\
	\Delta\big(		C		\big) &= C \tens C, 
\end{align}
where $C_1 := C \tens 1, C_2 := 1 \tens C$, is an algebra homomorphism. It is called the coproduct of $U_{q_1,q_2,q_3}(\widehat{\widehat{\fraks{gl}}}_1)$. 
\end{prop}

Next, we will discuss a representation of $U_{q_1,q_2,q_3}(\widehat{\widehat{\fraks{gl}}}_1)$. There are various known representations of $U_{q_1,q_2,q_3}(\widehat{\widehat{\fraks{gl}}}_1)$. However, among these representations, the horizontal Fock representation plays a crucial role in the study of quantum corner VOA. Thus, in this paper, we only discuss the horizontal Fock representation of $U_{q_1,q_2,q_3}(\widehat{\widehat{\fraks{gl}}}_1)$. 

To understand the horizontal Fock representation of $U_{q_1,q_2,q_3}(\widehat{\widehat{\fraks{gl}}}_1)$, we need to first discuss Heisenberg algebra.

\begin{dfn}
For each $i \in \{1,2,3\}$, we define the Heisenberg algebra $\cals{B}^{(i)}$ to be the algebra generated by generators $\left\{a^{(i)}_{\pm n}~|~ n \in \bb{Z}^{\geq 1}\right\}$, subject to the relation
\begin{align}
[a_n^{(i)},a_m^{(i)}] &= \frac{n}{\kappa_n}(q_i^{n/2} - q_i^{-n/2})
\delta_{n+m,0}. 
\end{align}
Here $\kappa_n := (q^n_1 - 1)(q^n_2 - 1)(q^n_3 - 1)$. 
\end{dfn}

Next, we define $\cals{H}^{(i)}$ to be the vector space
\begin{align}
	\cals{H}^{(i)} := \operatorname{span}
	\left\{
	a^{(i)}_{-\lambda_1}\cdots a^{(i)}_{-\lambda_m}|0\rangle
	\;\middle\vert\;
	\begin{array}{@{}l@{}}
		(1) \,\,  m \in \bb{Z}^{\geq 0}\\
		(2) \,\, \lambda_1 \geq \cdots \geq \lambda_m \geq 1
	\end{array}
	\right\}.
\end{align}
where $|0\rangle$ is the vacuum state, which is annihilated by the positive mode operators, i.e. $a_n|0\rangle = 0$. It is straightforward to see that $\cals{H}^{(i)}$ has the structure of a left $\cals{B}^{(i)}$-module.

We have now covered all the necessary ingredients to define the horizontal Fock representation of $U_{q_1,q_2,q_3}(\widehat{\widehat{\fraks{gl}}}_1)$. 
We will now discuss the horizontal Fock representation.

\begin{dfn-prp}
For each $i \in \{1,2,3\}$ and $u \in \bb{C}\backslash\{0\}$, the horizontal Fock representation of $U_{q_1,q_2,q_3}(\widehat{\widehat{\fraks{gl}}}_1)$ is an algebra homomorphism $\rho_{H,u}^{(i)} : U_{q_1,q_2,q_3}(\widehat{\widehat{\fraks{gl}}}_1)
\rightarrow \operatorname{End}(\cals{H}^{(i)})$ determined by the following equations:
\begin{align}
	\rho_{H,u}^{(i)}\left(
	K^+(z)
	\right)
	&= 
	\exp\left(
	\sum_{n = 1}^{\infty}\frac{\kappa_n}{n}q^{n/4}_ia^{(i)}_{n}z^{-n}
	\right),
	\\
	\rho_{H,u}^{(i)}
	\left(
	K^-(z)
	\right)
	&= 
	\exp\left(
	- \sum_{n = 1}^{\infty}\frac{\kappa_n}{n}q^{-n/4}_{i}
	a^{(i)}_{-n}z^n
	\right),
	\\
	\rho_{H,u}^{(i)}
	\left(
	E(z)
	\right)
	&= u\widetilde{d}_1
	\exp\left(
	\sum_{n = 1}^{\infty}
	\frac{\kappa_n}{n}
	\frac{q^{-n/4}_i}{q^{n/2}_{i} - q^{-n/2}_{i}}
	a^{(i)}_{-n}z^n
	\right)
	\exp\left(
	-\sum_{n = 1}^{\infty}
	\frac{\kappa_n}{n}
	\frac{q^{n/4}_i}{q^{n}_{i}  - 1}
	a^{(i)}_{n}z^{-n}
	\right),
	\\
	\rho_{H,u}^{(i)}
	\left(
	F(z)
	\right)
	&= 
	u^{-1}\widetilde{d}_2
	\exp\left(
	-\sum_{n = 1}^{\infty}
	\frac{\kappa_n}{n}
	\frac{q^{n/4}_i}{q^{n/2}_{i} - q^{-n/2}_{i}}
	a^{(i)}_{-n}z^{n}
	\right)
	\exp\left(
	\sum_{n = 1}^{\infty}
	\frac{\kappa_n}{n}
	\frac{q^{n/4}_i}{q^{n/2}_{i} - q^{-n/2}_{i}}
	a^{(i)}_{n}z^{-n}
	\right),
	\\
	\rho_{H,u}^{(i)}
	\left(
	C
	\right)
	&= q^{1/2}_{i},
\end{align}
provided that 
\begin{align}
	\widetilde{d}_1 &= 
	\frac{
	1 - q_i
	}{
	(1 - q_1)
	(1 - q_2)
	(1 - q_3)
	}, 
	\label{eqn315-2135}
	\\
	\widetilde{d}_2 &=
	\frac{
	(1 - q^{-1}_i)
	}{
		(q_1 - 1)(q_2 - 1)(q_3 - 1)
	}.
	\label{eqn316-2136}
\end{align}
Note that the $i^\prime, i^{\prime\prime}$, which appears in the equation \eqref{eqn315-2135}, are the elements in the set $\{1,2,3\}$ which are different from $i$. 
\end{dfn-prp}

Since $U_{q_1,q_2,q_3}(\widehat{\widehat{\fraks{gl}}}_1)$ has a Hopf algebra structure, it follows that we can construct representations of $U_{q_1,q_2,q_3}(\widehat{\widehat{\fraks{gl}}}_1)$ by using tensor products. More precisely, for each $\vec{c} = (c_1,\dots,c_n) \in \{1,2,3\}^n$ and $\vec{u} = (u_1,\dots,u_n) \in \bb{C}^n$, we define 
\begin{align}
	\rho^{\vec{c}}_{H,\vec{u}}
	: 
	U_{q_1,q_2,q_3}(\widehat{\widehat{\fraks{gl}}}_1)
	\rightarrow \operatorname{End}(\cals{H}^{(c_1)} \tens \cdots \tens \cals{H}^{(c_n)}), 
\end{align}
by 
\begin{align}
	\rho^{\vec{c}}_{H,\vec{u}}
	:= 
	\left(
	\rho_{H,u_1}^{(c_1)} \tens \rho_{H,u_2}^{(c_2)} \tens \cdots \tens \rho_{H,u_n}^{(c_n)}	
	\right) \comp 
	\Delta^{(n-1)}, 
\end{align}
where $\Delta^{(1)} := \Delta$, 
and for $n \in \bb{Z}^{\geq 2}$, 
\begin{align}
\Delta^{(n)} := (\Delta \tens \underbrace{		1 \tens \cdots \tens 1		}_{n-1} ) \comp \Delta^{(n-1)}. 
\end{align}
It is clear that $\rho^{\vec{c}}_{H,\vec{u}}$ is an algebra homomorphism. 

\subsection{Quantum corner VOA}
\label{subsec32-1955-4apr}

The goal of this subsection is to define the quantum corner VOA via the Miura transformation (see \textbf{Definition \ref{def38-2133-12jan}}). To do so, we will begin this subsection by introducing the so-called vertex operator, which plays a crucial role in the Miura transformation.

Define 
\begin{align}
	\alpha(z) &:= 
	\exp\left(
	-\sum_{r = 1}^{\infty}\frac{1}{C^r - C^{-r}}\widetilde{b}_{-r}z^r
	\right),
	\\
	\beta(z) &:=
	\exp\left(
	\sum_{r = 1}^{\infty}\frac{1}{C^r - C^{-r}}\widetilde{b}_rz^{-r}
	\right),
\end{align}
Here, $\widetilde{b}_{\pm r}$ are defined by the equation \eqref{eqn331-1625-12jan} below: 
\begin{align}
	K^+(z) &= 
	\exp\left(
	\sum_{r = 1}^{\infty}\widetilde{b}_rC^{r}z^{-r}
	\right),
	\hspace{1cm}
	K^-(z) =
	\exp\left(
	-\sum_{r = 1}^{\infty}\widetilde{b}_{-r}z^r
	\right). 
\label{eqn331-1625-12jan}
\end{align}

\begin{dfn}
For each $\vec{c} = (c_1,\dots,c_n) \in \{1,2,3\}^n$ and $\vec{u} = (u_1,\dots,u_n) \in (\bb{C} \backslash \{0\})^n$, we define the vertex operator $\widetilde{\Lambda}^{\vec{c},\vec{u}}_j(z) \,\, (j = 1,\dots,n)$ to be 
\begin{align}
\widetilde{\Lambda}^{\vec{c},\vec{u}}_j(z)
:= 
\rho^{\vec{c}}_{H,\vec{u}}\left(\alpha(z)\right)
\Lambda^{\vec{c},\vec{u}}_j(z)
\rho^{\vec{c}}_{H,\vec{u}}\left(\beta(z)\right)
\end{align}
Here 
\begin{align}
\Lambda^{\vec{c},\vec{u}}_j(z) := 
u_j
\varphi^{(c_1)}(q_{c_1}^{1/2}z;p) \tens \cdots \tens \varphi^{(c_{j-1})}(q_{c_1}^{1/2} \cdots q_{c_{j-1}}^{1/2}z;p) 	
\tens \eta^{(c_j)}(q_{c_1}^{1/2} \cdots q_{c_{j-1}}^{1/2}z;p) \tens 
\underbrace{	1 \tens \cdots \tens 1				}_{n - j}, 
\end{align}
where 
\begin{align}
	\varphi^{(i)}(z;p) &= 
	\exp\left(
	- \sum_{n = 1}^{\infty}
	\frac{\kappa_n}{n}q^{-n/4}_{i}
	a^{(i)}_{-n}z^n
	\right),
	\\
	\eta^{(i)}(z;p) &= 
	\exp\left(
	\sum_{n = 1}^{\infty}
	\frac{\kappa_n}{n}
	\frac{q^{-n/4}_{i}}{q^{n/2}_{i} - q^{-n/2}_{i}}
	a^{(i)}_{-n}z^n
	\right)
	\exp\left(
	-\sum_{n = 1}^{\infty}
	\frac{\kappa_n}{n}
	\frac{q^{n/4}_{i}}{q^{n}_{i}  - 1}
	a^{(i)}_{n}z^{-n}
	\right). 
\end{align}
\end{dfn}

\begin{dfn}
For each $\vec{c} = (c_1,\dots,c_n) \in \{1,2,3\}^n$ satisfying the condition $q_{\vec{c}} := \prod_{k = 1}^{n}q_{c_k} \neq 1$, we define a factor $f^{\vec{c}}_{r,m}(z)$ as follows: For $r \leq m$, 
\begin{align}
	f^{\vec{c}}_{r,m}(z)
	&:=
	\exp\bigg[
	\sum_{k = 1}^{\infty}\frac{z^k}{k}
	(q_3^{\frac{r}{2}k} - q_3^{-\frac{r}{2}k})(q_{\vec{c}}^{\frac{k}{2}}q^{-\frac{m}{2}k}_{3}
	- q_{\vec{c}}^{-\frac{k}{2}}q^{\frac{m}{2}k}_{3})
	\frac{
		(q^{\frac{k}{2}}_1 - q^{-\frac{k}{2}}_1)(q^{\frac{k}{2}}_2 - q^{-\frac{k}{2}}_2)
	}{
		(q^{\frac{k}{2}}_{\vec{c}} - q^{-\frac{k}{2}}_{\vec{c}})(q^{\frac{k}{2}}_3 - q^{-\frac{k}{2}}_3)
	}
	\bigg]
	\label{eqn337-2055-30mar}
\end{align}
and for $r \geq m$, 
\begin{align}
	f^{\vec{c}}_{r,m}(z)
	&:= 
	f^{\vec{c}}_{m,r}(z).
\end{align}
\end{dfn}

\begin{prop}
\label{prp37-1018-23jan}
	The following relations hold:
	\begin{align}
	\label{eqn339-2150-12jan}
		\widetilde{\Lambda}^{\vec{c},\vec{u}}_i(z)			\widetilde{\Lambda}^{\vec{c},\vec{u}}_j(w)		
		= 
		\begin{cases}
			\displaystyle
			f^{\vec{c}}_{11}\left(\frac{w}{z}	  	\right)^{-1}
			\Delta\left(	q_3^{\frac{1}{2}}\frac{w}{z}			\right)
			\normord{		\widetilde{\Lambda}^{\vec{c},\vec{u}}_i(z)	 	\widetilde{\Lambda}^{\vec{c},\vec{u}}_j(w)						}
			\text{ for } i < j 
			\\
			\displaystyle
			f^{\vec{c}}_{11}\left(\frac{w}{z}  \right)^{-1}
			\gamma_{c_i}\left(\frac{w}{z} \right)
			\normord{		\widetilde{\Lambda}^{\vec{c},\vec{u}}_i(z)		\widetilde{\Lambda}^{\vec{c},\vec{u}}_i(w)						}
			\text{ for } i = j
			\\
			\displaystyle
			f^{\vec{c}}_{11}\left(\frac{w}{z} \right)^{-1}
			\Delta\left(	q_3^{-\frac{1}{2}}\frac{w}{z} 			\right)
			\normord{		\widetilde{\Lambda}^{\vec{c},\vec{u}}_i(z)		\widetilde{\Lambda}^{\vec{c},\vec{u}}_j(w)				}
			\text{ for } i > j 
		\end{cases}
	\end{align}
	where
	\begin{align}
		\Delta(z) 
		&:= 
		\frac{
		(1 - q_1q_3^{\frac{1}{2}}z)
		(1 - q^{-1}_1q_3^{-\frac{1}{2}}z)
		}{
		(1 - q_3^{\frac{1}{2}}z)
		(1 - q_3^{-\frac{1}{2}}z)
		},
		\\
		\gamma_{c_i}(z) 
		\label{eqn340-1147-15apr}
		&:= 
		\frac{
		(1 - q_{c_i}z)
		(1 - q^{-1}_{c_i}z)
		}{
		(1 - q_{3}z)
		(1 - q^{-1}_{3}z)
		}. 
	\end{align} 
\end{prop}

\begin{dfn}[\cite{Misha} \cite{HMNW} \cite{PR18}]
\label{def38-2133-12jan}
Let $L, M, N \in \bb{Z}^{\geq 0}$, and let 
\begin{align}
	\vec{c} = (c_1,\dots,c_n) = (
	\underbrace{	3,\dots,3		}_{N},
	\underbrace{		1,\dots,1			}_{L},\underbrace{	2,\dots,2			}_{M}
	)
	= 
	(3^N1^L2^M). 
\end{align}
Fix $\vec{u} = (u_1,\dots,u_n) \in (\bb{C} \backslash \{0\})^n$. 
We define the quantum corner VOA $q\widetilde{Y}_{L,M,N}[\Psi]$ as the algebra generated by currents $\left\{\widetilde{T}^{\vec{c},\vec{u}}_m(z)\right\}_{m \in  \bb{Z}^{\geq 0}}$ , where these $\widetilde{T}^{\vec{c},\vec{u}}_m(z)$ can be expressed in terms of the previously defined vertex operators by the Miura transformation, which is defined by the equation \eqref{431-1247} below.
\begin{align}
	\normord{
	\widetilde{R}^{(c_1)}_{1}(z)\widetilde{R}^{(c_2)}_{2}(z)\cdots \widetilde{R}^{(c_n)}_{n}(z)
	}
	\,\, 
	= \sum_{m = 0}^{\infty}(-1)^m\widetilde{T}^{\vec{c},\vec{u}}_m(z)q_3^{-mD_z},
	\label{431-1247}
\end{align}
where 
\begin{align}
	\widetilde{R}^{(c)}_{i}(z) 
	&:=
	\sum_{k = 0}^{\infty}
	(-1)^kq_3^{\frac{1}{2}k^2}
	q_c^{-\frac{k}{2}}
	\bigg[
	\prod_{\ell = 0}^{k-1}
	\frac{1 - q_3^{\ell - k +1}q_c}{1 - q_3^{\ell + 1}}
	\bigg]
	\normord{
		\widetilde{\Lambda}_i^{\vec{c},\vec{u}}(z)
		\widetilde{\Lambda}_i^{\vec{c},\vec{u}}(q_3^{-1}z)
		\cdots
		\widetilde{\Lambda}_i^{\vec{c},\vec{u}}(q_3^{-(k-1)}z)
	}
	q_3^{-kD_z}, 
	\label{232-eqn}
\end{align}
and $D_z$ is the scaling operator, i.e. $q_3^{D_z}g(z) = g(q_3z)$, for any function $g(z)$. 
\end{dfn}

From the Miura transformation, equation \eqref{431-1247}, it can be shown that for any $m \in \bb{Z}^{\geq 0}$, 
\begin{align}
	\widetilde{T}^{\vec{c},\vec{u}}_m(z)
	&=
	\underbrace{			\sum_{k_1, \dots,k_n \in \bb{Z}^{\geq 0}}			}_{k_1 + \cdots + k_n = m}
	\left[
	\prod_{i = 1}^{n}
	\prod_{j_i = 1}^{k_i}
	\left(
	-(q^{1/2}_{c_i}q^{1/2}_{3})
	\frac{(1 - q_3^{j_i-1}q_{c_i}^{-1})}{(1 - q_3^{j_i})}
	\right)
	\right]
	\normord{
		\prod_{i = 1}^{n}\prod_{j_i = 1}^{k_i}
		\widetilde{		\Lambda		}^{\vec{c},\vec{u}}_i(q_3^{- \sum_{\ell = 1}^{i-1}k_\ell -j_i + 1}z)
	}. 
	\label{eqn314-1625}
\end{align}

Using equations \eqref{eqn339-2150-12jan} and \eqref{eqn314-1625}, we can calculate the products of the currents $\left\{\widetilde{T}^{\vec{c},\vec{u}}_m(z)\right\}_{m \in  \bb{Z}^{\geq 0}}$. It can be shown that these products satisfy the relations given in equation \eqref{eqn346-2152-12jan} below. We refer to these relations as the quadratic relations.

\begin{prop}[\cite{pc2406}]
\label{prop39-1328-8ap}
For each $\vec{c} = (c_1,\dots,c_n) = 
(3^N1^L2^M)$ satisfying the condition $q_{\vec{c}} := \prod_{k = 1}^{n}q_{c_k} \neq 1$, for each $\vec{u} = (u_1,\dots,u_n) \in (\bb{C} \backslash \{0\})^n$, 
and for each $r, m \in \bb{Z}^{\geq 1}$ such that $r \leq m$, we have 
\begin{align}
	&f^{\vec{c}}_{r,m}\left(
	q^{\frac{r - m}{2}}_{3}\frac{w}{z} 
	\right)
	\widetilde{	T	}^{\vec{c},\vec{u}}_r(z )\widetilde{	T		}^{\vec{c},\vec{u}}_m(w ) 
	- f^{\vec{c}}_{m,r}\left(
	q^{\frac{m - r}{2}}_{3}\frac{z}{w} 
	\right)
	\widetilde{T}^{\vec{c},\vec{u}}_m(w )\widetilde{T}^{\vec{c},\vec{u}}_r(z )
	\notag \\
	&= 
	\frac{
	(1 - q_1)(1 - q_2)
	}{
	(1 - q_3^{-1})
	}
	\sum_{k = 1}^{r}
	\left(
	\prod_{\ell = 1}^{k-1}
	\frac{(1 -  q_1q_3^{-\ell})( 1 - q_2q_3^{-\ell})}{( 1 - q_3^{-\ell - 1})( 1 - q_3^{-\ell})}
	\right)
	\biggl\{
	\delta\left(q_3^k\frac{w}{z}\right)f^{\vec{c}}_{r-k,m+k}(q_3^{\frac{r-m}{2}} )\widetilde{T}^{\vec{c},\vec{u}}_{r-k}(q_3^{-k}z)\widetilde{T}^{\vec{c},\vec{u}}_{m+k}(q_3^kw)
	\notag \\ 
	&\hspace{0.3cm}- 
	\delta\left(q_3^{r-m-k}\frac{w}{z}\right)f^{\vec{c}}_{r-k,m+k}(q_3^{\frac{m-r}{2}})\widetilde{T}^{\vec{c},\vec{u}}_{r-k}(z )\widetilde{T}^{\vec{c},\vec{u}}_{m+k}(w)
	\biggr\}. 
	\label{eqn346-2152-12jan}
\end{align}
\end{prop}

\begin{rem}
From equation \eqref{eqn346-2152-12jan}, it is evident that the quadratic relations are independent of the choice of $\vec{u} \in (\bb{C} \backslash \{0\})^n$. Since the quantum corner VOA $q\widetilde{Y}_{L,M,N}[\Psi]$ is determined by these quadratic relations, we can conclude that the resulting $q\widetilde{Y}_{L,M,N}[\Psi]$ remains the same regardless of the choice of $\vec{u}$. This justifies the omission of $\vec{u}$ from the notation $q\widetilde{Y}_{L,M,N}[\Psi]$. 
\end{rem}

\begin{rem}
In \cite{pc2406}, the quantum corner VOA $q\widetilde{Y}_{L,M,N}[\Psi]$ was defined using $\vec{c} = (1^L2^M3^N)$. However, in this paper, we choose to employ $\vec{c} = (3^N1^L2^M)$ instead. From the derivation presented in \cite{pc2406}, it is evident that the quadratic relations \eqref{eqn346-2152-12jan} of the quantum corner VOA $q\widetilde{Y}_{L,M,N}[\Psi]$ are independent of the ordering of the numbers in $\vec{c}$. Consequently, the quadratic relations \eqref{eqn346-2152-12jan} remain valid even with the choice of $\vec{c} = (3^N1^L2^M)$. 
\end{rem}

\section{Quantum Corner VOA/Super Macdonald correspondence}
\label{sec4-1950-4apr}

In this section, we state the main theorem of this paper (\textbf{Theorem \ref{thm42-main-1022}}). By using \textbf{Lemma \ref{lemm43-1156-22jan}}, we can show that instead of directly proving \textbf{Theorem \ref{thm42-main-1022}}, we can prove \textbf{Lemma \ref{lemm44-1106-29jan}}, which is considerably less complicated. 

\begin{dfn}
\label{dfn41-2057-30mar}
	For each $\lambda \in \operatorname{Par}(k)$, we define 
	$\dualmap$
	to be the map which sends $f(z_1,\dots,z_k) \in \bb{C}(z_1,\dots,z_k)$ to the element 
	\begin{align}
		&f(y,qy,\dots,q^{\lambda_1 - 1}y
		\notag	\\
		&\hspace{0.4cm} \xi y,q\xi y,\dots,q^{\lambda_2 - 1}\xi y
		\notag \\
		&\hspace{2.8cm}\vdots
		\notag \\
		&\hspace{0.4cm} \xi^{\ell(\lambda) - 1} y,q\xi^{\ell(\lambda) - 1} y,\dots,q^{\lambda_{\ell(\lambda)} - 1}\xi^{\ell(\lambda) - 1} y). 
	\end{align}
	of $\bb{C}(\xi, y)$. 
\end{dfn}

It should be noted that the map $\dualmap$ is not defined for all $f(z_1,\dots,z_k) \in \bb{C}(z_1,\dots,z_k)$. However, the $f(z_1,\dots,z_k) \in \bb{C}(z_1,\dots,z_k)$ considered in this paper yield well-defined images under the map $\dualmap$.

\begin{thm}[Quantum Corner VOA/Super Macdonald correspondence]
\label{thm42-main-1022}
Suppose that $\vec{c} = (3^N1^M)$, $\vec{u} = (u_1,\dots,u_{N+M})$, 
and $\lambda \in \operatorname{Par}(k)$ such that $\lambda \in H_{N,M}$. Then, 
\begin{align}
	&
	\lim_{\xi \rightarrow t^{-1}}\,\,
	(
	\dualmap
	\comp 
	\bigg|_{
		\substack{
			q_1 = q, \\
			q_2 = q^{-1}t,\\
			q_3 = t^{-1} \\
		}
	}
	)
	\left(
	\cals{N}_{\lambda}(z_1,\dots,z_k )
	\times
	\prod_{1 \leq i < j \leq k}f^{\vec{c}}_{11}\left(\frac{z_j}{z_i} \right)
	\times
	\langle 0 |\widetilde{T}^{\vec{c},\vec{u}}_{1}(z_1 )\cdots \widetilde{T}^{\vec{c},\vec{u}}_{1}(z_k )|0\rangle
	\right)
	\label{eqn-42-1021-29jan}
	\\
	&= 
	\supermac_\lambda(u_1,\dots,u_N;
	q^{-\frac{1}{2}}t^{-\frac{1}{2}}u_{N+1},\dots,q^{-\frac{1}{2}}t^{-\frac{1}{2}}u_{N+M};q,t)
	\notag 
\end{align}
where 
%$\widetilde{T}^{\vec{c},\vec{u}}_{1}(z) := \sum_{k = 1}^{N+M}y_k\widetilde{\Lambda}^{\vec{c},\vec{u}}_k(z)$ and
\begin{align}
	\cals{N}_{\lambda}(z_1,\dots,z_k )
	&:= 
	\prod_{1 \leq c < d \leq \ell(\lambda)}
	\prod_{
		\substack{
			i \in I^{(c)}
			\\
			j \in I^{(d)}
		}
	}
	\Delta\left(	q_3^{-\frac{1}{2}}\frac{z_j}{z_i}				\right)^{-1}, 
	\label{eqn32-1044}
\end{align}
provided that $\Delta(z)$ is as defined in equation \eqref{eqn340-1147-15apr} and 
\begin{align*}
	I^{(1)} &= \{1,\dots,\lambda_1\},
	\\
	I^{(2)} &= \{\lambda_1 + 1,\dots,\lambda_1 + \lambda_2\},
	\\
	&\vdots
	\\
	I^{(\ell(\lambda))} &= \{\sum_{j = 1}^{\ell(\lambda) - 1}\lambda_j+ 1, \cdots, \sum_{j = 1}^{\ell(\lambda)}\lambda_j\}. 
\end{align*}
\end{thm}

The next lemma will tell us that we can write the (LHS) of equation \eqref{eqn-42-1021-29jan} as a summation of elements in the set $\operatorname{RSSYBT}(N,M;\lambda)$.

\begin{lem}
\label{lemm43-1156-22jan}
\begin{align}
	&
	\lim_{\xi \rightarrow t^{-1}}\,\,
	(
	\dualmap
	\comp 
	\bigg|_{
		\substack{
			q_1 = q, \\
			q_2 = q^{-1}t,\\
			q_3 = t^{-1} \\
		}
	}
	)
	\left(
	\cals{N}_{\lambda}(z_1,\dots,z_k )
	\times
	\prod_{1 \leq i < j \leq k}f^{\vec{c}}_{11}\left(\frac{z_j}{z_i} \right)
	\times
	\langle 0 |\widetilde{T}^{\vec{c},\vec{u}}_{1}(z_1 )\cdots \widetilde{T}^{\vec{c},\vec{u}}_{1}(z_k )|0\rangle
	\right)
	\label{eqn13-1426}
	\\
	&= 
	\underbrace{				
		\sum_{i_1 = 1}^{N+M}
		\cdots
		\sum_{i_k = 1}^{N+M}
	}_{
		(i_1,\dots,i_k) \in 
		\operatorname{RSSYBT}(N,M;\lambda)
	}
	\lim_{\xi \rightarrow t^{-1}}\,\,
	(
	\dualmap
	\comp 
	\bigg|_{
		\substack{
			q_1 = q, \\
			q_2 = q^{-1}t,\\
			q_3 = t^{-1} \\
		}
	}
	)
	\bigg[
	y_{i_1}\cdots y_{i_k}
	u_{i_1}\cdots u_{i_k}
	\notag \\
	&\hspace{7.5cm}\times
	\cals{N}_\lambda(z_1,\dots,z_k)
	\times
	\prod_{1 \leq a < b \leq k}
	\cals{D}^{(i_a,i_b)}\left(
	\frac{z_b}{z_a}
	; q, t
	\right)
	\bigg]
	\notag 
\end{align}
where $\displaystyle y_j :=
\frac{
q_{c_j}^{\frac{1}{2}} - q_{c_j}^{-\frac{1}{2}}
}{
q_{3}^{\frac{1}{2}} - q_{3}^{-\frac{1}{2}}
}
\,\, (j = 1,\dots,N+M)$, 
\begin{align}
	\cals{D}^{(i,j)}\left(
	\frac{z_b}{z_a}
	; q, t
	\right)
	:= 
	\begin{cases}
		\displaystyle 
		\frac{
			\left(1 - q_1^{-1}\frac{z_b}{z_a}\right)
			\left(1 - q_2^{-1}\frac{z_b}{z_a}\right)
		}{
			\left(1 - q_3\frac{z_b}{z_a}\right)
			\left(1 - \frac{z_b}{z_a}\right)
		}
		\hspace{0.3cm}
		&\text{ if } i < j
		\\
		\displaystyle 
		\frac{
			\left(1 - q_1^{-1}\frac{z_b}{z_a}\right)
			\left(1 - q_1\frac{z_b}{z_a}\right)
		}{
			\left(1 - q_3^{-1}\frac{z_b}{z_a}\right)
			\left(1 - q_3\frac{z_b}{z_a}\right)
		}
		\hspace{0.3cm}
		&\text{ if } i = j = \text{ super-number }
		\\
		1
		\hspace{0.3cm}
		&\text{ if } i = j = \text{ ordinary-number }
		\\
		\displaystyle
		\frac{
			\left(1 - q_1\frac{z_b}{z_a}\right)
			\left(1 - q_2\frac{z_b}{z_a}\right)
		}{
			\left(1 - q_3^{-1}\frac{z_b}{z_a}\right)
			\left(1 - \frac{z_b}{z_a}\right)
		}
		\hspace{0.3cm}
		&\text{ if } i > j 
	\end{cases}
	\label{eqn126-1300-4dec}
\end{align}
and 
\small 
\begin{align}
	\operatorname{RSSYBT}(N,M;\lambda)
	:= 
	\left\{
	(i_1,\dots,i_k) \in \{1,\dots,N+M\}^k
	\;\middle\vert\;
	\begin{array}{@{}l@{}}
		\begin{ytableau}
			i_1 & i_2  & \none[\dots] & i_{\lambda_1 - 1}
			& i_{\lambda_1}
			\\
			i_{\lambda_1 + 1} & i_{\lambda_1 + 2} &  \none[\dots]
			& i_{\lambda_1 + \lambda_2} \\
			\none[\vdots] & \none[\vdots]
			& \none[\vdots]
		\end{ytableau}
		\\
		\text{is a RSSYBT with shape $\lambda$}
	\end{array}
	\right\}. 
\label{eqn46-1534-28jan}
\end{align}
\normalsize
\end{lem}
\begin{proof}
The proof of this lemma is relegated to Appendix \ref{appA-1155-22jan}. 
\end{proof}

\begin{rem}
Note that since there is a natural bijection between the set $\operatorname{RSSYBT}(N,M;\lambda)$ defined in equation \eqref{eqn46-1534-28jan} and those defined in \textbf{Definition \ref{dfn211-1443}}, we decide to write this set using the same notation. 
It will be clear from the context whether the set $\operatorname{RSSYBT}(N,M;\lambda)$ that we are considering is the one in equation \eqref{eqn46-1534-28jan}  or the one in \textbf{Definition \ref{dfn211-1443}}. 
\end{rem}

From \textbf{Lemma \ref{lemm43-1156-22jan}} and the combinatorial formula for super Macdonald polynomials, \textbf{Proposition \ref{prp228-1020-29jan}}, we deduce that both the (LHS) and (RHS) of equation \eqref{eqn-42-1021-29jan} can be expressed as summations of elements of $\text{RSSYBT}(N,M;\lambda)$. Therefore, to prove \textbf{Theorem \ref{thm42-main-1022}}, it suffices to prove the following lemma:

\begin{lem}
\label{lemm44-1106-29jan}
Let $(N,M) \in (	\bb{Z}^{\geq 0}	\times \bb{Z}^{\geq 0})\backslash \{(0,0)\}$ and 
$\lambda \in \operatorname{Par}(k)$ where $\lambda \in H_{N,M}$. Then, for each $(i_1,\dots,i_k) \in \operatorname{RSSYBT}(N,M;\lambda)$ whose super number part has shape $\mu \in \operatorname{Par}$, we have 
\begin{align}
	&\lim_{\xi \rightarrow t^{-1}}\,\,
	(
	\dualmap
	\comp 
	\bigg|_{
		\substack{
			q_1 = q, \\
			q_2 = q^{-1}t,\\
			q_3 = t^{-1} \\
		}
	}
	)
	\bigg[
	\bigg(
	\frac{q_1^{\frac{1}{2}} - q_1^{-\frac{1}{2}}}{
		q_3^{\frac{1}{2}} - q_3^{-\frac{1}{2}}
	}
	\bigg)^{|\mu|}
	\cals{N}_\lambda(z_1,\dots,z_k)
	\prod_{1 \leq a < b \leq k}
	\cals{D}^{(i_a,i_b)}\left(
	\frac{z_b}{z_a}
	; q, t
	\right)
	\bigg]
	\label{eqn47-1211-1feb}
	\\
	&=
	\widetilde{\psi}_{T(i_1,\dots,i_k ; \lambda)}(q,t^{-1})
	\left(q^{-\frac{1}{2}}t^{-\frac{1}{2}}\right)^{|\mu|}.
	\notag 
\end{align}
Here 
\begin{align}
T(i_1,\dots,i_k ; \lambda) = 
\begin{ytableau}
	i_1 & i_2  & \none[\dots] & i_{\lambda_1 - 1}
	& i_{\lambda_1}
	\\
	i_{\lambda_1 + 1} & i_{\lambda_1 + 2} &  \none[\dots]
	& i_{\lambda_1 + \lambda_2} \\
	\none[\vdots] & \none[\vdots]
	& \none[\vdots]
\end{ytableau}
\end{align}
is the tableau obtained by arranging $(i_1,\dots,i_k) \in \{1,\dots,N+M\}^k$ into the form of a Young diagram $\lambda$. 
\end{lem}

The proof of this lemma will proceed by considering three separate cases:
\begin{enumerate}[(1)]
\item $(N,0)$ where $N \in \bb{Z}^{>0}$,
\item $(N,M)$ where $N, M \in \bb{Z}^{>0}$,
\item $(0,M)$ where $M \in \bb{Z}^{>0}$.
\end{enumerate}
A detailed proof of each case will be provided in the subsequent sections. 

\section{Proof of lemma \ref{lemm44-1106-29jan} for the case $(N,0)$ where $N \in \bb{Z}^{>0}$}
\label{sec5-1041-7apr}

In the case $(N,0)$ where $N \in \bb{Z}^{>0}$, equation \eqref{eqn47-1211-1feb} reduces to 
\begin{align}
	&\lim_{\xi \rightarrow t^{-1}}\,\,
	(
	\dualmap
	\comp 
	\bigg|_{
		\substack{
			q_1 = q, \\
			q_2 = q^{-1}t,\\
			q_3 = t^{-1} \\
		}
	}
	)
	\bigg[
	\cals{N}_\lambda(z_1,\dots,z_k)
	\prod_{1 \leq a < b \leq k}
	\cals{D}^{(i_a,i_b)}\left(
	\frac{z_b}{z_a}
	; q, t
	\right)
	\bigg]
	= \psi_{T(i_1,\dots,i_k ; \lambda)}(q,t). 
	\label{eqn126-1958}
\end{align}
This simplification arises because the case $(N,0)$ involves no super number. In other word, in this case, $\mu = \emptyset$. We will prove equation \eqref{eqn126-1958} using mathematical induction on $N$. 

\subsection{Basis step $N = 1$}
Consider $(i_1,\dots,i_k) \in \operatorname{RSSYBT}(1,0;\lambda)$. From \textbf{Definition \ref{dfn210-1443}}, 
it follows that $\lambda$ must have only one row. Consequently, $T(i_1,\dots,i_k ; \lambda)$ takes the form 
\begin{align}
T(i_1,\dots,i_k ; \lambda)
= 
T(1,\dots,1 ; (k))
= 
\underbrace{		
\begin{ytableau}
	1 & 1  & \none[\dots] & 1
	& 1
\end{ytableau}
}_{
k \text{ boxes }
}
\end{align}
According to equation \eqref{eqn32-1044}, we have $\cals{N}_{(k)}(z_1,\dots,z_k) = 1$. Furthermore, from \eqref{eqn126-1300-4dec}, we can show that for $T(1,\dots,1 ; (k))$, we have $\prod_{1 \leq a < b \leq k}
\cals{D}^{(i_a,i_b)}\left(
\frac{z_b}{z_a}
; q, t
\right) = 1$. 

On the other side, we know from \textbf{Definition \ref{dfn223-2118}} that 
\begin{align*}
\psi_{T(1,\dots,1 ; (k))}(q,t) = \psi_{(\lambda) / \emptyset} = 1. 
\end{align*}
Thus, we have shown that (LHS) of \eqref{eqn126-1958} = (RHS) of \eqref{eqn126-1958} = 1. Therefore, 
the basis step $N = 1$ of the induction is proved.

\subsection{Inductive step}

We assume that equation \eqref{eqn126-1958} holds for $1, \dots, N - 1$ where $N \geq 2$. Let $(i_1,\dots,i_k) \in \operatorname{RSSYBT}(N,0;\lambda)$. Without loss of generality, we can assume that
\begin{align}
\{1,\dots,N\} = \{i_1,\dots,i_k\}. 
\label{eqn53-1641-1feb}
\end{align}
meaning that all elements in the set $\{1,\dots,N\}$ appear in $(i_1,\dots,i_k)$. This assumption is justified because if only numbers in a proper subset of $\{1,\dots,N\}$ appear in $(i_1,…,i_k)$, we can directly apply the induction hypothesis in conjunction with \textbf{Proposition \ref{prp226-1326-15feb}}
to conclude that equation \eqref{eqn126-1958} holds.

Under the assumption stated in equation \eqref{eqn53-1641-1feb}, we know that the tableau $T(i_1,\dots,i_k ; \lambda)$ must contain a box with the number $1$. 

\begin{prop}
\label{prp51-1812-1feb}
The positions of the boxes containing the number $1$ in the tableau $T(i_1,\dots,i_k ; \lambda)$ must satisfy the following conditions:
\begin{enumerate}[(1)]
\item \label{con1-1135-3feb} The boxes containing the number $1$ must be located in the last row of a row interval. 
\item \label{con2-1135-3feb} There must be no boxes below the boxes containing the number $1$. 
\end{enumerate}
\end{prop}
\begin{proof}
Follow directly from \textbf{Definition \ref{dfn210-1443}}. 
\end{proof}

\begin{rem}
It is important to note that the position of the box containing the number $1$ must satisfy conditions \eqref{con1-1135-3feb} and \eqref{con2-1135-3feb} stated in \textbf{Proposition \ref{prp51-1812-1feb}}. However,
the converse is not necessarily true; that is, a box satisfying conditions 
\eqref{con1-1135-3feb} and \eqref{con2-1135-3feb} is not necessarily the box containing the number $1$. 
\end{rem}

The Young tableau in figure \ref{fig1-1543-15apr} illustrates \textbf{Proposition \ref{prp51-1812-1feb}}.
The boxes containing the number $1$ can only appear in the positions highlighted in yellow. 

\begin{figure}[!h]
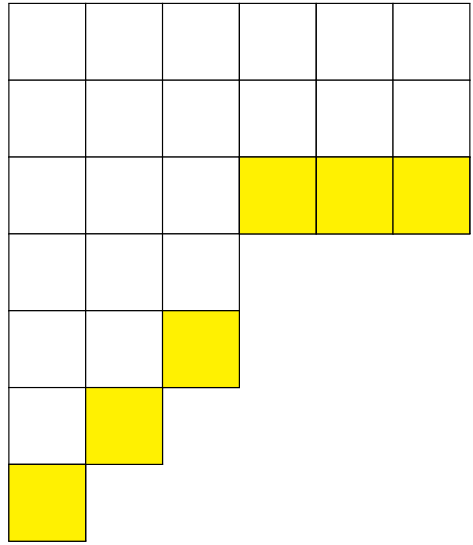

	\centering
	\small 
	\begin{tabular}{r@{}l}
		\ydiagram[*(yellow)]
		{6+0,6+0,3+3,3+0,2+1,1+1,0+1}
		*[*(white)]{6,6,6,3,3,2,1}
	\end{tabular}
	\caption{The positions highlighted in yellow represent the allowed positions for the boxes containing the number $1$. }
	\label{fig1-1543-15apr}
\end{figure}
\normalsize

\begin{cor}
\label{coro52-1feb-2026}
Let $\widetilde{	T		}(i_1,\dots,i_k ; \widetilde{\lambda})$ be the tableau with shape $\widetilde{\lambda}$ obtained by removing all boxes containing the number $1$
from $T(i_1,\dots,i_k ; \lambda)$. Then, $\widetilde{	T		}(i_1,\dots,i_k ; \widetilde{\lambda})$ is also a RSSYT. 
\end{cor}
\begin{proof}
This corollary is a direct consequence of \textbf{Proposition \ref{prp51-1812-1feb}}. 
\end{proof}

It is important to note that \textbf{Corollary \ref{coro52-1feb-2026}} tells us particularly that we can talk about $\psi_{\widetilde{	T		}(i_1,\dots,i_k ; \widetilde{\lambda})}(q,t)$. 

\begin{lem}
\label{lemm53-1415-2feb}
Under the induction hypothesis, for $(i_1,\dots,i_k) \in \operatorname{RSSYBT}(N,0;\lambda)$ satisfying the assumption in equation \eqref{eqn53-1641-1feb}, we have 
\begin{align}
&\lim_{\xi \rightarrow t^{-1}}\,\,
(
\dualmap
\comp 
\bigg|_{
	\substack{
		q_1 = q, \\
		q_2 = q^{-1}t,\\
		q_3 = t^{-1} \\
	}
}
)
\bigg[
\cals{N}_\lambda(z_1,\dots,z_k)
\prod_{1 \leq a < b \leq k}
\cals{D}^{(i_a,i_b)}\left(
\frac{z_b}{z_a}
; q, t
\right)
\bigg]
= \psi_{T(i_1,\dots,i_k ; \lambda)}(q,t). 
\label{eqn54-1054-2feb}
\end{align}
if and only if 

\tiny  
\begin{align}
	&\psi_{\lambda^{(1)}/\lambda^{(2)}}(q,t)
	\label{eqn55-1407-2feb}
	\\
	&= 
	\left(
	\lim_{\xi \rightarrow t^{-1}}\,\,
	\dualmap
	\right)
	\Bigg[
	\notag \\ 
	&\underbrace{		\prod_{1 \leq \alpha < \beta \leq \ell(\lambda)}				}_{
		\substack{
		\alpha \in \operatorname{Row}(T(i_1,\dots,i_k ; \lambda)|1)
			\\
			\beta \in \operatorname{Row}(T(i_1,\dots,i_k ; \lambda)|1)
		}
	}
	\prod_{
		\substack{
		a \in I^{(\alpha)}
			\\
		b \in I^{(\beta)}
			\\
		a \text{ is a box with number $1$}
			\\
	b \text{ is a box with number $1$}
		}
	}
	\frac{
		\left(1 - \frac{z_b}{z_a}\right)
		\left(1 - t\frac{z_b}{z_a}\right)
	}{
		\left(1 - q\frac{z_b}{z_a}\right)
		\left(1 - q^{-1}t\frac{z_b}{z_a}\right)
	}
	\times 
	\underbrace{		\prod_{1 \leq \alpha < \beta \leq \ell(\lambda)}				}_{
		\substack{
	\alpha \in \operatorname{Row}(T(i_1,\dots,i_k ; \lambda)|1)
			\\
	\beta \in \operatorname{Row}(T(i_1,\dots,i_k ; \lambda)|1)
		}
	}
	\prod_{
		\substack{
			a \in I^{(\alpha)}
			\\
			b \in I^{(\beta)}
			\\
			a \text{ is a box with number $1$}
			\\
			b \text{ is \textbf{not} a box with number $1$}
		}
	}
	\frac{
		\left(1 - \frac{z_b}{z_a}\right)
		\left(1 - t\frac{z_b}{z_a}\right)
	}{
		\left(1 - q\frac{z_b}{z_a}\right)
		\left(1 - q^{-1}t\frac{z_b}{z_a}\right)
	}
	\notag \\
	&
	\times
	\underbrace{			\prod_{1 \leq \alpha < \beta \leq \ell(\lambda)}					}_{
		\substack{
			\alpha \in \operatorname{Row}(T(i_1,\dots,i_k ; \lambda)|1)
			\\
			\beta \notin  \operatorname{Row}(T(i_1,\dots,i_k ; \lambda)|1)
		}
	}
	\prod_{
		\substack{
			a \in I^{(\alpha)}
			\\
			b \in I^{(\beta)}
			\\
			a \text{ is a box with number $1$}
			\\
			b \text{ is \textbf{not} a box with number $1$}
		}
	}
	\frac{
		\left(1 - \frac{z_b}{z_a}\right)
		\left(1 - t\frac{z_b}{z_a}\right)
	}{
		\left(1 - q\frac{z_b}{z_a}\right)
		\left(1 - q^{-1}t\frac{z_b}{z_a}\right)
	}
	\times
	\underbrace{				\prod_{1 \leq a < b \leq k}							}_{
		\substack{
			a \text{ is \textbf{not} a box with number $1$}
			\\
			b \text{ is a box with number $1$}
			\\
			a \text{ and } b \text{ are in the same row}
		}
	}
	\frac{
		\left(1 - q\frac{z_b}{z_a}\right)
		\left(1 - q^{-1}t\frac{z_b}{z_a}\right)
	}{
		\left(1 - t\frac{z_b}{z_a}\right)
		\left(1 - \frac{z_b}{z_a}\right)
	}
	\notag \\
	&\times\underbrace{			\prod_{1 \leq a < b \leq k}					}_{
		\substack{
			a \text{ is a box with number $1$}
			\\
			b \text{ is \textbf{not} a box with number $1$}
		}
	}
	\frac{
		\left(1 - q^{-1}\frac{z_b}{z_a}\right)
		\left(1 - qt^{-1}\frac{z_b}{z_a}\right)
	}{
		\left(1 - t^{-1}\frac{z_b}{z_a}\right)
		\left(1 - \frac{z_b}{z_a}\right)
	}
	\Bigg]
	\notag 
\end{align}
\normalsize
where $\operatorname{Row}(T(i_1,\dots,i_k ; \lambda)|1)$ denotes the set of rows in the tableau $T(i_1,\dots,i_k ; \lambda)$ containing the boxes with number $1$ (as defined in \textbf{Definition \ref{def220-1540-19mar}}), and the sets $I^{(\alpha)} \, (\alpha = 1 \dots, \ell(\lambda))$ are as defined in \textbf{Theorem \ref{thm42-main-1022}}. 
\end{lem}

\begin{proof}
First note that 

\tiny
\begin{align}
	&\left(
	\bigg|_{
		\substack{
			q_1 = q 
			\\
			q_2 = q^{-1}t
			\\
			q_3 = t^{-1}
		}
	}
	\right)
	\bigg[
	\cals{N}_{\lambda}(z_1,\dots,z_k )
	\times
	\prod_{1 \leq a < b \leq k}
	\cals{D}^{(i_a,i_b)}\left(
	\frac{z_b}{z_a}
	; q, t
	\right)
	\bigg]
	\\
	&= 
	\Biggl[
	\underbrace{		\prod_{1 \leq \alpha < \beta \leq \ell(\lambda)}				}_{
		\substack{
			\alpha \in \operatorname{Row}(T(i_1,\dots,i_k ; \lambda)|1)
			\\
			\beta \in \operatorname{Row}(T(i_1,\dots,i_k ; \lambda)|1)
		}
	}
	\prod_{
		\substack{
			a \in I^{(\alpha)}
			\\
			b \in I^{(\beta)}
			\\
			a \text{ is a box with number $1$}
			\\
			b \text{ is a box with number $1$}
		}
	}
	\times 
	\underbrace{		\prod_{1 \leq \alpha < \beta \leq \ell(\lambda)}				}_{
		\substack{
			\alpha \in \operatorname{Row}(T(i_1,\dots,i_k ; \lambda)|1)
			\\
			\beta \in \operatorname{Row}(T(i_1,\dots,i_k ; \lambda)|1)
		}
	}
	\prod_{
		\substack{
			a \in I^{(\alpha)}
			\\
			b \in I^{(\beta)}
			\\
			a \text{ is a box with number $1$}
			\\
			b \text{ is \textbf{not} a box with number $1$}
		}
	}
	\notag 
	\\
	&\times 
	\underbrace{		\prod_{1 \leq \alpha < \beta \leq \ell(\lambda)}				}_{
		\substack{
			\alpha \in \operatorname{Row}(T(i_1,\dots,i_k ; \lambda)|1)
			\\
			\beta \in \operatorname{Row}(T(i_1,\dots,i_k ; \lambda)|1)
		}
	}
	\prod_{
		\substack{
			a \in I^{(\alpha)}
			\\
			b \in I^{(\beta)}
			\\
			a \text{ is \textbf{not} a box with number $1$}
			\\
			b \text{ is a box with number $1$}
		}
	}
	\times 
	\underbrace{		\prod_{1 \leq \alpha < \beta \leq \ell(\lambda)}				}_{
		\substack{
			\alpha \in \operatorname{Row}(T(i_1,\dots,i_k ; \lambda)|1)
			\\
			\beta \in \operatorname{Row}(T(i_1,\dots,i_k ; \lambda)|1)
		}
	}
	\prod_{
		\substack{
			a \in I^{(\alpha)}
			\\
			b \in I^{(\beta)}
			\\
			a \text{ is \textbf{not} a box with number $1$}
			\\
			b \text{ is \textbf{not} a box with number $1$}
		}
	}
	\notag 
	\\
	&\times
	\underbrace{			\prod_{1 \leq \alpha < \beta \leq \ell(\lambda)}					}_{
		\substack{
			\alpha \in \operatorname{Row}(T(i_1,\dots,i_k ; \lambda)|1)
			\\
			\beta \notin \operatorname{Row}(T(i_1,\dots,i_k ; \lambda)|1)
		}
	}
	\prod_{
		\substack{
			a \in I^{(\alpha)}
			\\
			b \in I^{(\beta)}
			\\
			a \text{ is a box with number $1$}
			\\
			b \text{ is \textbf{not} a box with number $1$}
		}
	}
	\times
	\underbrace{			\prod_{1 \leq \alpha < \beta \leq \ell(\lambda)}					}_{
		\substack{
			\alpha \in \operatorname{Row}(T(i_1,\dots,i_k ; \lambda)|1)
			\\
			\beta \notin \operatorname{Row}(T(i_1,\dots,i_k ; \lambda)|1)
		}
	}
	\prod_{
		\substack{
			a \in I^{(\alpha)}
			\\
			b \in I^{(\beta)}
			\\
			a \text{ is \textbf{not} a box with number $1$}
			\\
			b \text{ is \textbf{not} a box with number $1$}
		}
	}
	\notag 
	\\
	&\times 
	\underbrace{		\prod_{1 \leq \alpha < \beta \leq \ell(\lambda)}				}_{
		\substack{
			\alpha \notin \operatorname{Row}(T(i_1,\dots,i_k ; \lambda)|1)
			\\
			\beta \in \operatorname{Row}(T(i_1,\dots,i_k ; \lambda)|1)
		}
	}
	\prod_{
		\substack{
			a \in I^{(\alpha)}
			\\
			b \in I^{(\beta)}
			\\
			a \text{ is \textbf{not} a box with number $1$}
			\\
			b \text{ is a box with number $1$}
		}
	}
	\times 
	\underbrace{		\prod_{1 \leq \alpha < \beta \leq \ell(\lambda)}				}_{
		\substack{
			\alpha \notin \operatorname{Row}(T(i_1,\dots,i_k ; \lambda)|1)
			\\
			\beta \in \operatorname{Row}(T(i_1,\dots,i_k ; \lambda)|1)
		}
	}
	\prod_{
		\substack{
			a \in I^{(\alpha)}
			\\
			b \in I^{(\beta)}
			\\
			a \text{ is \textbf{not} a box with number $1$}
			\\
			b \text{ is \textbf{not} a box with number $1$}
		}
	}
	\notag 
	\\
	&\times
	\underbrace{			\prod_{1 \leq \alpha < \beta \leq \ell(\lambda)}					}_{
		\substack{
			\alpha \notin \operatorname{Row}(T(i_1,\dots,i_k ; \lambda)|1)
			\\
			\beta \notin \operatorname{Row}(T(i_1,\dots,i_k ; \lambda)|1)
		}
	}
	\prod_{
		\substack{
			a \in I^{(\alpha)}
			\\
			b \in I^{(\beta)}
		}
	}
	\Biggr]
	\frac{
		\left(1 - \frac{z_b}{z_a}\right)
		\left(1 - t\frac{z_b}{z_a}\right)
	}{
		\left(1 - q\frac{z_b}{z_a}\right)
		\left(1 - q^{-1}t\frac{z_b}{z_a}\right)
	}
	\times\underbrace{			\prod_{1 \leq a < b \leq k}					}_{
		\substack{
			a \text{ is a box with number $1$}
			\\
			b \text{ is \textbf{not} a box with number $1$}
		}
	}
	\frac{
		\left(1 - q^{-1}\frac{z_b}{z_a}\right)
		\left(1 - qt^{-1}\frac{z_b}{z_a}\right)
	}{
		\left(1 - t^{-1}\frac{z_b}{z_a}\right)
		\left(1 - \frac{z_b}{z_a}\right)
	}
	\notag 
	\\
	&\times
	\underbrace{				\prod_{1 \leq a < b \leq k}							}_{
		\substack{
			a \text{ is \textbf{not} a box with number $1$}
			\\
			b \text{ is a box with number $1$}
		}
	}
	\frac{
		\left(1 - q\frac{z_b}{z_a}\right)
		\left(1 - q^{-1}t\frac{z_b}{z_a}\right)
	}{
		\left(1 - t\frac{z_b}{z_a}\right)
		\left(1 - \frac{z_b}{z_a}\right)
	}
	\times
	\underbrace{		\prod_{1 \leq a < b \leq k}					}_{
		\substack{
			a \text{ is \textbf{not} a box with number $1$}
			\\
			b \text{ is \textbf{not} a box with number $1$}
		}
	}
	\cals{D}^{(i_a,i_b)}\left(
	\frac{z_b}{z_a}
	; q, t
	\right). 
	\notag 
\end{align}
\normalsize
Therefore, by using the induction hypothesis, we obtain that 

\tiny
\begin{align}
&\lim_{\xi \rightarrow t^{-1}}\,\,
(
\dualmap
\comp 
\bigg|_{
	\substack{
		q_1 = q, \\
		q_2 = q^{-1}t,\\
		q_3 = t^{-1} \\
	}
}
)
\bigg[
\cals{N}_\lambda(z_1,\dots,z_k)
\prod_{1 \leq a < b \leq k}
\cals{D}^{(i_a,i_b)}\left(
\frac{z_b}{z_a}
; q, t
\right)
\bigg]
\\
&= 
\left(
\lim_{\xi \rightarrow t^{-1}}\,\,
\dualmap
\right)
\Bigg[
\Biggl\{
\underbrace{		\prod_{1 \leq \alpha < \beta \leq \ell(\lambda)}				}_{
	\substack{
		\alpha \in \operatorname{Row}(T(i_1,\dots,i_k ; \lambda)|1)
		\\
		\beta \in \operatorname{Row}(T(i_1,\dots,i_k ; \lambda)|1)
	}
}
\prod_{
	\substack{
		a \in I^{(\alpha)}
		\\
		b \in I^{(\beta)}
		\\
		a \text{ is a box with number $1$}
		\\
		b \text{ is a box with number $1$}
	}
}
\times 
\underbrace{		\prod_{1 \leq \alpha < \beta \leq \ell(\lambda)}				}_{
	\substack{
		\alpha \in \operatorname{Row}(T(i_1,\dots,i_k ; \lambda)|1)
		\\
		\beta \in \operatorname{Row}(T(i_1,\dots,i_k ; \lambda)|1)
	}
}
\prod_{
	\substack{
		a \in I^{(\alpha)}
		\\
		b \in I^{(\beta)}
		\\
		a \text{ is a box with number $1$}
		\\
		b \text{ is \textbf{not} a box with number $1$}
	}
}
\notag 
\\
&\times 
\underbrace{		\prod_{1 \leq \alpha < \beta \leq \ell(\lambda)}				}_{
	\substack{
		\alpha \in \operatorname{Row}(T(i_1,\dots,i_k ; \lambda)|1)
		\\
		\beta \in \operatorname{Row}(T(i_1,\dots,i_k ; \lambda)|1)
	}
}
\prod_{
	\substack{
		a \in I^{(\alpha)}
		\\
		b \in I^{(\beta)}
		\\
		a \text{ is \textbf{not} a box with number $1$}
		\\
		b \text{ is a box with number $1$}
	}
}
\times
\underbrace{			\prod_{1 \leq \alpha < \beta \leq \ell(\lambda)}					}_{
	\substack{
		\alpha \in \operatorname{Row}(T(i_1,\dots,i_k ; \lambda)|1)
		\\
		\beta \notin \operatorname{Row}(T(i_1,\dots,i_k ; \lambda)|1)
	}
}
\prod_{
	\substack{
		a \in I^{(\alpha)}
		\\
		b \in I^{(\beta)}
		\\
		a \text{ is a box with number $1$}
		\\
		b \text{ is \textbf{not} a box with number $1$}
	}
}
\notag \\
&\times 
\underbrace{		\prod_{1 \leq \alpha < \beta \leq \ell(\lambda)}				}_{
	\substack{
		\alpha \notin \operatorname{Row}(T(i_1,\dots,i_k ; \lambda)|1)
		\\
		\beta \in \operatorname{Row}(T(i_1,\dots,i_k ; \lambda)|1)
	}
}
\prod_{
	\substack{
		a \in I^{(\alpha)}
		\\
		b \in I^{(\beta)}
		\\
		a \text{ is \textbf{not} a box with number $1$}
		\\
		b \text{ is a box with number $1$}
	}
}
\Biggr\}
\frac{
	\left(1 - \frac{z_b}{z_a}\right)
	\left(1 - t\frac{z_b}{z_a}\right)
}{
	\left(1 - q\frac{z_b}{z_a}\right)
	\left(1 - q^{-1}t\frac{z_b}{z_a}\right)
}
\notag \\
&\times\underbrace{			\prod_{1 \leq a < b \leq k}					}_{
	\substack{
		a \text{ is a box with number $1$}
		\\
		b \text{ is \textbf{not} a box with number $1$}
	}
}
\frac{
	\left(1 - q^{-1}\frac{z_b}{z_a}\right)
	\left(1 - qt^{-1}\frac{z_b}{z_a}\right)
}{
	\left(1 - t^{-1}\frac{z_b}{z_a}\right)
	\left(1 - \frac{z_b}{z_a}\right)
}
\times
\underbrace{				\prod_{1 \leq a < b \leq k}							}_{
	\substack{
		a \text{ is \textbf{not} a box with number $1$}
		\\
		b \text{ is a box with number $1$}
	}
}
\frac{
	\left(1 - q\frac{z_b}{z_a}\right)
	\left(1 - q^{-1}t\frac{z_b}{z_a}\right)
}{
	\left(1 - t\frac{z_b}{z_a}\right)
	\left(1 - \frac{z_b}{z_a}\right)
}
\Bigg]
\times 
\psi_{\widetilde{	T		}(i_1,\dots,i_k ; \widetilde{\lambda})}(q,t). 
\notag
\end{align}
\normalsize
From this, we see that \eqref{eqn54-1054-2feb} holds 
if and only if 

\tiny 
\begin{align}
&\frac{\psi_{T(i_1,\dots,i_k ; \lambda)}(q,t)}{
\psi_{\widetilde{	T		}(i_1,\dots,i_k ; \widetilde{\lambda})}(q,t)
}
\label{eqn57-1406-2feb}
\\
&= 
\left(
\lim_{\xi \rightarrow t^{-1}}\,\,
\dualmap
\right)
\Bigg[
\Biggl\{
\underbrace{		\prod_{1 \leq \alpha < \beta \leq \ell(\lambda)}				}_{
	\substack{
		\alpha \in \operatorname{Row}(T(i_1,\dots,i_k ; \lambda)|1)
		\\
		\beta \in \operatorname{Row}(T(i_1,\dots,i_k ; \lambda)|1)
	}
}
\prod_{
	\substack{
		a \in I^{(\alpha)}
		\\
		b \in I^{(\beta)}
		\\
		a \text{ is a box with number $1$}
		\\
		b \text{ is a box with number $1$}
	}
}
\times 
\underbrace{		\prod_{1 \leq \alpha < \beta \leq \ell(\lambda)}				}_{
	\substack{
		\alpha \in \operatorname{Row}(T(i_1,\dots,i_k ; \lambda)|1)
		\\
		\beta \in \operatorname{Row}(T(i_1,\dots,i_k ; \lambda)|1)
	}
}
\prod_{
	\substack{
		a \in I^{(\alpha)}
		\\
		b \in I^{(\beta)}
		\\
		a \text{ is a box with number $1$}
		\\
		b \text{ is \textbf{not} a box with number $1$}
	}
}
\notag 
\\
&
\times
\underbrace{			\prod_{1 \leq \alpha < \beta \leq \ell(\lambda)}					}_{
	\substack{
		\alpha \in \operatorname{Row}(T(i_1,\dots,i_k ; \lambda)|1)
		\\
		\beta \notin \operatorname{Row}(T(i_1,\dots,i_k ; \lambda)|1)
	}
}
\prod_{
	\substack{
		a \in I^{(\alpha)}
		\\
		b \in I^{(\beta)}
		\\
		a \text{ is a box with number $1$}
		\\
		b \text{ is \textbf{not} a box with number $1$}
	}
}
\Biggr\}
\frac{
	\left(1 - \frac{z_b}{z_a}\right)
	\left(1 - t\frac{z_b}{z_a}\right)
}{
	\left(1 - q\frac{z_b}{z_a}\right)
	\left(1 - q^{-1}t\frac{z_b}{z_a}\right)
}
\notag \\
&\times\underbrace{			\prod_{1 \leq a < b \leq k}					}_{
	\substack{
		a \text{ is a box with number $1$}
		\\
		b \text{ is \textbf{not} a box with number $1$}
	}
}
\frac{
	\left(1 - q^{-1}\frac{z_b}{z_a}\right)
	\left(1 - qt^{-1}\frac{z_b}{z_a}\right)
}{
	\left(1 - t^{-1}\frac{z_b}{z_a}\right)
	\left(1 - \frac{z_b}{z_a}\right)
}
\times
\underbrace{				\prod_{1 \leq a < b \leq k}							}_{
	\substack{
		a \text{ is \textbf{not} a box with number $1$}
		\\
		b \text{ is a box with number $1$}
		\\
		a \text{ and } b \text{ are in the same rows }
	}
}
\frac{
	\left(1 - q\frac{z_b}{z_a}\right)
	\left(1 - q^{-1}t\frac{z_b}{z_a}\right)
}{
	\left(1 - t\frac{z_b}{z_a}\right)
	\left(1 - \frac{z_b}{z_a}\right)
}
\Bigg] 
\notag
\end{align}
\normalsize
Here we have used the fact that 

\tiny
\begin{align}
&\Biggl\{
\underbrace{		\prod_{1 \leq \alpha < \beta \leq \ell(\lambda)}				}_{
	\substack{
		\alpha \in \operatorname{Row}(T(i_1,\dots,i_k ; \lambda)|1)
		\\
		\beta \in \operatorname{Row}(T(i_1,\dots,i_k ; \lambda)|1)
	}
}
\prod_{
	\substack{
		a \in I^{(\alpha)}
		\\
		b \in I^{(\beta)}
		\\
		a \text{ is \textbf{not} a box with number $1$}
		\\
		b \text{ is a box with number $1$}
	}
}
\times
\underbrace{		\prod_{1 \leq \alpha < \beta \leq \ell(\lambda)}				}_{
	\substack{
		\alpha \notin \operatorname{Row}(T(i_1,\dots,i_k ; \lambda)|1)
		\\
		\beta \in \operatorname{Row}(T(i_1,\dots,i_k ; \lambda)|1)
	}
}
\prod_{
	\substack{
		a \in I^{(\alpha)}
		\\
		b \in I^{(\beta)}
		\\
		a \text{ is \textbf{not} a box with number $1$}
		\\
		b \text{ is a box with number $1$}
	}
}
\Biggr\}
\frac{
	\left(1 - \frac{z_b}{z_a}\right)
	\left(1 - t\frac{z_b}{z_a}\right)
}{
	\left(1 - q\frac{z_b}{z_a}\right)
	\left(1 - q^{-1}t\frac{z_b}{z_a}\right)
}
\\
&\times
\underbrace{				\prod_{1 \leq a < b \leq k}							}_{
	\substack{
		a \text{ is \textbf{not} a box with number $1$}
		\\
		b \text{ is a box with number $1$}
	}
}
\frac{
	\left(1 - q\frac{z_b}{z_a}\right)
	\left(1 - q^{-1}t\frac{z_b}{z_a}\right)
}{
	\left(1 - t\frac{z_b}{z_a}\right)
	\left(1 - \frac{z_b}{z_a}\right)
}
\notag 
\\
&= 
\underbrace{				\prod_{1 \leq a < b \leq k}							}_{
	\substack{
		a \text{ is \textbf{not} a box with number $1$}
		\\
		b \text{ is a box with number $1$}
		\\
		a \text{ and } b \text{ are in the same rows }
	}
}
\frac{
	\left(1 - q\frac{z_b}{z_a}\right)
	\left(1 - q^{-1}t\frac{z_b}{z_a}\right)
}{
	\left(1 - t\frac{z_b}{z_a}\right)
	\left(1 - \frac{z_b}{z_a}\right)
}
\notag 
\end{align}
\normalsize

From \textbf{Definition \ref{dfn223-2118}}, we know that 
\begin{align}
	\psi_{T(i_1,\dots,i_k ; \lambda)}(q,t) &=
	\psi_{\lambda^{(1)}/\lambda^{(2)}}(q,t)
	\cdots
	\psi_{\lambda^{(N-2)}/\lambda^{(N-1)}}(q,t)
	\psi_{\lambda^{(N-1)}/\lambda^{(N)}}(q,t), 
	\\
	\psi_{\widetilde{	T		}(i_1,\dots,i_k ; \widetilde{\lambda})}(q,t) &= 
	\psi_{\lambda^{(2)}/\lambda^{(3)}}(q,t)
	\cdots
	\psi_{\lambda^{(N-2)}/\lambda^{(N-1)}}(q,t)
	\psi_{\lambda^{(N-1)}/\lambda^{(N)}}(q,t). 
\end{align}
Thus, we get that 
\begin{align}
&\frac{\psi_{T(i_1,\dots,i_k ; \lambda)}(q,t)}{
	\psi_{\widetilde{	T		}(i_1,\dots,i_k ; \widetilde{\lambda})}(q,t)
}
= 
\psi_{\lambda^{(1)}/\lambda^{(2)}}(q,t). 
\label{eqn511-1404-2feb}
\end{align}
Substituting \eqref{eqn511-1404-2feb} into \eqref{eqn57-1406-2feb}, we then obtain \eqref{eqn55-1407-2feb}. So we have proved the lemma. 
\end{proof}

\begin{lem}
\label{lemm54-1412-2feb}
Under the induction hypothesis, for $(i_1,\dots,i_k) \in \operatorname{RSSYBT}(N,0;\lambda)$ satisfying the assumption in equation \eqref{eqn53-1641-1feb}, we have 

\tiny  
\begin{align*}
	&\psi_{\lambda^{(1)}/\lambda^{(2)}}(q,t)
	\\
	&= 
	\left(
	\lim_{\xi \rightarrow t^{-1}}\,\,
	\dualmap
	\right)
	\Bigg[
	\notag \\ 
	&\underbrace{		\prod_{1 \leq \alpha < \beta \leq \ell(\lambda)}				}_{
		\substack{
			\alpha \in \operatorname{Row}(T(i_1,\dots,i_k ; \lambda)|1)
			\\
			\beta \in \operatorname{Row}(T(i_1,\dots,i_k ; \lambda)|1)
		}
	}
	\prod_{
		\substack{
			a \in I^{(\alpha)}
			\\
			b \in I^{(\beta)}
			\\
			a \text{ is a box with number $1$}
			\\
			b \text{ is a box with number $1$}
		}
	}
	\frac{
		\left(1 - \frac{z_b}{z_a}\right)
		\left(1 - t\frac{z_b}{z_a}\right)
	}{
		\left(1 - q\frac{z_b}{z_a}\right)
		\left(1 - q^{-1}t\frac{z_b}{z_a}\right)
	}
	\times 
	\underbrace{		\prod_{1 \leq \alpha < \beta \leq \ell(\lambda)}				}_{
		\substack{
			\alpha \in \operatorname{Row}(T(i_1,\dots,i_k ; \lambda)|1)
			\\
			\beta \in \operatorname{Row}(T(i_1,\dots,i_k ; \lambda)|1)
		}
	}
	\prod_{
		\substack{
			a \in I^{(\alpha)}
			\\
			b \in I^{(\beta)}
			\\
			a \text{ is a box with number $1$}
			\\
			b \text{ is \textbf{not} a box with number $1$}
		}
	}
	\frac{
		\left(1 - \frac{z_b}{z_a}\right)
		\left(1 - t\frac{z_b}{z_a}\right)
	}{
		\left(1 - q\frac{z_b}{z_a}\right)
		\left(1 - q^{-1}t\frac{z_b}{z_a}\right)
	}
	\notag \\
	&
	\times
	\underbrace{			\prod_{1 \leq \alpha < \beta \leq \ell(\lambda)}					}_{
		\substack{
			\alpha \in \operatorname{Row}(T(i_1,\dots,i_k ; \lambda)|1)
			\\
			\beta \notin  \operatorname{Row}(T(i_1,\dots,i_k ; \lambda)|1)
		}
	}
	\prod_{
		\substack{
			a \in I^{(\alpha)}
			\\
			b \in I^{(\beta)}
			\\
			a \text{ is a box with number $1$}
			\\
			b \text{ is \textbf{not} a box with number $1$}
		}
	}
	\frac{
		\left(1 - \frac{z_b}{z_a}\right)
		\left(1 - t\frac{z_b}{z_a}\right)
	}{
		\left(1 - q\frac{z_b}{z_a}\right)
		\left(1 - q^{-1}t\frac{z_b}{z_a}\right)
	}
	\times
	\underbrace{				\prod_{1 \leq a < b \leq k}							}_{
		\substack{
			a \text{ is \textbf{not} a box with number $1$}
			\\
			b \text{ is a box with number $1$}
			\\
			a \text{ and } b \text{ are in the same row}
		}
	}
	\frac{
		\left(1 - q\frac{z_b}{z_a}\right)
		\left(1 - q^{-1}t\frac{z_b}{z_a}\right)
	}{
		\left(1 - t\frac{z_b}{z_a}\right)
		\left(1 - \frac{z_b}{z_a}\right)
	}
	\notag \\
	&\times\underbrace{			\prod_{1 \leq a < b \leq k}					}_{
		\substack{
			a \text{ is a box with number $1$}
			\\
			b \text{ is \textbf{not} a box with number $1$}
		}
	}
	\frac{
		\left(1 - q^{-1}\frac{z_b}{z_a}\right)
		\left(1 - qt^{-1}\frac{z_b}{z_a}\right)
	}{
		\left(1 - t^{-1}\frac{z_b}{z_a}\right)
		\left(1 - \frac{z_b}{z_a}\right)
	}
	\Bigg]
\end{align*}
\normalsize
\end{lem}

\begin{proof}
The proof of this lemma is relegated to Appendix \ref{secappB-1410-2feb}. 
\end{proof}

By applying \textbf{Lemmas \ref{lemm53-1415-2feb}} and \textbf{\ref{lemm54-1412-2feb}}, the inductive step is proved. Therefore, the induction is complete.

\section{Proof of lemma \ref{lemm44-1106-29jan} for the case $(N,M)$ where $N, M \in \bb{Z}^{>0}$}
\label{sec6-2226-13mar}

In this section, we will prove the \textbf{Lemma \ref{lemm44-1106-29jan}} for the case $(N,M)$ where $N, M \in \bb{Z}^{>0}$. The proof will proceed by induction on $M$. 

\subsection{Basis step $M = 1$}

Consider $(i_1,\dots,i_k) \in \operatorname{RSSYBT}(N,1;\lambda)$. If $N+1 \notin \{	i_1,\dots,i_k		\}$, then 
the result from the case $(N,0)$ can be applied to conclude that \textbf{Lemma \ref{lemm44-1106-29jan}} holds. Therefore, we assume that  $N+1 \in \{	i_1,\dots,i_k		\}$ and there are exactly $c$ boxes in $T(i_1,\dots,i_k ; \lambda)$ labeled by $N+1$. 

Since $(i_1,\dots,i_k) \in \operatorname{RSSYBT}(N,1;\lambda)$, $N+1$ is the only super number. Consequently, 
the boxes labeled by $N+1$ must be in the first column, as illustrated in the figure \ref{fig1-1303-10mar} below. 

\begin{figure}[!h]
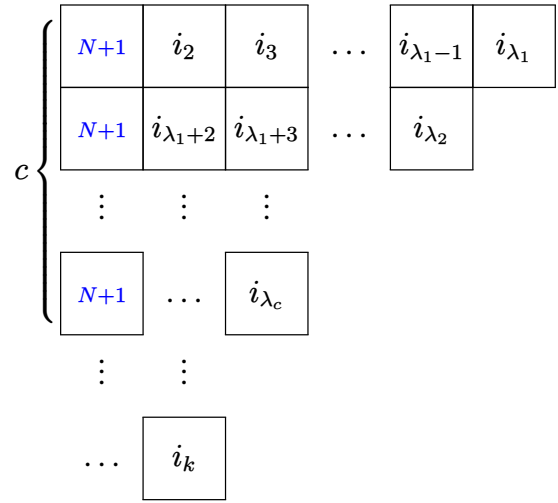

	\centering
	\begin{tabular}{r@{}l}
		\raisebox{5.9ex}{$c\left\{\vphantom{\begin{array}{c}~\\[17ex] ~
			\end{array}}\right.$} &
		\begin{ytableau}
		\scriptstyle	\textcolor{blue}{N+1} & i_2 & i_3 & \none[\dots]
			& i_{\lambda_1 - 1} & i_{\lambda_1} \\
		\scriptstyle	\textcolor{blue}{N+1} & i_{\lambda_1 + 2} & i_{\lambda_1 + 3} & \none[\dots]
			& i_{\lambda_2} \\
			\none[\vdots] & \none[\vdots]
			& \none[\vdots] \\
			\scriptstyle	\textcolor{blue}{N+1} & \none[\dots] & i_{\lambda_c} \\
				\none[\vdots] & \none[\vdots] \\
				\none[\dots] & i_{k} 
		\end{ytableau}
	\end{tabular}
	\caption{$T(i_1,\dots,i_k ; \lambda) \in \operatorname{RSSYBT}(N,1;\lambda)$ with $c$ boxes labeled by $N+1$.}
\label{fig1-1303-10mar}
\end{figure}
That is, 
\begin{align}
	(i_1,\dots,i_k)
	= (\textcolor{blue}{N+1},i_2,\dots,i_{\lambda_1},
	\textcolor{blue}{N+1},i_{\lambda_1 + 2}, \cdots, i_{\lambda_1 + \lambda_2}, \cdots,
	\textcolor{blue}{N+1},i_{\lambda_1 + \cdots + \lambda_{c-1} +2}, \cdots, i_k). 
	\label{eqn21-1433-6dec}
\end{align}

Next, we construct $(i^\prime_1,\dots,i^\prime_k) \in \operatorname{RSSYBT}(N+c,0;\lambda)$ by replacing all $\textcolor{blue}{N+1}$ with $N+1,N+2,\dots,N+c$, as illustrated in the figure \ref{fig2-1303-10mar} below. 

\begin{figure}[!h]
	\centering
	\begin{tabular}{r@{}l}
		\raisebox{5.9ex}{$c\left\{\vphantom{\begin{array}{c}~\\[17ex] ~
			\end{array}}\right.$} &
		\begin{ytableau}
			\scriptstyle N+c & i_2 & i_3 & \none[\dots]
			& i_{\lambda_1 - 1} & i_{\lambda_1} \\
			\scriptstyle	\scriptstyle N+c-1 & i_{\lambda_1 + 2} & i_{\lambda_1 + 3} & \none[\dots]
			& i_{\lambda_2} \\
			\none[\vdots] & \none[\vdots]
			& \none[\vdots] \\
			\scriptstyle	N+1 & \none[\dots] & i_{\lambda_c} \\
			\none[\vdots] & \none[\vdots] \\
			\none[\dots] & i_{k} 
		\end{ytableau}
	\end{tabular}
	\caption{$T(i^\prime_1,\dots,i^\prime_k ; \lambda)$ obtained by replacing all $\textcolor{blue}{N+1}$ in $T(i_1,\dots,i_k ; \lambda)$
	with $N+1,N+2,\dots,N+c$.}
	\label{fig2-1303-10mar}
\end{figure}
So, we get that 
\begin{align}
	(i^\prime_1,\dots,i^\prime_k)
	= (N+c,i_2,\dots,i_{\lambda_1},
N+c-1,i_{\lambda_1 + 2}, \cdots, i_{\lambda_1 + \lambda_2}, \cdots,
N+1,i_{\lambda_1 + \cdots + \lambda_{c-1} +2}, \cdots, i_k). 
\label{eqn62-1517-10mar}
\end{align}

\begin{lem}\mbox{}
\label{lem61-1347-10mar}
\begin{align}
	&\bigg|_{
		\substack{
			q_1 = q, \\
			q_2 = q^{-1}t,\\
			q_3 = t^{-1} \\
		}
	}
	\prod_{1 \leq a < b \leq k}					
	\cals{D}^{(i_a,i_b)}\left(
	\frac{z_b}{z_a}
	; q, t
	\right)
	\\
	&= 
	\underbrace{		\prod_{1 \leq a < b \leq k}					}_{
		\substack{
		a, b \in \{
		1,\lambda_1 + 1, \dots, \sum_{j = 1}^{c-1}\lambda_j + 1
		\}
		}
	}
	\frac{
		\left(1 - q^{-1}\frac{z_{b}}{z_a}\right)
		\left(1 - \frac{z_{b}}{z_a}\right)
	}{
		\left(1 - t^{-1}\frac{z_{b}}{z_a}\right)
		\left(1 - q^{-1}t\frac{z_{b}}{z_a}\right)
	}
	\times
	\bigg|_{
		\substack{
			q_1 = q, \\
			q_2 = q^{-1}t,\\
			q_3 = t^{-1} \\
		}
	}
	\prod_{1 \leq a < b \leq k}					
	\cals{D}^{(i^\prime_a,i^\prime_b)}\left(
	\frac{z_b}{z_a}
	; q, t
	\right)
\end{align}
\end{lem}
\begin{proof}
Straightforward calculation. 
\end{proof}

\begin{cor}\mbox{}
\label{cor62-1538-10mar}
\begin{align}
	&\lim_{\xi \rightarrow t^{-1}}\,\,
	(
	\dualmap
	\comp 
	\bigg|_{
		\substack{
			q_1 = q, \\
			q_2 = q^{-1}t,\\
			q_3 = t^{-1} \\
		}
	}
	)
	\bigg[
	\bigg(
	\frac{q_1^{\frac{1}{2}} - q_1^{-\frac{1}{2}}}{
		q_3^{\frac{1}{2}} - q_3^{-\frac{1}{2}}
	}
	\bigg)^{|\mu|}
	\cals{N}_\lambda(z_1,\dots,z_k)
	\prod_{1 \leq a < b \leq k}
	\cals{D}^{(i_a,i_b)}\left(
	\frac{z_b}{z_a}
	; q, t
	\right)
	\bigg]
	\label{eqn215-1512-6dec}
	\\
	&=
	\bigg(
	\frac{q^{\frac{1}{2}} - q^{-\frac{1}{2}}}{
		t^{-\frac{1}{2}} - t^{\frac{1}{2}}
	}
	\bigg)^{c}
	\times
	\prod_{a = 1}^{c-1}
	\Bigg[
	\left(
	\frac{(1 - q^{-1}t^{-a})(1 - t^{-a})}{(1 - t^{-1}t^{-a})(1 - q^{-1}tt^{-a})}
	\right)^{c-a}
	\Bigg]
	\notag \\	
	&\hspace{0.3cm}\times 
	\psi_{T(i^\prime_1,\dots,i^\prime_k ; \lambda)^{(1)}/T(i^\prime_1,\dots,i^\prime_k ; \lambda)^{(2)}}(q,t)
	\cdots
	\psi_{T(i^\prime_1,\dots,i^\prime_k ; \lambda)^{(N-1)}/T(i^\prime_1,\dots,i^\prime_k ; \lambda)^{(N)}}(q,t)
	\notag	\\
	&\hspace{0.3cm}\times 
	\psi_{T(i^\prime_1,\dots,i^\prime_k ; \lambda)^{(N)}/T(i^\prime_1,\dots,i^\prime_k ; \lambda)^{(N + 1)}}(q,t)
	\cdots
	\psi_{T(i^\prime_1,\dots,i^\prime_k ; \lambda)^{(N+c-1)}/T(i^\prime_1,\dots,i^\prime_k ; \lambda)^{(N+c)}}(q,t)
	\notag 
\end{align}
\end{cor}
\begin{proof}
By using \textbf{Lemma \ref{lem61-1347-10mar}}, we obtain that 
\begin{align}
	&\lim_{\xi \rightarrow t^{-1}}\,\,
	(
	\dualmap
	\comp 
	\bigg|_{
		\substack{
			q_1 = q, \\
			q_2 = q^{-1}t,\\
			q_3 = t^{-1} \\
		}
	}
	)
	\bigg[
	\bigg(
	\frac{q_1^{\frac{1}{2}} - q_1^{-\frac{1}{2}}}{
		q_3^{\frac{1}{2}} - q_3^{-\frac{1}{2}}
	}
	\bigg)^{|\mu|}
	\cals{N}_\lambda(z_1,\dots,z_k)
	\prod_{1 \leq a < b \leq k}
	\cals{D}^{(i_a,i_b)}\left(
	\frac{z_b}{z_a}
	; q, t
	\right)
	\bigg]
	\label{eqn65-1450-10mar}
	\\
	&=
	\bigg(
	\frac{q^{\frac{1}{2}} - q^{-\frac{1}{2}}}{
		t^{-\frac{1}{2}} - t^{\frac{1}{2}}
	}
	\bigg)^{c}
	\lim_{\xi \rightarrow t^{-1}}\,\,
	(
	\dualmap
	\comp 
	\bigg|_{
		\substack{
			q_1 = q, \\
			q_2 = q^{-1}t,\\
			q_3 = t^{-1} \\
		}
	}
	)
	\bigg[
	\underbrace{		\prod_{1 \leq a < b \leq k}					}_{
		\substack{
		a, b \in \{
		1,\lambda_1 + 1, \dots, \sum_{j = 1}^{c-1}\lambda_j + 1
		\}
		}
	}
	\frac{
		\left(1 - q^{-1}\frac{z_{b}}{z_a}\right)
		\left(1 - \frac{z_{b}}{z_a}\right)
	}{
		\left(1 - t^{-1}\frac{z_{b}}{z_a}\right)
		\left(1 - q^{-1}t\frac{z_{b}}{z_a}\right)
	}
	\bigg]
	\notag	\\
	&\hspace{0.3cm}\times 
	\lim_{\xi \rightarrow t^{-1}}\,\,
	(
	\dualmap
	\comp 
	\bigg|_{
		\substack{
			q_1 = q, \\
			q_2 = q^{-1}t,\\
			q_3 = t^{-1} \\
		}
	}
	)
	\bigg[
	\cals{N}_\lambda(z_1,\dots,z_k)
	\prod_{1 \leq a < b \leq k}					
	\cals{D}^{(i^\prime_a,i^\prime_b)}\left(
	\frac{z_b}{z_a}
	; q, t
	\right)
	\bigg]. 
	\notag 
\end{align}
It is straightforward to show that 
\begin{align}
&\lim_{\xi \rightarrow t^{-1}}\,\,
(
\dualmap
\comp 
\bigg|_{
	\substack{
		q_1 = q, \\
		q_2 = q^{-1}t,\\
		q_3 = t^{-1} \\
	}
}
)
\bigg[
\underbrace{		\prod_{1 \leq a < b \leq k}					}_{
	\substack{
		a, b \in \{
		1,\lambda_1 + 1, \dots, \sum_{j = 1}^{c-1}\lambda_j + 1
		\}
	}
}
\frac{
	\left(1 - q^{-1}\frac{z_{b}}{z_a}\right)
	\left(1 - \frac{z_{b}}{z_a}\right)
}{
	\left(1 - t^{-1}\frac{z_{b}}{z_a}\right)
	\left(1 - q^{-1}t\frac{z_{b}}{z_a}\right)
}
\bigg]
\label{eqn66-1450-10mar}
\\
&= 
\lim_{\xi \rightarrow t^{-1}}\,\,
\Bigg[
\left(
\frac{(1 - q^{-1}\xi)(1 - \xi)}{(1 - t^{-1}\xi)(1 - q^{-1}t\xi)}
\right)^{c-1}
\cdots 
\left(
\frac{(1 - q^{-1}\xi^{c-1})(1 - \xi^{c-1})}{(1 - t^{-1}\xi^{c-1})(1 - q^{-1}t\xi^{c-1})}
\right)^{1}
\Bigg]
\notag \\
&= 
\prod_{a = 1}^{c-1}
\Bigg[
\left(
\frac{(1 - q^{-1}t^{-a})(1 - t^{-a})}{(1 - t^{-1}t^{-a})(1 - q^{-1}tt^{-a})}
\right)^{c-a}
\Bigg].
\notag 
\end{align}
Furthermore, from the results of the preceding section, we know that 
\begin{align}
&\lim_{\xi \rightarrow t^{-1}}\,\,
(
\dualmap
\comp 
\bigg|_{
	\substack{
		q_1 = q, \\
		q_2 = q^{-1}t,\\
		q_3 = t^{-1} \\
	}
}
)
\bigg[
\cals{N}_\lambda(z_1,\dots,z_k)
\prod_{1 \leq a < b \leq k}					
\cals{D}^{(i^\prime_a,i^\prime_b)}\left(
\frac{z_b}{z_a}
; q, t
\right)
\bigg]
\label{eqn67-1450-10mar}
\\
&= 
\widetilde{\psi}_{T(i^\prime_1,\dots,i^\prime_k ; \lambda)}(q,t^{-1})
= 
\psi_{T(i^\prime_1,\dots,i^\prime_k ; \lambda)}(q,t)
\notag 
\\
&= 
\psi_{T(i^\prime_1,\dots,i^\prime_k ; \lambda)^{(1)}/T(i^\prime_1,\dots,i^\prime_k ; \lambda)^{(2)}}(q,t)
\cdots
\psi_{T(i^\prime_1,\dots,i^\prime_k ; \lambda)^{(N+c-1)}/T(i^\prime_1,\dots,i^\prime_k ; \lambda)^{(N+c)}}(q,t). 
\notag 
\end{align}
Substituting equations \eqref{eqn66-1450-10mar} and \eqref{eqn67-1450-10mar} into \eqref{eqn65-1450-10mar}
directly yields equation \eqref{eqn215-1512-6dec}. 
\end{proof}

\begin{lem}\mbox{}
\begin{align}
&\bigg(
\frac{q^{\frac{1}{2}} - q^{-\frac{1}{2}}}{
	t^{-\frac{1}{2}} - t^{\frac{1}{2}}
}
\bigg)^{c}
\times
\prod_{a = 1}^{c-1}
\Bigg[
\left(
\frac{(1 - q^{-1}t^{-a})(1 - t^{-a})}{(1 - t^{-1}t^{-a})(1 - q^{-1}tt^{-a})}
\right)^{c-a}
\Bigg]
\label{eqn68-1542-10mar}
\\	
&\hspace{0.3cm}\times 
\psi_{T(i^\prime_1,\dots,i^\prime_k ; \lambda)^{(N + 1)}/T(i^\prime_1,\dots,i^\prime_k ; \lambda)^{(N + 2)}}(q,t)
\cdots
\psi_{T(i^\prime_1,\dots,i^\prime_k ; \lambda)^{(N+c-1)}/T(i^\prime_1,\dots,i^\prime_k ; \lambda)^{(N+c)}}(q,t)
\notag 
\\
&=
\psi_{T^\prime_1(i_1,\dots,i_k ; \lambda)}(t,q)
\frac{
	H(\mu,q,t^{-1})
}{
	H(\mu^\prime,t^{-1},q)
}
\left(q^{-\frac{1}{2}}t^{-\frac{1}{2}}\right)^{c}
\notag 
\end{align}
where $T_1(i_1,\dots,i_k ; \lambda)$ denotes the super number part of the reverse SSYBT $T(i_1,\dots,i_k ; \lambda)$, and $\mu = (\underbrace{		1,\dots,1			}_{c})$ is the shape of $T_1(i_1,\dots,i_k ; \lambda)$. 
\end{lem}
\begin{proof}
From \textbf{Definition \ref{dfn223-2118}}, we know that for each $j \in \{1,\dots,c - 1\}$, we have 
\begin{align}
	\psi_{T(i^\prime_1,\dots,i^\prime_k ; \lambda)^{(N+j)}/ T(i^\prime_1,\dots,i^\prime_k ; \lambda)^{(N+j +1)}}(q,t) 
	= \prod_{1 \leq \alpha \leq  \beta \leq c - j}
	\frac{
		f_{q,t}(t^{j-i})
		f_{q,t}(t^{j-i})
	}{
		f_{q,t}(t^{j-i})
		f_{q,t}(t^{j-i})
	}
	= 1.
\end{align}
Therefore, 
\begin{align}
\psi_{T(i^\prime_1,\dots,i^\prime_k ; \lambda)^{(N + 1)}/T(i^\prime_1,\dots,i^\prime_k ; \lambda)^{(N + 2)}}(q,t)
\cdots
\psi_{T(i^\prime_1,\dots,i^\prime_k ; \lambda)^{(N+c-1)}/T(i^\prime_1,\dots,i^\prime_k ; \lambda)^{(N+c)}}(q,t)
= 1. 
\end{align}
Also, it is easy to show that 
\begin{align}
&\bigg(
\frac{q^{\frac{1}{2}} - q^{-\frac{1}{2}}}{
	t^{-\frac{1}{2}} - t^{\frac{1}{2}}
}
\bigg)^{c}
\times
\prod_{a = 1}^{c-1}
\Bigg[
\left(
\frac{(1 - q^{-1}t^{-a})(1 - t^{-a})}{(1 - t^{-1}t^{-a})(1 - q^{-1}tt^{-a})}
\right)^{c-a}
\Bigg]
=
\left(		q^{-1/2}t^{-1/2}			\right)^c
\prod_{\beta = 0}^{c - 1}
\Bigg[
\frac{
	q - t^{-\beta}
}{
	t^{-1-\beta} - 1
}
\Bigg]. 
\end{align}
Therefore, (LHS) of equation \eqref{eqn68-1542-10mar} is equal to 
\begin{align}
\left(		q^{-1/2}t^{-1/2}			\right)^c
\prod_{\beta = 0}^{c - 1}
\Bigg[
\frac{
	q - t^{-\beta}
}{
	t^{-1-\beta} - 1
}
\Bigg].
\end{align}

Next, let's consider (RHS) of equation \eqref{eqn68-1542-10mar}. Since $T^\prime_1(i_1,\dots,i_k ; \lambda) = (c)$, we get that
\begin{align}
\psi_{T^\prime_1(i_1,\dots,i_k ; \lambda)}(t,q) = 1. 
\end{align}
Since $\mu = (\underbrace{		1,\dots,1			}_{c})$, we get that 
\begin{align}
	H(\mu,q,t) &= q^{\frac{c(c-1)}{2}} (q - t^{c-1}) \cdots (q - t)(q - 1),
	\\
	H(\mu^\prime,q,t) &= t^{\frac{c(c-1)}{2}}(q^c - 1)\cdots (q^2 - 1)(q-1). 
\end{align}
Thus, (RHS) of equation \eqref{eqn68-1542-10mar} is equal to 
\begin{align}
	\frac{
		H(\mu,q,t^{-1})
	}{
		H(\mu^\prime,t^{-1},q)
	}
	\left(q^{-\frac{1}{2}}t^{-\frac{1}{2}}\right)^{c}
	&=
	\frac{
		(q - t^{-(c-1)}) \cdots (q - t^{-1})(q - 1)
	}{
		(t^{-c} - 1)\cdots (t^{-2} - 1)(t^{-1}-1)
	}
	\left(q^{-\frac{1}{2}}t^{-\frac{1}{2}}\right)^{c}
	\\
	&= 
	\left(		q^{-1/2}t^{-1/2}			\right)^c
	\prod_{\beta = 0}^{c - 1}
	\Bigg[
	\frac{
		q - t^{-\beta}
	}{
		t^{-1-\beta} - 1
	}
	\Bigg].
\end{align}
So, we have shown that (LHS) and (RHS) of equation \eqref{eqn68-1542-10mar} are equal. 
\end{proof}

From \textbf{Definition \ref{dfn223-2118}}, we know that 
\begin{align}
	&\psi_{T_0(i_1,\dots,i_k ; \lambda)}(q,t)
	\\
	&= 
	\psi_{T(i_1,\dots,i_k;\lambda)^{(1)} / T(i_1,\dots,i_k ; \lambda)^{(2)}}(q,t)
	\cdots
	\psi_{T(i_1,\dots,i_k;\lambda)^{(N - 1)} / T(i_1,\dots,i_k ; \lambda)^{(N)}}(q,t)
	\psi_{T(i_1,\dots,i_k ; \lambda)^{(N)} / \mu }(q,t).
	\notag 
\end{align}
However, from equations \eqref{eqn21-1433-6dec} and \eqref{eqn62-1517-10mar}, it is clear that for each $j \in \{1,2,\dots,N\}$, 
\begin{align}
	T(i_1,\dots,i_k ; \lambda)^{(j)}
	= 
	T(i^\prime_1,\dots,i^\prime_k ; \lambda)^{(j)}, 
\end{align}
and $\mu = T(i^\prime_1,\dots,i^\prime_k ; \lambda)^{(N+1)}$. 
Thus, we get that 
\begin{align}
	&\psi_{T_0(i_1,\dots,i_k ; \lambda)}(q,t)
	\label{eqn611-1542-10mar}
	\\
	&= 
	\psi_{T(i^\prime_1,\dots,i^\prime_k ; \lambda)^{(1)}/T(i^\prime_1,\dots,i^\prime_k ; \lambda)^{(2)}}(q,t)
	\cdots
	\psi_{T(i^\prime_1,\dots,i^\prime_k ; \lambda)^{(N-1)}/T(i^\prime_1,\dots,i^\prime_k ; \lambda)^{(N)}}(q,t)
	\psi_{T(i^\prime_1,\dots,i^\prime_k ; \lambda)^{(N)}/T(i^\prime_1,\dots,i^\prime_k ; \lambda)^{(N + 1)}}(q,t). 
	\notag 
\end{align}

From equations \eqref{eqn215-1512-6dec}, \eqref{eqn68-1542-10mar}, and \eqref{eqn611-1542-10mar}, we deduce that 
\begin{align}
	&\lim_{\xi \rightarrow t^{-1}}\,\,
	(
	\dualmap
	\comp 
	\bigg|_{
		\substack{
			q_1 = q, \\
			q_2 = q^{-1}t,\\
			q_3 = t^{-1} \\
		}
	}
	)
	\bigg[
	\bigg(
	\frac{q_1^{\frac{1}{2}} - q_1^{-\frac{1}{2}}}{
		q_3^{\frac{1}{2}} - q_3^{-\frac{1}{2}}
	}
	\bigg)^{|\mu|}
	\cals{N}_\lambda(z_1,\dots,z_k)
	\prod_{1 \leq a < b \leq k}
	\cals{D}^{(i_a,i_b)}\left(
	\frac{z_b}{z_a}
	; q, t
	\right)
	\bigg]
	\\
	&= 
	\psi_{T^\prime_1(i_1,\dots,i_k ; \lambda)}(t,q)
	\psi_{T_0(i_1,\dots,i_k ; \lambda)}(q,t)
	\frac{
		H(\mu,q,t^{-1})
	}{
		H(\mu^\prime,t^{-1},q)
	}
	\left(q^{-\frac{1}{2}}t^{-\frac{1}{2}}\right)^{c}
	\notag 
	\\
	&= 
	\widetilde{\psi}_{T(i_1,\dots,i_k ; \lambda)}(q,t^{-1})
	\left(q^{-\frac{1}{2}}t^{-\frac{1}{2}}\right)^{c}.
	\label{eqn622-1215-11mar}
\end{align}
It is important to note that to obtain equation \eqref{eqn622-1215-11mar}, we used \textbf{Corollary \ref{cor226-1214-11mar}} to conclude that 
\begin{align}
\psi_{T^\prime_1(i_1,\dots,i_k ; \lambda)}(t,q) = \psi_{T^\prime_1(i_1,\dots,i_k ; \lambda)}(t^{-1},q^{-1}). 
\end{align} 
Thus, we have proved the basis step $M = 1$.

\subsection{Inductive step}

Now, assume that \textbf{Lemma \ref{lemm44-1106-29jan}} holds for $1, \dots, M - 1$ where $M \geq 2$. 
Consider $(i_1,\dots,i_k) \in \operatorname{RSSYBT}(N,M;\lambda)$ such that $T_1(i_1,\dots,i_k)$ has shape $\mu$. If $\{	i_1,\dots,i_k			\}$ is a proper subset of $\{1,\dots,N+M\}$, then we can apply \textbf{Proposition \ref{prp226-1326-15feb}} and the induction hypothesis to conclude that \textbf{Lemma \ref{lemm44-1106-29jan}} holds for this case. Thus, we assume that $\{	i_1,\dots,i_k			\} = \{1,\dots,N+M\}$. Under this assumption, there must be boxes labeled by the super number $N+1$ in the reverse SSYBT $T(i_1,\dots,i_k)$. 

Next, we analyze the possible positions of the boxes labeled by super number $N+1$ in $T_1(i_1,\dots,i_k)$. According to \textbf{Definition \ref{dfn211-1443}} of a reverse SSYBT, 
we know that a box containing the super number $N+1$ must be the rightmost box in any row of $T_1(i_1,\dots,i_k)$ that contains it. This is illustrated in figure \ref{fig3-1516-11mar} below. The blue boxes in figure \ref{fig3-1516-11mar} represent the possible positions of the boxes containing the super number $N+1$ in $T_1(i_1,\dots,i_k)$ with shape $(6,6,6,3,3,2,1)$. 

\begin{figure}[!h]
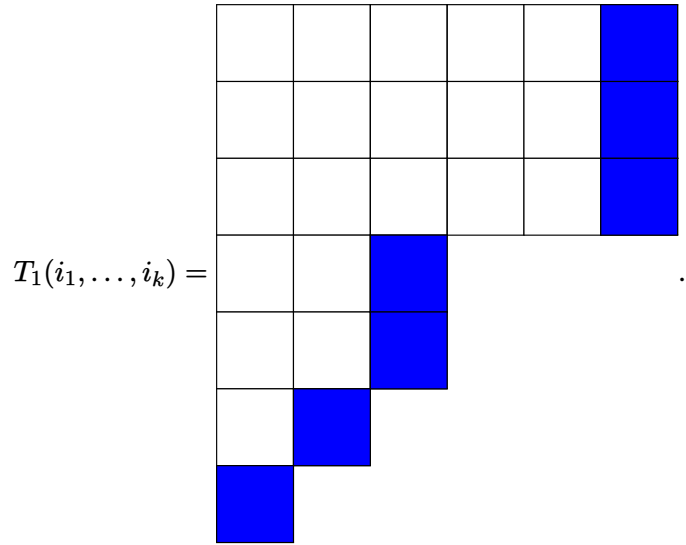

	\centering
	\small 
	\begin{align*}
		T_1(i_1,\dots,i_k) 
		= 
		\ydiagram[*(blue)]
		{5+1,5+1,5+1,2+1,2+1,1+1,0+1}
		*[*(white)]{6,6,6,3,3,2,1}.
	\end{align*}
	\normalsize
	\caption{The blue boxes represent the possible positions of the boxes with super number $N+1$.}
	\label{fig3-1516-11mar}
\end{figure}

From the preceding explanation, the following lemma is immediately obtained. 

\begin{lem}
Let 
$(i_1,\dots,i_k) \in \operatorname{RSSYBT}(N,M;\lambda)$ where $T_1(i_1,\dots,i_k) \neq \emptyset$. Then, the following statements hold:
\begin{enumerate}[(1)]
\item If a row interval of $T_1(i_1,\dots,i_k)$ contains boxes labeled by the super number $N+1$, these boxes must be located in the rightmost column of that row interval. 
\item Suppose that $\{a+1,\dots,a+n\}$ is the set of rows comprising 
a row interval of $T_1(i_1,\dots,i_k)$. If $a+j$ is the first row in this row interval containing a box with the 
super number $N+1$, then all rows $a+j+1,\dots,a+n$ must also contain boxes with super number $N+1$. 
\end{enumerate}
\end{lem}

\begin{assump}
\label{assum65-1451-18mar}
Recall that $T_1(i_1,\dots,i_k)$ has shape $\mu$. As explained in equation \eqref{eqn21-1749-1jan}, we can express $\mu$ as 
\begin{align}
	\mu = (\mu_1^{m_1}\mu_{m_1 + 1}^{m_2}\mu_{m_1 + m_2 + 1}^{m_3} \cdots \mu_{m_1 + \cdots + m_{n-1} + 1}^{m_n})
\end{align}
where $m_1,\dots,m_n \in \bb{Z}^{\geq 1}$ and $\mu_1, \mu_{m_1 + 1}, \mu_{m_1 + m_2 + 1}, \dots, \mu_{m_1 + \cdots + m_{n-1} + 1}$ are distinct positive integers. We assume that the row intervals containing at least one box with the super number $N+1$ are $\gamma_1,\dots,\gamma_r$ where $\gamma_1 < \gamma_2 < \dots < \gamma_r$. Furthermore, for each $j \in \{1,\dots,r\}$, we assume there are exactly $c_j$ boxes with the super number $N+1$ within the row interval $\gamma_j$. For future convenience, we define $c := c_1 + \cdots + c_r$. 
\end{assump}

\begin{rem}
It is important to note that all propositions and lemmas that appear in this section from this point onwards, as well as in appendix \ref{appC-1214-13mar}, are subject to \textbf{Assumption \ref{assum65-1451-18mar}}. 
\end{rem}

The figure \ref{fig4-1039-12mar} below illustrates the \textbf{Assumption \ref{assum65-1451-18mar}}. 

\begin{figure}[!h]
	\centering
	\small 
	\begin{tabular}{r@{}l}
		\begin{ytableau}
			i_1 & i_2 & i_3 & \none[\dots]
			& i_{\lambda_1 - 1} & i_{\lambda_1} \\
			i_{\lambda_1 + 1} & i_{\lambda_1 + 2} & i_{\lambda_1 + 3} & \none[\dots]
			& i_{\lambda_2 - 1} & i_{\lambda_2} \\
			i_{\lambda_2 + 1} & i_{\lambda_2 + 2} & i_{\lambda_2 + 3} & \none[\dots]
			& i_{\lambda_3} \\
			i_{\lambda_3 + 1} & i_{\lambda_3 + 2} & i_{\lambda_3 + 3} & \none[\dots]
			& \textcolor{blue}{N+1}
			\\
			i_{\lambda_4 + 1} & i_{\lambda_4 + 2} & i_{\lambda_4 + 3} & \none[\dots]
			& \textcolor{blue}{N+1} \\
%			\none[\vdots] & \none[\vdots]
%			& \none[\vdots] \\
			i_{\lambda_5 + 1} & i_{\lambda_5 + 2} & i_{\lambda_6} \\
			i_{\lambda_6 + 1} & i_{\lambda_6 + 2} & i_{\lambda_7} \\
			i_{\lambda_7 + 1} & i_{\lambda_7 + 2} & \textcolor{blue}{N+1} \\
			\textcolor{blue}{N+1} 
		\end{ytableau}
		\raisebox{-25.2ex}{$
			\left.
			\vphantom{
				\begin{array}{c}~\\[0.6ex] ~
				\end{array}
			}
			\right\}c_3 = 1$} &
		\raisebox{-19.2ex}{$
			\left.
			\vphantom{
				\begin{array}{c}~\\[0.6ex] ~
				\end{array}
			}
			\right\}c_2 = 1$} 
		\raisebox{2.9ex}{$
			\left.
			\vphantom{
				\begin{array}{c}~\\[5ex] ~
				\end{array}
			}
			\right\}c_1 = 2$} 
	\end{tabular}
	\caption{The reverse SSYBT $T_1(i_1,\dots,i_k)$ contains boxes assigned the super number $N+1$ within the row intervals $\gamma_1 = 2, \gamma_2 =3$, and $\gamma_3 = 4$. The number of boxes in each interval is $c_1 = 2, c_2 = 1,$ and $c_3 = 1$, respectively.}
	\label{fig4-1039-12mar}
\end{figure}

We now construct $(\widetilde{i}_1,\dots,\widetilde{i}_k) \in \operatorname{RSSYBT}(N+c,M-1;\lambda)$ from the given $(i_1,\dots,i_k) \in \operatorname{RSSYBT}(N,M;\lambda)$ using the following procedure:
\begin{enumerate}[(1)]
\item Convert each super numbers $N+1$ in $(i_1,\dots,i_k)$ to the ordinary numbers $N+1, N+2, \dots, N+c$ in $(\widetilde{i}_1,\dots,\widetilde{i}_k)$, performed sequentially from bottom to top. 
\item For each $j \in \{2,\dots,M\}$, convert the super number $N+j$ in $(i_1,\dots,i_k)$ to the super number $N+c+j-1$ in $(\widetilde{i}_1,\dots,\widetilde{i}_k)$. 
\item The ordinary numbers $1,\dots,N$ from $(i_1,\dots,i_k)$ are retained without modification in $(\widetilde{i}_1,\dots,\widetilde{i}_k)$. 
\end{enumerate}

Figure \ref{fig5-1039-12mar} below illustrates this procedure. 
\begin{figure}[!h]
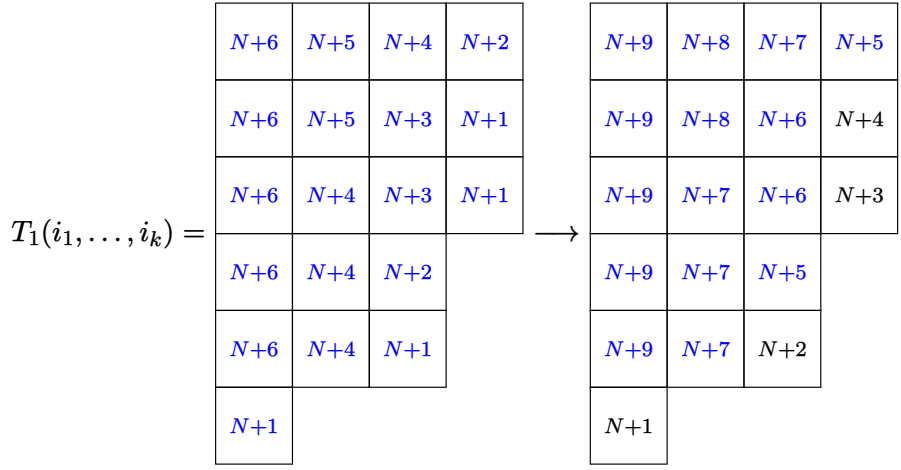

	\centering
	\small 
	\begin{tabular}{r@{}l}
		$T_1(i_1,\dots,i_k) =$
		\begin{ytableau}
			\scriptstyle \textcolor{blue}{N+6} & \scriptstyle \textcolor{blue}{N+5} & \scriptstyle \textcolor{blue}{N+4}
			& \scriptstyle \textcolor{blue}{N+2} \\
			\scriptstyle \textcolor{blue}{N+6} & \scriptstyle \textcolor{blue}{N+5} & \scriptstyle \textcolor{blue}{N+3}
			& \scriptstyle \textcolor{blue}{N+1}
			\\
			\scriptstyle \textcolor{blue}{N+6} & \scriptstyle \textcolor{blue}{N+4} & \scriptstyle \textcolor{blue}{N+3}
			& \scriptstyle \textcolor{blue}{N+1} \\
			%			\none[\vdots] & \none[\vdots]
			%			& \none[\vdots] \\
			\scriptstyle	\textcolor{blue}{N+6} & \scriptstyle\textcolor{blue}{N+4} &\scriptstyle\textcolor{blue}{N+2}\\
			\scriptstyle\textcolor{blue}{N+6} & \scriptstyle\textcolor{blue}{N+4} & \scriptstyle\textcolor{blue}{N+1} \\
			\scriptstyle\textcolor{blue}{N+1} 
		\end{ytableau}
		$\longrightarrow$
		\begin{ytableau}
			\scriptstyle \textcolor{blue}{N+9} & \scriptstyle \textcolor{blue}{N+8} & \scriptstyle \textcolor{blue}{N+7}
			& \scriptstyle \textcolor{blue}{N+5} \\
			\scriptstyle \textcolor{blue}{N+9} & \scriptstyle \textcolor{blue}{N+8} & \scriptstyle \textcolor{blue}{N+6}
			& \scriptstyle N+4
			\\
			\scriptstyle \textcolor{blue}{N+9} & \scriptstyle \textcolor{blue}{N+7} & \scriptstyle \textcolor{blue}{N+6}
			& \scriptstyle N+3 \\
			%			\none[\vdots] & \none[\vdots]
			%			& \none[\vdots] \\
			\scriptstyle	\textcolor{blue}{N+9} & \scriptstyle\textcolor{blue}{N+7} &\scriptstyle\textcolor{blue}{N+5}\\
			\scriptstyle\textcolor{blue}{N+9} & \scriptstyle\textcolor{blue}{N+7} & \scriptstyle N+2 \\
			\scriptstyle N+1
		\end{ytableau}
	\end{tabular}
	\caption{This figure illustrates the procedure of converting super number part of $(i_1,\dots,i_k)$}
	\label{fig5-1039-12mar}
\end{figure}

\begin{prop}
If $\widetilde{\mu}$ is the shape of $T_1(\widetilde{i}_1,\dots,\widetilde{i}_k)$, then 
\begin{align}
\widetilde{\mu} = \mu \backslash \{	\text{ the boxes labeled by super number $N+1$ in $(i_1,\dots,i_k)$}		\},
\end{align}
where $\mu$ is the shape of $T_1(i_1,\dots,i_k)$. 
\end{prop}  
\begin{proof}
Follows directly from the procedure described above. 
\end{proof}

\begin{lem}
\label{lem66-2142-12mar}
Let $\operatorname{Pos}((i_1,\dots,i_k);\alpha) =
\{	\beta \in \{1,\dots,k\} ~|~ i_\beta = \alpha		\}
$. Then, 
\small 
\begin{gather}
	\underbrace{			\prod_{1 \leq a < b \leq k}									}_{
		\substack{
		a,b \notin \operatorname{Pos}((i_1,\dots,i_k); N + 1)
		}
	}
	\cals{D}^{(\widetilde{i}_a,\widetilde{i}_b)}\left(
	\frac{z_b}{z_a}
	; q, t
	\right)
	= 
	\underbrace{			\prod_{1 \leq a < b \leq k}									}_{
		\substack{
		a,b \notin \operatorname{Pos}((i_1,\dots,i_k); N + 1)
		}
	}
	\cals{D}^{(i_a,i_b)}\left(
	\frac{z_b}{z_a}
	; q, t
	\right),
	\label{eqn219-1540-12dec}
	\\
	%%%% equation 2
	\underbrace{			\prod_{1 \leq a < b \leq k}									}_{
		\substack{
			a \in \operatorname{Pos}((i_1,\dots,i_k); N + 1)
			\\
			b \notin \operatorname{Pos}((i_1,\dots,i_k); N + 1)
		}
	}
	\cals{D}^{(\widetilde{i}_a,\widetilde{i}_b)}\left(
	\frac{z_b}{z_a}
	; q, t
	\right)
	= 
	\underbrace{			\prod_{1 \leq a < b \leq k}									}_{
		\substack{
			a \in \operatorname{Pos}((i_1,\dots,i_k); N + 1)
			\\
			b \notin \operatorname{Pos}((i_1,\dots,i_k); N + 1)
		}
	}
	\cals{D}^{(i_a,i_b)}\left(
	\frac{z_b}{z_a}
	; q, t
	\right),
	\\
	%%%% equation 3
	\underbrace{			\prod_{1 \leq a < b \leq k}									}_{
		\substack{
			a \notin \operatorname{Pos}((i_1,\dots,i_k); N + 1)
			\\
			b \in \operatorname{Pos}((i_1,\dots,i_k); N + 1)
		}
	}
	\cals{D}^{(\widetilde{i}_a,\widetilde{i}_b)}\left(
	\frac{z_b}{z_a}
	; q, t
	\right)
	= 
	\underbrace{			\prod_{1 \leq a < b \leq k}									}_{
		\substack{
			a \notin \operatorname{Pos}((i_1,\dots,i_k); N + 1)
			\\
			b \in \operatorname{Pos}((i_1,\dots,i_k); N + 1)
		}
	}
	\cals{D}^{(i_a,i_b)}\left(
	\frac{z_b}{z_a}
	; q, t
	\right),
\end{gather}
and 
\begin{align*}
&\bigg|_{
	\substack{
		q_1 = q, \\
		q_2 = q^{-1}t,\\
		q_3 = t^{-1} \\
	}
}
\underbrace{		\prod_{1 \leq a < b \leq k}			}_{
	\substack{
		a,b \in \operatorname{Pos}((i_1,\dots,i_k); N + 1)
	}
}
\cals{D}^{(\widetilde{i}_a,\widetilde{i}_b)}\left(
\frac{z_b}{z_a}
; q, t
\right)
\\
&= 
\underbrace{		\prod_{1 \leq a < b \leq k}			}_{
	\substack{
		a,b \in \operatorname{Pos}((i_1,\dots,i_k); N + 1)
	}
}
\frac{
	\left(1 - t^{-1}\frac{z_b}{z_a}\right)
	\left(1 - q^{-1}t\frac{z_b}{z_a}\right)
}{
	\left(1 - q^{-1}\frac{z_b}{z_a}\right)
	\left(1 - \frac{z_b}{z_a}\right)
}
\times 
\bigg|_{
	\substack{
		q_1 = q, \\
		q_2 = q^{-1}t,\\
		q_3 = t^{-1} \\
	}
}
\underbrace{		\prod_{1 \leq a < b \leq k}			}_{
	\substack{
		a,b \in \operatorname{Pos}((i_1,\dots,i_k); N + 1)
	}
}
\cals{D}^{(i_a,i_b)}\left(
\frac{z_b}{z_a}
; q, t
\right). 
\end{align*}
\end{lem}
\begin{proof}
The formulas are evident from the construction procedure of $(\widetilde{i}_1,\dots,\widetilde{i}_k)$. 
\end{proof}

From \textbf{Lemma \ref{lem66-2142-12mar}}, we know that 
\begin{align}
&\lim_{\xi \rightarrow t^{-1}}\,\,
(
\dualmap
\comp 
\bigg|_{
	\substack{
		q_1 = q, \\
		q_2 = q^{-1}t,\\
		q_3 = t^{-1} \\
	}
}
)
\bigg[
\bigg(
\frac{q_1^{\frac{1}{2}} - q_1^{-\frac{1}{2}}}{
	q_3^{\frac{1}{2}} - q_3^{-\frac{1}{2}}
}
\bigg)^{|\mu|}
\cals{N}_\lambda(z_1,\dots,z_k)
\prod_{1 \leq a < b \leq k}
\cals{D}^{(i_a,i_b)}\left(
\frac{z_b}{z_a}
; q, t
\right)
\bigg]
\\
&= 
\lim_{\xi \rightarrow t^{-1}}\,\,
(
\dualmap
\comp 
\bigg|_{
	\substack{
		q_1 = q, \\
		q_2 = q^{-1}t,\\
		q_3 = t^{-1} \\
	}
}
)
\bigg[
\bigg(
\frac{q_1^{\frac{1}{2}} - q_1^{-\frac{1}{2}}}{
	q_3^{\frac{1}{2}} - q_3^{-\frac{1}{2}}
}
\bigg)^{|\mu|}
\cals{N}_\lambda(z_1,\dots,z_k)
\prod_{1 \leq a < b \leq k}
\cals{D}^{(\widetilde{i}_a,\widetilde{i}_b)}\left(
\frac{z_b}{z_a}
; q, t
\right)
\bigg]
\notag 
\\
&\hspace{0.3cm}\times
\lim_{\xi \rightarrow t^{-1}}\,\,
\dualmap
\bigg[
\underbrace{		\prod_{1 \leq a < b \leq k}			}_{
	\substack{
		a,b \in \operatorname{Pos}((i_1,\dots,i_k); N + 1)
	}
}
\frac{
	\left(1 - q^{-1}\frac{z_b}{z_a}\right)
	\left(1 - \frac{z_b}{z_a}\right)
}{
	\left(1 - t^{-1}\frac{z_b}{z_a}\right)
	\left(1 - q^{-1}t\frac{z_b}{z_a}\right)
}
\bigg]. 
\notag 
\end{align}
However, we know from induction hypothesis that 
\begin{align}
&\lim_{\xi \rightarrow t^{-1}}\,\,
(
\dualmap
\comp 
\bigg|_{
	\substack{
		q_1 = q, \\
		q_2 = q^{-1}t,\\
		q_3 = t^{-1} \\
	}
}
)
\bigg[
\bigg(
\frac{q_1^{\frac{1}{2}} - q_1^{-\frac{1}{2}}}{
	q_3^{\frac{1}{2}} - q_3^{-\frac{1}{2}}
}
\bigg)^{|\mu| - c}
\cals{N}_\lambda(z_1,\dots,z_k)
\prod_{1 \leq a < b \leq k}
\cals{D}^{(\widetilde{i}_a,\widetilde{i}_b)}\left(
\frac{z_b}{z_a}
; q, t
\right)
\bigg]
\\
&=
\widetilde{\psi}_{T(\widetilde{i}_1,\dots, \widetilde{i}_k;\lambda)}(q,t^{-1})
\left(q^{-\frac{1}{2}}t^{-\frac{1}{2}}\right)^{|\mu| - c}. 
\notag 
\end{align}
Therefore, to prove the inductive step, we have to show that 
\begin{align}
&\frac{1}{
	\left(	q^{-\frac{1}{2}}t^{-\frac{1}{2}}			\right)^c
}
\left(		\frac{q^{\frac{1}{2}} - q^{- \frac{1}{2}}}{t^{-\frac{1}{2}} - t^{\frac{1}{2}}}			\right)^c
\lim_{\xi \rightarrow t^{-1}}\,\,
\dualmap
\bigg[
\underbrace{		\prod_{1 \leq a < b \leq k}			}_{
	\substack{
		a,b \in \operatorname{Pos}((i_1,\dots,i_k); N + 1)
	}
}
\frac{
	\left(1 - q^{-1}\frac{z_b}{z_a}\right)
	\left(1 - \frac{z_b}{z_a}\right)
}{
	\left(1 - t^{-1}\frac{z_b}{z_a}\right)
	\left(1 - q^{-1}t\frac{z_b}{z_a}\right)
}
\bigg]
\label{eqn630-1637-12mar}
\\
&= 
\frac{
	\widetilde{\psi}_{T(i_1,\dots,i_k;\lambda)}(q,t^{-1})
}{
	\widetilde{\psi}_{T(\widetilde{i}_1,\dots, \widetilde{i}_k;\lambda)}(q,t^{-1})
}. 
\notag 
\end{align}

\begin{lem}\mbox{}
\label{lem67-1637-12mar}
\small 
\begin{align}
&\frac{
	\widetilde{\psi}_{T(i_1,\dots,i_k;\lambda)}(q,t^{-1})
}{
	\widetilde{\psi}_{T(\widetilde{i}_1,\dots, \widetilde{i}_k;\lambda)}(q,t^{-1})
}
\label{eqn632-1210-13mar}
\\
&= 
\frac{
	\psi_{T^{\prime}(i_1,\dots,i_k;\lambda)^{(N+1)} 	/	T^{\prime}(i_1,\dots,i_k;\lambda)^{(N+2) }		 }(t,q)
}{
	\psi_{T(\widetilde{i}_1,\dots, \widetilde{i}_k;\lambda)^{(N+1)} / T(\widetilde{i}_1,\dots, \widetilde{i}_k;\lambda)^{(N+2)}}(q,t)
	\times
	\cdots 
	\times
	\psi_{T(\widetilde{i}_1,\dots, \widetilde{i}_k ; \lambda)^{(N+c-1)} / T(\widetilde{i}_1,\dots, \widetilde{i}_k;\lambda)^{(N+c)}}(q,t)
	\times
	\psi_{T(\widetilde{i}_1,\dots, \widetilde{i}_k;\lambda)^{(N+c)} / 	\widetilde{\mu}			}(q,t)
}
\notag 
\\
&\hspace{0.3cm}\times 
\prod_{s \in \mu}
\frac{	q^{a_{\mu}(s) + 1} - t^{-\ell_{\mu}(s)}			}{
	t^{-(\ell_{\mu}(s) + 1)} - q^{a_{\mu}(s)}
}
\times 
\left(		
\prod_{s \in \widetilde{\mu}		}
\frac{	q^{a_{\widetilde{			\mu				}}(s) + 1} - t^{-\ell_{\widetilde{		\mu		}}(s)}			}{
	t^{-(\ell_{\widetilde{\mu}}(s) + 1)} - q^{a_{\widetilde{\mu}}(s)}
}		
\right)^{-1}. 
\notag 
\end{align}
\normalsize
where $\mu$ and $\widetilde{\mu}$ are the shape of $T_1(i_1,\dots,i_k;\lambda)$ and $T_1(\widetilde{i}_1,\dots, \widetilde{i}_k;\lambda)$, respectively. 
\end{lem}
\begin{proof}
We know from \textbf{Proposition \ref{prp228-1020-29jan}} that 
\begin{align}
	\widetilde{\psi}_{T(i_1,\dots,i_k;\lambda)}(q,t^{-1})
	&= \psi_{T^\prime_{1}(i_1,\dots,i_k ; \lambda)}(t,q)
	\psi_{T_0(i_1,\dots,i_k ; \lambda)}(q,t)
	\frac{
		H(\mu,q,t^{-1})
	}{
		H(\mu^\prime,t^{-1},q)
	}
	\\
	\widetilde{\psi}_{T(\widetilde{i}_1,\dots, \widetilde{i}_k ; \lambda)}(q,t^{-1})
	&= 
	\psi_{T^\prime_1(\widetilde{i}_1,\dots, \widetilde{i}_k ; \lambda)}(t,q)
	\psi_{T_0(\widetilde{i}_1,\dots, \widetilde{i}_k ; \lambda)}(q,t)
	\frac{
		H(\widetilde{\mu},q,t^{-1})
	}{
		H(\widetilde{\mu}^\prime,t^{-1},q)
	}
\end{align}
We also know from \textbf{Definition \ref{dfn223-2118}} that 
\begin{align}
	\psi_{T_0(\widetilde{i}_1,\dots, \widetilde{i}_k;\lambda)}(q,t)
	=&
	\psi_{T(\widetilde{i}_1,\dots, \widetilde{i}_k;\lambda)^{(1)} / T(\widetilde{i}_1,\dots, \widetilde{i}_k;\lambda)^{(2)}}(q,t)
	\times
	\cdots 
	\times
	\psi_{T(\widetilde{i}_1,\dots, \widetilde{i}_k;\lambda)^{(N+c-1)} / T(\widetilde{i}_1,\dots, \widetilde{i}_k;\lambda)^{(N+c)}}(q,t)
	\\
	&\times
	\psi_{T(\widetilde{i}_1,\dots, \widetilde{i}_k;\lambda)^{(N+c)} / 	\widetilde{\mu}			}(q,t),
	\notag 
	\\
	\psi_{T_0(i_1,\dots,i_k;\lambda)}(q,t)
	=& 
	\psi_{T(i_1,\dots,i_k;\lambda)^{(1)} / T(i_1,\dots,i_k;\lambda)^{(2)}}(q,t)
	\times
	\cdots
	\times
	\psi_{T(i_1,\dots,i_k;\lambda)^{(N-1)} / T(i_1,\dots,i_k;\lambda)^{(N)}}(q,t)
	\\
	&\times
	\psi_{T(i_1,\dots,i_k;\lambda)^{(N)} / 	\mu					}(q,t),
	\notag 
	\\
	\psi_{T^\prime_{1}(i_1,\dots,i_k;\lambda)}(t,q)
	=& 
	\psi_{T^{\prime}(i_1,\dots,i_k;\lambda)^{(N+1)} 	/	T^{\prime}(i_1,\dots,i_k;\lambda)^{(N+2)}		 }(t,q)
	\\
	&\times
	\cdots
	\times
	\psi_{T^{\prime}(i_1,\dots,i_k;\lambda)^{(N+M-1)} 	/	T^{\prime}(i_1,\dots,i_k;\lambda)^{(N+M)}		 }(t,q),
	\notag 
	\\
	\psi_{T^\prime_1(\widetilde{i}_1,\dots, \widetilde{i}_k;\lambda)}(t,q)
	=& 
	\psi_{T^{\prime}(\widetilde{i}_1,\dots, \widetilde{i}_k;\lambda)^{(N+c+1)}	/		T^{\prime}(\widetilde{i}_1,\dots, \widetilde{i}_k;\lambda)^{(N+c+2)}	}(t,q)
	\\
	&\times
	\cdots
	\times
	\psi_{T^{\prime}(\widetilde{i}_1,\dots, \widetilde{i}_k;\lambda)^{(N+c+M - 2)}	/		T^{\prime}(\widetilde{i}_1,\dots, \widetilde{i}_k;\lambda)^{(N+c+M-1)}	}(t,q).
	\notag 
\end{align}
It is clear that $T(\widetilde{i}_1,\dots, \widetilde{i}_k;\lambda)^{(N+1)} = \mu$, and 
\begin{align}
T(\widetilde{i}_1,\dots, \widetilde{i}_k;\lambda)^{(\alpha)}
&= 
T(i_1,\dots,i_k;\lambda)^{(\alpha)},
\\
T^{\prime}(i_1,\dots,i_k;\lambda)^{(N+\beta)}	&= T^{\prime}(\widetilde{i}_1,\dots, \widetilde{i}_k;\lambda)^{(N+c+\beta-1)}, 
\end{align}
for $\alpha \in \{1,\dots,N\}$ and $\beta \in \{2,\dots,M\}$. 
Therefore, we obtain that 
\small
\begin{align}
	&\psi_{T_0(\widetilde{i}_1,\dots, \widetilde{i}_k;\lambda)}(q,t)
	\label{227-eqn-1219}
	\\
	&= 
	\psi_{T_0(i_1,\dots,i_k;\lambda)}(q,t)
	\times
	\psi_{T(\widetilde{i}_1,\dots, \widetilde{i}_k;\lambda)^{(N+1)} / T(\widetilde{i}_1,\dots, \widetilde{i}_k;\lambda)^{(N+2)}}(q,t)
	\times
	\cdots 
	\times
	\notag 
	\\
	&\hspace{0.3cm}\times
	\psi_{T(\widetilde{i}_1,\dots, \widetilde{i}_k;\lambda)^{(N+c-1)} / T(\widetilde{i}_1,\dots, \widetilde{i}_k;\lambda)^{(N+c)}}(q,t)
	\times
	\psi_{T(\widetilde{i}_1,\dots, \widetilde{i}_k;\lambda)^{(N+c)} / 	\widetilde{\mu}			}(q,t). 
	\notag 
\end{align}
\normalsize
and 
\begin{align}
	\psi_{T^\prime_{1}(i_1,\dots,i_k;\lambda)}(t,q)
	= 
	\psi_{T^\prime_1(\widetilde{i}_1,\dots, \widetilde{i}_k;\lambda)}(t,q)
	\times
	\psi_{T^{\prime}(i_1,\dots,i_k;\lambda)^{(N+1)} 	/	T^{\prime}(i_1,\dots,i_k;\lambda)^{(N+2)}		 }(t,q).
	\label{228-eqn-1219}
\end{align}
From equations \eqref{227-eqn-1219} \eqref{228-eqn-1219}, we immediately obtain equation \eqref{eqn632-1210-13mar}. 
\end{proof}

We can deduce from \textbf{Lemma \ref{lem67-1637-12mar}} and equation \eqref{eqn630-1637-12mar}, that in order to prove the inductive step, we have to show \textbf{Lemma \ref{lem68-1212-13mar}} below. 

\begin{lem}\mbox{}
\label{lem68-1212-13mar}
\small 
\begin{align}
	&\frac{1}{
		\left(	q^{-\frac{1}{2}}t^{-\frac{1}{2}}			\right)^c
	}
	\left(		\frac{q^{\frac{1}{2}} - q^{- \frac{1}{2}}}{t^{-\frac{1}{2}} - t^{\frac{1}{2}}}			\right)^c
	\lim_{\xi \rightarrow t^{-1}}\,\,
	\dualmap
	\bigg[
	\underbrace{		\prod_{1 \leq a < b \leq k}			}_{
		\substack{
			a,b \in \operatorname{Pos}((i_1,\dots,i_k); N + 1)
		}
	}
	\frac{
		\left(1 - q^{-1}\frac{z_b}{z_a}\right)
		\left(1 - \frac{z_b}{z_a}\right)
	}{
		\left(1 - t^{-1}\frac{z_b}{z_a}\right)
		\left(1 - q^{-1}t\frac{z_b}{z_a}\right)
	}
	\bigg]
	\\
	&= 
	\frac{
		\psi_{T^{\prime}(i_1,\dots,i_k;\lambda)^{(N+1)} 	/	T^{\prime}(i_1,\dots,i_k;\lambda)^{(N+2) }		 }(t,q)
	}{
		\psi_{T(\widetilde{i}_1,\dots, \widetilde{i}_k;\lambda)^{(N+1)} / T(\widetilde{i}_1,\dots, \widetilde{i}_k;\lambda)^{(N+2)}}(q,t)
		\times
		\cdots 
		\times
		\psi_{T(\widetilde{i}_1,\dots, \widetilde{i}_k ; \lambda)^{(N+c-1)} / T(\widetilde{i}_1,\dots, \widetilde{i}_k;\lambda)^{(N+c)}}(q,t)
		\times
		\psi_{T(\widetilde{i}_1,\dots, \widetilde{i}_k;\lambda)^{(N+c)} / 	\widetilde{\mu}			}(q,t)
	}
	\notag 
	\\
	&\hspace{0.3cm}\times 
	\prod_{s \in \mu}
	\frac{	q^{a_{\mu}(s) + 1} - t^{-\ell_{\mu}(s)}			}{
		t^{-(\ell_{\mu}(s) + 1)} - q^{a_{\mu}(s)}
	}
	\times 
	\left(		
	\prod_{s \in \widetilde{\mu}		}
	\frac{	q^{a_{\widetilde{			\mu				}}(s) + 1} - t^{-\ell_{\widetilde{		\mu		}}(s)}			}{
		t^{-(\ell_{\widetilde{\mu}}(s) + 1)} - q^{a_{\widetilde{\mu}}(s)}
	}		
	\right)^{-1}
	\notag 
\end{align}
\normalsize
\end{lem}
\begin{proof}
The proof of this lemma is relegated to Appendix \ref{appC-1214-13mar}. 
\end{proof}

As a consequence, we have proved the inductive step. Therefore, the induction is complete.

\section{Proof of lemma \ref{lemm44-1106-29jan} for the case $(0,M)$ where $M \in \bb{Z}^{>0}$}
\label{sec7-1042-7apr}

In this section, we will prove \textbf{Lemma \ref{lemm44-1106-29jan}} for the case $(0,M)$ where $M \in \bb{Z}^{>0}$. This will be achieved by utilizing the results from the preceding section. 
Let $(i_1,\dots,i_k) \in \operatorname{RSSYBT}(0,M;\lambda)$. We define 
$(\widetilde{i}_1,\dots,\widetilde{i}_k) \in \operatorname{RSSYBT}(1,M;\lambda)$ by 
\begin{align}
(\widetilde{i}_1,\dots,\widetilde{i}_k) := (i_1 + 1,\dots,i_k + 1). 
\end{align}
It is clear that $T(\widetilde{i}_1,\dots,\widetilde{i}_k;\lambda)$ and $T(i_1,\dots,i_k;\lambda)$ have the same structure (see \textbf{Definition \ref{dfn227-2239-13mar}}). Then, by using \textbf{Proposition \ref{prp226-1326-15feb}}, we obtain that 
\begin{align}
\psi_{T^\prime(i_1,\dots,i_k;\lambda)}(q,t) = \psi_{T^\prime(\widetilde{i}_1,\dots,\widetilde{i}_k;\lambda)}(q,t). 
\label{eqn72-1045-14mar}
\end{align}
From \textbf{Proposition \ref{prp228-1020-29jan}}, we know that 
\begin{align}
\widetilde{\psi}_{T(i_1,\dots,i_k;\lambda)}(q,t^{-1}) &= \psi_{T^\prime(i_1,\dots,i_k;\lambda)}(t^{-1},q^{-1})\frac{H(\lambda,q,t^{-1})}{H(\lambda^\prime,t^{-1},q)}, 
\\
\widetilde{\psi}_{T(\widetilde{i}_1,\dots,\widetilde{i}_k;\lambda)}(q,t^{-1}) &= \psi_{T^\prime(\widetilde{i}_1,\dots,\widetilde{i}_k;\lambda)}(t^{-1},q^{-1})\frac{H(\lambda,q,t^{-1})}{H(\lambda^\prime,t^{-1},q)}. 
\end{align}
From equation \eqref{eqn72-1045-14mar}, we conclude that 
\begin{align}
\widetilde{\psi}_{T(i_1,\dots,i_k;\lambda)}(q,t^{-1})
=
\widetilde{\psi}_{T(\widetilde{i}_1,\dots,\widetilde{i}_k;\lambda)}(q,t^{-1}). 
\end{align}

Since $(\widetilde{i}_1,\dots,\widetilde{i}_k) \in \operatorname{RSSYBT}(1,M;\lambda)$, we can conclude from the case $(1,M)$ of \textbf{Lemma \ref{lemm44-1106-29jan}} that 
\begin{align}
&\lim_{\xi \rightarrow t^{-1}}\,\,
(
\dualmap
\comp 
\bigg|_{
	\substack{
		q_1 = q, \\
		q_2 = q^{-1}t,\\
		q_3 = t^{-1} \\
	}
}
)
\bigg[
\bigg(
\frac{q_1^{\frac{1}{2}} - q_1^{-\frac{1}{2}}}{
	q_3^{\frac{1}{2}} - q_3^{-\frac{1}{2}}
}
\bigg)^{|\lambda|}
\cals{N}_\lambda(z_1,\dots,z_k)
\prod_{1 \leq a < b \leq k}
\cals{D}^{(\widetilde{i}_a,\widetilde{i}_b)}\left(
\frac{z_b}{z_a}
; q, t
\right)
\bigg]
\\
&= 
\widetilde{\psi}_{T(\widetilde{i}_1,\dots,\widetilde{i}_k;\lambda)}(q,t^{-1})
\left(	q^{-\frac{1}{2}}t^{-\frac{1}{2}}				\right)^{|\lambda|}. 
\notag 
\end{align}
It is obvious that 
\begin{align}
\cals{D}^{(\widetilde{i}_a,\widetilde{i}_b)}\left(
\frac{z_b}{z_a}
; q, t
\right)
= 
\cals{D}^{(i_a,i_b)}\left(
\frac{z_b}{z_a}
; q, t
\right).
\end{align}
Thus, we get that 
\begin{align}
	&\lim_{\xi \rightarrow t^{-1}}\,\,
	(
	\dualmap
	\comp 
	\bigg|_{
		\substack{
			q_1 = q, \\
			q_2 = q^{-1}t,\\
			q_3 = t^{-1} \\
		}
	}
	)
	\bigg[
	\bigg(
	\frac{q_1^{\frac{1}{2}} - q_1^{-\frac{1}{2}}}{
		q_3^{\frac{1}{2}} - q_3^{-\frac{1}{2}}
	}
	\bigg)^{|\lambda|}
	\cals{N}_\lambda(z_1,\dots,z_k)
	\prod_{1 \leq a < b \leq k}
	\cals{D}^{(i_a,i_b)}\left(
	\frac{z_b}{z_a}
	; q, t
	\right)
	\bigg]
	\\
	&= 
	\widetilde{\psi}_{T(\widetilde{i}_1,\dots,\widetilde{i}_k;\lambda)}(q,t^{-1})
	\left(	q^{-\frac{1}{2}}t^{-\frac{1}{2}}				\right)^{|\lambda|}
	= 
	\widetilde{\psi}_{T(i_1,\dots,i_k;\lambda)}(q,t^{-1})
	\left(	q^{-\frac{1}{2}}t^{-\frac{1}{2}}				\right)^{|\lambda|}. 
	\notag 
\end{align}
Thus, we have proved \textbf{Lemma \ref{lemm44-1106-29jan}} for the case $(0,M)$.

\section{Future Research Directions}

In this section, we outline several directions for future research that naturally extend the results presented in this paper.

\begin{enumerate}[(1)]
\item \textbf{Generalization to quantum corner VOA of type $\vec{c} = (3^N1^M2^L)$.}
In this work, we showed that the correlation functions of the currents of the quantum corner VOA of type $\vec{c} = (3^N1^M)$ are given by super Macdonald polynomials (\textbf{Theorem \ref{thm42-main-1022}}). A natural next step would be to investigate the class of polynomials corresponding to the correlation functions of the currents of the quantum corner VOA of type $\vec{c} = (3^N1^M2^L)$. A key question is whether these polynomials would also possess the property of being partially symmetric.
\item \textbf{Exploring higher-order currents.} 
As indicated in \textbf{Theorem \ref{thm42-main-1022}}, our analysis only utilized the currents $\widetilde{T}^{\vec{c},\vec{u}}_{1}(z)$. A promising direction would be to investigate the correlation functions of the higher currents $\widetilde{T}^{\vec{c},\vec{u}}_{r}(z) \,\, (r \geq 2)$. This raises the question of whether it would still be possible to construct a class of polynomials from these higher-order currents.
\end{enumerate}

\appendix

\section{Proof of lemma \ref{lemm43-1156-22jan}}
\label{appA-1155-22jan}

This appendix aims to prove \textbf{Lemma \ref{lemm43-1156-22jan}}. For the reader's convenience, we restate the lemma here:

\small 
\begin{align}
	&
	\lim_{\xi \rightarrow t^{-1}}\,\,
	(
	\dualmap
	\comp 
	\bigg|_{
		\substack{
			q_1 = q, \\
			q_2 = q^{-1}t,\\
			q_3 = t^{-1} \\
		}
	}
	)
	\left(
	\cals{N}_{\lambda}(z_1,\dots,z_k )
	\times
	\prod_{1 \leq i < j \leq k}f^{\vec{c}}_{11}\left(\frac{z_j}{z_i} \right)
	\times
	\langle 0 |\widetilde{T}^{\vec{c},\vec{u}}_{1}(z_1 )\cdots \widetilde{T}^{\vec{c},\vec{u}}_{1}(z_k )|0\rangle
	\right)
	\\
	&= 
	\underbrace{				
		\sum_{i_1 = 1}^{N+M}
		\cdots
		\sum_{i_k = 1}^{N+M}
	}_{
		(i_1,\dots,i_k) \in 
		\text{RSSYBT}(N,M;\lambda)
	}
	\lim_{\xi \rightarrow t^{-1}}\,\,
	(
	\dualmap
	\comp 
	\bigg|_{
		\substack{
			q_1 = q, \\
			q_2 = q^{-1}t,\\
			q_3 = t^{-1} \\
		}
	}
	)
	\bigg[
	y_{i_1}\cdots y_{i_k}
	u_{i_1}\cdots u_{i_k}
	\notag \\
	&\hspace{7.5cm}\times
	\cals{N}_\lambda(z_1,\dots,z_k)
	\times
	\prod_{1 \leq a < b \leq k}
	\cals{D}^{(i_a,i_b)}\left(
	\frac{z_b}{z_a}
	; q, t
	\right)
	\bigg],
	\notag 
\end{align}
\normalsize
where the explicit formulas of $\cals{N}_{\lambda}(z_1,\dots,z_k )$
and 
$
\cals{D}^{(i,j)}\left(
\frac{z_b}{z_a}
; q, t
\right)$ 
are provided in equations \eqref{eqn32-1044} and \eqref{eqn126-1300-4dec}, respectively. Note that
\begin{align}
	\bigg|_{
		\substack{
			q_1 = q, \\
			q_2 = q^{-1}t,\\
			q_3 = t^{-1} \\
		}
	}
	\cals{N}_{\lambda}(z_1,\dots,z_k )
	= 
	\prod_{1 \leq c < d \leq \ell(\lambda)}
	\prod_{
		\substack{
			i \in I^{(c)}
			\\
			j \in I^{(d)}
		}
	}
	\frac{
		\left(1 - \frac{z_j}{z_i}\right)
		\left(1 - t\frac{z_j}{z_i}\right)
	}{
		\left(1 - q\frac{z_j}{z_i}\right)
		\left(1 - q^{-1}t\frac{z_j}{z_i}\right)
	}.
	\label{eqnD3-1517-13Jan}
\end{align}

Recall that the map $\dualmap$ sends $f(z_1,\dots,z_m)$ to
\begin{align}
	&f(y,qy,\dots,q^{\lambda_1 - 1}y
	\notag	\\
	&\hspace{0.4cm} \xi y,q\xi y,\dots,q^{\lambda_2 - 1}\xi y
	\notag \\
	&\hspace{2.8cm}\vdots
	\notag \\
	&\hspace{0.4cm} \xi^{\ell(\lambda) - 1} y,q\xi^{\ell(\lambda) - 1} y,\dots,q^{\lambda_{\ell(\lambda)} - 1}\xi^{\ell(\lambda) - 1} y). 
\end{align}
Consequently, under the map 
\begin{align}
(
\dualmap
\comp 
\bigg|_{
	\substack{
		q_1 = q, \\
		q_2 = q^{-1}t,\\
		q_3 = t^{-1} \\
	}
}
), 
\end{align}
the only factors in equation \eqref{eqnD3-1517-13Jan} that can generate a factor of $(1-t\xi)$ are 
\begin{align}
	\left(1 - t\frac{z_j}{z_i}\right)
	\hspace{0.3cm}
	\text{ and }
	\hspace{0.3cm}
	\left(1 - q^{-1}t\frac{z_j}{z_i}\right).
\end{align}

More precisely, the factor $(1 - t\frac{z_j}{z_i})$ generates $(1-t\xi)$ if 
the boxes corresponding to $i$ and $j$ are adjacent in the same column
\small 
\begin{align}
\begin{ytableau}
i  \\ 
j 
\end{ytableau}
\hspace{0.2cm}. 
\label{eqn-a4-1332-22jan}
\end{align}
\normalsize
Conversely, 
the factor $(1 - q^{-1}t\frac{z_j}{z_i})$ generates $(1-t\xi) $ 
if the corresponding boxes are adjacent diagonally. 
\begin{align}
	\begin{ytableau}
	i & \none
		\\
		\none & j
	\end{ytableau}\hspace{0.2cm}.
	\label{eqn-a5-1927-23jan}
\end{align}

Since $(1 - t\frac{z_j}{z_i})$ appears in the numerator and $(1 - q^{-1}t\frac{z_j}{z_i})$ appears in the denominator, 
we can conclude that the contribution from a triangular-shaped diagram 
\begin{align}
\ydiagram{2,1+1} 
\end{align}
will not contain the factor $(1 - t\xi)$. Based on this reasoning, to determine the total number of factors $(1-t\xi)$ appearing in equation \eqref{eqnD3-1517-13Jan}, it suffices to consider only pairs of boxes
of the form \eqref{eqn-a4-1332-22jan} within the first column of $\lambda$. This directly leads to the following proposition. 

\begin{prop}
For any $\lambda \in \operatorname{Par}(k)$,  
$\cals{N}_{\lambda}(z_1,\dots,z_k )$ contains exactly the factor $(1-t\xi)^{\ell(\lambda) - 1}$.
\end{prop}

Furthermore, based on the preceding explanation, we know that  $\cals{N}_{\lambda}(z_1,\dots,z_k )$ can be decomposed into the product 
\small 
\begin{align}
	&\cals{N}_{\lambda}(z_1,\dots,z_k )
	= 
	\Delta\left(	q_3^{-\frac{1}{2}}\frac{z_{\lambda_1 + 1}}{z_1}				\right)^{-1}
	\Delta\left(	q_3^{-\frac{1}{2}}\frac{z_{\lambda_1 + \lambda_2 + 1}}{z_{\lambda_1 + 1}}				\right)^{-1}
	\times 
	\cdots 
	\times 
	\Delta\left(	q_3^{-\frac{1}{2}}\frac{z_{\sum_{j = 1}^{\ell(\lambda) - 1}\lambda_j+1}}{z_{\sum_{j = 1}^{\ell(\lambda) - 2}\lambda_j+1}}				\right)^{-1}
	\times
	\cals{A}
	\label{d7-1950-21jan}
\end{align}
\normalsize
where neither the numerator nor the denominator of 
\small 
\begin{align}
	(
	\dualmap
	\comp 
	\bigg|_{
		\substack{
			q_1 = q, \\
			q_2 = q^{-1}t,\\
			q_3 = t^{-1} \\
		}
	}
	)
	\left(
	\cals{A}
	\right)
\end{align}
\normalsize
contain any factors of the form $(1-t\xi)$.
This implies, in particular, that 
\small 
\begin{align}
\lim_{\xi \rightarrow t^{-1}}\,\,
(
\dualmap
\comp 
\bigg|_{
	\substack{
		q_1 = q, \\
		q_2 = q^{-1}t,\\
		q_3 = t^{-1} \\
	}
}
)
\left(
\cals{A}
\right)
\end{align}
\normalsize
is a nonzero number.

From equation \eqref{d7-1950-21jan}, we obtain that 
\small 
\begin{align}
	&
	\lim_{\xi \rightarrow t^{-1}}\,\,
	(
	\dualmap
	\comp 
	\bigg|_{
		\substack{
			q_1 = q, \\
			q_2 = q^{-1}t,\\
			q_3 = t^{-1} \\
		}
	}
	)
	\left(
	\cals{N}_{\lambda}(z_1,\dots,z_k )
	\times
	\prod_{1 \leq i < j \leq k}f^{\vec{c}}_{11}\left(\frac{z_j}{z_i} \right)
	\times
	\langle 0 |\widetilde{T}^{\vec{c},\vec{u}}_{1}(z_1 )\cdots \widetilde{T}^{\vec{c},\vec{u}}_{1}(z_k )|0\rangle
	\right)
	\label{d7-eqn-1830-21jan}
	\\
	&= 
	\lim_{\xi \rightarrow t^{-1}}\,\,
	(
	\dualmap
	\comp 
	\bigg|_{
		\substack{
			q_1 = q, \\
			q_2 = q^{-1}t,\\
			q_3 = t^{-1} \\
		}
	}
	)
	\left(
	\cals{A}
	\right)
	\times 
	\lim_{\xi \rightarrow t^{-1}}\,\,
	(
	\dualmap
	\comp 
	\bigg|_{
		\substack{
			q_1 = q, \\
			q_2 = q^{-1}t,\\
			q_3 = t^{-1} \\
		}
	}
	)
	\Bigg[
	\Delta\left(	q_3^{-\frac{1}{2}}\frac{z_{\lambda_1 + 1}}{z_1}				\right)^{-1}
	\times 
	\cdots 
	\times 
	\Delta\left(	q_3^{-\frac{1}{2}}\frac{z_{\sum_{j = 1}^{\ell(\lambda) - 1}\lambda_j+1}}{z_{\sum_{j = 1}^{\ell(\lambda) - 2}\lambda_j+1}}				\right)^{-1}
	\notag \\
	&\hspace{6cm}\times
	\prod_{1 \leq i < j \leq k}f^{\vec{c}}_{11}\left(\frac{z_j}{z_i} \right)
	\times
	\langle 0 |\widetilde{T}^{\vec{c},\vec{u}}_{1}(z_1 )\cdots \widetilde{T}^{\vec{c},\vec{u}}_{1}(z_k )|0\rangle
	\Bigg]
	\notag 
	\\
	&= 
	\lim_{\xi \rightarrow t^{-1}}\,\,
	(
	\dualmap
	\comp 
	\bigg|_{
		\substack{
			q_1 = q, \\
			q_2 = q^{-1}t,\\
			q_3 = t^{-1} \\
		}
	}
	)
	\left(
	\cals{A}
	\right)
	\times
	\notag 
	\\
	&\hspace{0.3cm} \times \sum_{i_1 = 1}^{N+M}
	\cdots
	\sum_{i_k = 1}^{N+M}
	\lim_{\xi \rightarrow t^{-1}}\,\,
	(
	\dualmap
	\comp 
	\bigg|_{
		\substack{
			q_1 = q, \\
			q_2 = q^{-1}t,\\
			q_3 = t^{-1} \\
		}
	}
	)
	\Bigg[
	\Delta\left(	q_3^{-\frac{1}{2}}\frac{z_{\lambda_1 + 1}}{z_1}				\right)^{-1}
	\times 
	\cdots 
	\times 
	\Delta\left(	q_3^{-\frac{1}{2}}\frac{z_{\sum_{j = 1}^{\ell(\lambda) - 1}\lambda_j+1}}{z_{\sum_{j = 1}^{\ell(\lambda) - 2}\lambda_j+1}}				\right)^{-1}
	\notag 
	\\
	&\hspace{5.1cm}\times
	\prod_{1 \leq i < j \leq k}f^{\vec{c}}_{11}\left(\frac{z_j}{z_i} \right)
	\times
	y_{i_1} \cdots y_{i_k}
	\times 
	\langle 0 |			\normord{
		\widetilde{\Lambda}^{\vec{c},\vec{u}}_{i_1}(z_1)
	}
	\times
	\cdots
	\times
	\normord{
		\widetilde{\Lambda}^{\vec{c},\vec{u}}_{i_k}(z_k)
	}				
	|0\rangle
	\Bigg].
	\notag 
	\\
	&= 
	\lim_{\xi \rightarrow t^{-1}}\,\,
	(
	\dualmap
	\comp 
	\bigg|_{
		\substack{
			q_1 = q, \\
			q_2 = q^{-1}t,\\
			q_3 = t^{-1} \\
		}
	}
	)
	\left(
	\cals{A}
	\right)
	\times
	\notag 
	\\
	&\hspace{0.3cm} \times \sum_{i_1 = 1}^{N+M}
	\cdots
	\sum_{i_k = 1}^{N+M}
	\lim_{\xi \rightarrow t^{-1}}\,\,
	(
	\dualmap
	\comp 
	\bigg|_{
		\substack{
			q_1 = q, \\
			q_2 = q^{-1}t,\\
			q_3 = t^{-1} \\
		}
	}
	)
	\Bigg[
	y_{i_1}\cdots y_{i_k}
	u_{i_1}\cdots u_{i_k}
	\notag 
	\\
	&\hspace{4.7cm}\times
	\Delta\left(	q_3^{-\frac{1}{2}}\frac{z_{\lambda_1 + 1}}{z_1}				\right)^{-1}
	\times 
	\cdots 
	\times 
	\Delta\left(	q_3^{-\frac{1}{2}}\frac{z_{\sum_{j = 1}^{\ell(\lambda) - 1}\lambda_j+1}}{z_{\sum_{j = 1}^{\ell(\lambda) - 2}\lambda_j+1}}				\right)^{-1}
	\times
	\prod_{1 \leq a < b \leq k}
	\cals{D}^{(i_a,i_b)}\left(
	\frac{z_b}{z_a}
	; q, t
	\right)
	\Bigg]
	\notag
\end{align}
\normalsize
Note that to obtain the final equality of the above equation, we used proposition \ref{prp37-1018-23jan}.

\begin{lem}
\label{lema2-1452-28jan}
If the arrangement of $(i_1,\dots,i_k) \in \{1,\dots,N+M\}^k$ into a Young diagram $\lambda$ (as in equation \eqref{a11-eqn-1511-22jan}) 
\begin{align}
	\begin{ytableau}
		i_1 & i_2  & \none[\dots] & i_{\lambda_1 - 1}
		& i_{\lambda_1}
		\\
		i_{\lambda_1 + 1} & i_{\lambda_1 + 2} &  \none[\dots]
		& i_{\lambda_1 + \lambda_2} \\
		\none[\vdots] & \none[\vdots]
		& \none[\vdots]
	\end{ytableau}
\label{a11-eqn-1511-22jan}
\end{align}
results in at least one row that violates the conditions for being an element of 
$\text{RSSYBT}(N,M;\lambda)$, then 
\small 
\begin{align*}
&\lim_{\xi \rightarrow t^{-1}}\,\,
(
\dualmap
\comp 
\bigg|_{
	\substack{
		q_1 = q, \\
		q_2 = q^{-1}t,\\
		q_3 = t^{-1} \\
	}
}
)
\Bigg[
y_{i_1}\cdots y_{i_k}
u_{i_1}\cdots u_{i_k}
\notag 
\\
&\hspace{2.9cm}\times
\Delta\left(	q_3^{-\frac{1}{2}}\frac{z_{\lambda_1 + 1}}{z_1}				\right)^{-1}
\times 
\cdots 
\times 
\Delta\left(	q_3^{-\frac{1}{2}}\frac{z_{\sum_{j = 1}^{\ell(\lambda) - 1}\lambda_j+1}}{z_{\sum_{j = 1}^{\ell(\lambda) - 2}\lambda_j+1}}				\right)^{-1}
\times
\prod_{1 \leq a < b \leq k}
\cals{D}^{(i_a,i_b)}\left(
\frac{z_b}{z_a}
; q, t
\right)
\Bigg]
\notag
\\
&= 0. 
\end{align*}
\normalsize 
\end{lem}
\begin{proof}
First, recall from equation \eqref{eqn126-1300-4dec} that 
\begin{align}
	\bigg|_{
		\substack{
			q_1 = q, \\
			q_2 = q^{-1}t,\\
			q_3 = t^{-1} \\
		}
	}
	\cals{D}^{(i,j)}\left(
	\frac{z_b}{z_a}
	; q, t
	\right)
	:= 
	\begin{cases}
		\displaystyle 
		\frac{
			\left(1 - q^{-1}\frac{z_b}{z_a}\right)
			\left(1 - qt^{-1}\frac{z_b}{z_a}\right)
		}{
			\left(1 - t^{-1}\frac{z_b}{z_a}\right)
			\left(1 - \frac{z_b}{z_a}\right)
		}
		\hspace{0.3cm}
		&\text{ if } i < j
		\\
		\displaystyle 
		\frac{
			\left(1 - q^{-1}\frac{z_b}{z_a}\right)
			\left(1 - q\frac{z_b}{z_a}\right)
		}{
			\left(1 - t^{-1}\frac{z_b}{z_a}\right)
			\left(1 - t\frac{z_b}{z_a}\right)
		}
		\hspace{0.3cm}
		&\text{ if } i = j = \text{ super-number }
		\\
		1
		\hspace{0.3cm}
		&\text{ if } i = j = \text{ ordinary-number }
		\\
		\displaystyle
		\frac{
			\left(1 - q\frac{z_b}{z_a}\right)
			\left(1 - q^{-1}t\frac{z_b}{z_a}\right)
		}{
			\left(1 - t\frac{z_b}{z_a}\right)
			\left(1 - \frac{z_b}{z_a}\right)
		}
		\hspace{0.3cm}
		&\text{ if } i > j 
	\end{cases}
\label{eqn-a12-1300-23jan}
\end{align}
Next, we assume that at least one row in the Young tableau \eqref{a11-eqn-1511-22jan} violates the conditions for being an element of $\text{RSSYBT}(N,M;\lambda)$. This leads to two possible cases:
\begin{enumerate}[(1)]
	\item \label{case1-1600}
	There exists a pair of adjacent boxes within the same row that are not weakly decreasing (i.e., they are strictly increasing).
	\item \label{case2-1600}
	All boxes in the row are weakly decreasing, but there exists a pair of adjacent boxes within the same row containing the same super number.
\end{enumerate}

We begin by considering case \eqref{case1-1600}. Let us assume that 
\begin{ytableau}
	i_j & i_{j+1}
\end{ytableau}
denotes a pair of adjacent boxes within the same row where $i_j < i_{j+1}$. 
Then from \eqref{eqn-a12-1300-23jan}, the factor 
\begin{align}
	\frac{
		\left(1 - q^{-1}\frac{z_{j+1}}{z_j}\right)
		\left(1 - qt^{-1}\frac{z_{j+1}}{z_j}\right)
	}{
		\left(1 - \frac{z_{j+1}}{z_j}		\right)
		\left(1 - t^{-1}\frac{z_{j+1}}{z_j}\right)
	}
\end{align}
will appear. It is evident that 
\begin{align}
	\dualmap
	\left(
	\frac{
		\left(1 - q^{-1}\frac{z_{j+1}}{z_j}\right)
		\left(1 - qt^{-1}\frac{z_{j+1}}{z_j}\right)
	}{
		\left(1 - \frac{z_{j+1}}{z_j}\right)
		\left(1 - t^{-1}\frac{z_{j+1}}{z_j}\right)
	}
	\right)
	= 0,
\end{align}
due to the factor $\left(1 - q^{-1}\frac{z_{j+1}}{z_j}\right)$. 

The question is whether this factor $\left(1 - q^{-1}\frac{z_{j+1}}{z_j}\right)$ can be cancelled by other contributions from $\cals{D}^{(i_a,i_b)}\left(
\frac{z_b}{z_a}
; q, t
\right)$ and
\begin{align}
	\Delta\left(	q_3^{-\frac{1}{2}}\frac{z_{\lambda_1 + 1}}{z_1}				\right)^{-1}
	\times 
	\cdots 
	\times 
	\Delta\left(	q_3^{-\frac{1}{2}}\frac{z_{\sum_{j = 1}^{\ell(\lambda) - 1}\lambda_j+1}}{z_{\sum_{j = 1}^{\ell(\lambda) - 2}\lambda_j+1}}				\right)^{-1}. 
	\label{a15-1448-23jan}
\end{align}
We claim that such cancellation is not possible. Indeed, from \eqref{eqn-a12-1300-23jan}, we can see that the possible denominators are of the form
\begin{align}
	\left(1 - \frac{z_b}{z_a}\right), \hspace{0.3cm}
	\left(1 - t\frac{z_b}{z_a}\right), \hspace{0.3cm}
	\left(1 - t^{-1}\frac{z_b}{z_a}\right)
	\hspace{1cm}
	(1 \leq a < b \leq k)
\end{align}
Therefore, under the map $\dualmap$, the factors appearing in the denominator will be of the form
\begin{align}
	(1 - q^m\xi^n), \hspace{0.3cm}
	(1 - tq^m\xi^n), \hspace{0.3cm} 
	(1 - t^{-1}q^m\xi^n), 
	\label{a18-eqn-1550-22jan}
\end{align}
where $(m,n) \in ( \bb{Z} \times \bb{Z}^{\geq 0}) \backslash \{(0,0)\}$. It is clear that for any $(m,n) \in ( \bb{Z} \times \bb{Z}^{\geq 0}) \backslash \{(0,0)\}$, the factors in \eqref{a18-eqn-1550-22jan} can never equal zero.
Thus, the factor $\left(1 - q^{-1}\frac{z_{j+1}}{z_j}\right)$ cannot be cancelled by other contributions from $\cals{D}^{(i_a,i_b)}\left(
\frac{z_b}{z_a}
; q, t
\right)$.

Furthermore, given that 
\begin{align}
\bigg|_{
	\substack{
		q_1 = q, \\
		q_2 = q^{-1}t,\\
		q_3 = t^{-1} \\
	}
}
\Delta\left(q_3^{-\frac{1}{2}}	\frac{z_j}{z_i}		\right)
&= 
\frac{
	\left(1 - q\frac{z_j}{z_i}\right)
	\left(1 - q^{-1}t\frac{z_j}{z_i}\right)
}{
	\left(1 - \frac{z_j}{z_i}\right)
	\left(1 - t\frac{z_j}{z_i}\right)
}, 
\end{align}
a similar argument demonstrates that the factor  $\left(1 - q^{-1}\frac{z_{j+1}}{z_j}\right)$ cannot be cancelled by other contributions from
\begin{align}
	\Delta\left(	q_3^{-\frac{1}{2}}\frac{z_{\lambda_1 + 1}}{z_1}				\right)^{-1}
	\times 
	\cdots 
	\times 
	\Delta\left(	q_3^{-\frac{1}{2}}\frac{z_{\sum_{j = 1}^{\ell(\lambda) - 1}\lambda_j+1}}{z_{\sum_{j = 1}^{\ell(\lambda) - 2}\lambda_j+1}}				\right)^{-1}. 
\end{align}
Thus, in case \eqref{case1-1600}, we have 
\begin{align}
	&\lim_{\xi \rightarrow t^{-1}}\,\,
	(
	\dualmap
	\comp 
	\bigg|_{
		\substack{
			q_1 = q, \\
			q_2 = q^{-1}t,\\
			q_3 = t^{-1} \\
		}
	}
	)
	\Bigg[
	y_{i_1}\cdots y_{i_k}
	u_{i_1}\cdots u_{i_k}
	\\
	&\hspace{2.9cm}\times
	\Delta\left(	q_3^{-\frac{1}{2}}\frac{z_{\lambda_1 + 1}}{z_1}				\right)^{-1}
	\times 
	\cdots 
	\times 
	\Delta\left(	q_3^{-\frac{1}{2}}\frac{z_{\sum_{j = 1}^{\ell(\lambda) - 1}\lambda_j+1}}{z_{\sum_{j = 1}^{\ell(\lambda) - 2}\lambda_j+1}}				\right)^{-1}
	\times
	\prod_{1 \leq a < b \leq k}
	\cals{D}^{(i_a,i_b)}\left(
	\frac{z_b}{z_a}
	; q, t
	\right)
	\Bigg]
	\notag
	\\
	&= 0. 
	\notag 
\end{align}

Next, we consider case \eqref{case2-1600}. Let us assume that \begin{ytableau}
	i_j & i_{j+1}
\end{ytableau}
denotes a pair of adjacent boxes in the same row such that $i_j = i_{j+1} = \text{super number}$. From \eqref{eqn-a12-1300-23jan}, we see that the factor 
\begin{align}
\frac{
	\left(1 - q^{-1}		\frac{z_{j+1}}{z_j}		\right)
	\left(1 - q				\frac{z_{j+1}}{z_j}				\right)
}{
	\left(1 - t^{-1}		\frac{z_{j+1}}{z_j} \right)
	\left(1 - t\frac{z_{j+1}}{z_j}\right)
}
\end{align}
appears. As a consequence of the factor $\left(1 - q^{-1}\frac{z_{j+1}}{z_j}\right)$, 
\begin{align}
\dualmap\left(
\frac{
	\left(1 - q^{-1}		\frac{z_{j+1}}{z_j}		\right)
	\left(1 - q				\frac{z_{j+1}}{z_j}				\right)
}{
	\left(1 - t^{-1}		\frac{z_{j+1}}{z_j} \right)
	\left(1 - t\frac{z_{j+1}}{z_j}\right)
}
\right)
= 0. 
\end{align}
Similar to case  \eqref{case1-1600}, it is necessary to verify that $\left(1 - q^{-1}\frac{z_{j+1}}{z_j}\right)$ will not be cancelled by contributions from other factors
$\cals{D}^{(i_a,i_b)}\left(
\frac{z_b}{z_a}
; q, t
\right)$ and the factor in \eqref{a15-1448-23jan}. Applying the same reasoning as in case \eqref{case1-1600}, we conclude that $\left(1 - q^{-1}\frac{z_{j+1}}{z_j}\right)$ cannot be cancelled. Hence, 
\small 
\begin{align*}
	&\lim_{\xi \rightarrow t^{-1}}\,\,
	(
	\dualmap
	\comp 
	\bigg|_{
		\substack{
			q_1 = q, \\
			q_2 = q^{-1}t,\\
			q_3 = t^{-1} \\
		}
	}
	)
	\Bigg[
	y_{i_1}\cdots y_{i_k}
	u_{i_1}\cdots u_{i_k}
	\notag 
	\\
	&\hspace{2.9cm}\times
	\Delta\left(	q_3^{-\frac{1}{2}}\frac{z_{\lambda_1 + 1}}{z_1}				\right)^{-1}
	\times 
	\cdots 
	\times 
	\Delta\left(	q_3^{-\frac{1}{2}}\frac{z_{\sum_{j = 1}^{\ell(\lambda) - 1}\lambda_j+1}}{z_{\sum_{j = 1}^{\ell(\lambda) - 2}\lambda_j+1}}				\right)^{-1}
	\times
	\prod_{1 \leq a < b \leq k}
	\cals{D}^{(i_a,i_b)}\left(
	\frac{z_b}{z_a}
	; q, t
	\right)
	\Bigg]
	\notag
	\\
	&= 0 
\end{align*}
\normalsize 
\end{proof}

\subsection{Triangle-form cancellation}
We have previously demonstrated that if any row in the tableau \eqref{a11-eqn-1511-22jan} constructed from $(i_1,\dots,i_k) \in \{1,\dots,N+M\}^k$ violates the reverse SSYBT conditions, then
\begin{align}
	&\lim_{\xi \rightarrow t^{-1}}\,\,
	(
	\dualmap
	\comp 
	\bigg|_{
		\substack{
			q_1 = q, \\
			q_2 = q^{-1}t,\\
			q_3 = t^{-1} \\
		}
	}
	)
	\Bigg[
	y_{i_1}\cdots y_{i_k}
	u_{i_1}\cdots u_{i_k}
	\label{a20-eqn-1515-23jan}
	\\
	&\hspace{2.9cm}\times
	\Delta\left(	q_3^{-\frac{1}{2}}\frac{z_{\lambda_1 + 1}}{z_1}				\right)^{-1}
	\times 
	\cdots 
	\times 
	\Delta\left(	q_3^{-\frac{1}{2}}\frac{z_{\sum_{j = 1}^{\ell(\lambda) - 1}\lambda_j+1}}{z_{\sum_{j = 1}^{\ell(\lambda) - 2}\lambda_j+1}}				\right)^{-1}
	\times
	\prod_{1 \leq a < b \leq k}
	\cals{D}^{(i_a,i_b)}\left(
	\frac{z_b}{z_a}
	; q, t
	\right)
	\Bigg]
	\notag
\end{align}
equals 0. Next, we would like to show that if the reverse SSYBT conditions are violated in the column direction, then the quantity in equation \eqref{a20-eqn-1515-23jan} will also be 0. To show this, we must first introduce the concept of triangle-form cancellation in reverse SSYBT. Therefore, this subsection will be dedicated to 
defining and discussing triangle-form cancellation in reverse SSYBT.

\begin{prop}
\label{prpa3-1848-23jan}
If a triangular-shaped boxes 
\begin{align}
	\begin{ytableau}
		i_{c}  &   i_{c+1} \\
		\none & i_{d}
	\end{ytableau}
	\label{D22-1725-17jan}
\end{align}
is part of a reverse SSYBT, then 
\begin{align}
(
\dualmap
\comp 
\bigg|_{
	\substack{
		q_1 = q, \\
		q_2 = q^{-1}t,\\
		q_3 = t^{-1} \\
	}
}
)
\Bigg[
\underbrace{	\prod_{1 \leq a < b \leq k}				}_{a,b \in \{c,c+1,d\}}
\cals{D}^{(i_a,i_b)}\left(
\frac{z_b}{z_a}
; q, t
\right)
\Bigg], 
\label{A22-eqn-1555-23jan}
\end{align}
will not contain a factor of $(1 - t\xi)$ in
either the numerator or the denominator.
\end{prop}
\begin{proof}
We will divide the proof into two cases: the first where $i_c$ is a super number, and the second where 
$i_c$ is an ordinary number. Let us begin with the first case.

Since $i_c$ is a super number and the triangular-shaped diagram \ref{D22-1725-17jan} is part of a reverse SSYBT, we can immediately conclude that $i_d < i_c$. If $i_{c+1}$ is a super number and $i_d < i_{c+1}$, then 
\small  
\begin{align}
	&(
	\dualmap
	\comp 
	\bigg|_{
		\substack{
			q_1 = q, \\
			q_2 = q^{-1}t,\\
			q_3 = t^{-1} \\
		}
	}
	)
	\Bigg[
	\underbrace{	\prod_{1 \leq a < b \leq k}				}_{a,b \in \{c,c+1,d\}}
	\cals{D}^{(i_a,i_b)}\left(
	\frac{z_b}{z_a}
	; q, t
	\right)
	\Bigg]
	= 
	\frac{(1 - q^2)(1 - t)}{(1 - q)(1 - tq)}
	\times
	\frac{(1 - q\xi)(1 - q^{-1}t\xi)}{(1 - \xi)
		\cancel{	(1 - t\xi)			}}
	\times
	\frac{
		(1 - q^2\xi)\cancel{		(1 - t\xi)		}		
	}{
		(1 - q\xi)(1 - qt\xi)
	}. 
\end{align}
\normalsize
If $i_{c+1}$ is a super number and $i_d = i_{c+1}$, then 
\small 
\begin{align}
&(
\dualmap
\comp 
\bigg|_{
	\substack{
		q_1 = q, \\
		q_2 = q^{-1}t,\\
		q_3 = t^{-1} \\
	}
}
)
\Bigg[
\underbrace{	\prod_{1 \leq a < b \leq k}				}_{a,b \in \{c,c+1,d\}}
\cals{D}^{(i_a,i_b)}\left(
\frac{z_b}{z_a}
; q, t
\right)
\Bigg]
=
	\frac{(1 - q^2)(1 - t)}{(1 - q)(1 - tq)}
	\times
	\frac{
		(1 - q\xi)(1 - q^{-1}\xi)
	}{
		\cancel{	(1 - t\xi)		}(1 - t^{-1}\xi)
	}
	\times
	\frac{
		(1 - q^2\xi)\cancel{		(1 - t\xi)		}		
	}{
		(1 - q\xi)(1 - qt\xi)
	}. 
\end{align}
\normalsize
If $i_{c+1}$ is an ordinary number, then
\small 
\begin{align}
&(
\dualmap
\comp 
\bigg|_{
	\substack{
		q_1 = q, \\
		q_2 = q^{-1}t,\\
		q_3 = t^{-1} \\
	}
}
)
\Bigg[
\underbrace{	\prod_{1 \leq a < b \leq k}				}_{a,b \in \{c,c+1,d\}}
\cals{D}^{(i_a,i_b)}\left(
\frac{z_b}{z_a}
; q, t
\right)
\Bigg]
=
\frac{(1 - q^2)(1 - t)}{(1 - q)(1 - tq)}
\times
\frac{(1 - q\xi)(1 - q^{-1}t\xi)}{(1 - \xi)
	\cancel{	(1 - t\xi)			}}
\times
\frac{
	(1 - q^2\xi)\cancel{		(1 - t\xi)		}		
}{
	(1 - q\xi)(1 - qt\xi)
}
\end{align}
\normalsize
Therefore, we have shown that in the first case, the quantity \eqref{A22-eqn-1555-23jan} does not contain a factor of $(1-t\xi)$.

Next, we consider the second case where $i_c$ is an ordinary number. Since the triangular-shaped diagram in  \ref{D22-1725-17jan} is part of a reverse SSYBT, $i_{c+1}$ must also be an ordinary number, and $i_d < i_{c+1}$. Thus, we have $i_d < i_{c+1} \leq i_c$. If $i_{c+1} < i_c$, then 
\small 
\begin{align}
&(
\dualmap
\comp 
\bigg|_{
	\substack{
		q_1 = q, \\
		q_2 = q^{-1}t,\\
		q_3 = t^{-1} \\
	}
}
)
\Bigg[
\underbrace{	\prod_{1 \leq a < b \leq k}				}_{a,b \in \{c,c+1,d\}}
\cals{D}^{(i_a,i_b)}\left(
\frac{z_b}{z_a}
; q, t
\right)
\Bigg]
= 
\frac{(1 - q^2)(1 - t)}{(1 - q)(1 - tq)}
\times
\frac{(1 - q\xi)(1 - q^{-1}t\xi)}{(1 - \xi)
	\cancel{	(1 - t\xi)			}}
\times
\frac{
	(1 - q^2\xi)\cancel{		(1 - t\xi)		}		
}{
	(1 - q\xi)(1 - qt\xi)
},
\end{align}
\normalsize
while if $i_{c+1} = i_c$, then 
\small 
\begin{align}
&(
\dualmap
\comp 
\bigg|_{
	\substack{
		q_1 = q, \\
		q_2 = q^{-1}t,\\
		q_3 = t^{-1} \\
	}
}
)
\Bigg[
\underbrace{	\prod_{1 \leq a < b \leq k}				}_{a,b \in \{c,c+1,d\}}
\cals{D}^{(i_a,i_b)}\left(
\frac{z_b}{z_a}
; q, t
\right)
\Bigg]
= 
\frac{(1 - q\xi)(1 - q^{-1}t\xi)}{(1 - \xi)
	\cancel{	(1 - t\xi)			}}
\times
\frac{
	(1 - q^2\xi)\cancel{		(1 - t\xi)		}		
}{
	(1 - q\xi)(1 - qt\xi)
}. 
\end{align}
\normalsize
Therefore, we have shown that in the second case as well, the quantity \eqref{A22-eqn-1555-23jan} does not contain a factor of $(1-t\xi)$. 
\end{proof}

\begin{prop}
Let
\begin{align}
	\begin{ytableau}
		i_1 & i_2  & \none[\dots] & i_{\lambda_1 - 1}
		& i_{\lambda_1}
		\\
		i_{\lambda_1 + 1} & i_{\lambda_1 + 2} &  \none[\dots]
		& i_{\lambda_1 + \lambda_2} \\
		\none[\vdots] & \none[\vdots]
		& \none[\vdots]
	\end{ytableau}
\label{tableau-a23-2053-23jan}
\end{align}
be a reverse SSYBT of shape $\lambda$ where $|\lambda| = k$. Then, we can decompose  
\begin{align}
	(
	\dualmap
	\comp 
	\bigg|_{
		\substack{
			q_1 = q, \\
			q_2 = q^{-1}t,\\
			q_3 = t^{-1} \\
		}
	}
	)
	\Bigg[
	\prod_{1 \leq a < b \leq k}			
	\cals{D}^{(i_a,i_b)}\left(
	\frac{z_b}{z_a}
	; q, t
	\right)
	\Bigg]
	=
	\frac{1}{(1 - t\xi)^{\ell(\lambda)-1}} \times \cals{F}, 
\end{align}
where neither the numerator nor the denominator of $\cals{F}$ contains the factor $(1 - t\xi)$.
\end{prop}
\begin{proof}
To determine the total number of $(1-t\xi)$ factors remaining after cancellations, 
it is sufficient to consider only the $\cals{D}^{(i,j)}\left(
\frac{z_b}{z_a}
; q, t
\right)$ arising from pairs of boxes of the forms  \eqref{eqn-a4-1332-22jan} and 
\eqref{eqn-a5-1927-23jan}.
Moreover, by using proposition \ref{prpa3-1848-23jan}, 
we know that we only need to examine the first column
\begin{align}
\begin{ytableau}
	i_1 
	\\
	i_{\lambda_1 + 1}  \\
	\none[\vdots] 
\end{ytableau}
\end{align}

Given that the tableau \eqref{tableau-a23-2053-23jan} is a reverse SSYBT, it follows that for each 
$j \in \{1,2,\dots,\ell(\lambda) - 1\}$, 
\begin{align}
\cals{D}^{(i_{\sum_{\alpha = 1}^{j-1}\lambda_\alpha+1},i_{\sum_{\alpha = 1}^{j}\lambda_\alpha+1})}\left(
\frac{z_{\sum_{\alpha = 1}^{j}\lambda_\alpha+1}}{z_{\sum_{\alpha = 1}^{j-1}\lambda_\alpha+1}}
; q, t
\right)
\end{align}
contains exactly one factor of $(1-t\xi)^{-1}$. 
Therefore, the contribution from the first column, 
\begin{align}
	(
	\dualmap
	\comp 
	\bigg|_{
		\substack{
			q_1 = q, \\
			q_2 = q^{-1}t,\\
			q_3 = t^{-1} \\
		}
	}
	)
	\Bigg[
	\underbrace{		\prod_{1 \leq a < b \leq k}						}_{
	\substack{
	a,b \in
	\left\{	\sum_{\alpha = 1}^{j}\lambda_\alpha+1 ~|~ j \in \{0,1,2,\dots,\ell(\lambda) - 1\}	\right\}
	}
	}
	\cals{D}^{(i_a,i_b)}\left(
	\frac{z_b}{z_a}
	; q, t
	\right)
	\Bigg]
\end{align}
contains exactly the factor $(1-t\xi)^{-(\ell(\lambda) - 1)}$. 
This is because if a pair of boxes is not adjacent, the corresponding $\cals{D}^{(i_a,i_b)}\left(
\frac{z_b}{z_a}
; q, t
\right)$ will not contain $(1- t\xi)$. 
Consequently, we can conclude that 
\begin{align}
	(
	\dualmap
	\comp 
	\bigg|_{
		\substack{
			q_1 = q, \\
			q_2 = q^{-1}t,\\
			q_3 = t^{-1} \\
		}
	}
	)
	\Bigg[
	\prod_{1 \leq a < b \leq k}			
	\cals{D}^{(i_a,i_b)}\left(
	\frac{z_b}{z_a}
	; q, t
	\right)
	\Bigg]
\end{align}
also contains exactly the factor $(1-t\xi)^{-(\ell(\lambda) - 1)}$. 
\end{proof}

\subsection{Breaking pair, breaking triangle, and breaking band}

In subsection \ref{subsec-a3-1523-27jan}, we will show that 
if the reverse SSYBT conditions are violated in the column direction, then the quantity in equation \eqref{a20-eqn-1515-23jan} will also be 0. To demonstrate this, we will introduce the concepts of a breaking pair, breaking triangle, and breaking band, which will be defined in this subsection.

\begin{dfn}[Breaking pair]
Let 
\begin{align}
	\begin{ytableau}
		i_1 & i_2  & \none[\dots] & i_{\lambda_1 - 1}
		& i_{\lambda_1}
		\\
		i_{\lambda_1 + 1} & i_{\lambda_1 + 2} &  \none[\dots]
		& i_{\lambda_1 + \lambda_2} \\
		\none[\vdots] & \none[\vdots]
		& \none[\vdots]
	\end{ytableau}
	\label{A28-eqn-2158-23jan}
\end{align}
be a tableau of shape $\lambda$ constructed by arranging $(i_1,\dots,i_k) \in \{1,\dots,N+M\}^k$ into the form of a Young diagram $\lambda$. Suppose that every row of the tableau \eqref{A28-eqn-2158-23jan} satisfies the reverse SSYBT conditions. Consider a pair of adjacent boxes
\begin{align}
	\begin{ytableau}
		i_a
		\\
		i_b
	\end{ytableau}
\label{eqn-a29-1554-26jan}
\end{align}
in the tableau \eqref{A28-eqn-2158-23jan}. We say that it is a breaking pair if it violates the reverse SSYBT conditions. If the adjacent pair of boxes \eqref{eqn-a29-1554-26jan}
is not a breaking pair, we then call it a non-breaking pair. 
\end{dfn}

\begin{dfn}[Breaking triangle]
\label{dfna6-1448-17apr}
Let 
\begin{align}
	\begin{ytableau}
		i_1 & i_2  & \none[\dots] & i_{\lambda_1 - 1}
		& i_{\lambda_1}
		\\
		i_{\lambda_1 + 1} & i_{\lambda_1 + 2} &  \none[\dots]
		& i_{\lambda_1 + \lambda_2} \\
		\none[\vdots] & \none[\vdots]
		& \none[\vdots]
	\end{ytableau}
	\label{A30-eqn-2158-23jan}
\end{align}
be a tableau of shape $\lambda$ constructed by arranging $(i_1,\dots,i_k) \in \{1,\dots,N+M\}^k$ into the form of a Young diagram $\lambda$. Suppose that every row of the tableau \eqref{A30-eqn-2158-23jan} satisfies the reverse SSYBT conditions. Let 
\begin{align}
	\begin{ytableau}
		i_a
		\\
		i_b
	\end{ytableau}
\label{a31-2204-23jan-eqn}
\end{align}
be a breaking pair in the tableau \eqref{A30-eqn-2158-23jan}. If the breaking pair \eqref{a31-2204-23jan-eqn} is not in the first column, we then call 
\begin{align}
	\begin{ytableau}
		i_{a-1} & i_a
		\\
		\none & i_b
	\end{ytableau}
\end{align}
a breaking triangle of the tableau \eqref{A30-eqn-2158-23jan}.
\end{dfn}

Now, let 
\begin{align}
	\begin{ytableau}
		i_{a-1} & i_a
		\\
		\none & i_b
	\end{ytableau}
	\label{a33-eqn-1436-26jan}
\end{align}
be a breaking triangle. Recall from \textbf{Definition \ref{dfna6-1448-17apr}}
that we assumed every row satisfies the reverse SSYBT conditions. Therefore, we can conclude that $i_{a-1} > i_a$ or $i_{a-1} = i_a = \text{ ordinary number }$. On the other hand, since we know that
\begin{align}
	\begin{ytableau}
		i_a
		\\
		i_b
	\end{ytableau}
\end{align}
violates the reverse SSYBT conditions, we can conclude that $i_a < i_b$ or $i_a = i_b = \text{ordinary number}$. 

From this, we can deduce that the possible cases are
\begin{enumerate}[(1)]
\item \label{case1-1429}$i_a < i_b < i_{a-1}$
\item  \label{case2-1429} $i_a< i_b = i_{a-1}$
\item  \label{case3-1429}$i_a < i_{a-1} < i_b$
\item  \label{case4-1429}$\text{ordinary number} = i_a = i_{a-1} < i_b$
\item  \label{case5-1429}$\text{ordinary number} = i_a = i_b < i_{a-1}$
\item  \label{case6-1429}$\text{ordinary number} = i_a = i_b = i_{a-1}$
\end{enumerate}
In the following, we will analyze the contributions from breaking triangles in each of these cases. 

\textbf{\underline{Case \eqref{case1-1429} $i_a < i_b < i_{a-1}$ }}
In this case, we get that 
\small 
\begin{align*}
		&(
		\dualmap
		\comp 
		\bigg|_{
			\substack{
				q_1 = q, \\
				q_2 = q^{-1}t,\\
				q_3 = t^{-1} \\
			}
		}
		)
		\Bigg[
		\underbrace{		\prod_{1 \leq c < d \leq k}			}_{c,d \in \{a-1,a,b\}}
		\cals{D}^{(i_c,i_d)}\left(
		\frac{z_d}{z_c}
		; q, t
		\right)
		\Bigg]
		\\
	&=
	\frac{(1 - q^2)(1 - t)}{(1 - q)(1 - tq)}
	\times
	\frac{(1 - q^2\xi)(1 - t\xi)}{(1 - q\xi)(1 - qt\xi)}
	\times
	\frac{
		(1 - q^{-1}\xi)(1 - qt^{-1}\xi)
	}{
		(1 - \xi)(1 - t^{-1}\xi)
	}.
\end{align*}
\normalsize
We can see that there is one factor of $(1-t\xi)$. Furthermore, since each row must satisfy the reverse SSYBT conditions, the number in the box below $i_{a-1}$ (denoted as $i_{b-1}$) can be less than, equal to, or greater than $i_{a-1}$, but must be greater than or equal to $i_b$.

\begin{itemize}
\item if $i_{b-1} < i_{a-1}$, then we get that 
\begin{align}
	\begin{ytableau}
		i_{a-1}
		\\
		i_{b-1}
	\end{ytableau}
\label{a35-eqn-1551-26jan}
\end{align}
is not a breaking pair. 
\item if $i_{b-1} = i_{a-1} = \text{super number}$, then we get that the pair of boxes in \eqref{a35-eqn-1551-26jan}
is not a breaking pair. 
\item if $i_{b-1} = i_{a-1} = \text{ordinary number}$, then we get that the pair of boxes in \eqref{a35-eqn-1551-26jan}
is a breaking pair. 
\item if $i_{b-1} > i_{a-1}$, then we get that the pair of boxes in \eqref{a35-eqn-1551-26jan}
is a breaking pair. 
\end{itemize}
Thus, in case \eqref{case1-1429}, the adjacent pair of boxes
\begin{align}
\begin{ytableau}
	i_{a-1}
	\\
	i_{b-1}
\end{ytableau} 
\end{align}
can be either a breaking pair or a non-breaking pair. 
\vspace{0.2cm}

\textbf{\underline{Case \eqref{case2-1429} $i_a< i_b = i_{a-1}$ }}
In this case, if $i_b = i_{a-1} = \text{super number}$, then 
\small 
\begin{align}
	&(
	\dualmap
	\comp 
	\bigg|_{
		\substack{
			q_1 = q, \\
			q_2 = q^{-1}t,\\
			q_3 = t^{-1} \\
		}
	}
	)
	\Bigg[
	\underbrace{		\prod_{1 \leq c < d \leq k}			}_{c,d \in \{a-1,a,b\}}
	\cals{D}^{(i_c,i_d)}\left(
	\frac{z_d}{z_c}
	; q, t
	\right)
	\Bigg]
	\\
	&=
	\frac{(1 - q^2)(1 - t)}{(1 - q)(1 - tq)}
	\times
	\frac{
		(1 - q^{-1}\xi)(1 - qt^{-1}\xi)
	}{
		(1 - \xi)(1 - t^{-1}\xi)
	}
	\times
	\frac{
		(1 - \xi)(1 - q^2\xi)
	}{
		(1 - qt^{-1}\xi)(1 - qt\xi)
	}.
\notag 
\end{align}
\normalsize
However, if $i_b = i_{a-1} = \text{ordinary number}$, then  
\small 
\begin{align}
	(
	\dualmap
	\comp 
	\bigg|_{
		\substack{
			q_1 = q, \\
			q_2 = q^{-1}t,\\
			q_3 = t^{-1} \\
		}
	}
	)
	\Bigg[
	\underbrace{		\prod_{1 \leq c < d \leq k}			}_{c,d \in \{a-1,a,b\}}
	\cals{D}^{(i_c,i_d)}\left(
	\frac{z_d}{z_c}
	; q, t
	\right)
	\Bigg]
	= 
	\frac{(1 - q^2)(1 - t)}{(1 - q)(1 - tq)}
	\times
	\frac{
		(1 - q^{-1}\xi)(1 - qt^{-1}\xi)
	}{
		(1 - \xi)(1 - t^{-1}\xi)
	}. 
\end{align}
\normalsize
We see that regardless of whether $i_b = i_{a-1} = \text{super number}$ or $i_b = i_{a-1} = \text{ordinary number}$, the factor $(1-t\xi)$ does not appear. 
Furthermore, we can show that in case \eqref{case2-1429}, the adjacent pair of boxes 
\begin{align}
	\begin{ytableau}
		i_{a-1}
		\\
		i_{b-1}
	\end{ytableau} 
\end{align}
is necessarily a breaking pair. 
\vspace{0.2cm}

\textbf{\underline{Case \eqref{case3-1429} $i_a < i_{a-1} < i_b$ }}
In this case, we get that 
\small 
\begin{align*}
	&(
	\dualmap
	\comp 
	\bigg|_{
		\substack{
			q_1 = q, \\
			q_2 = q^{-1}t,\\
			q_3 = t^{-1} \\
		}
	}
	)
	\Bigg[
	\underbrace{		\prod_{1 \leq c < d \leq k}			}_{c,d \in \{a-1,a,b\}}
	\cals{D}^{(i_c,i_d)}\left(
	\frac{z_d}{z_c}
	; q, t
	\right)
	\Bigg]
	\\
	&=
	\frac{(1 - q^2)(1 - t)}{(1 - q)(1 - tq)}
	\times
	\frac{
		\cancel{	(1 - \xi)		}(1 - q^2t^{-1}\xi)
	}{
		(1 - q\xi)\cancel{		(1 - qt^{-1}\xi)		}
	}
	\times
	\frac{
		(1 - q^{-1}\xi)	 \cancel{	(1 - qt^{-1}\xi)	}	
	}{
		\cancel{	(1 - \xi)	}(1 - t^{-1}\xi)
	}.
\end{align*}
\normalsize
We see that the factor $(1-t\xi)$ does not appear. Furthermore, we can show that in case \eqref{case3-1429}, 
the adjacent pair of boxes 
\begin{align}
	\begin{ytableau}
		i_{a-1}
		\\
		i_{b-1}
	\end{ytableau} 
\end{align}
is necessarily a breaking pair. 
\vspace{0.2cm}

\textbf{\underline{Case \eqref{case4-1429} $\text{ordinary number} = i_a = i_{a-1} < i_b$}}
In this case, we get that 
\small 
\begin{align*}
	&(
	\dualmap
	\comp 
	\bigg|_{
		\substack{
			q_1 = q, \\
			q_2 = q^{-1}t,\\
			q_3 = t^{-1} \\
		}
	}
	)
	\Bigg[
	\underbrace{		\prod_{1 \leq c < d \leq k}			}_{c,d \in \{a-1,a,b\}}
	\cals{D}^{(i_c,i_d)}\left(
	\frac{z_d}{z_c}
	; q, t
	\right)
	\Bigg]
	=
	\frac{
		\cancel{	(1 - \xi)		}(1 - q^2t^{-1}\xi)
	}{
		(1 - q\xi)\cancel{		(1 - qt^{-1}\xi)		}
	}
	\times
	\frac{
		(1 - q^{-1}\xi)	 \cancel{	(1 - qt^{-1}\xi)	}	
	}{
		\cancel{	(1 - \xi)	}(1 - t^{-1}\xi)
	}
\end{align*}
\normalsize
We see that the factor $(1-t\xi)$ does not appear. Furthermore, we can show that in case \eqref{case4-1429}, the adjacent pair of boxes 
\begin{align}
	\begin{ytableau}
		i_{a-1}
		\\
		i_{b-1}
	\end{ytableau} 
\end{align}
is necessarily a breaking pair. 
\vspace{0.2cm}

\textbf{\underline{Case \eqref{case5-1429} $\text{ordinary number} = i_a = i_b < i_{a-1}$}}
In this case, we get that
\small  
\begin{align*}
	&(
	\dualmap
	\comp 
	\bigg|_{
		\substack{
			q_1 = q, \\
			q_2 = q^{-1}t,\\
			q_3 = t^{-1} \\
		}
	}
	)
	\Bigg[
	\underbrace{		\prod_{1 \leq c < d \leq k}			}_{c,d \in \{a-1,a,b\}}
	\cals{D}^{(i_c,i_d)}\left(
	\frac{z_d}{z_c}
	; q, t
	\right)
	\Bigg]
	=
	\frac{(1 - q^2)(1 - t)}{(1 - q)(1 - tq)}
	\times
	\frac{(1 - q^2\xi)(1 - t\xi)}{(1 - q\xi)(1 - qt\xi)}
\end{align*}
\normalsize
We see that there is one factor of $(1-t\xi)$. Furthermore, as in case \eqref{case1-1429}, the adjacent pair of boxes
\begin{align}
	\begin{ytableau}
		i_{a-1}
		\\
		i_{b-1}
	\end{ytableau} 
\end{align}
can be either a breaking pair or a non-breaking pair. 
\vspace{0.2cm}

\textbf{\underline{Case \eqref{case6-1429} $\text{ordinary number} = i_a = i_b = i_{a-1}$}}
In this case, we get that
\small 
\begin{align*}
&(
\dualmap
\comp 
\bigg|_{
	\substack{
		q_1 = q, \\
		q_2 = q^{-1}t,\\
		q_3 = t^{-1} \\
	}
}
)
\Bigg[
\underbrace{		\prod_{1 \leq c < d \leq k}			}_{c,d \in \{a-1,a,b\}}
\cals{D}^{(i_c,i_d)}\left(
\frac{z_d}{z_c}
; q, t
\right)
\Bigg]
=
1
\end{align*}
\normalsize
We see that the factor $(1-t\xi)$ does not appear. 
Furthermore, we can show that in case \eqref{case6-1429}, the adjacent pair of boxes
\begin{align}
	\begin{ytableau}
		i_{a-1}
		\\
		i_{b-1}
	\end{ytableau} 
\end{align}
is necessarily a breaking pair. 
\vspace{0.3cm}

The preceding analysis of \textbf{Case \eqref{case1-1429}} - \textbf{Case \eqref{case6-1429}} will lead to 
the definition of breaking triangle types as stated in \textbf{Definition \ref{dfna7-2052-26jan}} below. 

\begin{dfn}[Types of breaking triangle]
\label{dfna7-2052-26jan}
Suppose
	\begin{align}
			\begin{ytableau}
					i_{a-1} & i_a
					\\
					\none & i_b
				\end{ytableau}
		\end{align}
	is a breaking triangle.
	We say that it is an unstoppable breaking triangle if exactly one of the following conditions holds:
	\begin{enumerate}[(1)]
	\item 	$i_a< i_b = i_{a-1}$,
	\item  $i_a < i_{a-1} < i_b$,
	\item  $\text{ordinary number} = i_a = i_{a-1} < i_b$,
	\item  $\text{ordinary number} = i_a = i_b = i_{a-1}$.
	\end{enumerate}
	Conversely, we say that it is a stoppable breaking triangle if it is not an unstoppable breaking triangle. In other word, it is a a stoppable breaking triangle if exactly one of the following conditions holds:
	\begin{enumerate}[(1)]
		\item	$i_a < i_b < i_{a-1}$,
		\item  $\text{ordinary number} = i_a = i_b < i_{a-1}$. 
	\end{enumerate}
\end{dfn}

\begin{prop}
If 
\begin{align}
	\begin{ytableau}
		i_{a-1} & i_a
		\\
		\none & i_b
	\end{ytableau}
\end{align}
be an unstoppable breaking triangle, then 
\begin{align}
	\begin{ytableau}
		i_{a-1}
		\\
		i_{b-1}
	\end{ytableau} 
\end{align}
is a breaking pair. 
\end{prop}
\begin{proof}
This follows directly from the analysis we did in \textbf{Case \eqref{case1-1429}} - \textbf{Case \eqref{case6-1429}} above. 
\end{proof}

\begin{prop}
\label{prpa9-1257-27jan}
If 
\begin{align}
	\begin{ytableau}
		i_{a-1} & i_a
		\\
		\none & i_b
	\end{ytableau}
\label{eqn-a58-1043-18apr}
\end{align}
is an unstoppable breaking triangle, then 
\begin{align}
&(
\dualmap
\comp 
\bigg|_{
	\substack{
		q_1 = q, \\
		q_2 = q^{-1}t,\\
		q_3 = t^{-1} \\
	}
}
)
\Bigg[
\underbrace{		\prod_{1 \leq c < d \leq k}			}_{c,d \in \{a-1,a,b\}}
\cals{D}^{(i_c,i_d)}\left(
\frac{z_d}{z_c}
; q, t
\right)
\Bigg]
\label{eqn-A-48-1041}
\end{align}
will not contain a factor of $(1-t\xi)$ in either its numerator or denominator. Conversely, if \eqref{eqn-a58-1043-18apr} is a stoppable breaking triangle, then exactly one factor of $(1-t\xi)$ will appear in the numerator of \eqref{eqn-A-48-1041}. 
\end{prop}
\begin{proof}
This follows directly from the analysis we did in \textbf{Case \eqref{case1-1429}} - \textbf{Case \eqref{case6-1429}} above. 
\end{proof}

\begin{dfn}[Breaking band]
Let 
\begin{align}
	\begin{ytableau}
		i_1 & i_2  & \none[\dots] & i_{\lambda_1 - 1}
		& i_{\lambda_1}
		\\
		i_{\lambda_1 + 1} & i_{\lambda_1 + 2} &  \none[\dots]
		& i_{\lambda_1 + \lambda_2} \\
		\none[\vdots] & \none[\vdots]
		& \none[\vdots]
	\end{ytableau}
	\label{A70-eqn-2158-27jan}
\end{align}
be a tableau of shape $\lambda$ constructed by arranging $(i_1,\dots,i_k) \in \{1,\dots,N+M\}^k$ into the form of a Young diagram $\lambda$. Assume that every row of the tableau \eqref{A70-eqn-2158-27jan} satisfies the reverse SSYBT conditions. Let 
\begin{align}
	\begin{ytableau}
		i_a & i_{a+1}
		\\
		i_b & i_{b+1}
	\end{ytableau}
	\cdots
	\begin{ytableau}
		i_{a + \beta}
		\\
		i_{b + \beta}
	\end{ytableau}
\label{eqn-a50-1051-27jan}
\end{align}
be a part of tableau \eqref{A70-eqn-2158-27jan}. We define the boxes in \eqref{eqn-a50-1051-27jan} as a breaking band starting from the column $j$ to the column $i$ (where $j \geq i$) if all of the following conditions hold:
\begin{enumerate}[(1)]
\item For each $\gamma \in \{0,\dots,\beta\}$
\begin{align}
	\begin{ytableau}
		i_{a+\gamma}
		\\
		i_{b+\gamma}
	\end{ytableau} 
\end{align}
is a breaking pair. 
\item The breaking pair 
\begin{align}
	\begin{ytableau}
		i_{a+\beta}
		\\
		i_{b+\beta}
	\end{ytableau} 
\end{align}
is in column $j$.
\item
The breaking pair 
\begin{align}
	\begin{ytableau}
		i_{a}
		\\
		i_{b}
	\end{ytableau} 
\end{align}
is in column $i$. 
\item The pair of boxes 
\begin{align}
	\begin{ytableau}
		i_{a+\beta + 1}
		\\
		i_{b+\beta + 1}
	\end{ytableau} 
\end{align}
(if it exists in tableau \eqref{A70-eqn-2158-27jan}) is not a breaking pair. 
\item The pair of boxes 
\begin{align}
	\begin{ytableau}
		i_{a - 1}
		\\
		i_{b - 1}
	\end{ytableau} 
\end{align}
(if it exists in tableau \eqref{A70-eqn-2158-27jan}) is not a breaking pair. 
\end{enumerate}
\end{dfn}

Note that for a breaking band starting from the column $j$ to the column $i$, we sometimes refer to $j$ and $i$ are the starting and ending columns of this breaking band, respectively.
It is also important to note that a breaking band is always considered as a part of a tableau $T$, where every row of $T$ satisfies the reverse SSYBT conditions, as in \eqref{A70-eqn-2158-27jan}. 

\begin{prop}
\label{prpa11-1205-28jan}
Let 
\small 
\begin{align}
	\begin{ytableau}
		i_a & i_{a+1}
		\\
		i_b & i_{b+1}
	\end{ytableau}
	\cdots
	\begin{ytableau}
		i_{a + \beta}
		\\
		i_{b + \beta}
	\end{ytableau}
\label{eqna56-1256-27jan}
\end{align}
\normalsize
be a breaking band starting from the column $j$ to the column $i$ in the tableau \eqref{A70-eqn-2158-27jan}. Then, the following statements hold:
\begin{enumerate}[(1)]
\item 
\label{statementa-1259}
If the ending column $i > 1$, then the contribution from the boxes 
\small 
\begin{align}
	\begin{ytableau}
		i_{a-1}	& i_a & i_{a+1}
		\\
		\none & i_b & i_{b+1}
	\end{ytableau}
	\cdots
	\begin{ytableau}
		i_{a + \beta}
		\\
		i_{b + \beta}
	\end{ytableau}
\label{a57-eqn-1212-27jan}
\end{align}
\normalsize
to $(
\dualmap
\comp 
\bigg|_{
	\substack{
		q_1 = q, \\
		q_2 = q^{-1}t,\\
		q_3 = t^{-1} \\
	}
}
)
\Bigg[
\prod_{1 \leq c < d \leq k}			
\cals{D}^{(i_c,i_d)}\left(
\frac{z_d}{z_c}
; q, t
\right)
\Bigg]$
contains a factor of $(1-t\xi)^n$, where $n \in \bb{Z}^{\geq 1}$. More precisely, we can write 
\small 
\begin{align}
	&(
	\dualmap
	\comp 
	\bigg|_{
		\substack{
			q_1 = q, \\
			q_2 = q^{-1}t,\\
			q_3 = t^{-1} \\
		}
	}
	)
	\Bigg[
	\underbrace{		\prod_{1 \leq c < d \leq k}			}_{c,d \in \{a-1,a,\dots,a+\beta, b, \dots, b+\beta\}}
	\cals{D}^{(i_c,i_d)}\left(
	\frac{z_d}{z_c}
	; q, t
	\right)
	\Bigg]
	= 
	(1-t\xi)^n
	\times
	\text{other factor}, 
\label{a58-1147-27jan}
\end{align}
\normalsize
for some integer $n \in \bb{Z}^{\geq 1}$ and \say{other factor} in \eqref{a58-1147-27jan} does not contain any factor of $(1-t\xi)$ in either its numerator or denominator. 
\item 
\label{statement-b-1437}
If the ending column $i = 1$, then the contribution from the boxes 
\small 
\begin{align}
	\begin{ytableau}
		i_a & i_{a+1}
		\\
		i_b & i_{b+1}
	\end{ytableau}
	\cdots
	\begin{ytableau}
		i_{a + \beta}
		\\
		i_{b + \beta}
	\end{ytableau}
\label{a70-1203-18apr}
\end{align}
\normalsize
to $(
\dualmap
\comp 
\bigg|_{
	\substack{
		q_1 = q, \\
		q_2 = q^{-1}t,\\
		q_3 = t^{-1} \\
	}
}
)
\Bigg[
\prod_{1 \leq c < d \leq k}			
\cals{D}^{(i_c,i_d)}\left(
\frac{z_d}{z_c}
; q, t
\right)
\Bigg]$
contains a factor of $(1-t\xi)^n$, where $n \in \bb{Z}^{\geq 0}$. More precisely, we can write 
\small 
\begin{align}
	&(
	\dualmap
	\comp 
	\bigg|_{
		\substack{
			q_1 = q, \\
			q_2 = q^{-1}t,\\
			q_3 = t^{-1} \\
		}
	}
	)
	\Bigg[
	\underbrace{		\prod_{1 \leq c < d \leq k}			}_{c,d \in \{a,\dots,a+\beta, b, \dots, b+\beta\}}
	\cals{D}^{(i_c,i_d)}\left(
	\frac{z_d}{z_c}
	; q, t
	\right)
	\Bigg]
	= 
	(1-t\xi)^n
	\times
	\text{other factor}, 
\label{a59-1147-27jan}
\end{align}
\normalsize
where $n \in \bb{Z}^{\geq 0}$
and \say{other factor} in \eqref{a59-1147-27jan} does not contain any factor of $(1-t\xi)$ in either its numerator or denominator.  
\end{enumerate}
\end{prop}
\begin{proof}
\eqref{statementa-1259}
Let the ending column be $i > 1$. It is clear that we can write 
\small 
\begin{align}
&(
\dualmap
\comp 
\bigg|_{
	\substack{
		q_1 = q, \\
		q_2 = q^{-1}t,\\
		q_3 = t^{-1} \\
	}
}
)
\Bigg[
\underbrace{		\prod_{1 \leq c < d \leq k}			}_{c,d \in \{a-1,a,\dots,a+\beta, b, \dots, b+\beta\}}
\cals{D}^{(i_c,i_d)}\left(
\frac{z_d}{z_c}
; q, t
\right)
\Bigg]
= 
C_{BT} \times C_{OP}, 
\end{align}
\normalsize
where $C_{BT}$ denotes the contribution from breaking triangles in \eqref{a57-eqn-1212-27jan} and $C_{OP}$ denotes the contribution from other pairs of boxes in \eqref{a57-eqn-1212-27jan}. 

The contribution $C_{OP}$ from other pairs of boxes in \eqref{a57-eqn-1212-27jan} does not contain 
$(1-t\xi)$ as a factor in either its numerator or denominator. This is because $(1-t\xi)$ arise only from contributions associated with adjacent boxes (vertical, horizontal, or diagonal).

Therefore, to demonstrate that 
\begin{align}
	&(
	\dualmap
	\comp 
	\bigg|_{
		\substack{
			q_1 = q, \\
			q_2 = q^{-1}t,\\
			q_3 = t^{-1} \\
		}
	}
	)
	\Bigg[
	\underbrace{		\prod_{1 \leq c < d \leq k}			}_{c,d \in \{a-1,a,\dots,a+\beta, b, \dots, b+\beta\}}
	\cals{D}^{(i_c,i_d)}\left(
	\frac{z_d}{z_c}
	; q, t
	\right)
	\Bigg]
\end{align}
contains a factor of $(1-t\xi)$, we need only analyze $C_{BT}$. Since $i > 1$, the breaking band \eqref{eqna56-1256-27jan} contains a stoppable breaking triangle. 
By proposition \ref{prpa9-1257-27jan}, the contribution from breaking triangles, $C_{BT}$, must contain a factor of $(1-t\xi)^n$ for some integer $n \in \bb{Z}^{\geq 1}$. This completes the proof of statement \eqref{statementa-1259}. 

\eqref{statement-b-1437}
Let the ending column be $i = 1$. It is clear that we can write 
\small 
\begin{align}
&(
\dualmap
\comp 
\bigg|_{
	\substack{
		q_1 = q, \\
		q_2 = q^{-1}t,\\
		q_3 = t^{-1} \\
	}
}
)
\Bigg[
\underbrace{		\prod_{1 \leq c < d \leq k}			}_{c,d \in \{a,\dots,a+\beta, b, \dots, b+\beta\}}
\cals{D}^{(i_c,i_d)}\left(
\frac{z_d}{z_c}
; q, t
\right)
\Bigg]
= C_{BT} \times C_{BP} \times C_{OP}, 
\label{a64-27jan-1435}
\end{align}
\normalsize
where $C_{BT}$ is the contribution from the breaking triangles 
\begin{align}
\begin{ytableau}
	i_{a}	& i_{a+1} & i_{a+2}
	\\
	\none & i_{b+1} & i_{b+2}
\end{ytableau}
\cdots
\begin{ytableau}
	i_{a + \beta}
	\\
	i_{b + \beta}
\end{ytableau}
\end{align}
in \eqref{a70-1203-18apr}, $C_{BP}$ is the contribution from the breaking pair 
\begin{align}
\begin{ytableau}
	i_{a}
	\\
	i_{b}
\end{ytableau}
\end{align}
in \eqref{a70-1203-18apr}, 
and $C_{OP}$ is the contribution from other pairs of boxes in \eqref{a70-1203-18apr}. 

Similar to the case $i > 1$, the contribution $C_{OP}$ does not contain $(1-t\xi)$ as a factor in its numerator or denominator. This is because $(1-t\xi)$ arise only from contributions associated with adjacent boxes. 

Note that the breaking triangles appearing in 
\begin{align}
\begin{ytableau}
	i_{a+1} & i_{a+2}
	\\
	i_{b+1} & i_{b+2}
\end{ytableau}
\cdots
\begin{ytableau}
	i_{a + \beta}
	\\
	i_{b + \beta}
\end{ytableau}
\end{align}
and the breaking triangle 
\begin{align}
\begin{ytableau}
	i_{a}	& i_{a+1} 
	\\
	\none & i_{b+1} 
\end{ytableau}
\end{align}
could be either unstoppable or stoppable. Therefore, based on \textbf{Proposition \ref{prpa9-1257-27jan}}, we can conclude that there will be a factor of $(1-t\xi)^n$ where $n \in \bb{Z}^{\geq 0}$ in $C_{BT}$. 
Moreover, since 
\begin{align}
\begin{ytableau}
	i_{a}
	\\
	i_{b}
\end{ytableau} 
\end{align}
is a breaking pair, there will be no factor of $(1-t\xi)$ in $C_{BP}$. 

In conclusion, $C_{BT} \times C_{BP} \times C_{OP}$ contains a factor of $(1-t\xi)^n$ where $n \in \bb{Z}^{\geq 0}$.
This completes the proof of statement \eqref{statement-b-1437}. 
\end{proof}

\begin{thm}
\label{a12-28jan-1342}
Let 
\begin{align}
	\begin{ytableau}
		i_1 & i_2  & \none[\dots] & i_{\lambda_1 - 1}
		& i_{\lambda_1}
		\\
		i_{\lambda_1 + 1} & i_{\lambda_1 + 2} &  \none[\dots]
		& i_{\lambda_1 + \lambda_2} \\
		\none[\vdots] & \none[\vdots]
		& \none[\vdots]
	\end{ytableau}
	\label{A71-eqn-1552-27jan}
\end{align}
be a tableau of shape $\lambda$ constructed by arranging $(i_1,\dots,i_k) \in \{1,\dots,N+M\}^k$ into the form of a Young diagram $\lambda$. Suppose that every row of the tableau \eqref{A71-eqn-1552-27jan} satisfies the RSSYBT condition. If the tableau \eqref{A71-eqn-1552-27jan} contains at least one breaking band, then 
\small 
\begin{align}
	&(
	\dualmap
	\comp 
	\bigg|_{
		\substack{
			q_1 = q, \\
			q_2 = q^{-1}t,\\
			q_3 = t^{-1} \\
		}
	}
	)
	\Bigg[
	\prod_{1 \leq c < d \leq k}
	\cals{D}^{(i_c,i_d)}\left(
	\frac{z_d}{z_c}
	; q, t
	\right)
	\Bigg]
	= 
	\frac{1}{	(1-t\xi)^n	}
	\times
	\text{other factors}
	\label{a72-eqn-1600-27jan}
\end{align}
\normalsize
where $n \leq \ell(\lambda) - 2$ and \say{other factors} in \eqref{a72-eqn-1600-27jan}
does not contain any factor of $(1-t\xi)$ in either its numerator or denominator.  
\end{thm}

\begin{nota}\mbox{}
\begin{align*}
&T\left( \,\,
\begin{ytableau}
	i_{a}	& i_{a+1} & i_{a+2}
	\\
	i_b & i_{b+1} & i_{b+2}
\end{ytableau}
\cdots
\begin{ytableau}
	i_{a + \beta}
	\\
	i_{b + \beta}
\end{ytableau}
\,\,
\right)
=
T\left(		\,\,
\begin{ytableau}
	i_{a}	& i_{a+1} & i_{a+2}
	\\
	\none & i_{b+1} & i_{b+2}
\end{ytableau}
\cdots
\begin{ytableau}
	i_{a + \beta}
	\\
	i_{b + \beta}
\end{ytableau}
\,\,
\right)
\\
&:=
(
\dualmap
\comp 
\bigg|_{
	\substack{
		q_1 = q, \\
		q_2 = q^{-1}t,\\
		q_3 = t^{-1} \\
	}
}
)
\Bigg[
\underbrace{	\prod_{1 \leq c < d \leq k}					}_{
c,d \in \{a,a+1,b+1\}
}
\cals{D}^{(i_c,i_d)}\left(
\frac{z_d}{z_c}
; q, t
\right)
\times \cdots
\times
\underbrace{	\prod_{1 \leq c < d \leq k}					}_{
	c,d \in \{a+\beta-1,a+\beta,b+\beta\}
}
\cals{D}^{(i_c,i_d)}\left(
\frac{z_d}{z_c}
; q, t
\right)
\Bigg]
\end{align*}
\begin{align*}
P\left(
\,\,
\begin{ytableau}
	i_{a}
	\\
	i_{b}
\end{ytableau} 
\,\,
\right)
:=
(
\dualmap
\comp 
\bigg|_{
	\substack{
		q_1 = q, \\
		q_2 = q^{-1}t,\\
		q_3 = t^{-1} \\
	}
}
)
\Bigg[
\cals{D}^{(i_a,i_b)}\left(
\frac{z_b}{z_a}
; q, t
\right)
\Bigg]
\end{align*}
\end{nota}

\begin{proof}
By the assumption of the theorem, the RSSYBT conditions is satisfied for every row of the tableau
\eqref{A71-eqn-1552-27jan}. This implies that any breaking occurs only in the column direction. 
Therefore, the vertical pairs of boxes 
\begin{align}
\begin{ytableau}
	i_{a}
	\\
	i_{b}
\end{ytableau} 
\end{align}
in tableau \eqref{A71-eqn-1552-27jan} can be classified into three types:
\begin{enumerate}[(1)]
\item Vertical pairs of boxes which are breaking pairs, i.e. contained in a breaking band. 
\item 
\label{con2-1645-27jan}
Vertical pairs of boxes which are not breaking pairs, and is not in the first column of tableau \eqref{A71-eqn-1552-27jan}.
\item 
\label{con3-1645-27jan}
Vertical pairs of boxes which are not breaking pairs, and is in the first column of tableau \eqref{A71-eqn-1552-27jan}.
\end{enumerate}
Therefore, we can write 
\footnotesize
\begin{align}
&(
\dualmap
\comp 
\bigg|_{
	\substack{
		q_1 = q, \\
		q_2 = q^{-1}t,\\
		q_3 = t^{-1} \\
	}
}
)
\Bigg[
\prod_{1 \leq c < d \leq k}
\cals{D}^{(i_c,i_d)}\left(
\frac{z_d}{z_c}
; q, t
\right)
\Bigg]
\\
&= 
\prod
\left\{
T\left(
\begin{ytableau}
	i_{a-1}	& i_a & i_{a+1}
	\\
	\none & i_b & i_{b+1}
\end{ytableau}
\cdots
\begin{ytableau}
	i_{a + \beta}
	\\
	i_{b + \beta}
\end{ytableau}
\right)
\;\middle\vert\;
\begin{array}{@{}l@{}}
	\begin{ytableau}
		i_{a}	& i_{a+1} & i_{a+2}
		\\
		i_b & i_{b+1} & i_{b+2}
	\end{ytableau}
	\cdots
	\begin{ytableau}
		i_{a + \beta}
		\\
		i_{b + \beta}
	\end{ytableau}
	\\
	\text{ is a breaking band such that } 
	\\
	\begin{ytableau}
		i_{a}
		\\
		i_{b}
	\end{ytableau} 
	\text{ is not in the first column}. 
\end{array}
\right\}
\label{a74-28jan-1241}
\\
&\times
\prod
\left\{
T\left(
\begin{ytableau}
	i_{a}	& i_{a+1} & i_{a+2}
	\\
	i_b & i_{b+1} & i_{b+2}
\end{ytableau}
\cdots
\begin{ytableau}
	i_{a + \beta}
	\\
	i_{b + \beta}
\end{ytableau}
\right)
\label{a75-28jan-1241}
\;\middle\vert\;
\begin{array}{@{}l@{}}
\begin{ytableau}
	i_{a}	& i_{a+1} & i_{a+2}
	\\
	i_b & i_{b+1} & i_{b+2}
\end{ytableau}
\cdots
\begin{ytableau}
	i_{a + \beta}
	\\
	i_{b + \beta}
\end{ytableau}
\\
\text{ is a breaking band such that } 
\\
\begin{ytableau}
	i_{a}
	\\
	i_{b}
\end{ytableau} 
\text{ is in the first column}. 
\end{array}
\right\}
\\
&\times\prod
\left\{
P\left(
\begin{ytableau}
	i_{a}
	\\
	i_{b}
\end{ytableau} 
\right)
\;\middle\vert\;
\begin{array}{@{}l@{}}
	\begin{ytableau}
		i_{a}	& i_{a+1} & i_{a+2}
		\\
		i_b & i_{b+1} & i_{b+2}
	\end{ytableau}
	\cdots
	\begin{ytableau}
		i_{a + \beta}
		\\
		i_{b + \beta}
	\end{ytableau}
	\\
	\text{ is a breaking band such that } 
	\\
	\begin{ytableau}
		i_{a}
		\\
		i_{b}
	\end{ytableau} 
	\text{ is in the first column}. 
\end{array}
\right\}
\times\prod
\left\{
T\left(
\begin{ytableau}
	i_{a-1} & i_a
	\\
	\none & i_b
\end{ytableau}
\right)
\;\middle\vert\;
\begin{array}{@{}l@{}}
	\begin{ytableau}
		i_{a}
		\\
		i_{b}
	\end{ytableau} 
	\text{ belongs to the type \eqref{con2-1645-27jan}}
\end{array}
\right\}
\\
&\times\prod
\left\{
P\left(
\begin{ytableau}
	i_{a}
	\\
	i_{b}
\end{ytableau}
\right)
\label{a76-28jan-1241}
\;\middle\vert\;
\begin{array}{@{}l@{}}
	\begin{ytableau}
		i_{a}
		\\
		i_{b}
	\end{ytableau} 
	\text{ belongs to the type \eqref{con3-1645-27jan}}
\end{array}
\right\}
\\
&\times\prod
\left\{
\text{Contribution from other pairs of boxes in } \eqref{A71-eqn-1552-27jan}
\right\},
\label{eqn-a88-2144-17aug}
\end{align}
\normalsize
where in equations \eqref{a74-28jan-1241} - \eqref{eqn-a88-2144-17aug}, the notation $\{\cdot\}$ denotes a \textbf{multiset}. Unlike a standard set where the number of occurrences of elements are not taken into account, a multiset is a collection of elements in which the number of occurrences of each element is significant. 

Following the reasoning similar to that used in the proof of 
\textbf{Proposition \ref{prpa11-1205-28jan}}, we can conclude that 
\small 
\begin{align}
\prod
\left\{
\text{Contribution from other pairs of boxes in } \eqref{A71-eqn-1552-27jan}
\right\}, 
\end{align}
\normalsize
does not contain any factor of $(1-t\xi)$ in either the numerator or the denominator. 
Furthermore, the formula for $\cals{D}^{(i_c,i_d)}\left(
\frac{z_d}{z_c}
; q, t
\right)$,
as given in equation \eqref{eqn126-1300-4dec} indicates that 
the contribution from breaking pairs likewise does not contain the 
factor $(1-t\xi)$. Moreover, by applying triangle-form cancellation, as described in \textbf{Proposition \ref{prpa3-1848-23jan}}, it can be shown that the product 
\small 
\begin{align}
\prod
\left\{
T\left(
\begin{ytableau}
	i_{a-1} & i_a
	\\
	\none & i_b
\end{ytableau}
\right)
\;\middle\vert\;
\begin{array}{@{}l@{}}
	\begin{ytableau}
		i_{a}
		\\
		i_{b}
	\end{ytableau} 
	\text{ belongs to the constituent \eqref{con2-1645-27jan}}
\end{array}
\right\}
\end{align}
\normalsize
does not contain any factor of the form $(1-t\xi)$. 

Consequently, to determine the overall power of the factor $(1-t\xi)$ appearing in 
\begin{align*}
(
\dualmap
\comp 
\bigg|_{
	\substack{
		q_1 = q, \\
		q_2 = q^{-1}t,\\
		q_3 = t^{-1} \\
	}
}
)
\Bigg[
\prod_{1 \leq c < d \leq k}
\cals{D}^{(i_c,i_d)}\left(
\frac{z_d}{z_c}
; q, t
\right)
\Bigg]
\end{align*}
it is sufficient to consider only the quantities in 
\eqref{a74-28jan-1241}, 
\eqref{a75-28jan-1241}, and 
\eqref{a76-28jan-1241}. 

Now, suppose the tableau \eqref{A71-eqn-1552-27jan} contains a total of $g$ breaking bands. Let $g_1$ be the number of these bands whose ending column is column $1$. Thus, $g-g_1$ bands end in a column other than the column $1$. Applying \textbf{Proposition \ref{prpa11-1205-28jan}}, we find that \eqref{a74-28jan-1241} contributes a factor of $(1 - t\xi)^{n_1}$ where $n_1 \in \bb{Z}^{\geq g - g_1}$. Similarly, from the same proposition, \eqref{a75-28jan-1241} contributes a factor of $(1-t\xi)^{n_2}$ where $n_2 \in \bb{Z}^{\geq 0}$.

Since there are $g_1$ bands ending in the column $1$, there are $\ell(\lambda) - 1 - g_1$ pairs of adjacent boxes in the column $1$ that do not break the reverse SSYBT conditions.
This results in a factor of 
\begin{align}
\frac{1}{(1-t\xi)^{\ell(\lambda) - 1 - g_1}}
\end{align}
in \eqref{a76-28jan-1241}. 

Specifically,
\begin{align}
\eqref{a74-28jan-1241}
\times 
\eqref{a75-28jan-1241}
\times
\eqref{a76-28jan-1241}
=
\frac{
	(1 - t\xi)^{n_1 + n_2}
}{
	(1 - t\xi)^{\ell(\lambda) - 1 - g_1}
}
\times
\text{other factors}
\end{align}
where the \say{other factors} do not contain $(1-t\xi)$ in either their numerator or denominator.

Note that $n_1 + n_2 + g_1 \in \bb{Z}^{\geq g}$. Since we assume that there is at least one breaking band in 
tableau \eqref{A71-eqn-1552-27jan}, we have $g \geq 1$.
It then follows from $n_1 + n_2 + g_1 \geq g$ that 
$n_1 + n_2 + g_1 \geq 1$. This inequality implies the assertion of \textbf{Theorem \ref{a12-28jan-1342}}. 
\end{proof}

\subsection{Proof of lemma \ref{lemm43-1156-22jan}}
\label{subsec-a3-1523-27jan}

From equation \eqref{d7-eqn-1830-21jan}, we know that 
\small 
\begin{align}
&
\lim_{\xi \rightarrow t^{-1}}\,\,
(
\dualmap
\comp 
\bigg|_{
	\substack{
		q_1 = q, \\
		q_2 = q^{-1}t,\\
		q_3 = t^{-1} \\
	}
}
)
\left(
\cals{N}_{\lambda}(z_1,\dots,z_k )
\times
\prod_{1 \leq i < j \leq k}f^{\vec{c}}_{11}\left(\frac{z_j}{z_i} \right)
\times
\langle 0 |\widetilde{T}^{\vec{c},\vec{u}}_{1}(z_1 )\cdots \widetilde{T}^{\vec{c},\vec{u}}_{1}(z_k )|0\rangle
\right)
\\
&= 
\lim_{\xi \rightarrow t^{-1}}\,\,
(
\dualmap
\comp 
\bigg|_{
	\substack{
		q_1 = q, \\
		q_2 = q^{-1}t,\\
		q_3 = t^{-1} \\
	}
}
)
\left(
\cals{A}
\right)
\times
\notag 
\\
&\hspace{0.3cm} \times \sum_{i_1 = 1}^{N+M}
\cdots
\sum_{i_k = 1}^{N+M}
\lim_{\xi \rightarrow t^{-1}}\,\,
(
\dualmap
\comp 
\bigg|_{
	\substack{
		q_1 = q, \\
		q_2 = q^{-1}t,\\
		q_3 = t^{-1} \\
	}
}
)
\Bigg[
y_{i_1}\cdots y_{i_k}
u_{i_1}\cdots u_{i_k}
\notag 
\\
&\hspace{4.7cm}\times
\Delta\left(	q_3^{-\frac{1}{2}}\frac{z_{\lambda_1 + 1}}{z_1}				\right)^{-1}
\times 
\cdots 
\times 
\Delta\left(	q_3^{-\frac{1}{2}}\frac{z_{\sum_{j = 1}^{\ell(\lambda) - 1}\lambda_j+1}}{z_{\sum_{j = 1}^{\ell(\lambda) - 2}\lambda_j+1}}				\right)^{-1}
\times
\prod_{1 \leq a < b \leq k}
\cals{D}^{(i_a,i_b)}\left(
\frac{z_b}{z_a}
; q, t
\right)
\Bigg]
\notag
\end{align}
\normalsize
where $\lim_{\xi \rightarrow t^{-1}}\,\,
(
\dualmap
\comp 
\bigg|_{
	\substack{
		q_1 = q, \\
		q_2 = q^{-1}t,\\
		q_3 = t^{-1} \\
	}
}
)
\left(
\cals{A}
\right)$
is a nonzero number. By using \textbf{Lemma \ref{lema2-1452-28jan}}, we get that 
\small 
\begin{align}
	&
	\lim_{\xi \rightarrow t^{-1}}\,\,
	(
	\dualmap
	\comp 
	\bigg|_{
		\substack{
			q_1 = q, \\
			q_2 = q^{-1}t,\\
			q_3 = t^{-1} \\
		}
	}
	)
	\left(
	\cals{N}_{\lambda}(z_1,\dots,z_k )
	\times
	\prod_{1 \leq i < j \leq k}f^{\vec{c}}_{11}\left(\frac{z_j}{z_i} \right)
	\times
	\langle 0 |\widetilde{T}^{\vec{c},\vec{u}}_{1}(z_1 )\cdots \widetilde{T}^{\vec{c},\vec{u}}_{1}(z_k )|0\rangle
	\right)
	\label{eqn-a82-1643-28jan}
	\\
	&= 
	\lim_{\xi \rightarrow t^{-1}}\,\,
	(
	\dualmap
	\comp 
	\bigg|_{
		\substack{
			q_1 = q, \\
			q_2 = q^{-1}t,\\
			q_3 = t^{-1} \\
		}
	}
	)
	\left(
	\cals{A}
	\right)
	\times
	\notag 
	\\
	&\hspace{0.3cm} \times 
	\underbrace{				
	\sum_{i_1 = 1}^{N+M}
	\cdots
	\sum_{i_k = 1}^{N+M}
	}_{
	(i_1,\dots,i_k) \in \text{Row-RSSYBT}(N,M;\lambda)
	}
	\lim_{\xi \rightarrow t^{-1}}\,\,
	(
	\dualmap
	\comp 
	\bigg|_{
		\substack{
			q_1 = q, \\
			q_2 = q^{-1}t,\\
			q_3 = t^{-1} \\
		}
	}
	)
	\Bigg[
	y_{i_1}\cdots y_{i_k}
	u_{i_1}\cdots u_{i_k}
	\notag 
	\\
	&\hspace{4.7cm}\times
	\Delta\left(	q_3^{-\frac{1}{2}}\frac{z_{\lambda_1 + 1}}{z_1}				\right)^{-1}
	\times 
	\cdots 
	\times 
	\Delta\left(	q_3^{-\frac{1}{2}}\frac{z_{\sum_{j = 1}^{\ell(\lambda) - 1}\lambda_j+1}}{z_{\sum_{j = 1}^{\ell(\lambda) - 2}\lambda_j+1}}				\right)^{-1}
	\times
	\prod_{1 \leq a < b \leq k}
	\cals{D}^{(i_a,i_b)}\left(
	\frac{z_b}{z_a}
	; q, t
	\right)
	\Bigg]
	\notag
\end{align}
\normalsize
where 
\small 
\begin{align}
\text{Row-RSSYBT}(N,M;\lambda)
:= 
\left\{
(i_1,\dots,i_k) \in \{1,\dots,N+M\}^k
\;\middle\vert\;
\begin{array}{@{}l@{}}
\begin{ytableau}
	i_1 & i_2  & \none[\dots] & i_{\lambda_1 - 1}
	& i_{\lambda_1}
	\\
	i_{\lambda_1 + 1} & i_{\lambda_1 + 2} &  \none[\dots]
	& i_{\lambda_1 + \lambda_2} \\
	\none[\vdots] & \none[\vdots]
	& \none[\vdots]
\end{ytableau}
\\
\text{has shape $\lambda$, and all of }
\\
\text{its rows satisfy RSSYBT conditions.}
\end{array}
\right\}. 
\end{align}
\normalsize

We know that in 
\small 
\begin{align}
(
\dualmap
\comp 
\bigg|_{
	\substack{
		q_1 = q, \\
		q_2 = q^{-1}t,\\
		q_3 = t^{-1} \\
	}
}
)
\Bigg[
\Delta\left(	q_3^{-\frac{1}{2}}\frac{z_{\lambda_1 + 1}}{z_1}				\right)^{-1}
\times 
\cdots 
\times 
\Delta\left(	q_3^{-\frac{1}{2}}\frac{z_{\sum_{j = 1}^{\ell(\lambda) - 1}\lambda_j+1}}{z_{\sum_{j = 1}^{\ell(\lambda) - 2}\lambda_j+1}}				\right)^{-1}
\Bigg], 
\end{align}
\normalsize
there are exactly $\ell(\lambda) - 1$ numbers of $(1-t\xi)$. Therefore, according to \textbf{Theorem \ref{a12-28jan-1342}}, we obtain that 
\small 
\begin{align}
&\underbrace{				
	\sum_{i_1 = 1}^{N+M}
	\cdots
	\sum_{i_k = 1}^{N+M}
}_{
	(i_1,\dots,i_k) \in \text{Row-RSSYBT}(N,M;\lambda)
}
\lim_{\xi \rightarrow t^{-1}}\,\,
(
\dualmap
\comp 
\bigg|_{
	\substack{
		q_1 = q, \\
		q_2 = q^{-1}t,\\
		q_3 = t^{-1} \\
	}
}
)
\Bigg[
y_{i_1}\cdots y_{i_k}
u_{i_1}\cdots u_{i_k}
\label{eqn-a85-1642-28jan}
\\
&\hspace{4.7cm}\times
\Delta\left(	q_3^{-\frac{1}{2}}\frac{z_{\lambda_1 + 1}}{z_1}				\right)^{-1}
\times 
\cdots 
\times 
\Delta\left(	q_3^{-\frac{1}{2}}\frac{z_{\sum_{j = 1}^{\ell(\lambda) - 1}\lambda_j+1}}{z_{\sum_{j = 1}^{\ell(\lambda) - 2}\lambda_j+1}}				\right)^{-1}
\times
\prod_{1 \leq a < b \leq k}
\cals{D}^{(i_a,i_b)}\left(
\frac{z_b}{z_a}
; q, t
\right)
\Bigg]
\notag
\\
&= 
\underbrace{				
	\sum_{i_1 = 1}^{N+M}
	\cdots
	\sum_{i_k = 1}^{N+M}
}_{
	(i_1,\dots,i_k) \in \text{RSSYBT}(N,M;\lambda)
}
\lim_{\xi \rightarrow t^{-1}}\,\,
(
\dualmap
\comp 
\bigg|_{
	\substack{
		q_1 = q, \\
		q_2 = q^{-1}t,\\
		q_3 = t^{-1} \\
	}
}
)
\Bigg[
y_{i_1}\cdots y_{i_k}
u_{i_1}\cdots u_{i_k}
\notag 
\\
&\hspace{4.7cm}\times
\Delta\left(	q_3^{-\frac{1}{2}}\frac{z_{\lambda_1 + 1}}{z_1}				\right)^{-1}
\times 
\cdots 
\times 
\Delta\left(	q_3^{-\frac{1}{2}}\frac{z_{\sum_{j = 1}^{\ell(\lambda) - 1}\lambda_j+1}}{z_{\sum_{j = 1}^{\ell(\lambda) - 2}\lambda_j+1}}				\right)^{-1}
\times
\prod_{1 \leq a < b \leq k}
\cals{D}^{(i_a,i_b)}\left(
\frac{z_b}{z_a}
; q, t
\right)
\Bigg]
\notag
\end{align}
\normalsize
where $\text{RSSYBT}(N,M;\lambda)$ is the set defined in \eqref{eqn46-1534-28jan}. Substituting \eqref{eqn-a85-1642-28jan} into \eqref{eqn-a82-1643-28jan}, and employing \eqref{d7-1950-21jan}, we then obtain  
\small 
\begin{align}
	&
	\lim_{\xi \rightarrow t^{-1}}\,\,
	(
	\dualmap
	\comp 
	\bigg|_{
		\substack{
			q_1 = q, \\
			q_2 = q^{-1}t,\\
			q_3 = t^{-1} \\
		}
	}
	)
	\left(
	\cals{N}_{\lambda}(z_1,\dots,z_k )
	\times
	\prod_{1 \leq i < j \leq k}f^{\vec{c}}_{11}\left(\frac{z_j}{z_i} \right)
	\times
	\langle 0 |\widetilde{T}^{\vec{c}}_{1}(z_1 )\cdots \widetilde{T}^{\vec{c}}_{1}(z_k )|0\rangle
	\right)
	\\
	&= 
	\underbrace{				
		\sum_{i_1 = 1}^{N+M}
		\cdots
		\sum_{i_k = 1}^{N+M}
	}_{
		(i_1,\dots,i_k) \in 
		\text{RSSYBT}(N,M;\lambda)
	}
	\lim_{\xi \rightarrow t^{-1}}\,\,
	(
	\dualmap
	\comp 
	\bigg|_{
		\substack{
			q_1 = q, \\
			q_2 = q^{-1}t,\\
			q_3 = t^{-1} \\
		}
	}
	)
	\bigg[
	y_{i_1}\cdots y_{i_k}
	u_{i_1}\cdots u_{i_k}
	\notag \\
	&\hspace{7.5cm}\times
	\cals{N}_\lambda(z_1,\dots,z_k)
	\times
	\prod_{1 \leq a < b \leq k}
	\cals{D}^{(i_a,i_b)}\left(
	\frac{z_b}{z_a}
	; q, t
	\right)
	\bigg],
	\notag 
\end{align}
\normalsize
as desired. So we have proved \textbf{Lemma \ref{lemm43-1156-22jan}}.

\section{Proof of lemma \ref{lemm54-1412-2feb}}
\label{secappB-1410-2feb}

The goal of this appendix is to provide a proof for \textbf{Lemma \ref{lemm54-1412-2feb}}. It is important to note that the statement of \textbf{Lemma \ref{lemm54-1412-2feb}} holds under the induction hypothesis and the assumption given in equation \eqref{eqn53-1641-1feb}. Recall from equation \eqref{eqn21-1749-1jan} that $\lambda^{(1)}$ can always be written as 

\begin{align}
	\lambda^{(1)} = (\lambda^{(1) \, m_1}_1		\lambda^{(1) \, m_2}_{m_1 + 1}		\lambda^{(1) \, m_3}_{m_1 + m_2 + 1}	\cdots \lambda^{(1) \, m_n}_{m_1 + \cdots + m_{n-1} + 1}	). 
\end{align}
where $\lambda_1^{(1)}, \lambda_{m_1 + 1}^{(1)}, \dots, \lambda_{\sum_{i = 1}^{n-1}m_i+ 1}^{(1)} \in \bb{Z}^{\geq 1}$ and $m_1,m_2,\dots, m_n \in \bb{Z}^{\geq 1}$ satisfy $m_1 + \cdots + m_n = \ell(\lambda^{(1)})$.

From \textbf{Proposition \ref{prp51-1812-1feb}}, it follows that any box with the number $1$ must be in the last row of its corresponding row interval. Therefore, without loss of generality, we may assume that the rows 
containing boxes with the number $1$ are precisely the rows $m_1 + \cdots + m_{\zeta_1}, \dots, m_1 + \cdots + m_{\zeta_r}$ (where $1 \leq \zeta_1 < \dots < \zeta_r \leq n$), and that the number of boxes with the number $1$ in these rows are $c_1, \dots, c_r$, respectively.

\begin{lem}
\label{lemb1-1353-3feb}
	\begin{align}
	\psi_{\lambda^{(1)}/\lambda^{(2)}}(q,t)
	&= 
	\prod_{j = 1}^{r}
	\prod_{\gamma = 0}^{n - \zeta_j}
	\prod_{\beta = 0}^{c_j - 1}
	\frac{
		1 - q^{\lambda^{(2)}_{m_1 + \cdots + m_{\zeta_j}} - \lambda^{(2)}_{m_1 + \cdots + m_{\zeta_j + \gamma}} + \beta}t^{
			m_1 + \cdots + m_{\zeta_j + \gamma} - (m_1 + \cdots + m_{\zeta_j}) + 1
		}
	}{
		1 - q^{\lambda^{(2)}_{m_1 + \cdots + m_{\zeta_j}} - \lambda^{(1)}_{m_1 + \cdots + m_{\zeta_j + \gamma} + 1} + \beta}t^{
			m_1 + \cdots + m_{\zeta_j + \gamma} - (m_1 + \cdots + m_{\zeta_j}) + 1
		}
	}
	\\
	&\times
	\prod_{j = 1}^{r}
	\prod_{\gamma = 0}^{n - \zeta_j}
	\prod_{\beta = 0}^{c_j - 1}
	\frac{
		1 - q^{\lambda^{(2)}_{m_1 + \cdots + m_{\zeta_j}} - \lambda^{(1)}_{m_1 + \cdots + m_{\zeta_j + \gamma} + 1} + 1 + \beta}t^{
			m_1 + \cdots + m_{\zeta_j + \gamma} - (m_1 + \cdots + m_{\zeta_j})
		}
	}{
		1 - q^{\lambda^{(2)}_{m_1 + \cdots + m_{\zeta_j}} - \lambda^{(2)}_{m_1 + \cdots + m_{\zeta_j + \gamma}} + 1 + \beta}t^{
			m_1 + \cdots + m_{\zeta_j + \gamma} - (m_1 + \cdots + m_{\zeta_j})
		}
	}
	\notag 
\end{align}
\end{lem}
\begin{proof}
It is known from equation \eqref{eqn27-2052} that 

\small
\begin{align}
	\psi_{\lambda^{(1)}/\lambda^{(2)}}(q,t)
	&= 
	\prod_{
		1 \leq \alpha \leq \beta \leq \ell(\lambda^{(2)})		}
	\frac{
		f_{q,t}\left(q^{\lambda^{(2)}_{\alpha} - \lambda^{(2)}_{\beta}}t^{\beta - \alpha}\right)
		f_{q,t}\left(q^{\lambda^{(1)}_{\alpha} - \lambda^{(1)}_{\beta + 1}}t^{\beta - \alpha}\right)
	}{
		f_{q,t}\left(q^{\lambda^{(1)}_{\alpha} - \lambda^{(2)}_{\beta}}t^{\beta - \alpha}\right)
		f_{q,t}\left(q^{\lambda^{(2)}_{\alpha} - \lambda^{(1)}_{\beta + 1}}t^{\beta - \alpha}\right)
	}
	\\
	&= 
	\prod_{j = 1}^{r}
	\underbrace{		\prod_{
			1 \leq \alpha \leq \beta \leq \ell(\lambda^{(2)})		}						}_{
		\alpha = m_1 + \cdots + m_{\zeta_j}
	}
	\frac{
		f_{q,t}\left(q^{\lambda^{(2)}_{\alpha} - \lambda^{(2)}_{\beta}}t^{\beta - \alpha}\right)
		f_{q,t}\left(q^{\lambda^{(2)}_{\alpha} - \lambda^{(1)}_{\beta + 1} + c_j}t^{\beta - \alpha}\right)
	}{
		f_{q,t}\left(q^{\lambda^{(2)}_{\alpha} - \lambda^{(1)}_{\beta + 1}}t^{\beta - \alpha}\right)
		f_{q,t}\left(q^{\lambda^{(2)}_{\alpha} - \lambda^{(2)}_{\beta} + c_j}t^{\beta - \alpha}\right)
	}
	\notag 
\end{align}
\normalsize
We note that an index $\beta$ contributes $1$ if $\lambda^{(2)}_{\beta} = \lambda^{(1)}_{\beta + 1}$. 
Therefore, it suffices to consider only $\beta$ for which $\lambda^{(2)}_{\beta} \neq \lambda^{(1)}_{\beta + 1}$. 

So, for $\alpha = m_1 + \cdots + m_{\zeta_j}$, we consider only $\beta$ within the set 
\begin{align}
	\beta \in 
	\{
	m_1 + \cdots + m_{\zeta_j}, 
	m_1 + \cdots + m_{\zeta_j+1}, m_1 + \cdots + m_{\zeta_j + 2}, \cdots, m_1 + \cdots + m_n
	\}.
\end{align}
Thus, we obtain that 

\small 
\begin{align}
	\psi_{\lambda^{(1)}/\lambda^{(2)}}(q,t)
	&= 
	\prod_{j = 1}^{r}
	\prod_{\gamma = 0}^{n - \zeta_j}
	\Bigg[
	\frac{
		f_{q,t}\left(q^{\lambda^{(2)}_{m_1 + \cdots + m_{\zeta_j}} - \lambda^{(2)}_{m_1 + \cdots + m_{\zeta_j + \gamma}} } t^{m_1 + \cdots + m_{\zeta_j + \gamma} - (m_1 + \cdots + m_{\zeta_j})}\right)
	}{
		f_{q,t}\left(q^{\lambda^{(2)}_{m_1 + \cdots + m_{\zeta_j}} - \lambda^{(1)}_{m_1 + \cdots + m_{\zeta_j + \gamma} + 1} } t^{m_1 + \cdots + m_{\zeta_j + \gamma} - (m_1 + \cdots + m_{\zeta_j})}\right)
	}
	\label{eqn12-1140-18} 
	\\
	&\hspace{1.7cm}\times 
	\frac{
	f_{q,t}\left(q^{\lambda^{(2)}_{m_1 + \cdots + m_{\zeta_j}} - \lambda^{(1)}_{m_1 + \cdots + m_{\zeta_j + \gamma} + 1} + c_j} t^{m_1 + \cdots + m_{\zeta_j + \gamma} - (m_1 + \cdots + m_{\zeta_j})}\right)
	}{
	f_{q,t}\left(q^{\lambda^{(2)}_{m_1 + \cdots + m_{\zeta_j}} - \lambda^{(2)}_{m_1 + \cdots + m_{\zeta_j + \gamma} } + c_j} t^{m_1 + \cdots + m_{\zeta_j + \gamma} - (m_1 + \cdots + m_{\zeta_j})}\right)
	}
	\Bigg]
	\notag 
\end{align}
\normalsize
Since $\displaystyle f(u) := \frac{(tu;q)_\infty}{(qu;q)_\infty}$, we obtain that 
\begin{align}
	\psi_{\lambda^{(1)}/\lambda^{(2)}}(q,t)
	&= 
	\prod_{j = 1}^{r}
	\prod_{\gamma = 0}^{n - \zeta_j}
	\prod_{\beta = 0}^{c_j - 1}
	\frac{
		1 - q^{\lambda^{(2)}_{m_1 + \cdots + m_{\zeta_j}} - \lambda^{(2)}_{m_1 + \cdots + m_{\zeta_j + \gamma}} + \beta}t^{
			m_1 + \cdots + m_{\zeta_j + \gamma} - (m_1 + \cdots + m_{\zeta_j}) + 1
		}
	}{
		1 - q^{\lambda^{(2)}_{m_1 + \cdots + m_{\zeta_j}} - \lambda^{(1)}_{m_1 + \cdots + m_{\zeta_j + \gamma} + 1} + \beta}t^{
			m_1 + \cdots + m_{\zeta_j + \gamma} - (m_1 + \cdots + m_{\zeta_j}) + 1
		}
	}
	\\
	&\times
	\prod_{j = 1}^{r}
	\prod_{\gamma = 0}^{n - \zeta_j}
	\prod_{\beta = 0}^{c_j - 1}
	\frac{
		1 - q^{\lambda^{(2)}_{m_1 + \cdots + m_{\zeta_j}} - \lambda^{(1)}_{m_1 + \cdots + m_{\zeta_j + \gamma} + 1} + 1 + \beta}t^{
			m_1 + \cdots + m_{\zeta_j + \gamma} - (m_1 + \cdots + m_{\zeta_j})
		}
	}{
		1 - q^{\lambda^{(2)}_{m_1 + \cdots + m_{\zeta_j}} - \lambda^{(2)}_{m_1 + \cdots + m_{\zeta_j + \gamma}} + 1 + \beta}t^{
			m_1 + \cdots + m_{\zeta_j + \gamma} - (m_1 + \cdots + m_{\zeta_j})
		}
	}. 
	\notag 
\end{align}
\end{proof}

For later convenience, we introduce

\small 
\begin{align}
&\operatorname{Factor 1} := 
\left(
\lim_{\xi \rightarrow t^{-1}}\,\,
\dualmap
\right)
\Bigg[
\underbrace{		\prod_{1 \leq \alpha < \beta \leq \ell(\lambda)}				}_{
	\substack{
		\alpha \in \operatorname{Row}(T(i_1,\dots,i_k ; \lambda)|1)
		\\
		\beta \in \operatorname{Row}(T(i_1,\dots,i_k ; \lambda)|1)
	}
}
\prod_{
	\substack{
		a \in I^{(\alpha)}
		\\
		b \in I^{(\beta)}
		\\
		a \text{ is a box with number $1$}
		\\
		b \text{ is a box with number $1$}
	}
}
\frac{
	\left(1 - \frac{z_b}{z_a}\right)
	\left(1 - t\frac{z_b}{z_a}\right)
}{
	\left(1 - q\frac{z_b}{z_a}\right)
	\left(1 - q^{-1}t\frac{z_b}{z_a}\right)
}
\Bigg], 
\label{eqnb6-1335-5mar}
\\
&\operatorname{Factor 2} :=  
\left(
\lim_{\xi \rightarrow t^{-1}}\,\,
\dualmap
\right)
\Bigg[
\underbrace{		\prod_{1 \leq \alpha < \beta \leq \ell(\lambda)}				}_{
	\substack{
		\alpha \in \operatorname{Row}(T(i_1,\dots,i_k ; \lambda)|1)
		\\
		\beta \in \operatorname{Row}(T(i_1,\dots,i_k ; \lambda)|1)
	}
}
\prod_{
	\substack{
		a \in I^{(\alpha)}
		\\
		b \in I^{(\beta)}
		\\
		a \text{ is a box with number $1$}
		\\
		b \text{ is \textbf{not} a box with number $1$}
	}
}
\frac{
	\left(1 - \frac{z_b}{z_a}\right)
	\left(1 - t\frac{z_b}{z_a}\right)
}{
	\left(1 - q\frac{z_b}{z_a}\right)
	\left(1 - q^{-1}t\frac{z_b}{z_a}\right)
}
\Bigg],
\\
&\operatorname{Factor 3} := 
\left(
\lim_{\xi \rightarrow t^{-1}}\,\,
\dualmap
\right)
\Bigg[
\underbrace{			\prod_{1 \leq \alpha < \beta \leq \ell(\lambda)}					}_{
	\substack{
		\alpha \in \operatorname{Row}(T(i_1,\dots,i_k ; \lambda)|1)
		\\
		\beta \notin  \operatorname{Row}(T(i_1,\dots,i_k ; \lambda)|1)
	}
}
\prod_{
	\substack{
		a \in I^{(\alpha)}
		\\
		b \in I^{(\beta)}
		\\
		a \text{ is a box with number $1$}
		\\
		b \text{ is \textbf{not} a box with number $1$}
	}
}
\frac{
	\left(1 - \frac{z_b}{z_a}\right)
	\left(1 - t\frac{z_b}{z_a}\right)
}{
	\left(1 - q\frac{z_b}{z_a}\right)
	\left(1 - q^{-1}t\frac{z_b}{z_a}\right)
}
\Bigg],
\\
&\operatorname{Factor 4} := 
\left(
\lim_{\xi \rightarrow t^{-1}}\,\,
\dualmap
\right)
\Bigg[
\underbrace{			\prod_{1 \leq a < b \leq k}					}_{
	\substack{
		a \text{ is a box with number $1$}
		\\
		b \text{ is \textbf{not} a box with number $1$}
	}
}
\frac{
	\left(1 - q^{-1}\frac{z_b}{z_a}\right)
	\left(1 - qt^{-1}\frac{z_b}{z_a}\right)
}{
	\left(1 - t^{-1}\frac{z_b}{z_a}\right)
	\left(1 - \frac{z_b}{z_a}\right)
}
\Bigg],
\\
&\operatorname{Factor 5} := 
\left(
\lim_{\xi \rightarrow t^{-1}}\,\,
\dualmap
\right)
\Bigg[
\underbrace{				\prod_{1 \leq a < b \leq k}							}_{
	\substack{
		a \text{ is \textbf{not} a box with number $1$}
		\\
		b \text{ is a box with number $1$}
		\\
		a \text{ and } b \text{ are in the same row}
	}
}
\frac{
	\left(1 - q\frac{z_b}{z_a}\right)
	\left(1 - q^{-1}t\frac{z_b}{z_a}\right)
}{
	\left(1 - t\frac{z_b}{z_a}\right)
	\left(1 - \frac{z_b}{z_a}\right)
}
\Bigg]. 
\label{eqnb10-1336-5mar}
\end{align}
\normalsize
So, to prove \textbf{Lemma \ref{lemm54-1412-2feb}}, we have to show that 
\begin{align}
	\psi_{\lambda^{(1)}/\lambda^{(2)}}(q,t)
	= 
	\operatorname{Factor 1}
	\times
	\operatorname{Factor 2}
	\times
	\operatorname{Factor 3}
	\times
	\operatorname{Factor 4}
	\times
	\operatorname{Factor 5}. 
\label{eqnb11-1602-5mar}
\end{align}
According to equations \eqref{eqnb6-1335-5mar} - \eqref{eqnb10-1336-5mar}, one can easily show that 
\small 
\begin{align}
&\operatorname{Factor 1} 
= 
\prod_{j = 1}^{r - 1}
\prod_{j^\prime = j+1}^{r}
\prod_{a = 1}^{c_j}
\frac{
1 - t^{m_1 + \cdots + m_{\zeta_{j^\prime}} - (m_1 + \cdots + m_{\zeta_j})}q^{\lambda^{(2)}_{	\sum_{\ell = 1}^{\zeta_j}		m_\ell} - \lambda^{(2)}_{	\sum_{\ell = 1}^{\zeta_{j^\prime}}		m_\ell} + a- 1}
}{
1 - t^{m_1 + \cdots + m_{\zeta_{j^\prime}} - (m_1 + \cdots + m_{\zeta_j})}q^{\lambda^{(2)}_{	\sum_{\ell = 1}^{\zeta_j}		m_\ell} - \lambda^{(2)}_{	\sum_{\ell = 1}^{\zeta_{j^\prime}}		m_\ell} + a-c_{j^\prime}-1}
}
\label{b12-eqn-1525-5mar}
\\
&\hspace{1.5cm}\times 
\prod_{j = 1}^{r - 1}
\prod_{j^\prime = j+1}^{r}
\prod_{a = 1}^{c_j}
\frac{
1 - t^{m_1 + \cdots + m_{\zeta_{j^\prime}} - (m_1 + \cdots + m_{\zeta_j}) - 1}q^{\lambda^{(2)}_{	\sum_{\ell = 1}^{\zeta_j}		m_\ell} - \lambda^{(2)}_{	\sum_{\ell = 1}^{\zeta_{j^\prime}}		m_\ell} + a-c_{j^\prime}}
}{
1 - t^{m_1 + \cdots + m_{\zeta_{j^\prime}} - (m_1 + \cdots + m_{\zeta_j}) - 1}q^{\lambda^{(2)}_{	\sum_{\ell = 1}^{\zeta_j}		m_\ell} - \lambda^{(2)}_{	\sum_{\ell = 1}^{\zeta_{j^\prime}}		m_\ell} +  a}
}, 
\notag 
\\
&\operatorname{Factor 2}
= 
\prod_{j = 1}^{r - 1}
\prod_{j^\prime = j+1}^{r}
\prod_{a =  1 }^{	c_j	 }
\frac{
1 - t^{m_1 + \cdots + m_{\zeta_{j^\prime}} - (m_1 + \cdots + m_{\zeta_j})}q^{\lambda^{(2)}_{	\sum_{\ell = 1}^{\zeta_j}		m_\ell}  + a-1}
}{
1 - t^{m_1 + \cdots + m_{\zeta_{j^\prime}} - (m_1 + \cdots + m_{\zeta_j})}q^{\lambda^{(2)}_{	\sum_{\ell = 1}^{\zeta_j}		m_\ell} - \lambda^{(2)}_{	\sum_{\ell = 1}^{\zeta_{j^\prime}}		m_\ell}  + a -1}
}
\label{b13-eqn-1525-5mar}
\\
&\hspace{1.5cm}\times
\prod_{j = 1}^{r - 1}
\prod_{j^\prime = j+1}^{r}
\prod_{a =  1 }^{	c_j	 }
\frac{
1 - t^{m_1 + \cdots + m_{\zeta_{j^\prime}} - (m_1 + \cdots + m_{\zeta_j}) - 1}q^{\lambda^{(2)}_{	\sum_{\ell = 1}^{\zeta_j}		m_\ell} - 		\lambda^{(2)}_{	\sum_{\ell = 1}^{\zeta_{j^\prime}}		m_\ell}	+ a}
}{
1 - t^{m_1 + \cdots + m_{\zeta_{j^\prime}} - (m_1 + \cdots + m_{\zeta_j}) - 1}q^{\lambda^{(2)}_{	\sum_{\ell = 1}^{\zeta_j}		m_\ell} +  a}
}, 
\notag 
\\
&\operatorname{Factor 3}
=
\prod_{j = 1}^{r - 1}
\underbrace{		\prod_{\gamma = 1}^{m_{\zeta_j + 1} + \cdots + m_n}					}_{
	\gamma \notin 
	\bigcup_{j^\prime = j+1}^{r}
	\{
	m_{\zeta_j + 1} + \cdots + m_{\zeta_{j^\prime}}
	\}
}
\prod_{a = 1}^{c_j}
\frac{
	\left(	1 - t^\gamma q^{\lambda^{(2)}_{\sum_{\ell = 1}^{\zeta_j}m_\ell} + a - 1}			\right)
	\left(	1 - t^{\gamma - 1} q^{\lambda^{(2)}_{\sum_{\ell = 1}^{\zeta_j}m_\ell} -	\lambda^{(2)}_{\sum_{\ell = 1}^{\zeta_j}m_\ell+\gamma}		+ a 			}			\right)
}{
	\left(	1 - t^\gamma q^{\lambda^{(2)}_{\sum_{\ell = 1}^{\zeta_j}m_\ell}  -\lambda^{(2)}_{\sum_{\ell = 1}^{\zeta_j}m_\ell+\gamma}+  a	-1}			\right)
	\left(	1 - t^{\gamma - 1} q^{\lambda^{(2)}_{\sum_{\ell = 1}^{\zeta_j}m_\ell} + a }			\right)
}
\label{b13-1335-3feb}
\\
&\hspace{1.4cm}
\times \prod_{\gamma = 1}^{m_{\zeta_r + 1} + \cdots + m_n }
\prod_{a = 1}^{c_r}
\frac{
	\left(	1 - t^\gamma q^{\lambda^{(2)}_{\sum_{\ell = 1}^{\zeta_r}m_\ell} + a - 1}			\right)
	\left(	1 - t^{\gamma - 1} q^{\lambda^{(2)}_{\sum_{\ell = 1}^{\zeta_r}m_\ell} -	\lambda^{(2)}_{\sum_{\ell = 1}^{\zeta_r}m_\ell+\gamma}		+ a 			}			\right)
}{
	\left(	1 - t^\gamma q^{\lambda^{(2)}_{\sum_{\ell = 1}^{\zeta_r}m_\ell}  -\lambda^{(2)}_{\sum_{\ell = 1}^{\zeta_r}m_\ell+\gamma}+  a	-1}			\right)
	\left(	1 - t^{\gamma - 1} q^{\lambda^{(2)}_{\sum_{\ell = 1}^{\zeta_r}m_\ell} + a }			\right)
}, 
\notag 
\end{align}
%------------------------------------
%------------------------------------
%------------------------------------
%------------------------------------
%------------------------------------
%------------------------------------
%------------------------------------
%------------------------------------
%------------------------------------
\begin{align}
\operatorname{Factor 4}
&=\prod_{j = 1}^{r - 1}
\underbrace{			\prod_{\gamma = 1}^{m_{\zeta_j +1} + \cdots + m_n}						}_{
	\gamma \notin 
	\bigcup_{j^\prime = j+1}^{r}
	\{
	m_{\zeta_j + 1} + \cdots + m_{\zeta_{j^\prime}}
	\}
}
\prod_{a = 	 1	}^{c_j}
\frac{
	\left(	1 - t^{\gamma}q^{\lambda^{(2)}_{\sum_{\ell = 1}^{\zeta_j}m_\ell}	+ a }							\right)
	\left(	1 - t^{\gamma + 1}q^{\lambda^{(2)}_{\sum_{\ell = 1}^{\zeta_j}m_\ell}	 - \lambda^{(2)}_{\sum_{\ell = 1}^{\zeta_j}m_\ell + \gamma} + a- 1}							\right)
}{
	\left(	1 - t^{\gamma}q^{\lambda^{(2)}_{\sum_{\ell = 1}^{\zeta_j}m_\ell}	-\lambda^{(2)}_{\sum_{\ell = 1}^{\zeta_j}m_\ell + \gamma} + a}							\right)
	\left(	1 - t^{\gamma + 1}q^{\lambda^{(2)}_{\sum_{\ell = 1}^{\zeta_j}m_\ell}	+a-1}							\right)
}
\label{b14-1335-3feb}
\\
&\times 
\prod_{j = 1}^{r - 1}
\underbrace{			\prod_{\gamma = 1}^{m_{\zeta_j +1} + \cdots + m_n}				}_{
	\gamma \in 
	\bigcup_{j^\prime = j+1}^{r}
	\{
	m_{\zeta_j + 1} + \cdots + m_{\zeta_{j^\prime}}
	\}
}
\prod_{a = 	 1	}^{c_j}
\frac{
	\left(	1 - t^{\gamma}q^{\lambda^{(2)}_{\sum_{\ell = 1}^{\zeta_j}m_\ell}	+ a }							\right)
	\left(	1 - t^{\gamma + 1}q^{\lambda^{(2)}_{\sum_{\ell = 1}^{\zeta_j}m_\ell}	 - \lambda^{(2)}_{\sum_{\ell = 1}^{\zeta_j}m_\ell + \gamma} + a- 1}							\right)
}{
	\left(	1 - t^{\gamma}q^{\lambda^{(2)}_{\sum_{\ell = 1}^{\zeta_j}m_\ell}	-\lambda^{(2)}_{\sum_{\ell = 1}^{\zeta_j}m_\ell + \gamma} + a}							\right)
	\left(	1 - t^{\gamma + 1}q^{\lambda^{(2)}_{\sum_{\ell = 1}^{\zeta_j}m_\ell}	+a-1}							\right)
}
\notag 
\\
&\times 
\prod_{\gamma = 1}^{m_{\zeta_r +1} + \cdots + m_n}
\prod_{a = 	 1	}^{c_r}
\frac{
	\left(	1 - t^{\gamma}q^{\lambda^{(2)}_{\sum_{\ell = 1}^{\zeta_r}m_\ell}	+ a }							\right)
	\left(	1 - t^{\gamma + 1}q^{\lambda^{(2)}_{\sum_{\ell = 1}^{\zeta_r}m_\ell}	 - \lambda^{(2)}_{\sum_{\ell = 1}^{\zeta_r}m_\ell + \gamma} + a- 1}							\right)
}{
	\left(	1 - t^{\gamma}q^{\lambda^{(2)}_{\sum_{\ell = 1}^{\zeta_r}m_\ell}	-\lambda^{(2)}_{\sum_{\ell = 1}^{\zeta_r}m_\ell + \gamma} + a}							\right)
	\left(	1 - t^{\gamma + 1}q^{\lambda^{(2)}_{\sum_{\ell = 1}^{\zeta_r}m_\ell}	+a-1}							\right)
},
\notag 
\\
\operatorname{Factor 5}
&= \prod_{j = 1}^{r}
\prod_{b = 1}^{c_j}
\frac{
	\left(1 - q^{\lambda^{(2)}_{\sum_{\ell = 1}^{\zeta_j}m_\ell} + b}\right)
	\left(1 - tq^{b - 1}\right)
}{
	\left(1 - q^{b 	}\right)
	\left(1 - tq^{\lambda^{(2)}_{\sum_{\ell = 1}^{\zeta_j}m_\ell} + b-1}\right)
}. 
\label{b15-1521-3feb}
\end{align}
%------------------------------------
%------------------------------------
%------------------------------------
\normalsize

\begin{lem}
\label{lemb2-1432-5mar}
\mbox{}
\small 
\begin{align}
	&\operatorname{Factor 3} \times \operatorname{Factor 4}
	\\
	&= 
	\widetilde{\circled{1}} 
	\times 
	\prod_{j = 1}^{r - 1}
	\prod_{a = 1}^{c_j}
	\prod_{\beta = 0}^{r - 1 - j}
	\frac{
		1 - t^{	m_{1} +  \cdots + m_{\zeta_{j + \beta + 1} - 1}	- (m_1 + \cdots + m_{\zeta_j})		}q^{ \lambda^{(2)}_{\sum_{\ell = 1}^{\zeta_j}m_\ell}	- 	\lambda^{(1)}_{\sum_{\ell = 1}^{\zeta_{j+\beta + 1}}m_\ell - 1}	 + a 	}
	}{
		1 - t^{	m_{1} +  \cdots + m_{\zeta_{j + \beta + 1} }	- (m_1 + \cdots + m_{\zeta_j}) - 1		}q^{ \lambda^{(2)}_{\sum_{\ell = 1}^{\zeta_j}m_\ell}	- 	\lambda^{(1)}_{\sum_{\ell = 1}^{\zeta_{j+\beta + 1}}m_\ell - 1}	 + a 	}
	}
	\notag 
	\\
	&\times 
	\prod_{j = 1}^{r - 1}
	\prod_{a = 1}^{c_j}
	\prod_{\beta = 0}^{r - 1 - j}
	\frac{
		1 - t^{	m_{1} +  \cdots + m_{\zeta_{j + \beta + 1} }	- (m_1 + \cdots + m_{\zeta_j}) 		}q^{ \lambda^{(2)}_{\sum_{\ell = 1}^{\zeta_j}m_\ell}	- 	\lambda^{(2)}_{\sum_{\ell = 1}^{\zeta_{j+\beta + 1}}m_\ell - 1}	 + a  - 1	}
	}{
		1 - t^{	m_{1} +  \cdots + m_{\zeta_{j + \beta + 1} - 1}	- (m_1 + \cdots + m_{\zeta_j})	+ 1	}q^{ \lambda^{(2)}_{\sum_{\ell = 1}^{\zeta_j}m_\ell}	- 	\lambda^{(2)}_{\sum_{\ell = 1}^{\zeta_{j+\beta + 1}}m_\ell - 1}	 + a  - 1	 }
	}
	\notag 
	\\
	&\times
	\prod_{j = 1}^{r - 1}
	\prod_{a = 1}^{c_j}
	\prod_{\beta = 0}^{r - j - 1}
	\frac{
		1 - t^{	m_{1} +  \cdots + m_{\zeta_{j + \beta } }	- (m_1 + \cdots + m_{\zeta_j})	+ 1	}q^{ \lambda^{(2)}_{\sum_{\ell = 1}^{\zeta_j}m_\ell}	+ a  - 1	 }
	}{
		1 - t^{	m_{1} +  \cdots + m_{\zeta_{j + \beta + 1} }	- (m_1 + \cdots + m_{\zeta_j}) 		}q^{ \lambda^{(2)}_{\sum_{\ell = 1}^{\zeta_j}m_\ell}	 + a  - 1	}
	}
	\notag 
	\\
	&\times 
	\prod_{j = 1}^{r - 1}
	\prod_{a = 1}^{c_j}
	\frac{
		1 - t^{	m_{1} +  \cdots + m_{\zeta_{r } }	- (m_1 + \cdots + m_{\zeta_j})	+ 1	}q^{ \lambda^{(2)}_{\sum_{\ell = 1}^{\zeta_j}m_\ell}	+ a  - 1	 }
	}{
		1 - t^{	m_{1} +  \cdots + m_{n}	- (m_1 + \cdots + m_{\zeta_j})	+ 1	}q^{ \lambda^{(2)}_{\sum_{\ell = 1}^{\zeta_j}m_\ell}	 + a  - 1	}
	}
	\notag 
	\\
	&\times 
	\prod_{j = 1}^{r - 1}
	\prod_{a = 1}^{c_j}
	\prod_{\beta = 0}^{r - 1 - j}
	\frac{
		1 - t^{	m_{1} +  \cdots + m_{\zeta_{j + \beta + 1} }	- (m_1 + \cdots + m_{\zeta_j}) 	- 1	}q^{ \lambda^{(2)}_{\sum_{\ell = 1}^{\zeta_j}m_\ell}	 + a  	}
	}{
		1 - t^{	m_{1} +  \cdots + m_{\zeta_{j + \beta } }	- (m_1 + \cdots + m_{\zeta_j})		}q^{ \lambda^{(2)}_{\sum_{\ell = 1}^{\zeta_j}m_\ell}	+ a  	 }
	}
	\notag 
	\\
	&\times 
	\prod_{j = 1}^{r - 1}
	\prod_{a = 1}^{c_j}
	\frac{
		1 - t^{	m_{1} +  \cdots + m_{n}	- (m_1 + \cdots + m_{\zeta_j})	}q^{ \lambda^{(2)}_{\sum_{\ell = 1}^{\zeta_j}m_\ell}	 + a  	}
	}{
		1 - t^{	m_{1} +  \cdots + m_{\zeta_{r } }	- (m_1 + \cdots + m_{\zeta_j})		}q^{ \lambda^{(2)}_{\sum_{\ell = 1}^{\zeta_j}m_\ell}	+ a  	 }
	}
	\notag 
	\\
	&\times 
	\widetilde{\circled{2}}
	\times 
	\prod_{a = 1}^{c_r}
	\frac{
		1 - tq^{ \lambda^{(2)}_{\sum_{\ell = 1}^{\zeta_r}m_\ell}	+ a  - 1	 }
	}{
		1 - t^{	m_{1} +  \cdots + m_{n}	- (m_1 + \cdots + m_{\zeta_r})	+ 1	}q^{ \lambda^{(2)}_{\sum_{\ell = 1}^{\zeta_r}m_\ell}	 + a  - 1	}
	}
	\times 
	\prod_{a = 1}^{c_r}
	\frac{
		1 - t^{	m_{1} +  \cdots + m_{n}	- (m_1 + \cdots + m_{\zeta_r})	}q^{ \lambda^{(2)}_{\sum_{\ell = 1}^{\zeta_r}m_\ell}	 + a  	}
	}{
		1 - q^{ \lambda^{(2)}_{\sum_{\ell = 1}^{\zeta_r}m_\ell}	+ a  	 }
	}
	\notag 
	\\
	&\times 
	\prod_{j = 1}^{r - 1}
	\underbrace{			\prod_{\gamma = 1}^{m_{\zeta_j +1} + \cdots + m_n}				}_{
		\gamma \in 
		\bigcup_{j^\prime = j+1}^{r}
		\{
		m_{\zeta_j + 1} + \cdots + m_{\zeta_{j^\prime}}
		\}
	}
	\prod_{a = 	 1	}^{c_j}
	\frac{
		\left(	1 - t^{\gamma}q^{\lambda^{(2)}_{\sum_{\ell = 1}^{\zeta_j}m_\ell}	+ a }							\right)
		\left(	1 - t^{\gamma + 1}q^{\lambda^{(2)}_{\sum_{\ell = 1}^{\zeta_j}m_\ell}	 - \lambda^{(2)}_{\sum_{\ell = 1}^{\zeta_j}m_\ell + \gamma} + a- 1}							\right)
	}{
		\left(	1 - t^{\gamma}q^{\lambda^{(2)}_{\sum_{\ell = 1}^{\zeta_j}m_\ell}	-\lambda^{(2)}_{\sum_{\ell = 1}^{\zeta_j}m_\ell + \gamma} + a}							\right)
		\left(	1 - t^{\gamma + 1}q^{\lambda^{(2)}_{\sum_{\ell = 1}^{\zeta_j}m_\ell}	+a-1}							\right)
	}
\notag 
\end{align}
\normalsize 
where 

\small 
\begin{align}
	\widetilde{\circled{1}}
	&:= 
	\prod_{j = 1}^{r - 1}
	\prod_{a = 1}^{c_j}
	\underbrace{			\prod_{\gamma = 0}^{n - \zeta_j - 1}							}_{
		\gamma \notin \{\zeta_{j+1} - \zeta_j - 1, \zeta_{j+2} - \zeta_j - 1, \dots, \zeta_r - \zeta_j - 1\}
	}
	\left(
	1 - q^{	\lambda^{(2)}_{m_1 + \cdots + m_{\zeta_j}} - \lambda^{(1)}_{m_1 + \cdots + m_{\zeta_j + \gamma} + 1 } 	+ a				}t^{m_1 + \cdots + m_{\zeta_j + \gamma} - (m_1 + \cdots + m_{\zeta_j}) }
	\right)
	\label{b17-1450-3feb}
	\\
	&\times 
	\prod_{j = 1}^{r - 1}
	\prod_{a = 1}^{c_j}
	\underbrace{			\prod_{\gamma = 1}^{n - \zeta_j }							}_{
		\gamma \notin \{\zeta_{j+1} - \zeta_j , \zeta_{j+2} - \zeta_j , \dots, \zeta_r - \zeta_j \}
	}
	\frac{1}{
		1 - q^{	\lambda^{(2)}_{m_1 + \cdots + m_{\zeta_j}} - \lambda^{(1)}_{m_1 + \cdots + m_{\zeta_j + \gamma} } 	+ a				}t^{m_1 + \cdots + m_{\zeta_j + \gamma} - (m_1 + \cdots + m_{\zeta_j}) }
	}
	\notag 
	\\
	&\times 
	\prod_{j = 1}^{r - 1}
	\prod_{a = 1}^{c_j}
	\underbrace{			\prod_{\gamma = 1}^{n - \zeta_j }							}_{
		\gamma \notin \{\zeta_{j+1} - \zeta_j , \zeta_{j+2} - \zeta_j , \dots, \zeta_r - \zeta_j \}
	}
	\frac{
	1 - q^{	\lambda^{(2)}_{m_1 + \cdots + m_{\zeta_j}} - \lambda^{(2)}_{m_1 + \cdots + m_{\zeta_j + \gamma} } 	+ a	- 1			}t^{m_1 + \cdots + m_{\zeta_j + \gamma } - (m_1 + \cdots + m_{\zeta_j}) + 1}
	}{
		1 - q^{	\lambda^{(2)}_{m_1 + \cdots + m_{\zeta_j}} - \lambda^{(2)}_{m_1 + \cdots + m_{\zeta_j + \gamma} } 	+ a	- 1			}t^{m_1 + \cdots + m_{\zeta_j + \gamma - 1} - (m_1 + \cdots + m_{\zeta_j}) + 1}
	}
	\notag 
	\\
	\widetilde{\circled{2}}
	&:= 
	\prod_{a = 1}^{c_r}
	\prod_{\gamma = 0}^{n - \zeta_r - 1}					
	\left(
	1 - q^{	\lambda^{(2)}_{m_1 + \cdots + m_{\zeta_r}} - \lambda^{(1)}_{m_1 + \cdots + m_{\zeta_r + \gamma} + 1 } 	+ a				}t^{m_1 + \cdots + m_{\zeta_r + \gamma} - (m_1 + \cdots + m_{\zeta_r}) }
	\right)
	\label{b18-1450-3feb}
	\\
	&\times 
	\prod_{a = 1}^{c_r}
	\prod_{\gamma = 1}^{n - \zeta_r }							
	\frac{1}{
		1 - q^{	\lambda^{(2)}_{m_1 + \cdots + m_{\zeta_r}} - \lambda^{(1)}_{m_1 + \cdots + m_{\zeta_r + \gamma} } 	+ a				}t^{m_1 + \cdots + m_{\zeta_r + \gamma} - (m_1 + \cdots + m_{\zeta_r}) }
	}
	\notag 
	\\
	&\times 
	\prod_{a = 1}^{c_r}
	\prod_{\gamma = 1}^{n - \zeta_r }						
	\frac{
		1 - q^{	\lambda^{(2)}_{m_1 + \cdots + m_{\zeta_r}} - \lambda^{(2)}_{m_1 + \cdots + m_{\zeta_r + \gamma} } 	+ a	- 1			}t^{m_1 + \cdots + m_{\zeta_r + \gamma } - (m_1 + \cdots + m_{\zeta_r}) + 1}
	}{
		1 - q^{	\lambda^{(2)}_{m_1 + \cdots + m_{\zeta_r}} - \lambda^{(2)}_{m_1 + \cdots + m_{\zeta_r + \gamma} } 	+ a	- 1			}t^{m_1 + \cdots + m_{\zeta_r + \gamma - 1} - (m_1 + \cdots + m_{\zeta_r}) + 1}
	}				
	\notag 
\end{align}
\normalsize
\end{lem}
\begin{proof}
First, one can easily show that the equations \eqref{b20-1428-5mar} - \eqref{eqnb23-1429-5mar} below hold. 
\small 
\begin{align}
&
\prod_{j = 1}^{r - 1}
\underbrace{		\prod_{\gamma = 1}^{m_{\zeta_j + 1} + \cdots + m_n}					}_{
	\gamma \notin 
	\bigcup_{j^\prime = j+1}^{r}
	\{
	m_{\zeta_j + 1} + \cdots + m_{\zeta_{j^\prime}}
	\}
}
\prod_{a = 1}^{c_j}
\frac{
		1 - t^{\gamma - 1} q^{\lambda^{(2)}_{\sum_{\ell = 1}^{\zeta_j}m_\ell} -	\lambda^{(2)}_{\sum_{\ell = 1}^{\zeta_j}m_\ell+\gamma}		+ a 			}			
}{
	1 - t^{\gamma}q^{\lambda^{(2)}_{\sum_{\ell = 1}^{\zeta_j}m_\ell}	-\lambda^{(2)}_{\sum_{\ell = 1}^{\zeta_j}m_\ell + \gamma} + a}						
}
\label{b20-1428-5mar}
\\
&\hspace{0.3cm}= 
\prod_{j = 1}^{r - 1}
\prod_{a = 1}^{c_j}
\underbrace{			\prod_{\gamma = 0}^{n - \zeta_j - 1}							}_{
	\gamma \notin \{\zeta_{j+1} - \zeta_j - 1, \zeta_{j+2} - \zeta_j - 1, \dots, \zeta_r - \zeta_j - 1\}
}
\left(
1 - q^{	\lambda^{(2)}_{m_1 + \cdots + m_{\zeta_j}} - \lambda^{(1)}_{m_1 + \cdots + m_{\zeta_j + \gamma} + 1 } 	+ a				}t^{m_1 + \cdots + m_{\zeta_j + \gamma} - (m_1 + \cdots + m_{\zeta_j}) }
\right)
\notag 
\\
&\hspace{0.3cm}\times 
\prod_{j = 1}^{r - 1}
\prod_{a = 1}^{c_j}
\underbrace{			\prod_{\gamma = 1}^{n - \zeta_j }							}_{
	\gamma \notin \{\zeta_{j+1} - \zeta_j , \zeta_{j+2} - \zeta_j , \dots, \zeta_r - \zeta_j \}
}
\frac{1}{
	1 - q^{	\lambda^{(2)}_{m_1 + \cdots + m_{\zeta_j}} - \lambda^{(1)}_{m_1 + \cdots + m_{\zeta_j + \gamma} } 	+ a				}t^{m_1 + \cdots + m_{\zeta_j + \gamma} - (m_1 + \cdots + m_{\zeta_j}) }
}
\notag 
\\
&\hspace{0.3cm}\times
\prod_{j = 1}^{r - 1}
\prod_{a = 1}^{c_j}
\prod_{\beta = 0}^{r - 1 - j}
\frac{
	1 - t^{	m_{1} +  \cdots + m_{\zeta_{j + \beta + 1} - 1}	- (m_1 + \cdots + m_{\zeta_j})		}q^{ \lambda^{(2)}_{\sum_{\ell = 1}^{\zeta_j}m_\ell}	- 	\lambda^{(1)}_{\sum_{\ell = 1}^{\zeta_{j+\beta + 1}}m_\ell - 1}	 + a 	}
}{
	1 - t^{	m_{1} +  \cdots + m_{\zeta_{j + \beta + 1} }	- (m_1 + \cdots + m_{\zeta_j}) - 1		}q^{ \lambda^{(2)}_{\sum_{\ell = 1}^{\zeta_j}m_\ell}	- 	\lambda^{(1)}_{\sum_{\ell = 1}^{\zeta_{j+\beta + 1}}m_\ell - 1}	 + a 	}
}, 
\notag 
\\
&\prod_{j = 1}^{r - 1}
\underbrace{		\prod_{\gamma = 1}^{m_{\zeta_j + 1} + \cdots + m_n}					}_{
	\gamma \notin 
	\bigcup_{j^\prime = j+1}^{r}
	\{
	m_{\zeta_j + 1} + \cdots + m_{\zeta_{j^\prime}}
	\}
}
\prod_{a = 1}^{c_j}
\frac{
	1 - t^{\gamma + 1}q^{\lambda^{(2)}_{\sum_{\ell = 1}^{\zeta_j}m_\ell}	 - \lambda^{(2)}_{\sum_{\ell = 1}^{\zeta_j}m_\ell + \gamma} + a- 1}			
}{
	1 - t^\gamma q^{\lambda^{(2)}_{\sum_{\ell = 1}^{\zeta_j}m_\ell}  -\lambda^{(2)}_{\sum_{\ell = 1}^{\zeta_j}m_\ell+\gamma}+  a	-1}			
}
\\
&\hspace{0.3cm}= 
\prod_{j = 1}^{r - 1}
\prod_{a = 1}^{c_j}
\underbrace{			\prod_{\gamma = 1}^{n - \zeta_j }							}_{
	\gamma \notin \{\zeta_{j+1} - \zeta_j , \zeta_{j+2} - \zeta_j , \dots, \zeta_r - \zeta_j \}
}
\frac{
1 - q^{	\lambda^{(2)}_{m_1 + \cdots + m_{\zeta_j}} - \lambda^{(2)}_{m_1 + \cdots + m_{\zeta_j + \gamma} } 	+ a	- 1			}t^{m_1 + \cdots + m_{\zeta_j + \gamma } - (m_1 + \cdots + m_{\zeta_j}) + 1}
}{
	1 - q^{	\lambda^{(2)}_{m_1 + \cdots + m_{\zeta_j}} - \lambda^{(2)}_{m_1 + \cdots + m_{\zeta_j + \gamma} } 	+ a	- 1			}t^{m_1 + \cdots + m_{\zeta_j + \gamma - 1} - (m_1 + \cdots + m_{\zeta_j}) + 1}
}
\notag 
\\
&\hspace{0.3cm}\times 
\prod_{j = 1}^{r - 1}
\prod_{a = 1}^{c_j}
\prod_{\beta = 0}^{r - 1 - j}
\frac{
	1 - t^{	m_{1} +  \cdots + m_{\zeta_{j + \beta + 1} }	- (m_1 + \cdots + m_{\zeta_j}) 		}q^{ \lambda^{(2)}_{\sum_{\ell = 1}^{\zeta_j}m_\ell}	- 	\lambda^{(2)}_{\sum_{\ell = 1}^{\zeta_{j+\beta + 1}}m_\ell - 1}	 + a  - 1	}
}{
	1 - t^{	m_{1} +  \cdots + m_{\zeta_{j + \beta + 1} - 1}	- (m_1 + \cdots + m_{\zeta_j})	+ 1	}q^{ \lambda^{(2)}_{\sum_{\ell = 1}^{\zeta_j}m_\ell}	- 	\lambda^{(2)}_{\sum_{\ell = 1}^{\zeta_{j+\beta + 1}}m_\ell - 1}	 + a  - 1	 }
}, 
\notag 
\\
&
\prod_{j = 1}^{r - 1}
\underbrace{		\prod_{\gamma = 1}^{m_{\zeta_j + 1} + \cdots + m_n}					}_{
	\gamma \notin 
	\bigcup_{j^\prime = j+1}^{r}
	\{
	m_{\zeta_j + 1} + \cdots + m_{\zeta_{j^\prime}}
	\}
}
\prod_{a = 1}^{c_j}
\frac{
	1 - t^\gamma q^{\lambda^{(2)}_{\sum_{\ell = 1}^{\zeta_j}m_\ell} + a - 1}			
}{
	1 - t^{\gamma + 1}q^{\lambda^{(2)}_{\sum_{\ell = 1}^{\zeta_j}m_\ell}	+a-1}
}
\\
&\hspace{0.3cm}= 
\prod_{j = 1}^{r - 1}
\prod_{a = 1}^{c_j}
\prod_{\beta = 0}^{r - j - 1}
\frac{
	1 - t^{	m_{1} +  \cdots + m_{\zeta_{j + \beta } }	- (m_1 + \cdots + m_{\zeta_j})	+ 1	}q^{ \lambda^{(2)}_{\sum_{\ell = 1}^{\zeta_j}m_\ell}	+ a  - 1	 }
}{
	1 - t^{	m_{1} +  \cdots + m_{\zeta_{j + \beta + 1} }	- (m_1 + \cdots + m_{\zeta_j}) 		}q^{ \lambda^{(2)}_{\sum_{\ell = 1}^{\zeta_j}m_\ell}	 + a  - 1	}
}
\notag 
\\
&\hspace{0.3cm}\times 
\prod_{j = 1}^{r - 1}
\prod_{a = 1}^{c_j}
\frac{
	1 - t^{	m_{1} +  \cdots + m_{\zeta_{r } }	- (m_1 + \cdots + m_{\zeta_j})	+ 1	}q^{ \lambda^{(2)}_{\sum_{\ell = 1}^{\zeta_j}m_\ell}	+ a  - 1	 }
}{
	1 - t^{	m_{1} +  \cdots + m_{n}	- (m_1 + \cdots + m_{\zeta_j})	+ 1	}q^{ \lambda^{(2)}_{\sum_{\ell = 1}^{\zeta_j}m_\ell}	 + a  - 1	}
},
\notag 
\\
&
\prod_{j = 1}^{r - 1}
\underbrace{		\prod_{\gamma = 1}^{m_{\zeta_j + 1} + \cdots + m_n}					}_{
	\gamma \notin 
	\bigcup_{j^\prime = j+1}^{r}
	\{
	m_{\zeta_j + 1} + \cdots + m_{\zeta_{j^\prime}}
	\}
}
\prod_{a = 1}^{c_j}
\frac{
	1 - t^{\gamma}q^{\lambda^{(2)}_{\sum_{\ell = 1}^{\zeta_j}m_\ell}	+ a }							
}{
	1 - t^{\gamma - 1} q^{\lambda^{(2)}_{\sum_{\ell = 1}^{\zeta_j}m_\ell} + a }		
}
\label{eqnb23-1429-5mar}
\\
&= 
\prod_{j = 1}^{r - 1}
\prod_{a = 1}^{c_j}
\prod_{\beta = 0}^{r - 1 - j}
\frac{
	1 - t^{	m_{1} +  \cdots + m_{\zeta_{j + \beta + 1} }	- (m_1 + \cdots + m_{\zeta_j}) 	- 1	}q^{ \lambda^{(2)}_{\sum_{\ell = 1}^{\zeta_j}m_\ell}	 + a  	}
}{
	1 - t^{	m_{1} +  \cdots + m_{\zeta_{j + \beta } }	- (m_1 + \cdots + m_{\zeta_j})		}q^{ \lambda^{(2)}_{\sum_{\ell = 1}^{\zeta_j}m_\ell}	+ a  	 }
}
\notag 
\\
&\times 
\prod_{j = 1}^{r - 1}
\prod_{a = 1}^{c_j}
\frac{
	1 - t^{	m_{1} +  \cdots + m_{n}	- (m_1 + \cdots + m_{\zeta_j})	}q^{ \lambda^{(2)}_{\sum_{\ell = 1}^{\zeta_j}m_\ell}	 + a  	}
}{
	1 - t^{	m_{1} +  \cdots + m_{\zeta_{r } }	- (m_1 + \cdots + m_{\zeta_j})		}q^{ \lambda^{(2)}_{\sum_{\ell = 1}^{\zeta_j}m_\ell}	+ a  	 }
}. 
\notag 
\end{align}
\normalsize
Applying the equations \eqref{b20-1428-5mar} - \eqref{eqnb23-1429-5mar} to the product of equations \eqref{b13-1335-3feb} \eqref{b14-1335-3feb}, we then immediately obtain the formula in the \textbf{Lemma \ref{lemb2-1432-5mar}}. 
\end{proof}

\begin{lem}
\label{lemmB3-1452-3feb}
\mbox{}
\small 
\begin{align*}
	&\frac{
		\psi_{\lambda^{(1)}/\lambda^{(2)}}(q,t)
	}{
		\widetilde{\circled{1}} 
		\times \widetilde{\circled{2}}
	}
	\\
	&=
	\prod_{j = 1}^{r-1}
	\prod_{\beta = 0}^{c_j - 1}
	\underbrace{		\prod_{\gamma = 0}^{n - \zeta_j}						}_{
		\gamma \in \{\zeta_{j+1} - \zeta_j - 1, \zeta_{j+2} - \zeta_j - 1, \dots, \zeta_r - \zeta_j - 1, n - \zeta_j\}
	}
	\frac{
	1 - q^{\lambda^{(2)}_{m_1 + \cdots + m_{\zeta_j}} - \lambda^{(1)}_{m_1 + \cdots + m_{\zeta_j + \gamma} + 1} + 1 + \beta}t^{
		m_1 + \cdots + m_{\zeta_j + \gamma} - (m_1 + \cdots + m_{\zeta_j})
	}
	}{
		1 - q^{\lambda^{(2)}_{m_1 + \cdots + m_{\zeta_j}} - \lambda^{(1)}_{m_1 + \cdots + m_{\zeta_j + \gamma} + 1} + \beta}t^{
			m_1 + \cdots + m_{\zeta_j + \gamma} - (m_1 + \cdots + m_{\zeta_j}) + 1
		}
	}
	\\
	&\times
	\prod_{j = 1}^{r-1}
	\prod_{\beta = 0}^{c_j - 1}
	\underbrace{			\prod_{\gamma = 0}^{n - \zeta_j}					}_{
		\gamma \in \{0,\zeta_{j+1} - \zeta_j , \zeta_{j+2} - \zeta_j , \dots, \zeta_r - \zeta_j \}
	}
	\frac{
	1 - q^{\lambda^{(2)}_{m_1 + \cdots + m_{\zeta_j}} - \lambda^{(2)}_{m_1 + \cdots + m_{\zeta_j + \gamma}} + \beta}t^{
		m_1 + \cdots + m_{\zeta_j + \gamma} - (m_1 + \cdots + m_{\zeta_j}) + 1
	}
	}{
		1 - q^{\lambda^{(2)}_{m_1 + \cdots + m_{\zeta_j}} - \lambda^{(2)}_{m_1 + \cdots + m_{\zeta_j + \gamma}} + 1 + \beta}t^{
			m_1 + \cdots + m_{\zeta_j + \gamma} - (m_1 + \cdots + m_{\zeta_j})
		}
	}
	\\
	&\times
	\prod_{\beta = 0}^{c_r - 1}
	\frac{1 - q^\beta t}{1 - q^{1 + \beta}}
	\times 
	\prod_{\beta = 0}^{c_r - 1}
	\frac{
		1 - q^{\lambda^{(2)}_{m_1 + \cdots + m_{\zeta_r}}  + 1 + \beta}t^{
			m_1 + \cdots + m_{n} - (m_1 + \cdots + m_{\zeta_r})
		}
	}{
		1 - q^{\lambda^{(2)}_{m_1 + \cdots + m_{\zeta_r}}   + \beta}t^{
			m_1 + \cdots + m_{n} - (m_1 + \cdots + m_{\zeta_r}) + 1
		}
	}
\end{align*}
\normalsize
\end{lem}
\begin{proof}
According to \textbf{Lemma \ref{lemb1-1353-3feb}}, we have 
	\begin{align}
	\psi_{\lambda^{(1)}/\lambda^{(2)}}(q,t)
	&= 
	\prod_{j = 1}^{r}
	\prod_{\gamma = 0}^{n - \zeta_j}
	\prod_{\beta = 0}^{c_j - 1}
	\frac{
		1 - q^{\lambda^{(2)}_{m_1 + \cdots + m_{\zeta_j}} - \lambda^{(2)}_{m_1 + \cdots + m_{\zeta_j + \gamma}} + \beta}t^{
			m_1 + \cdots + m_{\zeta_j + \gamma} - (m_1 + \cdots + m_{\zeta_j}) + 1
		}
	}{
		1 - q^{\lambda^{(2)}_{m_1 + \cdots + m_{\zeta_j}} - \lambda^{(1)}_{m_1 + \cdots + m_{\zeta_j + \gamma} + 1} + \beta}t^{
			m_1 + \cdots + m_{\zeta_j + \gamma} - (m_1 + \cdots + m_{\zeta_j}) + 1
		}
	}
	\\
	&\hspace{0.4cm}\times
	\prod_{j = 1}^{r}
	\prod_{\gamma = 0}^{n - \zeta_j}
	\prod_{\beta = 0}^{c_j - 1}
	\frac{
		1 - q^{\lambda^{(2)}_{m_1 + \cdots + m_{\zeta_j}} - \lambda^{(1)}_{m_1 + \cdots + m_{\zeta_j + \gamma} + 1} + 1 + \beta}t^{
			m_1 + \cdots + m_{\zeta_j + \gamma} - (m_1 + \cdots + m_{\zeta_j})
		}
	}{
		1 - q^{\lambda^{(2)}_{m_1 + \cdots + m_{\zeta_j}} - \lambda^{(2)}_{m_1 + \cdots + m_{\zeta_j + \gamma}} + 1 + \beta}t^{
			m_1 + \cdots + m_{\zeta_j + \gamma} - (m_1 + \cdots + m_{\zeta_j})
		}
	}
	\notag 
	\\
	&= 
	\prod_{j = 1}^{r-1}
	\prod_{\gamma = 0}^{n - \zeta_j}
	\prod_{\beta = 0}^{c_j - 1}
	\frac{
		1 - q^{\lambda^{(2)}_{m_1 + \cdots + m_{\zeta_j}} - \lambda^{(2)}_{m_1 + \cdots + m_{\zeta_j + \gamma}} + \beta}t^{
			m_1 + \cdots + m_{\zeta_j + \gamma} - (m_1 + \cdots + m_{\zeta_j}) + 1
		}
	}{
		1 - q^{\lambda^{(2)}_{m_1 + \cdots + m_{\zeta_j}} - \lambda^{(1)}_{m_1 + \cdots + m_{\zeta_j + \gamma} + 1} + \beta}t^{
			m_1 + \cdots + m_{\zeta_j + \gamma} - (m_1 + \cdots + m_{\zeta_j}) + 1
		}
	}
	\label{b20-1451-3feb}
	\\
	&\hspace{0.4cm}\times
	\prod_{j = 1}^{r-1}
	\prod_{\gamma = 0}^{n - \zeta_j}
	\prod_{\beta = 0}^{c_j - 1}
	\frac{
		1 - q^{\lambda^{(2)}_{m_1 + \cdots + m_{\zeta_j}} - \lambda^{(1)}_{m_1 + \cdots + m_{\zeta_j + \gamma} + 1} + 1 + \beta}t^{
			m_1 + \cdots + m_{\zeta_j + \gamma} - (m_1 + \cdots + m_{\zeta_j})
		}
	}{
		1 - q^{\lambda^{(2)}_{m_1 + \cdots + m_{\zeta_j}} - \lambda^{(2)}_{m_1 + \cdots + m_{\zeta_j + \gamma}} + 1 + \beta}t^{
			m_1 + \cdots + m_{\zeta_j + \gamma} - (m_1 + \cdots + m_{\zeta_j})
		}
	}
	\notag
	\\
	&\hspace{0.4cm}\times
	\prod_{\gamma = 0}^{n - \zeta_r}
	\prod_{\beta = 0}^{c_r - 1}
	\frac{
		1 - q^{\lambda^{(2)}_{m_1 + \cdots + m_{\zeta_r}} - \lambda^{(2)}_{m_1 + \cdots + m_{\zeta_r + \gamma}} + \beta}t^{
			m_1 + \cdots + m_{\zeta_r + \gamma} - (m_1 + \cdots + m_{\zeta_r}) + 1
		}
	}{
		1 - q^{\lambda^{(2)}_{m_1 + \cdots + m_{\zeta_r}} - \lambda^{(1)}_{m_1 + \cdots + m_{\zeta_r + \gamma} + 1} + \beta}t^{
			m_1 + \cdots + m_{\zeta_r + \gamma} - (m_1 + \cdots + m_{\zeta_r}) + 1
		}
	}
	\notag
	\\
	&\hspace{0.4cm}\times 
	\prod_{\gamma = 0}^{n - \zeta_r}
	\prod_{\beta = 0}^{c_r - 1}
	\frac{
		1 - q^{\lambda^{(2)}_{m_1 + \cdots + m_{\zeta_r}} - \lambda^{(1)}_{m_1 + \cdots + m_{\zeta_r + \gamma} + 1} + 1 + \beta}t^{
			m_1 + \cdots + m_{\zeta_r + \gamma} - (m_1 + \cdots + m_{\zeta_r})
		}
	}{
		1 - q^{\lambda^{(2)}_{m_1 + \cdots + m_{\zeta_r}} - \lambda^{(2)}_{m_1 + \cdots + m_{\zeta_r + \gamma}} + 1 + \beta}t^{
			m_1 + \cdots + m_{\zeta_r + \gamma} - (m_1 + \cdots + m_{\zeta_r})
		}
	}. 
	\notag
\end{align}
By examining the formulas for $\widetilde{\circled{1}}$ and $\widetilde{\circled{2}}$
given in equations \eqref{b17-1450-3feb} and \eqref{b18-1450-3feb}, we can see that there will be factors in $\widetilde{\circled{1}} \times \widetilde{\circled{2}}$ that cancel againt factors in equation \eqref{b20-1451-3feb}. 
After these cancellations are performed, we obtain the formula for 
$\frac{
	\psi_{\lambda^{(1)}/\lambda^{(2)}}(q,t)
}{
	\widetilde{\circled{1}} 
	\times \widetilde{\circled{2}}
}$ as stated in \textbf{Lemma \ref{lemmB3-1452-3feb}}.
\end{proof}

\begin{lem}
\label{lemb4-1600-5mar}
\mbox{}
\begin{align}
\psi_{\lambda^{(1)}/\lambda^{(2)}}(q,t)
= 
\operatorname{Factor 1}
\times
\operatorname{Factor 2}
\times
\operatorname{Factor 3}
\times
\operatorname{Factor 4}
\times
\operatorname{Factor 5}
\end{align}
if and only if 
\small 
\begin{align}
	&\operatorname{Factor 1}
	\times
	\operatorname{Factor 2}
	\label{eqn136-1322-2dec}
	\\
	&=
	\prod_{j = 1}^{r-1}
	\prod_{\beta = 0}^{c_j - 1}
	\underbrace{		\prod_{\gamma }^{}						}_{
		\gamma \in \{\zeta_{j+1} - \zeta_j , \zeta_{j+2} - \zeta_j , \dots, \zeta_r - \zeta_j \}
	}
	\frac{
		1 - q^{\lambda^{(2)}_{m_1 + \cdots + m_{\zeta_j}} - \lambda^{(1)}_{m_1 + \cdots + m_{\zeta_j + \gamma} } + 1 + \beta}t^{
			m_1 + \cdots + m_{\zeta_j + \gamma - 1} - (m_1 + \cdots + m_{\zeta_j})
		}
	}{
		1 - q^{\lambda^{(2)}_{m_1 + \cdots + m_{\zeta_j}} - \lambda^{(1)}_{m_1 + \cdots + m_{\zeta_j + \gamma}} + \beta}t^{
			m_1 + \cdots + m_{\zeta_j + \gamma - 1} - (m_1 + \cdots + m_{\zeta_j}) + 1
		}
	}
	\notag 
	\\
	&
	\times 
	\prod_{j = 1}^{r - 1}
	\prod_{a = 1}^{c_j}
	\prod_{\beta = 0}^{r - j - 1}
	\frac{
		1 - t^{	m_{1} +  \cdots + m_{\zeta_{j + \beta + 1} }	- (m_1 + \cdots + m_{\zeta_j}) - 1		}q^{ \lambda^{(2)}_{\sum_{\ell = 1}^{\zeta_j}m_\ell}	- 	\lambda^{(1)}_{\sum_{\ell = 1}^{\zeta_{j+\beta + 1}}m_\ell - 1}	 + a 	}
	}{
		1 - t^{	m_{1} +  \cdots + m_{\zeta_{j + \beta + 1} - 1}	- (m_1 + \cdots + m_{\zeta_j})		}q^{ \lambda^{(2)}_{\sum_{\ell = 1}^{\zeta_j}m_\ell}	- 	\lambda^{(1)}_{\sum_{\ell = 1}^{\zeta_{j+\beta + 1}}m_\ell - 1}	 + a 	}
	}
	\notag \\
	&\times 
	\prod_{j = 1}^{r - 1}
	\prod_{a = 1}^{c_j}
	\prod_{\beta = 0}^{r - j - 1}
	\frac{
		1 - t^{	m_{1} +  \cdots + m_{\zeta_{j + \beta + 1} - 1}	- (m_1 + \cdots + m_{\zeta_j})	+ 1	}q^{ \lambda^{(2)}_{\sum_{\ell = 1}^{\zeta_j}m_\ell}	- 	\lambda^{(2)}_{\sum_{\ell = 1}^{\zeta_{j+\beta + 1}}m_\ell - 1}	 + a  - 1	 }
	}{
		1 - t^{	m_{1} +  \cdots + m_{\zeta_{j + \beta + 1} }	- (m_1 + \cdots + m_{\zeta_j}) 		}q^{ \lambda^{(2)}_{\sum_{\ell = 1}^{\zeta_j}m_\ell}	- 	\lambda^{(2)}_{\sum_{\ell = 1}^{\zeta_{j+\beta + 1}}m_\ell - 1}	 + a  - 1	}
	}
	\notag	\\
	&\times
	\prod_{j = 1}^{r - 1}
	\prod_{a = 1}^{c_j}
	\prod_{\beta = 0}^{r - j - 1}
	\frac{
		1 - t^{	m_{1} +  \cdots + m_{\zeta_{j + \beta + 1} }	- (m_1 + \cdots + m_{\zeta_j}) 		}q^{ \lambda^{(2)}_{\sum_{\ell = 1}^{\zeta_j}m_\ell}	 + a  - 1	}
	}{
	1 - t^{	m_{1} +  \cdots + m_{\zeta_{j + \beta + 1} }	- (m_1 + \cdots + m_{\zeta_j}) 	- 1	}q^{ \lambda^{(2)}_{\sum_{\ell = 1}^{\zeta_j}m_\ell}	 + a  	}
	}
	\notag
\end{align}
\normalsize
\end{lem}
\begin{proof}
Using lemma \ref{lemmB3-1452-3feb} and equation \eqref{b15-1521-3feb}, we immediately obtain equation \eqref{eqn136-1322-2dec}. 
\end{proof}

\begin{lem}
\label{lemb5-1600-5mar}
\mbox{}
\small 
\begin{align}
&\operatorname{Factor 1}
\times
\operatorname{Factor 2}
\\
&=
\prod_{j = 1}^{r-1}
\prod_{\beta = 0}^{c_j - 1}
\underbrace{		\prod_{\gamma }^{}						}_{
	\gamma \in \{\zeta_{j+1} - \zeta_j , \zeta_{j+2} - \zeta_j , \dots, \zeta_r - \zeta_j \}
}
\frac{
	1 - q^{\lambda^{(2)}_{m_1 + \cdots + m_{\zeta_j}} - \lambda^{(1)}_{m_1 + \cdots + m_{\zeta_j + \gamma} } + 1 + \beta}t^{
		m_1 + \cdots + m_{\zeta_j + \gamma - 1} - (m_1 + \cdots + m_{\zeta_j})
	}
}{
	1 - q^{\lambda^{(2)}_{m_1 + \cdots + m_{\zeta_j}} - \lambda^{(1)}_{m_1 + \cdots + m_{\zeta_j + \gamma}} + \beta}t^{
		m_1 + \cdots + m_{\zeta_j + \gamma - 1} - (m_1 + \cdots + m_{\zeta_j}) + 1
	}
}
\notag 
\\
&
\times 
\prod_{j = 1}^{r - 1}
\prod_{a = 1}^{c_j}
\prod_{\beta = 0}^{r - j - 1}
\frac{
	1 - t^{	m_{1} +  \cdots + m_{\zeta_{j + \beta + 1} }	- (m_1 + \cdots + m_{\zeta_j}) - 1		}q^{ \lambda^{(2)}_{\sum_{\ell = 1}^{\zeta_j}m_\ell}	- 	\lambda^{(1)}_{\sum_{\ell = 1}^{\zeta_{j+\beta + 1}}m_\ell - 1}	 + a 	}
}{
	1 - t^{	m_{1} +  \cdots + m_{\zeta_{j + \beta + 1} - 1}	- (m_1 + \cdots + m_{\zeta_j})		}q^{ \lambda^{(2)}_{\sum_{\ell = 1}^{\zeta_j}m_\ell}	- 	\lambda^{(1)}_{\sum_{\ell = 1}^{\zeta_{j+\beta + 1}}m_\ell - 1}	 + a 	}
}
\notag \\
&\times 
\prod_{j = 1}^{r - 1}
\prod_{a = 1}^{c_j}
\prod_{\beta = 0}^{r - j - 1}
\frac{
	1 - t^{	m_{1} +  \cdots + m_{\zeta_{j + \beta + 1} - 1}	- (m_1 + \cdots + m_{\zeta_j})	+ 1	}q^{ \lambda^{(2)}_{\sum_{\ell = 1}^{\zeta_j}m_\ell}	- 	\lambda^{(2)}_{\sum_{\ell = 1}^{\zeta_{j+\beta + 1}}m_\ell - 1}	 + a  - 1	 }
}{
	1 - t^{	m_{1} +  \cdots + m_{\zeta_{j + \beta + 1} }	- (m_1 + \cdots + m_{\zeta_j}) 		}q^{ \lambda^{(2)}_{\sum_{\ell = 1}^{\zeta_j}m_\ell}	- 	\lambda^{(2)}_{\sum_{\ell = 1}^{\zeta_{j+\beta + 1}}m_\ell - 1}	 + a  - 1	}
}
\notag	\\
&\times
\prod_{j = 1}^{r - 1}
\prod_{a = 1}^{c_j}
\prod_{\beta = 0}^{r - j - 1}
\frac{
	1 - t^{	m_{1} +  \cdots + m_{\zeta_{j + \beta + 1} }	- (m_1 + \cdots + m_{\zeta_j}) 		}q^{ \lambda^{(2)}_{\sum_{\ell = 1}^{\zeta_j}m_\ell}	 + a  - 1	}
}{
	1 - t^{	m_{1} +  \cdots + m_{\zeta_{j + \beta + 1} }	- (m_1 + \cdots + m_{\zeta_j}) 	- 1	}q^{ \lambda^{(2)}_{\sum_{\ell = 1}^{\zeta_j}m_\ell}	 + a  	}
}
\notag
\end{align}
\normalsize
\end{lem}
\begin{proof}
First, one can easily show that 
\small 
\begin{align}
	&\prod_{a =  1 }^{	c_j	 } \frac{
	1 - t^{m_1 + \cdots + m_{\zeta_{j + \beta + 1}} - (m_1 + \cdots + m_{\zeta_j})}q^{\lambda^{(2)}_{	\sum_{\ell = 1}^{\zeta_j}		m_\ell} - \lambda^{(2)}_{	\sum_{\ell = 1}^{\zeta_{j + \beta + 1}}		m_\ell} + a- 1}
	}{
	1 - t^{m_1 + \cdots + m_{\zeta_{j + \beta + 1}} - (m_1 + \cdots + m_{\zeta_j})}q^{\lambda^{(2)}_{	\sum_{\ell = 1}^{\zeta_j}		m_\ell} - \lambda^{(2)}_{	\sum_{\ell = 1}^{\zeta_{j + \beta + 1}}		m_\ell} + a-c_{j + \beta + 1}-1}
	}
	\label{b29-eqn-1550-5mar}
	\\
	&\times
	\prod_{a =  1 }^{	c_j	 }
	\frac{
	1 - t^{m_1 + \cdots + m_{\zeta_{j + \beta + 1}} - (m_1 + \cdots + m_{\zeta_j}) - 1}q^{\lambda^{(2)}_{	\sum_{\ell = 1}^{\zeta_j}		m_\ell} - \lambda^{(2)}_{	\sum_{\ell = 1}^{\zeta_{j + \beta + 1}}		m_\ell} + a-c_{j + \beta + 1}}
	}{
	1 - t^{m_1 + \cdots + m_{\zeta_{j + \beta + 1}} - (m_1 + \cdots + m_{\zeta_j}) - 1}q^{\lambda^{(2)}_{	\sum_{\ell = 1}^{\zeta_j}		m_\ell} - \lambda^{(2)}_{	\sum_{\ell = 1}^{\zeta_{j + \beta + 1}}		m_\ell} +  a}
	}
	\notag 
	\\
	&\times 
	\prod_{a =  1 }^{	c_j	 }
	\frac{
	1 - t^{m_1 + \cdots + m_{\zeta_{j + \beta + 1}} - (m_1 + \cdots + m_{\zeta_j})}q^{\lambda^{(2)}_{	\sum_{\ell = 1}^{\zeta_j}		m_\ell}  + a-1}
	}{
	1 - t^{m_1 + \cdots + m_{\zeta_{j + \beta + 1}} - (m_1 + \cdots + m_{\zeta_j})}q^{\lambda^{(2)}_{	\sum_{\ell = 1}^{\zeta_j}		m_\ell} - \lambda^{(2)}_{	\sum_{\ell = 1}^{\zeta_{j + \beta + 1}}		m_\ell}  + a -1}
	}
	\notag 
	\\
	&\times
	\prod_{a =  1 }^{	c_j	 }
	\frac{
	1 - t^{m_1 + \cdots + m_{\zeta_{j + \beta + 1}} - (m_1 + \cdots + m_{\zeta_j}) - 1}q^{\lambda^{(2)}_{	\sum_{\ell = 1}^{\zeta_j}		m_\ell} - 		\lambda^{(2)}_{	\sum_{\ell = 1}^{\zeta_{j + \beta + 1}}		m_\ell}	+ a}
	}{
	1 - t^{m_1 + \cdots + m_{\zeta_{j + \beta + 1}} - (m_1 + \cdots + m_{\zeta_j}) - 1}q^{\lambda^{(2)}_{	\sum_{\ell = 1}^{\zeta_j}		m_\ell} +  a}
	}
	\notag 
	\\
	&
	\times \prod_{a =  1 }^{	c_j	 }
	\frac{
		1 - t^{	m_{1} +  \cdots + m_{\zeta_{j + \beta + 1} - 1}	- (m_1 + \cdots + m_{\zeta_j})		}q^{ \lambda^{(2)}_{\sum_{\ell = 1}^{\zeta_j}m_\ell}	- 	\lambda^{(1)}_{\sum_{\ell = 1}^{\zeta_{j+\beta + 1}}m_\ell - 1}	 + a 	}
	}{
		1 - t^{	m_{1} +  \cdots + m_{\zeta_{j + \beta + 1} }	- (m_1 + \cdots + m_{\zeta_j}) - 1		}q^{ \lambda^{(2)}_{\sum_{\ell = 1}^{\zeta_j}m_\ell}	- 	\lambda^{(1)}_{\sum_{\ell = 1}^{\zeta_{j+\beta + 1}}m_\ell - 1}	 + a 	}
	}
	\notag 
	\\
	&\times \prod_{a =  1 }^{	c_j	 }
	\frac{
		1 - t^{	m_{1} +  \cdots + m_{\zeta_{j + \beta + 1} }	- (m_1 + \cdots + m_{\zeta_j}) 		}q^{ \lambda^{(2)}_{\sum_{\ell = 1}^{\zeta_j}m_\ell}	- 	\lambda^{(2)}_{\sum_{\ell = 1}^{\zeta_{j+\beta + 1}}m_\ell - 1}	 + a  - 1	}
	}{
		1 - t^{	m_{1} +  \cdots + m_{\zeta_{j + \beta + 1} - 1}	- (m_1 + \cdots + m_{\zeta_j})	+ 1	}q^{ \lambda^{(2)}_{\sum_{\ell = 1}^{\zeta_j}m_\ell}	- 	\lambda^{(2)}_{\sum_{\ell = 1}^{\zeta_{j+\beta + 1}}m_\ell - 1}	 + a  - 1	 }
	}
	\notag 
	\\
	&\times \prod_{a =  1 }^{	c_j	 }
	\frac{
		1 - t^{	m_{1} +  \cdots + m_{\zeta_{j + \beta + 1} }	- (m_1 + \cdots + m_{\zeta_j}) 	- 1	}q^{ \lambda^{(2)}_{\sum_{\ell = 1}^{\zeta_j}m_\ell}	 + a  	}
	}{
		1 - t^{	m_{1} +  \cdots + m_{\zeta_{j + \beta + 1} }	- (m_1 + \cdots + m_{\zeta_j}) 		}q^{ \lambda^{(2)}_{\sum_{\ell = 1}^{\zeta_j}m_\ell}	 + a  - 1	}
	}
	\notag 
	\\
	&= 
	\prod_{a =  1 }^{	c_j	 }
	\frac{
		1 - q^{	\lambda^{(2)}_{m_1 + \cdots + m_{\zeta_j}} - \lambda^{(1)}_{m_1 + \cdots + m_{\zeta_{j + \beta + 1} } }		+ a }t^{m_1 + \cdots + m_{\zeta_{j + \beta + 1} - 1} - (m_1 + \cdots + m_{\zeta_j})}
	}{
		1 - q^{	\lambda^{(2)}_{m_1 + \cdots + m_{\zeta_j}} - \lambda^{(1)}_{m_1 + \cdots + m_{\zeta_{j + \beta + 1} } }	 + a - 1}t^{m_1 + \cdots + m_{\zeta_{j+ \beta + 1} - 1} - (m_1 + \cdots + m_{\zeta_j}) + 1}
	}. 
	\notag 
\end{align}
\normalsize
From equations \eqref{b12-eqn-1525-5mar} \eqref{b13-eqn-1525-5mar}, we know that 
\small 
\begin{align}
	&\operatorname{Factor 1}
	=
	\prod_{j = 1}^{r - 1}
	\prod_{\beta = 0}^{r - 1 - j}
	\prod_{a = 1}^{c_j}
	\frac{
	1 - t^{m_1 + \cdots + m_{\zeta_{j + \beta + 1}} - (m_1 + \cdots + m_{\zeta_j})}q^{\lambda^{(2)}_{	\sum_{\ell = 1}^{\zeta_j}		m_\ell} - \lambda^{(2)}_{	\sum_{\ell = 1}^{\zeta_{j + \beta + 1}}		m_\ell} + a- 1}
	}{
	1 - t^{m_1 + \cdots + m_{\zeta_{j + \beta + 1}} - (m_1 + \cdots + m_{\zeta_j})}q^{\lambda^{(2)}_{	\sum_{\ell = 1}^{\zeta_j}		m_\ell} - \lambda^{(2)}_{	\sum_{\ell = 1}^{\zeta_{j + \beta + 1}}		m_\ell} + a-c_{j + \beta + 1}-1}
	}
	\\
	&\hspace{1.5cm}\times
	\prod_{j = 1}^{r - 1}
	\prod_{\beta = 0}^{r - 1 - j}
	\prod_{a = 1}^{c_j}
	\frac{
	1 - t^{m_1 + \cdots + m_{\zeta_{j + \beta + 1}} - (m_1 + \cdots + m_{\zeta_j}) - 1}q^{\lambda^{(2)}_{	\sum_{\ell = 1}^{\zeta_j}		m_\ell} - \lambda^{(2)}_{	\sum_{\ell = 1}^{\zeta_{j + \beta + 1}}		m_\ell} + a-c_{j + \beta + 1}}
	}{
	1 - t^{m_1 + \cdots + m_{\zeta_{j + \beta + 1}} - (m_1 + \cdots + m_{\zeta_j}) - 1}q^{\lambda^{(2)}_{	\sum_{\ell = 1}^{\zeta_j}		m_\ell} - \lambda^{(2)}_{	\sum_{\ell = 1}^{\zeta_{j + \beta + 1}}		m_\ell} +  a}
	},
	\notag 
\end{align}
\normalsize
and 
\small 
\begin{align}
\operatorname{Factor 2}
&= 
\prod_{j = 1}^{r - 1}
\prod_{\beta = 0}^{r - 1 - j}
\prod_{a =  1 }^{	c_j	 }
\frac{
1 - t^{m_1 + \cdots + m_{\zeta_{j + \beta + 1}} - (m_1 + \cdots + m_{\zeta_j})}q^{\lambda^{(2)}_{	\sum_{\ell = 1}^{\zeta_j}		m_\ell}  + a-1}
}{
1 - t^{m_1 + \cdots + m_{\zeta_{j + \beta + 1}} - (m_1 + \cdots + m_{\zeta_j})}q^{\lambda^{(2)}_{	\sum_{\ell = 1}^{\zeta_j}		m_\ell} - \lambda^{(2)}_{	\sum_{\ell = 1}^{\zeta_{j + \beta + 1}}		m_\ell}  + a -1}
}
\\
&\times
\prod_{j = 1}^{r - 1}
\prod_{\beta = 0}^{r - 1 - j}
\prod_{a =  1 }^{	c_j	 }
\frac{
1 - t^{m_1 + \cdots + m_{\zeta_{j + \beta + 1}} - (m_1 + \cdots + m_{\zeta_j}) - 1}q^{\lambda^{(2)}_{	\sum_{\ell = 1}^{\zeta_j}		m_\ell} - 		\lambda^{(2)}_{	\sum_{\ell = 1}^{\zeta_{j + \beta + 1}}		m_\ell}	+ a}
}{
	1 - t^{m_1 + \cdots + m_{\zeta_{j + \beta + 1}} - (m_1 + \cdots + m_{\zeta_j}) - 1}q^{\lambda^{(2)}_{	\sum_{\ell = 1}^{\zeta_j}		m_\ell} +  a}
}. 
\notag 
\end{align}
\normalsize
By using equation \eqref{b29-eqn-1550-5mar}, we obtain that 
\small
\begin{align}
&\operatorname{Factor 1}
\times
\operatorname{Factor 2}
\\
&
\times 
\prod_{j = 1}^{r - 1}
\prod_{a = 1}^{c_j}
\prod_{\beta = 0}^{r - j - 1}
\frac{
	1 - t^{	m_{1} +  \cdots + m_{\zeta_{j + \beta + 1} - 1}	- (m_1 + \cdots + m_{\zeta_j})		}q^{ \lambda^{(2)}_{\sum_{\ell = 1}^{\zeta_j}m_\ell}	- 	\lambda^{(1)}_{\sum_{\ell = 1}^{\zeta_{j+\beta + 1}}m_\ell - 1}	 + a 	}
}{
	1 - t^{	m_{1} +  \cdots + m_{\zeta_{j + \beta + 1} }	- (m_1 + \cdots + m_{\zeta_j}) - 1		}q^{ \lambda^{(2)}_{\sum_{\ell = 1}^{\zeta_j}m_\ell}	- 	\lambda^{(1)}_{\sum_{\ell = 1}^{\zeta_{j+\beta + 1}}m_\ell - 1}	 + a 	}
}
\notag \\
&\times 
\prod_{j = 1}^{r - 1}
\prod_{a = 1}^{c_j}
\prod_{\beta = 0}^{r - j - 1}
\frac{
	1 - t^{	m_{1} +  \cdots + m_{\zeta_{j + \beta + 1} }	- (m_1 + \cdots + m_{\zeta_j}) 		}q^{ \lambda^{(2)}_{\sum_{\ell = 1}^{\zeta_j}m_\ell}	- 	\lambda^{(2)}_{\sum_{\ell = 1}^{\zeta_{j+\beta + 1}}m_\ell - 1}	 + a  - 1	}
}{
	1 - t^{	m_{1} +  \cdots + m_{\zeta_{j + \beta + 1} - 1}	- (m_1 + \cdots + m_{\zeta_j})	+ 1	}q^{ \lambda^{(2)}_{\sum_{\ell = 1}^{\zeta_j}m_\ell}	- 	\lambda^{(2)}_{\sum_{\ell = 1}^{\zeta_{j+\beta + 1}}m_\ell - 1}	 + a  - 1	 }
}
\notag	\\
&\times
\prod_{j = 1}^{r - 1}
\prod_{a = 1}^{c_j}
\prod_{\beta = 0}^{r - j - 1}
\frac{
	1 - t^{	m_{1} +  \cdots + m_{\zeta_{j + \beta + 1} }	- (m_1 + \cdots + m_{\zeta_j}) 	- 1	}q^{ \lambda^{(2)}_{\sum_{\ell = 1}^{\zeta_j}m_\ell}	 + a  	}
}{
	1 - t^{	m_{1} +  \cdots + m_{\zeta_{j + \beta + 1} }	- (m_1 + \cdots + m_{\zeta_j}) 		}q^{ \lambda^{(2)}_{\sum_{\ell = 1}^{\zeta_j}m_\ell}	 + a  - 1	}
}
\notag
\\
&=
\prod_{j = 1}^{r - 1}
\prod_{\beta = 0}^{r - 1 - j}
\prod_{a =  1 }^{	c_j	 }
\frac{
	1 - q^{	\lambda^{(2)}_{m_1 + \cdots + m_{\zeta_j}} - \lambda^{(1)}_{m_1 + \cdots + m_{\zeta_{j + \beta + 1} } }		+ a }t^{m_1 + \cdots + m_{\zeta_{j + \beta + 1} - 1} - (m_1 + \cdots + m_{\zeta_j})}
}{
	1 - q^{	\lambda^{(2)}_{m_1 + \cdots + m_{\zeta_j}} - \lambda^{(1)}_{m_1 + \cdots + m_{\zeta_{j + \beta + 1} } }	 + a - 1}t^{m_1 + \cdots + m_{\zeta_{j+ \beta + 1} - 1} - (m_1 + \cdots + m_{\zeta_j}) + 1}
}
\notag 
\\
&=
\prod_{j = 1}^{r-1}
\prod_{\beta = 0}^{c_j - 1}
\underbrace{		\prod_{\gamma }^{}						}_{
	\gamma \in \{\zeta_{j+1} - \zeta_j , \zeta_{j+2} - \zeta_j , \dots, \zeta_r - \zeta_j \}
}
\frac{
	1 - q^{\lambda^{(2)}_{m_1 + \cdots + m_{\zeta_j}} - \lambda^{(1)}_{m_1 + \cdots + m_{\zeta_j + \gamma} } + 1 + \beta}t^{
		m_1 + \cdots + m_{\zeta_j + \gamma - 1} - (m_1 + \cdots + m_{\zeta_j})
	}
}{
	1 - q^{\lambda^{(2)}_{m_1 + \cdots + m_{\zeta_j}} - \lambda^{(1)}_{m_1 + \cdots + m_{\zeta_j + \gamma}} + \beta}t^{
		m_1 + \cdots + m_{\zeta_j + \gamma - 1} - (m_1 + \cdots + m_{\zeta_j}) + 1
	}
}
\notag 
\end{align}
\normalsize
So we have proved \textbf{Lemma \ref{lemb5-1600-5mar}}. 
\end{proof}

From \textbf{Lemmas \ref{lemb4-1600-5mar}} and \textbf{\ref{lemb5-1600-5mar}}, it follows that we have proved
\textbf{Lemma \ref{lemm54-1412-2feb}}.

\section{Proof of lemma \ref{lem68-1212-13mar}}
\label{appC-1214-13mar}

The goal of this appendix is to prove the \textbf{Lemma \ref{lem68-1212-13mar}}.

\begin{lem}\mbox{}
\label{lemc1-1621-18mar}
\begin{align}
&\frac{1}{
	\left(	q^{-\frac{1}{2}}t^{-\frac{1}{2}}			\right)^c
}
\left(		\frac{q^{\frac{1}{2}} - q^{- \frac{1}{2}}}{t^{-\frac{1}{2}} - t^{\frac{1}{2}}}			\right)^c
\lim_{\xi \rightarrow t^{-1}}\,\,
\dualmap
\bigg[
\underbrace{		\prod_{1 \leq a < b \leq k}			}_{
	\substack{
		a,b \in \operatorname{Pos}((i_1,\dots,i_k); N + 1)
	}
}
\frac{
	\left(1 - q^{-1}\frac{z_b}{z_a}\right)
	\left(1 - \frac{z_b}{z_a}\right)
}{
	\left(1 - t^{-1}\frac{z_b}{z_a}\right)
	\left(1 - q^{-1}t\frac{z_b}{z_a}\right)
}
\bigg]
\label{c1-1522-18mar}
\\
&=
\prod_{\beta = 0}^{c_1 - 1}
\frac{	\left(q - t^{-\beta}\right) 			}{
	\left( t^{-1-\beta} - 1 \right)
}
\times
\cdots
\times
\prod_{\beta = 0}^{c_r - 1}
\frac{	\left(q - t^{-\beta}\right) 			}{
	\left( t^{-1-\beta} - 1 \right)
}
\notag 
\\
&\hspace{0.3cm}\times
\prod_{j = 1}^{r - 1}
\prod_{\sigma = j+1}^{r}
\underbrace{			\prod_{s \in \widetilde{			\mu				}}				}_{
\substack{
s \in \operatorname{Row}(T(i_1,\dots,i_k ; \lambda)|N+1),
		\\
s \text{ is in the row interval $\gamma_j$},
		\\
\text{col}(s) = \mu_{m_1 + \cdots + m_{\gamma_\sigma -1} + 1}
	}
}
\frac{	1 - q^{a_{\widetilde{			\mu				}}(s) + 2}  t^{\ell_{\widetilde{			\mu				}}(s) + c_\sigma		}			}{
	1 -q^{a_{\widetilde{			\mu				}}(s) + 1}t^{\ell_{\widetilde{			\mu				}}(s) + c_\sigma + 1} 
}
\frac{	1 -   q^{a_{\widetilde{\mu}}(s) + 1}	t^{\ell_{\widetilde{\mu}}(s) + 1}	}{
	1 - 	q^{a_{\widetilde{			\mu				}}(s) + 2} t^{\ell_{\widetilde{		\mu		}}(s)}	
}
\notag 
\end{align}
\end{lem}
\begin{proof}
It is clear that 
\begin{align}
&\lim_{\xi \rightarrow t^{-1}}\,\,
\dualmap
\bigg[
\underbrace{		\prod_{1 \leq a < b \leq k}			}_{
	\substack{
		a,b \in \operatorname{Pos}((i_1,\dots,i_k); N + 1)
	}
}
\frac{
	\left(1 - q^{-1}\frac{z_b}{z_a}\right)
	\left(1 - \frac{z_b}{z_a}\right)
}{
	\left(1 - t^{-1}\frac{z_b}{z_a}\right)
	\left(1 - q^{-1}t\frac{z_b}{z_a}\right)
}
\bigg]
\label{c2-1522-18mar}
\\
&= 
\prod_{\beta = 1}^{c_1 - 1}
\frac{	\left(1 - q^{-1}t^{-\beta}\right) \left(1 - t^{-1}\right)			}{
	\left(1 - q^{-1}\right)		\left(1 - t^{-1-\beta}\right)
}
\times
\cdots
\times
\prod_{\beta = 1}^{c_r - 1}
\frac{	\left(1 - q^{-1}t^{-\beta}\right) \left(1 - t^{-1}\right)			}{
	\left(1 - q^{-1}\right)		\left(1 - t^{-1-\beta}\right)
}
\notag
\\
&\times
\prod_{j = 1}^{r - 1}
\prod_{\sigma = j+1}^{r}
\prod_{\beta = 1}^{c_j}
\frac{
	\left(	1 - (q^{-1})^{\mu_{m_1 + \cdots + m_{\gamma_j -1} + 1}	- \mu_{m_1 + \cdots + m_{\gamma_\sigma -1} + 1}		+ 1}   (t^{-1})^{m_1 + \cdots + m_{\gamma_\sigma} - (m_1 + \cdots + m_{\gamma_j} - c_j + \beta)			}					\right)
}{
	\left(	1 - (q^{-1})^{\mu_{m_1 + \cdots + m_{\gamma_j -1} + 1}	- \mu_{m_1 + \cdots + m_{\gamma_\sigma -1} + 1}		}   (t^{-1})^{m_1 + \cdots + m_{\gamma_\sigma} - (m_1 + \cdots + m_{\gamma_j} - c_j + \beta) + 1}					\right)
}
\notag
\\
&\times
\prod_{j = 1}^{r - 1}
\prod_{\sigma = j+1}^{r}
\prod_{\beta = 1}^{c_j}
\frac{
	\left(	1 - (q^{-1})^{\mu_{m_1 + \cdots + m_{\gamma_j -1} + 1}	- \mu_{m_1 + \cdots + m_{\gamma_\sigma -1} + 1}		}   (t^{-1})^{m_1 + \cdots + m_{\gamma_\sigma} - c_\sigma + 1 - (m_1 + \cdots + m_{\gamma_j} - c_j + \beta)}				\right)
}{
	\left(	1 - (q^{-1})^{\mu_{m_1 + \cdots + m_{\gamma_j -1} + 1}	- \mu_{m_1 + \cdots + m_{\gamma_\sigma -1} + 1}		+ 1}   (t^{-1})^{m_1 + \cdots + m_{\gamma_\sigma} - c_\sigma - (m_1 + \cdots + m_{\gamma_j} - c_j + \beta)}					\right)
}. 
\notag 
\end{align}
Also, one can easily show that 
\begin{align}
&\prod_{\beta = 1}^{c_1 - 1}
\frac{	\left(1 - q^{-1}t^{-\beta}\right) \left(1 - t^{-1}\right)			}{
	\left(1 - q^{-1}\right)		\left(1 - t^{-1-\beta}\right)
}
\times
\cdots
\times
\prod_{\beta = 1}^{c_r - 1}
\frac{	\left(1 - q^{-1}t^{-\beta}\right) \left(1 - t^{-1}\right)			}{
	\left(1 - q^{-1}\right)		\left(1 - t^{-1-\beta}\right)
}
\label{c3-1522-18mar}
\\
&= 
\prod_{\beta = 0}^{c_1 - 1}
\frac{	\left(q - t^{-\beta}\right) 			}{
	\left( t^{-1-\beta} - 1 \right)
}
\times
\cdots
\times
\prod_{\beta = 0}^{c_r - 1}
\frac{	\left(q - t^{-\beta}\right) 			}{
	\left( t^{-1-\beta} - 1 \right)
}
\times 
\left(
\frac{t^{-1} - 1}{q - 1}
\right)^{c}, 
\notag 
\end{align}
and 
\begin{align}
&\prod_{j = 1}^{r - 1}
\prod_{\sigma = j+1}^{r}
\prod_{\beta = 1}^{c_j}
\frac{
	\left(	1 - (q^{-1})^{\mu_{m_1 + \cdots + m_{\gamma_j -1} + 1}	- \mu_{m_1 + \cdots + m_{\gamma_\sigma -1} + 1}		+ 1}   (t^{-1})^{m_1 + \cdots + m_{\gamma_\sigma} - (m_1 + \cdots + m_{\gamma_j} - c_j + \beta)			}					\right)
}{
	\left(	1 - (q^{-1})^{\mu_{m_1 + \cdots + m_{\gamma_j -1} + 1}	- \mu_{m_1 + \cdots + m_{\gamma_\sigma -1} + 1}		}   (t^{-1})^{m_1 + \cdots + m_{\gamma_\sigma} - (m_1 + \cdots + m_{\gamma_j} - c_j + \beta) + 1}					\right)
}
\label{c4-1522-18mar}
\\
&\times
\prod_{j = 1}^{r - 1}
\prod_{\sigma = j+1}^{r}
\prod_{\beta = 1}^{c_j}
\frac{
	\left(	1 - (q^{-1})^{\mu_{m_1 + \cdots + m_{\gamma_j -1} + 1}	- \mu_{m_1 + \cdots + m_{\gamma_\sigma -1} + 1}		}   (t^{-1})^{m_1 + \cdots + m_{\gamma_\sigma} - c_\sigma + 1 - (m_1 + \cdots + m_{\gamma_j} - c_j + \beta)}				\right)
}{
	\left(	1 - (q^{-1})^{\mu_{m_1 + \cdots + m_{\gamma_j -1} + 1}	- \mu_{m_1 + \cdots + m_{\gamma_\sigma -1} + 1}		+ 1}   (t^{-1})^{m_1 + \cdots + m_{\gamma_\sigma} - c_\sigma - (m_1 + \cdots + m_{\gamma_j} - c_j + \beta)}					\right)
}
\notag 
\\
&= 
\prod_{j = 1}^{r - 1}
\prod_{\sigma = j+1}^{r}
\underbrace{			\prod_{s \in \widetilde{			\mu				}}				}_{
	\substack{
s \in \operatorname{Row}(T(i_1,\dots,i_k ; \lambda)|N+1),
		\\
s \text{ is in the row interval $\gamma_j$},
		\\
\text{col}(s) = \mu_{m_1 + \cdots + m_{\gamma_\sigma -1} + 1}
	}
}
\frac{	1 - q^{a_{\widetilde{			\mu				}}(s) + 2}  t^{\ell_{\widetilde{			\mu				}}(s) + c_\sigma		}			}{
	1 -q^{a_{\widetilde{			\mu				}}(s) + 1}t^{\ell_{\widetilde{			\mu				}}(s) + c_\sigma + 1} 
}
\frac{	1 -   q^{a_{\widetilde{\mu}}(s) + 1}	t^{\ell_{\widetilde{\mu}}(s) + 1}	}{
	1 - 	q^{a_{\widetilde{			\mu				}}(s) + 2} t^{\ell_{\widetilde{		\mu		}}(s)}	
}. 
\notag 
\end{align}
From equations \eqref{c2-1522-18mar}
\eqref{c3-1522-18mar}
\eqref{c4-1522-18mar}, we immediately obtain equation \eqref{c1-1522-18mar}. 
\end{proof}

\begin{lem}\mbox{}
\label{lemc2-1621-18mar}
\begin{align*}
&\prod_{s \in \mu}
\frac{	q^{a_{\mu}(s) + 1} - t^{-\ell_{\mu}(s)}			}{
	t^{-(\ell_{\mu}(s) + 1)} - q^{a_{\mu}(s)}
}
\times 
\left(		
\prod_{s \in \widetilde{\mu}		}
\frac{	q^{a_{\widetilde{			\mu				}}(s) + 1} - t^{-\ell_{\widetilde{		\mu		}}(s)}			}{
	t^{-(\ell_{\widetilde{\mu}}(s) + 1)} - q^{a_{\widetilde{\mu}}(s)}
}		
\right)^{-1}
\notag
\\
&= 
\prod_{j = 1}^{r - 1}
\prod_{\sigma = j+1}^{r}
\underbrace{			\prod_{s \in \widetilde{			\mu				}}				}_{
	\substack{
s \in \operatorname{Row}(T(i_1,\dots,i_k ; \lambda)|N+1),
		\\
s \text{ is in the row interval $\gamma_j$},
		\\
\text{col}(s) = \mu_{m_1 + \cdots + m_{\gamma_\sigma -1} + 1}
	}
}
\frac{	q^{a_{\widetilde{			\mu				}}(s) + 2} - t^{-(	\ell_{\widetilde{			\mu				}}(s) + c_\sigma	)	}			}{
	t^{-(\ell_{\widetilde{			\mu				}}(s) + c_\sigma + 1)} - q^{a_{\widetilde{			\mu				}}(s) + 1}
}
\frac{	t^{-(\ell_{\widetilde{\mu}}(s) + 1)} - q^{a_{\widetilde{\mu}}(s)}		}{
	q^{a_{\widetilde{			\mu				}}(s) + 1} - t^{-\ell_{\widetilde{		\mu		}}(s)}	
}
\notag 
\\
&\times 
\prod_{j = 1}^{r - 1}
\underbrace{			\prod_{s \in \widetilde{			\mu				}}				}_{
	\substack{
s \in \operatorname{Row}(T(i_1,\dots,i_k ; \lambda)|N+1),
		\\
s \text{ is in the row interval $\gamma_j$},
		\\
\text{col}(s) \notin \{		\mu_{m_1 + \cdots + m_{\gamma_\sigma -1} + 1} ~|~ \sigma = j + 1, \dots, r			\}
	}
}
\frac{	q^{a_{\widetilde{			\mu				}}(s) + 2} - t^{-\ell_{\widetilde{			\mu				}}(s)}			}{
	t^{-(\ell_{\widetilde{			\mu				}}(s) + 1)} - q^{a_{\widetilde{			\mu				}}(s) + 1}
}
\frac{	t^{-(\ell_{\widetilde{\mu}}(s) + 1)} - q^{a_{\widetilde{\mu}}(s)}		}{
	q^{a_{\widetilde{			\mu				}}(s) + 1} - t^{-\ell_{\widetilde{		\mu		}}(s)}	
}
\notag 
\\
&\times 
\prod_{j = 1}^{r}
\prod_{\sigma = j + 1}^{r}
\underbrace{			\prod_{s \in \widetilde{			\mu				}			}					}_{
	\substack{
s \notin \operatorname{Row}(T(i_1,\dots,i_k ; \lambda)|N+1),
		\\
s \text{ is in the row interval $\gamma_j$},
		\\
\text{col}(s) = \mu_{m_1 + \cdots + m_{\gamma_\sigma -1} + 1}
	}
}
\frac{	
	q^{	 	\mu_{m_1 + \cdots + m_{\gamma_j -1} + 1}	- \mu_{m_1 + \cdots + m_{\gamma_\sigma -1} + 1}		+ 1} - t^{-		(\ell_{\widetilde{			\mu				}}(s)		+ c_\sigma	)					}	 	
	}{
	t^{-(\ell_{\widetilde{			\mu				}}(s)		+ c_\sigma			 + 1)} - q^{	\mu_{m_1 + \cdots + m_{\gamma_j -1} + 1}	- \mu_{m_1 + \cdots + m_{\gamma_\sigma -1} + 1}						}
}
\notag 
\\
&\times 
\prod_{j = 1}^{r}
\prod_{\sigma = j + 1}^{r}
\underbrace{			\prod_{s \in \widetilde{			\mu				}			}					}_{
	\substack{
s \notin \operatorname{Row}(T(i_1,\dots,i_k ; \lambda)|N+1),
		\\
s \text{ is in the row interval $\gamma_j$},
		\\
\text{col}(s) = \mu_{m_1 + \cdots + m_{\gamma_\sigma -1} + 1}
	}
}
\frac{	t^{-(\ell_{\widetilde{\mu}}(s) + 1)} - q^{				\mu_{m_1 + \cdots + m_{\gamma_j -1} + 1}	- \mu_{m_1 + \cdots + m_{\gamma_\sigma -1} + 1}							}	}{
	q^{\mu_{m_1 + \cdots + m_{\gamma_j -1} + 1}	- \mu_{m_1 + \cdots + m_{\gamma_\sigma -1} + 1}		+ 1} - t^{-\ell_{\widetilde{		\mu		}}(s)}		
}
\\
&\times 
\prod_{\sigma = 1}^{r}
\underbrace{		\prod_{s \in \widetilde{			\mu				}}				}_{
	\substack{
s \notin \operatorname{Row}(T(i_1,\dots,i_k ; \lambda)|N+1),
		\\
s \text{ is not in the row interval $\gamma_1, \dots, \gamma_r$},
		\\
\operatorname{row}(s) \in \{1,\dots,m_1 + \cdots + m_{\gamma_1 - 1}		\},
		\\
\operatorname{col}(s) = \mu_{m_1 + \cdots + m_{\gamma_\sigma -1} + 1}
	}
}
\frac{	q^{a_{\widetilde{			\mu				}}(s) + 1} - t^{-(\ell_{\widetilde{			\mu				}}(s)	+ c_\sigma)	}			}{
	t^{-(\ell_{\widetilde{			\mu				}}(s) + c_\sigma + 1)} - q^{a_{\widetilde{			\mu				}}(s)}
}
\frac{			t^{-(\ell_{\widetilde{\mu}}(s) + 1)} - q^{a_{\widetilde{\mu}}(s)}	}{
	q^{a_{\widetilde{			\mu				}}(s) + 1} - t^{-\ell_{\widetilde{		\mu		}}(s)}		
}	
\\
&\times 
\prod_{j = 1}^{r - 1}
\prod_{\sigma = j+1}^{r}
\underbrace{		\prod_{s \in \widetilde{			\mu				}}				}_{
	\substack{
s \notin \operatorname{Row}(T(i_1,\dots,i_k ; \lambda)|N+1),
		\\
s \text{ is not in the row interval $\gamma_1, \dots, \gamma_r$},
		\\
\operatorname{row}(s) \in \{m_1 + \cdots + m_{\gamma_j} + 1, \cdots, m_1 + \cdots + m_{\gamma_{j+1} - 1}					\},
		\\
\operatorname{col}(s) = \mu_{m_1 + \cdots + m_{\gamma_\sigma -1} + 1}
	}
}
\frac{	q^{a_{\widetilde{			\mu				}}(s) + 1} - t^{-(\ell_{\widetilde{			\mu				}}(s) + c_\sigma)}			}{
	t^{-(\ell_{\widetilde{			\mu				}}(s) + c_\sigma + 1)} - q^{a_{\widetilde{			\mu				}}(s)}
}
\frac{			t^{-(\ell_{\widetilde{\mu}}(s) + 1)} - q^{a_{\widetilde{\mu}}(s)}	}{
	q^{a_{\widetilde{			\mu				}}(s) + 1} - t^{-\ell_{\widetilde{		\mu		}}(s)}		
}	
\notag 
\\
&\times
\prod_{\beta = m_{\gamma_1} - c_1}^{m_{\gamma_1} - 1}
\frac{	q - t^{-\beta}			}{
	t^{-(\beta + 1)} - 1
}
\times
\cdots
\times
\prod_{\beta = m_{\gamma_r} - c_r}^{	m_{\gamma_r}	 - 1}
\frac{	q - t^{-\beta}			}{
	t^{-(\beta + 1)} - 1
}
\\
&\times 
\underbrace{			\prod_{s \in \widetilde{			\mu				}}				}_{
	\substack{
s \in \operatorname{Row}(T(i_1,\dots,i_k ; \lambda)|N+1),
		\\
s 
		\text{ is in the row interval $\gamma_r$}
	}
}
\frac{	q^{a_{\widetilde{			\mu				}}(s) + 2} - t^{-\ell_{\widetilde{			\mu				}}(s)}			}{
	t^{-(\ell_{\widetilde{			\mu				}}(s) + 1)} - q^{a_{\widetilde{			\mu				}}(s) + 1}
}
\frac{	t^{-(\ell_{\widetilde{\mu}}(s) + 1)} - q^{a_{\widetilde{\mu}}(s)}		}{
	q^{a_{\widetilde{			\mu				}}(s) + 1} - t^{-\ell_{\widetilde{		\mu		}}(s)}	
}
\notag 
\end{align*}
\end{lem}
\begin{proof}
Straightforward calculation. 
\end{proof}

\begin{lem}\mbox{}
\label{lemc3-1621-18mar}
\begin{align}
&\psi_{T^{\prime}(i_1,\dots,i_k;\lambda)^{(N+1)} 	/	T^{\prime}(i_1,\dots,i_k;\lambda)^{(N+2) }		 }(t,q)
\notag 
\\
&\times 
\prod_{j = 1}^{r}
\prod_{\sigma = j + 1}^{r}
\underbrace{			\prod_{s \in \widetilde{			\mu				}			}					}_{
	\substack{
s \notin \operatorname{Row}(T(i_1,\dots,i_k ; \lambda)|N+1),
		\\
s \text{ is in the row interval $\gamma_j$},
		\\
\text{col}(s) = \mu_{m_1 + \cdots + m_{\gamma_\sigma -1} + 1}
	}
}
\frac{	
	q^{	 	\mu_{m_1 + \cdots + m_{\gamma_j -1} + 1}	- \mu_{m_1 + \cdots + m_{\gamma_\sigma -1} + 1}		+ 1} - t^{-		(\ell_{\widetilde{			\mu				}}(s)		+ c_\sigma	)					}	 	
}{
	t^{-(\ell_{\widetilde{			\mu				}}(s)		+ c_\sigma			 + 1)} - q^{	\mu_{m_1 + \cdots + m_{\gamma_j -1} + 1}	- \mu_{m_1 + \cdots + m_{\gamma_\sigma -1} + 1}						}
}
\label{eqn-c5-1250}
\\
&\times 
\prod_{j = 1}^{r}
\prod_{\sigma = j + 1}^{r}
\underbrace{			\prod_{s \in \widetilde{			\mu				}			}					}_{
	\substack{
s \notin \operatorname{Row}(T(i_1,\dots,i_k ; \lambda)|N+1),
		\\
s 
		\text{ is in the row interval $\gamma_j$},
		\\
\text{col}(s) = \mu_{m_1 + \cdots + m_{\gamma_\sigma -1} + 1}
	}
}
\frac{	t^{-(\ell_{\widetilde{\mu}}(s) + 1)} - q^{				\mu_{m_1 + \cdots + m_{\gamma_j -1} + 1}	- \mu_{m_1 + \cdots + m_{\gamma_\sigma -1} + 1}							}	}{
	q^{\mu_{m_1 + \cdots + m_{\gamma_j -1} + 1}	- \mu_{m_1 + \cdots + m_{\gamma_\sigma -1} + 1}		+ 1} - t^{-\ell_{\widetilde{		\mu		}}(s)}		
}
\label{eqn-c6-1250}
\\
&\times 
\prod_{\sigma = 1}^{r}
\underbrace{		\prod_{s \in \widetilde{			\mu				}}				}_{
	\substack{
s \notin \operatorname{Row}(T(i_1,\dots,i_k ; \lambda)|N+1),
		\\
s \text{ is not in the row interval $\gamma_1, \dots, \gamma_r$},
		\\
\operatorname{row}(s) \in \{1,\dots,m_1 + \cdots + m_{\gamma_1 - 1}		\},
		\\
\operatorname{col}(s) = \mu_{m_1 + \cdots + m_{\gamma_\sigma -1} + 1}
	}
}
\frac{	q^{a_{\widetilde{			\mu				}}(s) + 1} - t^{-(\ell_{\widetilde{			\mu				}}(s)	+ c_\sigma)	}			}{
	t^{-(\ell_{\widetilde{			\mu				}}(s) + c_\sigma + 1)} - q^{a_{\widetilde{			\mu				}}(s)}
}
\frac{			t^{-(\ell_{\widetilde{\mu}}(s) + 1)} - q^{a_{\widetilde{\mu}}(s)}	}{
	q^{a_{\widetilde{			\mu				}}(s) + 1} - t^{-\ell_{\widetilde{		\mu		}}(s)}		
}	
\label{eqn-c7-1250}
\\
&\times 
\prod_{j = 1}^{r - 1}
\prod_{\sigma = j+1}^{r}
\underbrace{		\prod_{s \in \widetilde{			\mu				}}				}_{
	\substack{
s \notin \operatorname{Row}(T(i_1,\dots,i_k ; \lambda)|N+1),
		\\
s \text{ is not in the row interval $\gamma_1, \dots, \gamma_r$},
		\\
\operatorname{row}(s) \in \{m_1 + \cdots + m_{\gamma_j} + 1, \cdots, m_1 + \cdots + m_{\gamma_{j+1} - 1}					\},
		\\
\operatorname{col}(s) = \mu_{m_1 + \cdots + m_{\gamma_\sigma -1} + 1}
	}
}
\frac{	q^{a_{\widetilde{			\mu				}}(s) + 1} - t^{-(\ell_{\widetilde{			\mu				}}(s) + c_\sigma)}			}{
	t^{-(\ell_{\widetilde{			\mu				}}(s) + c_\sigma + 1)} - q^{a_{\widetilde{			\mu				}}(s)}
}
\frac{			t^{-(\ell_{\widetilde{\mu}}(s) + 1)} - q^{a_{\widetilde{\mu}}(s)}	}{
	q^{a_{\widetilde{			\mu				}}(s) + 1} - t^{-\ell_{\widetilde{		\mu		}}(s)}		
}	
\label{eqn-c8-1250}
\\
&\times
\prod_{\beta = c_1}^{m_{\gamma_1} - 1}
\frac{	q - t^{-\beta}			}{
	t^{-(\beta + 1)} - 1
}
\times
\cdots
\times
\prod_{\beta = c_r}^{	m_{\gamma_r}	 - 1}
\frac{	q - t^{-\beta}			}{
	t^{-(\beta + 1)} - 1
}
\times
\prod_{\beta = 0}^{m_{\gamma_1} - c_1 - 1}
\frac{	t^{-(\beta + 1)} - 1			}{
	q - t^{-\beta}
}
\times
\cdots
\times
\prod_{\beta = 0}^{m_{\gamma_r} - c_r - 1}
\frac{	t^{-(\beta + 1)} - 1			}{
	q - t^{-\beta}
}
\label{eqn-c9-1250}
\\
&= 1
\notag 
\end{align}
\end{lem}
\begin{proof}
From \textbf{Proposition \ref{prop226-1127-19mar}}, we know that 
\begin{align}
&\psi_{T^{\prime}(i_1,\dots,i_k;\lambda)^{(N+1)} 	/	T^{\prime}(i_1,\dots,i_k;\lambda)^{(N+2) }		 }(t,q)
\notag 
\\
&= 
\psi^\prime_{T(i_1,\dots,i_k;\lambda)^{(N+1)} 	/	T(i_1,\dots,i_k;\lambda)^{(N+2) }		 }(q,t)
=
\psi^\prime_{\mu/\widetilde{\mu}}(q,t) 
\notag 
\\
&= 
\underbrace{	\prod_{	1 \leq	\alpha < \beta \leq \ell(\mu)	}					}_{
	\substack{
\mu_\alpha = \widetilde{			\mu				}_\alpha,
		\\
\mu_\beta = \widetilde{			\mu				}_\beta + 1
	}
}
\frac{
	\left(1 - q^{\widetilde{			\mu				}_\alpha - \widetilde{			\mu				}_\beta}t^{\beta - \alpha - 1}\right)
	\left(1 - q^{\mu_\alpha - \mu_\beta}t^{\beta - \alpha + 1}\right)
}{
	\left(1 - q^{\widetilde{			\mu				}_\alpha - \widetilde{			\mu				}_\beta}t^{\beta - \alpha }\right)
	\left(1 - q^{\mu_\alpha - \mu_\beta}t^{\beta - \alpha }\right)
}
\notag 
\\
&= 
\prod_{j = 1}^{r}
\underbrace{	\prod_{		1 \leq	\alpha < \beta \leq \ell(\mu)	}					}_{
	\substack{
\mu_\alpha = \widetilde{			\mu				}_\alpha,
		\\
\mu_\beta = \widetilde{			\mu				}_\beta + 1,
		\\
\beta \text{ is in the row interval $\gamma_j$},
		\\
\alpha \in \{1,2,\dots, m_1 + \cdots + m_{\gamma_j - 1}\}
	}
}
\frac{
	\left(1 - q^{\widetilde{			\mu				}_\alpha - \widetilde{			\mu				}_\beta}t^{\beta - \alpha - 1}\right)
	\left(1 - q^{\mu_\alpha - \mu_\beta}t^{\beta - \alpha + 1}\right)
}{
	\left(1 - q^{\widetilde{			\mu				}_\alpha - \widetilde{			\mu				}_\beta}t^{\beta - \alpha }\right)
	\left(1 - q^{\mu_\alpha - \mu_\beta}t^{\beta - \alpha }\right)
}
\times
\prod_{j = 1}^{r}
\prod_{\alpha = 1}^{m_{\gamma_j} - c_j}
\frac{
	\left(1 - t^{c_j + \alpha}\right)
	\left(1 - qt^{\alpha - 1}\right)
}{
	\left(1 - qt^{c_j + \alpha - 1}\right)
	\left(1 - t^\alpha\right)
}
\notag 
\\
&= 
\prod_{\zeta = 1}^{\gamma_1 - 1}
\underbrace{			\prod_{s \in \widetilde{			\mu				}}				}_{
	\substack{
s \notin \operatorname{Row}(T(i_1,\dots,i_k ; \lambda)|N+1),
		\\
s \text{ is in the row interval $\zeta$},
		\\
\text{col}(s) = \mu_{m_1 + \cdots + m_{\gamma_j -1} + 1}
	}
}
\frac{
	\left(1 - q^{	a_{\widetilde{			\mu				}}(s)	+1}t^{\ell_{	\widetilde{			\mu				}		}(s)}\right)
	\left(1 - q^{	a_{\widetilde{			\mu				}}(s)	}t^{\ell_{	\widetilde{			\mu				}		}(s) + c_j + 1}\right)
}{
	\left(1 - q^{	a_{\widetilde{			\mu				}}(s)	+1}t^{\ell_{	\widetilde{			\mu				}		}(s) + c_j}\right)
	\left(1 - q^{	a_{\widetilde{			\mu				}}(s)	}t^{\ell_{	\widetilde{			\mu				}		}(s)  + 1}\right)
}
\label{eqn-c10-1250}
\\
&\hspace{0.5cm}
\times \prod_{j = 2}^{r}
\underbrace{			\prod_{\zeta = 1}^{\gamma_j - 1}					}_{
\zeta \neq \gamma_1, \dots, \gamma_{j - 1}
}
\underbrace{			\prod_{s \in \widetilde{			\mu				}}				}_{
	\substack{
\,\, s \notin \operatorname{Row}(T(i_1,\dots,i_k ; \lambda)|N+1),
		\\
\,\, s \text{ is in the row interval $\zeta$},
		\\
\text{col}(s) = \mu_{m_1 + \cdots + m_{\gamma_j -1} + 1}
	}
}
\frac{
	\left(1 - q^{	a_{\widetilde{			\mu				}}(s)	+1}t^{\ell_{	\widetilde{			\mu				}		}(s)}\right)
	\left(1 - q^{	a_{\widetilde{			\mu				}}(s)	}t^{\ell_{	\widetilde{			\mu				}		}(s) + c_j + 1}\right)
}{
	\left(1 - q^{	a_{\widetilde{			\mu				}}(s)	+1}t^{\ell_{	\widetilde{			\mu				}		}(s) + c_j}\right)
	\left(1 - q^{	a_{\widetilde{			\mu				}}(s)	}t^{\ell_{	\widetilde{			\mu				}		}(s)  + 1}\right)
}
\label{eqn-c11-1250}
\\
&\hspace{0.5cm}
\times \prod_{j = 2}^{r}
\underbrace{			\prod_{\zeta = 1}^{\gamma_j - 1}					}_{
	\zeta = \gamma_1, \dots, \gamma_{j - 1}
}
\underbrace{			\prod_{s \in \widetilde{			\mu				}}				}_{
	\substack{
\,\, s \notin \operatorname{Row}(T(i_1,\dots,i_k ; \lambda)|N+1),
		\\
\,\, s \text{ is in the row interval $\zeta$},
		\\
		\text{col}(s) = \mu_{m_1 + \cdots + m_{\gamma_j -1} + 1}
	}
}
\frac{
	\left(1 - q^{	a_{\widetilde{			\mu				}}(s)	+1}t^{\ell_{	\widetilde{			\mu				}		}(s)}\right)
	\left(1 - q^{	a_{\widetilde{			\mu				}}(s)	}t^{\ell_{	\widetilde{			\mu				}		}(s) + c_j + 1}\right)
}{
	\left(1 - q^{	a_{\widetilde{			\mu				}}(s)	+1}t^{\ell_{	\widetilde{			\mu				}		}(s) + c_j}\right)
	\left(1 - q^{	a_{\widetilde{			\mu				}}(s)	}t^{\ell_{	\widetilde{			\mu				}		}(s)  + 1}\right)
}  
\label{eqn-c12-1250}
\\
&\hspace{0.5cm}\times
\prod_{j = 1}^{r}
\prod_{\alpha = 1}^{m_{\gamma_j} - c_j}
\frac{
	\left(1 - t^{c_j + \alpha}\right)
	\left(1 - qt^{\alpha - 1}\right)
}{
	\left(1 - qt^{c_j + \alpha - 1}\right)
	\left(1 - t^\alpha\right)
}
\label{eqn-c13-1250}
\end{align}
One can easily check that 
\begin{align*}
\eqref{eqn-c7-1250}
\times
\eqref{eqn-c8-1250}
\times
\eqref{eqn-c10-1250}
\times
\eqref{eqn-c11-1250}
&= 
1,
\\
\eqref{eqn-c5-1250}
\times
\eqref{eqn-c6-1250}
\times
\eqref{eqn-c12-1250}
&= 1,
\\
\eqref{eqn-c9-1250}
\times
\eqref{eqn-c13-1250}
&= 1.
\end{align*}
From this, the formula in \textbf{Lemma \ref{lemc3-1621-18mar}} immediately follows. 
\end{proof}

\begin{lem}\mbox{}
\label{lemc4-1621-18mar}
\begin{align*}
&\psi_{T(\widetilde{i}_1,\dots, \widetilde{i}_k;\lambda)^{(N+1)} / T(\widetilde{i}_1,\dots, \widetilde{i}_k;\lambda)^{(N+2)}}(q,t)
\times
\cdots 
\times
\psi_{T(\widetilde{i}_1,\dots, \widetilde{i}_k ; \lambda)^{(N+c-1)} / T(\widetilde{i}_1,\dots, \widetilde{i}_k;\lambda)^{(N+c)}}(q,t)
\times
\psi_{T(\widetilde{i}_1,\dots, \widetilde{i}_k;\lambda)^{(N+c)} / 	\widetilde{\mu}			}(q,t)
\\
&=
\prod_{j = 1}^{r - 1}
\prod_{\sigma = j+1}^{r}
\underbrace{			\prod_{s \in \widetilde{			\mu				}}				}_{
	\substack{
		s \in \operatorname{Row}(T(i_1,\dots,i_k ; \lambda)|N+1),
		\\
		s 
		\text{ is in the row interval $\gamma_j$},
		\\
		\text{col}(s) = \mu_{m_1 + \cdots + m_{\gamma_\sigma -1} + 1}
	}
}
\frac{	
	1 - 	q^{a_{\widetilde{			\mu				}}(s) + 2} t^{\ell_{\widetilde{		\mu		}}(s)}	
}{
	1 -   q^{a_{\widetilde{\mu}}(s) + 1}	t^{\ell_{\widetilde{\mu}}(s) + 1}	
}
\frac{	1 - q^{a_{\widetilde{\mu}}(s)}	t^{\ell_{\widetilde{\mu}}(s) + 1} 	}{
	1 - q^{a_{\widetilde{			\mu				}}(s) + 1} t^{\ell_{\widetilde{		\mu		}}(s)}	
}
\notag 
\\
&\times 
\prod_{j = 1}^{r - 1}
\underbrace{			\prod_{s \in \widetilde{			\mu				}}				}_{
	\substack{
		s \in \operatorname{Row}(T(i_1,\dots,i_k ; \lambda)|N+1),
		\\
		s 
		\text{ is in the row interval $\gamma_j$},
		\\
		\text{col}(s) \notin \{		\mu_{m_1 + \cdots + m_{\gamma_\sigma -1} + 1} ~|~ \sigma = j + 1, \dots, r			\}
	}
}
\frac{	
	1 - 	q^{a_{\widetilde{			\mu				}}(s) + 2} t^{\ell_{\widetilde{		\mu		}}(s)}	
}{
	1 -   q^{a_{\widetilde{\mu}}(s) + 1}	t^{\ell_{\widetilde{\mu}}(s) + 1}	
}
\frac{	1 - q^{a_{\widetilde{\mu}}(s)}	t^{\ell_{\widetilde{\mu}}(s) + 1} 	}{
	1 - q^{a_{\widetilde{			\mu				}}(s) + 1} t^{\ell_{\widetilde{		\mu		}}(s)}	
}
\notag 
\\
&\times 
\underbrace{			\prod_{s \in \widetilde{			\mu				}}				}_{
	\substack{
	s \in \operatorname{Row}(T(i_1,\dots,i_k ; \lambda)|N+1),
		\\
	s 
		\text{ is in the row interval $\gamma_r$}
	}
}
\frac{	
	1 - 	q^{a_{\widetilde{			\mu				}}(s) + 2} t^{\ell_{\widetilde{		\mu		}}(s)}	
}{
	1 -   q^{a_{\widetilde{\mu}}(s) + 1}	t^{\ell_{\widetilde{\mu}}(s) + 1}	
}
\frac{	1 - q^{a_{\widetilde{\mu}}(s)}	t^{\ell_{\widetilde{\mu}}(s) + 1} 	}{
	1 - q^{a_{\widetilde{			\mu				}}(s) + 1} t^{\ell_{\widetilde{		\mu		}}(s)}	
}
\notag 
\end{align*}
\end{lem}
\begin{proof}
We know from equation \eqref{eqn27-2052} that 
\begin{align}
&\psi_{T(\widetilde{i}_1,\dots, \widetilde{i}_k;\lambda)^{(N+c)} / 	\widetilde{\mu}			}(q,t)
= 
\prod_{1 \leq \alpha \leq \beta \leq \ell(\widetilde{\mu}		)	}					
\frac{
	f_{q,t}\left(			q^{	\widetilde{			\mu				}_\alpha - \widetilde{			\mu				}_\beta		}t^{\beta - \alpha}				\right)
	f_{q,t}\left(			q^{	T(\widetilde{i}_1,\dots, \widetilde{i}_k;\lambda)^{(N+c)}_\alpha - T(\widetilde{i}_1,\dots, \widetilde{i}_k;\lambda)^{(N+c)}_{\beta+1}				}t^{\beta - \alpha}				\right)
}{
	f_{q,t}\left(			q^{		T(\widetilde{i}_1,\dots, \widetilde{i}_k;\lambda)^{(N+c)}_\alpha 		 - \widetilde{			\mu				}_\beta		}t^{\beta - \alpha}				\right)
	f_{q,t}\left(			q^{			\widetilde{			\mu				}_\alpha 		- T(\widetilde{i}_1,\dots, \widetilde{i}_k;\lambda)^{(N+c)}_{\beta+1}				}t^{\beta - \alpha}				\right)
}
\notag 
\\
&= 
\underbrace{		\prod_{	1 \leq \alpha \leq \beta \leq \ell(\widetilde{\mu}		)			}					}_{
	\substack{
	\alpha = m_1 + \cdots + m_{\gamma_1} - c_1 + 1, 
		\\
	\widetilde{			\mu				}_{\beta} \neq \widetilde{			\mu				}_{\beta+1}
	}
}
\frac{
	f\left(			q^{	\widetilde{			\mu				}_\alpha - \widetilde{			\mu				}_\beta		}t^{\beta - \alpha}				\right)
	f\left(			q^{	\widetilde{			\mu				}_\alpha - \widetilde{			\mu				}_{\beta+1} + 1				}t^{\beta - \alpha}				\right)
}{
	f\left(			q^{		\widetilde{			\mu				}_\alpha - \widetilde{			\mu				}_\beta	 + 1	}t^{\beta - \alpha}				\right)
	f\left(			q^{		\widetilde{			\mu				}_\alpha - \widetilde{			\mu				}_{\beta+1}		}t^{\beta - \alpha}				\right)
}
\notag 
\\
&=
\underbrace{		\prod_{	1 \leq \alpha \leq \beta \leq \ell(\widetilde{\mu}		)			}					}_{
	\substack{
		\alpha = m_1 + \cdots + m_{\gamma_1} - c_1 + 1, 
		\\
		\widetilde{			\mu				}_{\beta} \neq \widetilde{			\mu				}_{\beta+1}
	}
}
\frac{
	\left(1 - q^{	\widetilde{			\mu				}_\alpha - \widetilde{			\mu				}_\beta		}t^{\beta - \alpha + 1}		 \right)
	\left(1 - q^{		\widetilde{			\mu				}_\alpha - \widetilde{			\mu				}_{\beta+1}	+ 1	}t^{\beta - \alpha}	 \right)
}{
	\left(1 - q^{		\widetilde{			\mu				}_\alpha - \widetilde{			\mu				}_{\beta+1}		}t^{\beta - \alpha + 1}	 \right)
	\left(1 - q^{	\widetilde{			\mu				}_\alpha - \widetilde{			\mu				}_\beta	 + 1	}t^{\beta - \alpha}	 \right)
}
\notag 
\\
&= 
\underbrace{			\prod_{s \in \widetilde{			\mu				} }				}_{
	\substack{
	\operatorname{row}(s) \in \operatorname{Row}(	\widetilde{			\mu				}	\rightarrow		
	T(\widetilde{i}_1,\dots, \widetilde{i}_k;\lambda)^{(N+c)}
	)
	}
}
\frac{	1 -   q^{a_{\widetilde{\mu}}(s) + 2}	t^{\ell_{\widetilde{\mu}}(s) }	}{
	1 - 	q^{a_{\widetilde{			\mu				}}(s) + 1} t^{\ell_{\widetilde{		\mu		}}(s) + 1}	
}
\frac{	1 -   q^{a_{\widetilde{\mu}}(s) }	t^{\ell_{\widetilde{\mu}}(s) + 1}	}{
	1 - 	q^{a_{\widetilde{			\mu				}}(s) + 1} t^{\ell_{\widetilde{		\mu		}}(s)}	
},
\label{eqn-c14-1622}
\end{align}
where the set $\operatorname{Row}(	\widetilde{			\mu				}	\rightarrow		
T(\widetilde{i}_1,\dots, \widetilde{i}_k;\lambda)^{(N+c)}
)$ is as defined in \textbf{Definition \ref{def220-1540-19mar}}. 

Using a similar argument, we obtain that 
\begin{align}
	&\psi_{T(\widetilde{i}_1,\dots, \widetilde{i}_k ; \lambda)^{(N+c-1)} / T(\widetilde{i}_1,\dots, \widetilde{i}_k;\lambda)^{(N+c)}}(q,t)
	\label{eqn262-1458-23dec}
	\\
	&= 
	\underbrace{			\prod_{s \in T(\widetilde{i}_1,\dots, \widetilde{i}_k;\lambda)^{(N+c)}  }				}_{
		\substack{
			\operatorname{row}(s) \in \operatorname{Row}(T(\widetilde{i}_1,\dots, \widetilde{i}_k;\lambda)^{(N+c)}
			\rightarrow 
			T(\widetilde{i}_1,\dots, \widetilde{i}_k ; \lambda)^{(N+c-1)} 
			)
		}
	}
	\notag 
	\Bigg[
	\\
	&\hspace{0.4cm}
	\frac{	1 -   q^{a_{	T(\widetilde{i}_1,\dots, \widetilde{i}_k;\lambda)^{(N+c)}	 }(s) + 2}	t^{\ell_{	T(\widetilde{i}_1,\dots, \widetilde{i}_k;\lambda)^{(N+c)}		}(s) }	}{
		1 - 	q^{a_{	T(\widetilde{i}_1,\dots, \widetilde{i}_k;\lambda)^{(N+c)}			}(s) + 1} t^{\ell_{	T(\widetilde{i}_1,\dots, \widetilde{i}_k;\lambda)^{(N+c)}		}(s) + 1}	
	}
	\frac{	1 -   q^{a_{	T(\widetilde{i}_1,\dots, \widetilde{i}_k;\lambda)^{(N+c)} }(s) }	t^{\ell_{ T(\widetilde{i}_1,\dots, \widetilde{i}_k;\lambda)^{(N+c)} }(s) + 1}	}{
		1 - 	q^{a_{	T(\widetilde{i}_1,\dots, \widetilde{i}_k;\lambda)^{(N+c)}	}(s) + 1} t^{\ell_{T(\widetilde{i}_1,\dots, \widetilde{i}_k;\lambda)^{(N+c)}}(s)}	
	}
	\Bigg]. 
	\notag 
\end{align}
However, it is clear that for box $s$ such that $\operatorname{row}(s) \in \operatorname{Row}(T(\widetilde{i}_1,\dots, \widetilde{i}_k;\lambda)^{(N+c)}
\rightarrow 
T(\widetilde{i}_1,\dots, \widetilde{i}_k ; \lambda)^{(N+c-1)} 
)$, we have 
\begin{align}
a_{T(\widetilde{i}_1,\dots, \widetilde{i}_k;\lambda)^{(N+c)}}(s) = a_{\widetilde{			\mu				}}(s)
\hspace{0.5cm}
\text{ and }
\hspace{0.5cm}
\ell_{T(\widetilde{i}_1,\dots, \widetilde{i}_k;\lambda)^{(N+c)}}(s) = \ell_{\widetilde{			\mu				}}(s). 
\end{align} 
Therefore, we get that 
\begin{align}
&\psi_{T(\widetilde{i}_1,\dots, \widetilde{i}_k ; \lambda)^{(N+c-1)} / T(\widetilde{i}_1,\dots, \widetilde{i}_k;\lambda)^{(N+c)}}(q,t)
\label{eqn-c17-1622}
\\
&= 
\underbrace{			\prod_{s \in \widetilde{			\mu				} }				}_{
	\substack{
	\operatorname{row}(s) \in \operatorname{Row}(T(\widetilde{i}_1,\dots, \widetilde{i}_k;\lambda)^{(N+c)}
	\rightarrow 
	T(\widetilde{i}_1,\dots, \widetilde{i}_k ; \lambda)^{(N+c-1)} 
	)
	}
}
\frac{	1 -   q^{a_{\widetilde{\mu}}(s) + 2}	t^{\ell_{\widetilde{\mu}}(s) }	}{
	1 - 	q^{a_{\widetilde{			\mu				}}(s) + 1} t^{\ell_{\widetilde{		\mu		}}(s) + 1}	
}
\frac{	1 -   q^{a_{\widetilde{\mu}}(s) }	t^{\ell_{\widetilde{\mu}}(s) + 1}	}{
	1 - 	q^{a_{\widetilde{			\mu				}}(s) + 1} t^{\ell_{\widetilde{		\mu		}}(s)}	
}.
\notag
\end{align}

By using the similar arguments as in \eqref{eqn-c14-1622} \eqref{eqn-c17-1622}, we can show that 
\begin{align}
&\psi_{	T(\widetilde{i}_1,\dots, \widetilde{i}_k ; \lambda)^{(N+c_1 + \cdots + c_{r-1} + 1)} 	/	T(\widetilde{i}_1,\dots, \widetilde{i}_k ; \lambda)^{(N+ c_1 + \cdots + c_{r-1} + 2)} 	}(q,t)
\times
\cdots 
\times
\psi_{T(\widetilde{i}_1,\dots, \widetilde{i}_k ; \lambda)^{(N+c)} / 	\widetilde{\mu}			}(q,t)
\\
&= 
\underbrace{			\prod_{s \in \widetilde{			\mu				} }				}_{
	\substack{
	s \in \operatorname{Row}(T(i_1,\dots,i_k ; \lambda)|N+1), 
	\\
	s 
	\text{ is in the row interval $\gamma_1$}
	}
}
\frac{	1 -   q^{a_{\widetilde{\mu}}(s) + 2}	t^{\ell_{\widetilde{\mu}}(s) }	}{
	1 - 	q^{a_{\widetilde{			\mu				}}(s) + 1} t^{\ell_{\widetilde{		\mu		}}(s) + 1}	
}
\frac{	1 -   q^{a_{\widetilde{\mu}}(s) }	t^{\ell_{\widetilde{\mu}}(s) + 1}	}{
	1 - 	q^{a_{\widetilde{			\mu				}}(s) + 1} t^{\ell_{\widetilde{		\mu		}}(s)}	
}. 
\notag 
\end{align}
Thus, 
\begin{align}
&\psi_{T(\widetilde{i}_1,\dots, \widetilde{i}_k;\lambda)^{(N+1)} / T(\widetilde{i}_1,\dots, \widetilde{i}_k;\lambda)^{(N+2)}}(q,t)
\times
\cdots 
\times
\psi_{T(\widetilde{i}_1,\dots, \widetilde{i}_k;\lambda)^{(N+c-1)} / T(\widetilde{i}_1,\dots, \widetilde{i}_k;\lambda)^{(N+c)}}(q,t)
\times 
\psi_{T(\widetilde{i}_1,\dots, \widetilde{i}_k;\lambda)^{(N+c)} / 	\widetilde{\mu}			}(q,t)
\notag 
\\
&=
\left(		
\psi_{T(\widetilde{i}_1,\dots, \widetilde{i}_k ; \lambda)^{(N+1)} / T(\widetilde{i}_1,\dots, \widetilde{i}_k;\lambda)^{(N+2)}}(q,t)			
\times
\cdots	
\times
\psi_{T(\widetilde{i}_1,\dots, \widetilde{i}_k ; \lambda)^{(N+c_1)} / T(\widetilde{i}_1,\dots, \widetilde{i}_k ; \lambda)^{(N+c_1 + 1)}}(q,t)			
\right)
\notag
\\
&\hspace{0.4cm}\times
\left(
\psi_{T(\widetilde{i}_1,\dots, \widetilde{i}_k ; \lambda)^{(N+c_1 + 1)} / T(\widetilde{i}_1,\dots, \widetilde{i}_k ; \lambda)^{(N+c_1 + 2)}}(q,t)			
\times
\cdots
\times
\psi_{T(\widetilde{i}_1,\dots, \widetilde{i}_k ; \lambda)^{(N+c_1 + c_2)} / T(\widetilde{i}_1,\dots, \widetilde{i}_k ; \lambda)^{(N+c_1 + c_2 + 1)}}(q,t)
\right)
\notag
\\
&\hspace{0.4cm}\times \cdots \times 
\left(
\psi_{T(\widetilde{i}_1,\dots, \widetilde{i}_k ; \lambda)^{(N+c_1 + \cdots + c_{r-1} + 1)} / T(\widetilde{i}_1,\dots, \widetilde{i}_k ; \lambda)^{(N+c_1 + \cdots + c_{r-1} + 2)}}(q,t)
\times
\cdots
\times 
\psi_{T(\widetilde{i}_1,\dots, \widetilde{i}_k ; \lambda)^{(N+c_1 + \cdots + c_r)} / 	\widetilde{\mu}			}(q,t)
\right)
\notag
\\
&= 
\prod_{j = 1}^{r}
\underbrace{			\prod_{s \in \widetilde{			\mu				} }				}_{
	\substack{
		s \in \operatorname{Row}(T(i_1,\dots,i_k ; \lambda)|N+1),
		\\
		s 
		\text{ is in the row interval $\gamma_j$}
	}
}
\frac{	1 -   q^{a_{\widetilde{\mu}}(s) + 2}	t^{\ell_{\widetilde{\mu}}(s) }	}{
	1 - 	q^{a_{\widetilde{			\mu				}}(s) + 1} t^{\ell_{\widetilde{		\mu		}}(s) + 1}	
}
\frac{	1 -   q^{a_{\widetilde{\mu}}(s) }	t^{\ell_{\widetilde{\mu}}(s) + 1}	}{
	1 - 	q^{a_{\widetilde{			\mu				}}(s) + 1} t^{\ell_{\widetilde{		\mu		}}(s)}	
}. 
\end{align}
It is obvious that 
\begin{align}
&\prod_{j = 1}^{r}
\underbrace{			\prod_{s \in \widetilde{			\mu				} }				}_{
	\substack{
		s \in \operatorname{Row}(T(i_1,\dots,i_k ; \lambda)|N+1), 
		\\
		s 
		\text{ is in the row interval $\gamma_j$}
	}
}
\frac{	1 -   q^{a_{\widetilde{\mu}}(s) + 2}	t^{\ell_{\widetilde{\mu}}(s) }	}{
	1 - 	q^{a_{\widetilde{			\mu				}}(s) + 1} t^{\ell_{\widetilde{		\mu		}}(s) + 1}	
}
\frac{	1 -   q^{a_{\widetilde{\mu}}(s) }	t^{\ell_{\widetilde{\mu}}(s) + 1}	}{
	1 - 	q^{a_{\widetilde{			\mu				}}(s) + 1} t^{\ell_{\widetilde{		\mu		}}(s)}	
}
\\
&=
\prod_{j = 1}^{r - 1}
\prod_{\sigma = j+1}^{r}
\underbrace{			\prod_{s \in \widetilde{			\mu				}}				}_{
	\substack{
		s \in \operatorname{Row}(T(i_1,\dots,i_k ; \lambda)|N+1),
		\\
		s 
		\text{ is in the row interval $\gamma_j$},
		\\
		\text{col}(s) = \mu_{m_1 + \cdots + m_{\gamma_\sigma -1} + 1}
	}
}
\frac{	
	1 - 	q^{a_{\widetilde{			\mu				}}(s) + 2} t^{\ell_{\widetilde{		\mu		}}(s)}	
}{
	1 -   q^{a_{\widetilde{\mu}}(s) + 1}	t^{\ell_{\widetilde{\mu}}(s) + 1}	
}
\frac{	1 - q^{a_{\widetilde{\mu}}(s)}	t^{\ell_{\widetilde{\mu}}(s) + 1} 	}{
	1 - q^{a_{\widetilde{			\mu				}}(s) + 1} t^{\ell_{\widetilde{		\mu		}}(s)}	
}
\notag 
\\
&\times 
\prod_{j = 1}^{r - 1}
\underbrace{			\prod_{s \in \widetilde{			\mu				}}				}_{
	\substack{
		s \in \operatorname{Row}(T(i_1,\dots,i_k ; \lambda)|N+1),
		\\
		s 
		\text{ is in the row interval $\gamma_j$},
		\\
		\text{col}(s) \notin \{		\mu_{m_1 + \cdots + m_{\gamma_\sigma -1} + 1} ~|~ \sigma = j + 1, \dots, r			\}
	}
}
\frac{	1 -   q^{a_{\widetilde{\mu}}(s) + 2}	t^{\ell_{\widetilde{\mu}}(s) }	}{
	1 - 	q^{a_{\widetilde{			\mu				}}(s) + 1} t^{\ell_{\widetilde{		\mu		}}(s) + 1}	
}
\frac{	1 -   q^{a_{\widetilde{\mu}}(s) }	t^{\ell_{\widetilde{\mu}}(s) + 1}	}{
	1 - 	q^{a_{\widetilde{			\mu				}}(s) + 1} t^{\ell_{\widetilde{		\mu		}}(s)}	
}
\notag 
\\
&\times 
\underbrace{			\prod_{s \in \widetilde{			\mu				}}				}_{
	\substack{
	s \in \operatorname{Row}(T(i_1,\dots,i_k ; \lambda)|N+1), 
		\\
		s 
		\text{ is in the row interval $\gamma_r$}
	}
}
\frac{	1 -   q^{a_{\widetilde{\mu}}(s) + 2}	t^{\ell_{\widetilde{\mu}}(s) }	}{
	1 - 	q^{a_{\widetilde{			\mu				}}(s) + 1} t^{\ell_{\widetilde{		\mu		}}(s) + 1}	
}
\frac{	1 -   q^{a_{\widetilde{\mu}}(s) }	t^{\ell_{\widetilde{\mu}}(s) + 1}	}{
	1 - 	q^{a_{\widetilde{			\mu				}}(s) + 1} t^{\ell_{\widetilde{		\mu		}}(s)}	
}. 
\notag  
\end{align}
So we have proved \textbf{Lemma \ref{lemc4-1621-18mar}}. 
\end{proof}

\textbf{Lemma \ref{lem68-1212-13mar}} follows immediately from Lemmas  \textbf{\ref{lemc1-1621-18mar}},  \textbf{\ref{lemc2-1621-18mar}}, and \textbf{\ref{lemc3-1621-18mar}}.

%%%%%%%%%%%%%%%%%%%%%%%%%%%%%%%%%%%%%%%%%%%%%%%%%%%%%%%%%%%%%%%%%%%%%%%%%%
%\newpage


\begin{thebibliography}{99}

\bibitem{AGT}
L. F. Alday, D. Gaiotto, and Y. Tachikawa,
{\it Liouville correlation functions from four-dimensional gauge theories,}
Letters in Mathematical Physics 91, no. 2 (2010): 167-197,
arXiv:0906.3219. 

\bibitem{awata-excited}
H.~Awata, Y.~Matsuo, S.~Odake, and J.~Shiraishi, 
{\it Excited states of the Calogero-Sutherland model and singular vectors of the $W_N$ algebra,}
Nuclear Physics B 449, no. 1-2 (1995): 347-374, 
arXiv:9503043. 

\bibitem{AKOS-950}
H.~Awata, H.~Kubo, S.~Odake, and J.~Shiraishi, 
{\it Quantum $W_N$ algebras and Macdonald polynomials,}
Communications in Mathematical Physics 179(2): 401-416 (1996), 
arXiv:9508011

\bibitem{AY}
H.~Awata, and Y.~Yamada,
{\it Five-dimensional AGT conjecture and the deformed Virasoro algebra,}
Journal of High Energy Physics 2010, no. 1 (2010): 1-11, 
arXiv:0910.4431. 

\bibitem{AY2}
H.~Awata, and Y.~Yamada,
{\it Five-dimensional AGT relation and the deformed $\beta$-ensemble,}
Progress of theoretical physics 124, no. 2 (2010): 227-262, 
arXiv:1004.5122

\bibitem{Awata-note}
H. Awata, B. Feigin, A. Hoshino, M. Kanai, J. Shiraishi, and S. Yanagida,
{\it Notes on Ding-Iohara algebra and AGT conjecture,}
arXiv:1106.4088.

\bibitem{AFS}
H. Awata, B. Feigin, and J. Shiraishi,
{\it Quantum algebraic approach to refined topological vertex,}
Journal of High Energy Physics 2012, no. 3 (2012): 1-35, 
arXiv:1112.6074. 


\bibitem{AKMM17}
H. Awata, H. Kanno, A. Mironov, A. Morozov, A. Morozov, Y. Ohkubo, and Y. Zenkevich, 
{\it Generalized Knizhnik-Zamolodchikov equation for Ding-Iohara-Miki algebra}
Physical Review D 96, no. 2 (2017): 026021, 
arXiv:1703.06084.

\bibitem{AKMM17-2}
H. Awata, H. Kanno, A. Mironov, A. Morozov, K. Suetake, and Y. Zenkevich,
{\it $(q,t)$-KZ equations for quantum toroidal algebra and Nekrasov partition functions on ALE spaces,}
Journal of High Energy Physics 2018, no. 3 (2018): 1-70,
arXiv:1712.08016.

\bibitem{AKMM18}
H. Awata, H. Kanno, A. Mironov, A. Morozov, K. Suetake, and Y.Zenkevich,
{\it The MacMahon $R$-matrix}
Journal of High Energy Physics 2019, no. 4 (2019): 1-34, 
arXiv:1810.07676.

\bibitem{BPZ}
A.A.~Belavin, A.M.~Polyakov, and A.B.~Zamolodchikov
{\it Infinite conformal symmetry in two-dimensional quantum field theory,}
Nuclear Physics B 241, no. 2 (1984): 333-380. 

\bibitem{Misha}
M.~Bershtein, B.~Feigin, and G.~Merzon,
{\it Plane partitions with a “pit”: generating functions and representation theory,}
Selecta Mathematica 24 (2018): 21-62, 
arXiv:1512.08779.

\bibitem{BS1993}
P.~Bouwknegt, and K.~Schoutens,
{\it W symmetry in conformal field theory,}
Physics Reports 223, no. 4 (1993): 183-276,
arXiv:9210010. 	

\bibitem{BFM17}
J-E. Bourgine, M. Fukuda, K. Harada, Y. Matsuo, and R-D. Zhu, 
{\it $(p,q)$-webs of DIM representations, 5d $\mathcal {N}= 1$ instanton partition functions and qq-characters,}
Journal of High Energy Physics 2017, no. 11 (2017): 1-51, 
arXiv:1703.10759. 

\bibitem{BFM17-2}
J-E. Bourgine, M. Fukuda, Y. Matsuo, and R-D. Zhu,
{\it Reflection states in Ding-Iohara-Miki algebra and brane-web for D-type quiver,}
Journal of High Energy Physics 2017, no. 12 (2017): 1-27, 
arXiv:1709.01954. 

\bibitem{BJ}
J-E. Bourgine, and S. Jeong,
{\it New quantum toroidal algebras from 5D $\mathcal {N}$= 1 instantons on orbifolds,}
Journal of High Energy Physics 2020, no. 5 (2020): 1-52,
arXiv:1906.01625	

\bibitem{BS}
I. Burban and O. Schiffmann,
{\it On the Hall algebra of an elliptic curve, I,}
Duke Math. J. \textbf{161} (2012): 1171-1231,
arXiv:0505148.

\bibitem{CK}
P. Cheewaphutthisakun, and H. Kanno,
{\it MacMahon KZ equation for Ding-Iohara-Miki algebra}
Journal of High Energy Physics 2021, no. 4 (2021): 1-47, 
arXiv:2101.01420. 

\bibitem{CK2}
P. Cheewaphutthisakun, and H. Kanno, 
{\it Quasi-Hopf twist and elliptic Nekrasov factor}
Journal of High Energy Physics 2021, no. 12 (2021): 1-45,
arXiv:2110.03970. 

\bibitem{pc2406}
P.~Cheewaphutthisakun, J.~Shiraishi, and K.~Wiboonton,
{\it Elliptic deformation of the Gaiotto-Rapčák corner VOA and the associated partially symmetric polynoimals,}
Journal of High Energy Physics 2024, no. 8 (2024): 1-45, 
arXiv:2406.15860. 

\bibitem{DI} 
J. Ding, K. Iohara,
{\it Generalization of Drinfeld quantum affine algebras,}		
Lett. Math. Phys. {\bf 41} (1997) 181--193,
arXiv:9608002.

\bibitem{DR}
V.~Drinfeld, 
{\it A new realization of Yangians and quantized affine algebras,}
In Soviet Math. Dokl., vol. 32, pp. 212-216. 1988

 



\bibitem{FHH}
B. Feigin, K. Hashizume, A. Hoshino, J. Shiraishi, and S. Yanagida, 
{\it A commutative algebra on degenerate CP1 and Macdonald polynomials,}
Journal of Mathematical Physics 50, no. 9 (2009), 
arXiv:0904.2291

\bibitem{FHSSY}
B.~Feigin, A.~Hoshino, J.~Shibahara, J.~Shiraishi, and S.~Yanagida,
{\it Kernel function and quantum algebras,}
arXiv:1002.2485.



\bibitem{GR17}
D. Gaiotto, and M. Rapčák,
{\it Vertex algebras at the corner}
Journal of High Energy Physics 2019, no. 1 (2019): 1-88,
arXiv:1703.00982. 

\bibitem{HMNW}
K.~Harada, Y.~Matsuo, G.~Noshita, and A.~Watanabe,
{\it q-deformation of corner vertex operator algebras by Miura transformation,}
Journal of High Energy Physics 2021, no. 4 (2021): 1-49, 
arXiv:2101.03953

\bibitem{Jim-original}
M.~Jimbo
{\it A $q$-difference analogue of $U(\fraks{g})$ and the Yang-Baxter equation,}
Letters in Mathematical Physics 10 (1985): 63-69.

\bibitem{macbook-1998}
I.G.~Macdonald,
{\it Symmetric functions and Hall polynomials,}
Oxford university press, 1998.

\bibitem{Ma}
Y. Matsuo, S. Nawata, G. Noshita, and R-D. Zhu, 
{\it Quantum toroidal algebras and solvable structures in gauge/string theory,}
Physics Reports 1055 (2024): 1-144,
arXiv:2309.07596. 

\bibitem{mimachi-yamada-1995-singular}
K.~Mimachi, and Y.~Yamada,
{\it Singular vectors of the Virasoro algebra in terms of Jack symmetric polynomials,}
Communications in mathematical physics 174 (1995): 447-455. 

\bibitem{Miki} 
K. Miki, 
{\it A $(q, \gamma)$ anlog of the $W_{1+\infty}$ algebra,}
J. Math. Phys. {\bf 48} (2007) 123520.

\bibitem{nekrasov}
N. A. Nekrasov,
{\it Seiberg-Witten prepotential from instanton counting,}
Advances in Theoretical and Mathematical Physics 7 (2003): 831-864,
arXiv:0206161. 

\bibitem{noumi-mac}
M.~Noumi
{\it Macdonald Polynomials, Commuting Family of q-Difference Operators and Their Joint Eigenfunctions,}
SpringerBriefs in Mathematical Physics, Springer Singapore (2023).

\bibitem{PR17}
T. Procházka, and M. Rapčák,
{\it Webs of W-algebras,}
Journal of High Energy Physics 2018, no. 11 (2018): 1-89,
arXiv:1711.06888. 

\bibitem{PR18}
T. Procházka, and M. Rapčák,
{\it $\cals{W}$-algebra modules, free fields, and Gukov-Witten defects,}
Journal of High Energy Physics 2019, no. 5 (2019): 1-72,
arXiv:1808.08837. 

\bibitem{Sch}
O.~Schiffmann,
{\it Drinfeld realization of the elliptic Hall algebra,}
J. Alg.  Comb.  \textbf{35}  (2012): 237-262, 
arXiv:1004.2575.

\bibitem{sv2007-supermac}
A.N.~Sergeev, and A.P.~Veselov, 
{\it Deformed Macdonald-Ruijsenaars operators and super Macdonald polynomials,}
Communications in Mathematical Physics 288 (2009): 653-675, 
arXiv:0707.3129. 

\bibitem{SKAO95}
J.~Shiraishi, H.~Kubo, H.~Awata, and S.~Odake,
{\it A quantum deformation of the Virasoro algebra and the Macdonald symmetric functions}
Letters in Mathematical Physics 38 (1996): 33-51,
arXiv:9507034



\bibitem{Taki}
M. Taki, 
{\it On AGT-W conjecture and $q$-deformed $W$-algebra}
arXiv:1403.7016.

\bibitem{Wyl}
N. Wyllard,
{\it $A_{N-1}$ conformal Toda field theory correlation functions from conformal $\cals{N = 2} \,\, SU(N)$ quiver gauge theories,}
Journal of High Energy Physics 2009, no. 11 (2009): 002,
arXiv:0907.2189.  

\bibitem{Zamo}
A.B.~Zamolodchikov,
{\it Infinite additional symmetries in two-dimensional conformal quantum field theory,}
W-Symmetry. World Scientific (1995): 221-229.		

\bibitem{Zen}
Y.~Zenkevich,
{\it On pentagon identity in Ding-Iohara-Miki algebra,}
Journal of High Energy Physics 2023, no. 3 (2023): 1-14, 
arXiv:2112.14687.





		
		




















		
		
		
	\end{thebibliography}
\end{document}